\documentclass{article}

\usepackage{hyperref}  
\usepackage{amssymb, amsbsy}  
\usepackage{mathtools}  
\usepackage{xfrac}  
\usepackage{xspace}  
\usepackage{float}  
\usepackage{graphicx}  
\usepackage{subcaption}  
\usepackage{natbib}  
\usepackage[capitalise]{cleveref}  
\usepackage{appendix}  
\usepackage{enumitem}  
\usepackage{multicol}  
\usepackage[noend]{algpseudocode}  
\usepackage{algorithm}  
\usepackage{acro}  
\usepackage{xr}  
\usepackage{bm} 
\usepackage{tabularx}  
\usepackage{etoolbox}  
\usepackage{booktabs}
\usepackage{multirow}

\hypersetup{colorlinks=true, allcolors=blue}

\makeatletter
 
  \@addtoreset{equation}{section}
 \makeatother
\usepackage{bm}
\usepackage{amsfonts}
\usepackage{amssymb}
\usepackage{mathrsfs}
\usepackage{amsmath,amsthm}
\usepackage{authblk}

\usepackage{comment}
\usepackage{graphicx}
\usepackage{color}
\usepackage{indentfirst}
\usepackage{here}
\usepackage{natbib}

\theoremstyle{plain}  
\newtheorem{theorem}{Theorem}[section]
\newtheorem{lemma}{Lemma}[section]
\newtheorem{prop}{Proposition}[section]

\newtheorem{rem}{Remark}[section]

\newcommand{\trueparam}{\theta^\dagger}
\newcommand{\truebeta}{\beta^\dagger}
\newcommand{\truesigma}{\sigma^\dagger}
\newcommand{\truegamma}{\gamma^\dagger}
\newcommand{\trueE}{\mathbb{E}_{\trueparam}}
\newcommand{\truedist}{\nu_{\trueparam}}
\newcommand{\est}{\hat{\theta}_n}
\newcommand{\estgamma}{\hat{\gamma}_n}
\newcommand{\f}{\frac}
\newcommand{\mrm}{\mathrm} 
  
\newcommand{\limit}{$n \to \infty$, $\Delta_n \to 0$ and $n \Delta_n \to \infty$}

\newcommand*{\filtration}[1]{\mathcal{F}_{t_{#1}}} 
\newcommand*{\sample}[2]{{#1}_{t_{#2}}}

\newcommand*{\contrast}[1]{\ell_{n, \Delta}(#1)}
\newcommand*{\contrastG}[1]{\ell_{n, \Delta}^{(1)}(#1)}
\newcommand*{\contrastW}[1]{\ell_{n, \Delta}^{(2)}(#1)}
\newcommand{\probconv}{\xrightarrow{\mathbb{P}_{\trueparam}}}

\newcommand*{\deriv}[2] {\partial_{#1_{#2}}}
\providecommand{\keywords}[1]{\textit{Keywords:} #1}

\everymath {\displaystyle}

\newcommand{\ruby}[2]{
\leavevmode
\setbox0=\hbox{#1}
\setbox1=\hbox{\tiny #2}
\ifdim\wd0>\wd1 \dimen0=\wd0 \else \dimen0=\wd1 \fi
\hbox{
\kanjiskip=0pt plus 2fil
\xkanjiskip=0pt plus 2fil
\vbox{
\hbox to \dimen0{
\small \hfil#2\hfil}
\nointerlineskip
\hbox to \dimen0{\mathstrut\hfil#1\hfil}}}}
\allowdisplaybreaks[1]

\renewcommand{\theenumi}{H\arabic{enumi}}
\renewcommand{\labelenumi}{(\theenumi)}

\title{Parameter Estimation with Increased Precision for Elliptic and Hypo-elliptic Diffusions}

\author[1]{Yuga Iguchi}
\author[1]{Alexandros Beskos}
\author[1]{Matthew Graham}

\affil[1]{University College London, London, UK}

\begin{document}
\maketitle

\begin{abstract}
This work aims at making a comprehensive contribution in the general area of parametric inference for discretely observed 
diffusion processes. Established approaches for likelihood-based estimation invoke a time-discretisation scheme for the approximation of the intractable transition dynamics of the Stochastic Differential Equation (SDE) model over finite time periods. The scheme is applied for a step-size $\delta>0$, that is either user-selected or determined by the data. Recent research 
has highlighted the critical effect of the choice of numerical scheme on the behaviour of derived parameter estimates in the setting of \emph{hypo-elliptic} SDEs. In brief, in our work, first, we develop two weak second order \emph{sampling schemes} (to cover both hypo-elliptic and elliptic SDEs) and 
produce a \emph{small time expansion} for the density of the schemes to form a proxy for the true intractable SDE transition density. 
Then, we establish a collection of analytic results for likelihood-based parameter estimates obtained via the formed proxies, thus 
providing a theoretical framework that 
showcases advantages from the use of the developed methodology for SDE calibration. 
We present numerical results from carrying out classical or Bayesian inference, for both elliptic and hypo-elliptic SDEs.
\end{abstract}





\section{Introduction} 

Our work is placed within the general framework of parametric inference for diffusion processes. 
Calibration approaches under broad observation regimes, in both Bayesian and classical settings, necessitate the use of a numerical scheme used as proxy for the underlying, typically intractable, Markovian dynamics of the model over finite time steps. Following the latest contributions in the area, we aim to make connections between approximation schemes (i.e., a research area mainly within the remit of stochastic analysis) and their impact on the accuracy and performance of induced likelihood-based inferential approaches (i.e., the field of statistical calibration for Stochastic Differential Equations (SDEs)). 

Let $(\Omega, \mathcal{F}, \{\mathcal{F}_t\}_{t \geq 0}, \mathbb{P})$ be a filtered probability space and $B = (B_{1,t}, \ldots, B_{d_R,t})_{t \ge 0}$ a  standard $d_R$-dimensional Brownian motion defined thereon, $d_R\ge 1$. We use the convention $B_{0,t} = t$. Consider the following general class of SDEs:
\begin{align}
\label{eq:model}
\begin{aligned}
dX_t = \left[ \begin{array}{c} dX_{R,t} \\ dX_{S,t}
\end{array} \right] = \left[ \begin{array}{c}  V_{R,0}(X_t, \beta) \\  V_{S,0}(X_t, \gamma)
\end{array} \right]dt  +  \sum_{1 \le k \le d_R}  \left[ \begin{array}{c}  V_{R,k} (X_t, \sigma) \\ \mathbf{0}_{d_S} \end{array} \right] \,dB_{k,t}, \qquad X_0 = x_0 \in \mathbb{R}^{d}, 
\end{aligned}
\end{align}
with $V_{R,0}:\mathbb{R}^{d}\times \Theta_{\beta}\to \mathbb{R}^{d_R}$, $V_{S,0}:\mathbb{R}^{d}\times \Theta_{\gamma}\to \mathbb{R}^{d_S}$,   $V_{R,k}:\mathbb{R}^{d}\times \Theta_{\sigma}\to \mathbb{R}^{d_R}$, $1\le k\le N$, where $d_S\ge 0$ and $d=d_R+d_S$. 
Also,   $\Theta_\beta\subseteq \mathbb{R}^{d_\beta}$, $\Theta_\gamma\subseteq \mathbb{R}^{d_\gamma}$, $\Theta_\sigma\subseteq \mathbb{R}^{d_\sigma}$, $d_\beta\ge 1$, $d_\gamma\ge 0$, $d_\sigma\ge 1$.  We set $\theta = (\beta, \gamma, \sigma)$ and $\Theta =   \Theta_\beta \times \Theta_\gamma \times \Theta_\sigma$. $X_{R,t}$ and $X_{S,t}$ denote the  rough and smooth components, respectively, of process $X_t$ that solves (\ref{eq:model}). For $d_S>0$, $\{ X_t \}_{t \geq 0}$ is a \emph{hypo-elliptic diffusion} if the law of $X_t$, $t > 0$, admits a density with respect  to (w.r.t.) the Lebesgue measure, while for $d_S=0$ the process is an elliptic one. 
For such a diffusion model, 
we consider the problem of estimating  $\theta \in \Theta$ given observations for state $X_t$ at the discrete-time instances 
$ 0 \leq t_0 < t_1 < \cdots < t_{n-1} < t_n$,  $n \ge 0
$. For simplicity, we assume that observation times are equidistant and set   
$\Delta: = t_{i} - t_{i-1}$,   $1\le i \le  n$.
%
%
%

Non-linear SDEs do not permit, in general, analytical solutions, thus one must rely on approximate, time-discretisation schemes to generate diffusion sample paths and obtain closed-form expressions for the SDE transition density. Development of approximation schemes with high accuracy can lead to effective parameter estimation methods, e.g.~to Bayesian data augmentation without excessive imputation of latent variables to cover a non-small step-size $\Delta$. Indeed, several works on Markov chain Monte-Carlo (MCMC) methods for diffusion processes suggest use of high order approximations, e.g.~the Milstein or a strong 1.5 order scheme \citep{kloe:92}, but practical application of such schemes  is restricted to limited classes of SDEs due to involvement of intractable random variables, such as the \emph{L\'evy area} $\textstyle{\int_0^\Delta (B_{k_1 ,s} d B_{k_2 ,s} - B_{k_2 ,s} d B_{k_1 ,s}}), \, k_1 \neq k_2$. Motivated by the above, we develop our work under the following strategy. First, we propose an explicit sampling scheme for the generation of sample paths for model (\ref{eq:model}), with higher order accuracy in distributional sense compared to classical (conditionally) Gaussian numerical schemes, e.g.~the Euler-Maruyama scheme. Then, we derive an Edgeworth-type density expansion of the above sampling scheme. Instances of the expansion will be used to provide proxies for the true intractable transition density and, ultimately, likelihood-based parameter estimates. Finally, we provide analytic results showcasing advantages of the obtained parameter estimates in both high/low-frequency observation regimes, where in the former case one assumes $\Delta\rightarrow 0$ and in the latter that $\Delta$ is fixed and not small enough.

One set of analytic results that we provide in this work  -- and in accordance with recent contributions involving hypo-elliptic SDEs by, e.g., \cite{dit:19, glot:21} -- will correspond to asymptotic results for the Maximum Likelihood Estimator (MLE), in the high-frequency regime, with $n\rightarrow \infty$ and $\Delta = \Delta_n\rightarrow 0$, and in a setting of \emph{complete} discrete-time observations. That is, the dataset $Y^{c}$ is 
as follows: 
\begin{align}
\label{eq:data} 
Y^{c} := \{X_{t_0}, X_{t_1},\ldots, X_{t_n}\}.
\end{align}
We refer to the dataset in (\ref{eq:data}) in the high-frequency setting as the `complete observation regime' in the sequel. 
We stress that the discretisation schemes developed in the above context are practically relevant for generic observation regimes, e.g.~within Bayesian data augmentation methods. Consideration of these latter methods motivates the derivation of further analytic results, now in a low-frequency regime. 
In particular, for elliptic SDEs,
we explicitly connect the weak order of the numerical scheme with the proximity between the true transition density over a period of size $\Delta$ and the transition density produced by  convolution of $M\ge 1$ proxy transition densities over time intervals of size $\delta=\Delta/M$. 

We briefly review recent works in the high-frequency regime. In the hypo-elliptic setting, a first main contribution is the work of \cite{poke:09} that uses an It\^o-Taylor expansion to add a noise term of size $\mathcal{O}(\Delta^{3/2})$ in the numerical scheme for the smooth component of SDE~(\ref{eq:model}), thus obtaining a non-degenerate (conditionally) Gaussian approximation of the true transition density. The class of models considered in \cite{poke:09} is restrictive and no analytical results are provided. 
\cite{dit:19} make a major contribution by starting  with a strong 1.5 order scheme 
before removing terms that do not affect  their asymptotic results.
Noise of size $\mathcal{O}(\Delta^{3/2})$ is propagated onto the smooth component (as in   \cite{poke:09}) and quantities of size $\mathcal{O}(\Delta^{2})$ are retained in the mean terms. These latter components remove the  bias for the estimates of drift parameters observed experimentally (within a Bayesian data augmentation setting) in \cite{poke:09}.  \cite{dit:19} provide analytic asymptotic results for a contrast estimator in the complete observation regime (\ref{eq:data}), under the condition $\Delta_n=o(n^{-1/2})$; this is a type of condition referred to as `rapidly increasing experimental design' in early investigations for elliptic SDEs in \cite{prak:88}.  The  estimation procedure in \cite{dit:19} is separated into two contrast functions, one for
$\gamma$ assuming knowledge of the true values of $(\beta, \sigma)$, and vice-versa for the contrast function for $(\beta, \sigma)$. The analysis provides marginal CLTs for the estimates of the parameters rather than an ideal joint CLT. The class of models covered is restricted to a scalar smooth component and diagonal diffusion coefficient matrix for the rough component. 
\cite{glot:20, glot:21} provide the most recent contributions. Closer to our purposes, \cite{glot:20} describe a non-adaptive approach (as opposed to one-step adaptive methods in \cite{glot:21}) and 
prove a CLT in the complete observation regime for $\Delta_n = o(n^{-1/2})$, as in \cite{dit:19}, but without strong restrictions on the class of models. 
In the elliptic setting, \cite{kessler1997estimation} developed contrast functions for the scalar case that, within regime (\ref{eq:data}), deliver estimates satisfying a CLT for $\Delta_n = o(n^{-1/q})$, for any integer $q\ge 2$. 
The method  in \cite{kessler1997estimation} is based on the use of a Gaussian density for the approximation of the transition density of the SDE  for small $\Delta$, together with high order  expansions in $\Delta$ for the mean and variance of the SDE transitions. 
 \cite{uchi:12} extended such results to general model dimension based on multi-step adaptive estimates.

Our main contributions can be summarised as follows:
\vspace{-0.3 \baselineskip} 
\begin{description}
\item[(a)] 
In the hypo-elliptic case
  we propose a new weak second order sampling scheme. 
The scheme is explicit 
as one can produce SDE sample paths by generating Gaussian variates. 
A related weak second order scheme 
is also put forward for the elliptic case. 
%
\item[(b)]
The above sampling scheme possess, in general, an intractable transition density 
due to involvement of polynomials of Gaussian variates. 
We derive a closed-form \emph{small time density expansion} of the scheme by making use of tools from Malliavin calculus.   
%
\item[(c)]
In the high-frequency regime, 
we develop our contrast function by selecting appropriate high order terms from the density expansion.  
We then prove a joint CLT for the deduced parameter estimates under $\Delta_n = o (n^{-1/3})$, which is the largest step-size permitted in the hypo-elliptic case to the best of our knowledge. In particular, a main improvement by the new estimator --  when compared with the existing estimator satisfying the CLT under $\Delta_n = o(n^{-1/2})$ -- is observed in the estimation of diffusion parameter $\sigma$. 
%
\item[(d)] 
In the low-frequency regime, 
we choose particular 
high order terms from the small time density expansion to develop our local weak third order transition density scheme. 
We then study the use of the developed scheme in the practical setting 
where its density is iteratively
applied with user-specified steps of size $\delta= \Delta/M$, $M\ge 1$, to cover fixed inter-observation times of length $\Delta$. We prove that the induced bias is $\mathcal{O} (M^{-2})$ for elliptic SDEs. 
%
\item[(e)] We show numerical examples illustrating the benefits of the new schemes in applications with hypo-elliptic SDEs in the high-frequency regime. In the context of Bayesian data augmentation, in the low-frequency regime,
we apply the sampling scheme to an elliptic model and observe that the bias 
in the induced posterior is reduced when comparing with the Euler-Maruyama scheme. 
\end{description} 
\vspace{-0.3 \baselineskip}  
The remaining part of the paper is organised as follows.
%
%
Section \ref{sec:scheme} puts forward our sampling schemes for elliptic and hypo-elliptic SDEs. 
Section \ref{sec:density_scheme} derives the closed-form small time density expansion for the sampling scheme.  Section \ref{sec:main} provides a collection of analytical results for parameter inference carried out via appropriate choice of high order terms from the developed density expansion, 
both in the high and low-frequency regimes.
Section \ref{sec:sim_study} provides numerical experiments related to the analytic results.
Section \ref{sec:conclusions} concludes our work. 
\\[-0.2cm]

\noindent \textbf{Notation.}
We define: 
\begin{align*}
V_0 (\cdot, \beta,\gamma)
=\big[\,V_{R,0}(\cdot, \beta)^{\top}, V_{S,0}(\cdot, \gamma)^{\top}\,\big]^{\top}, \quad 
V_k (\cdot, \sigma) 
=\big[\,V_{R,k}(\cdot, \sigma)^{\top}, \mathbf{0}_{d_S}^\top \,\big]^{\top},
\quad  1 \le k \le d_R. 
\end{align*}
%
\noindent We set 
$V_R= V_R(x , \sigma) =  \big[ V_{R,1}(x, \sigma),\ldots, V_{R, d_R}(x, \sigma) \big]\in \mathbb{R}^{d_R\times d_R}$, and define:
%
\begin{align}
\label{eq:VAR}
a_R(x , \sigma) = V_R(x, \sigma) V_R(x,  \sigma)^{\top}\in \mathbb{R}^{d_R\times d_R}.
\end{align}
%
Let 
$C_b^{\infty} (\mathbb{R}^n ;  \mathbb{R}^m)$ 
(resp.~$C_p^{\infty} (\mathbb{R}^n ;  \mathbb{R}^m)$), $m,n\ge 1$, be the space of smooth (i.e.~infinitely differentiable) bounded (resp.~of polynomial growth) functions $f :\mathbb{R}^n \to \mathbb{R}^m$  with bounded derivatives (resp.~with derivatives of polynomial growth).
We write $\textstyle \partial_{u}  =  \big[ \tfrac{\partial }{\partial u^{1}}, \ldots,  \tfrac{\partial }{\partial u^{n} } \big]^{\top}, \;  
\partial^{2}_{u}   = \partial_{u} \partial_{u}^{\top} \equiv \big(\tfrac{\partial^{2}}{\partial u^{i}\partial u^{j} }\big)_{i,j=1}^{n}$ for the standard differential operators acting upon maps $\mathbb{R}^{n}\to \mathbb{R}$, $n\ge 1$. 
%
 %
As above with $u=(u^{1},\ldots,u^{n})^{\top}\in\mathbb{R}^{n}$, we use superscripts to specify co-ordinates of a vector when needed. 
For sufficiently smooth functions 
$f = f(x)=f(x_R,x_S):\mathbb{R}^d \to \mathbb{R}$,  
we define:  
\begin{align*}
   \hat{V}_0 f =   \hat{V}_0 f (x, \theta) 
  &:=  \big\langle V_0, \partial_x f \big\rangle
  + \tfrac{1}{2} \sum_{1 \le k \le d_R} V_{R,k}^{\top} \big(\partial_{x_R}^{2} f\big) V_{R,k}; \\ 
  \hat{V}_k f = \hat{V}_k f (x, \sigma)
  &:= \big\langle V_{k},  \partial_{x} f  \big\rangle
  = \big\langle V_{R,k},  \partial_{x_R} f  \big\rangle , \qquad 1\le  k \le d_R.  
\end{align*}
%
 $\hat{V}_0$,  $\hat{V}_k$ apply to vector-valued functions by separate consideration of the scalar co-ordinates.
%
%
Note that the original SDE (\ref{eq:model}) is equivalently given as the Stratonovitch-type SDE with drift function defined as: 
%
%
\begin{align}
\label{eq:VO}
 \widetilde{V}_0 (x ,  \theta) = V_0(x ,  \theta) - \tfrac{1}{2} \sum_{1 \le k \le d_R} \hat{V}_k V_0 (x ,  \theta). 
\end{align}
We define the vector-valued function:
\begin{equation} \label{eq:commutator0}
 [V_k, V_l] = \hat{V}_k V_l (x ,  \theta) - \hat{V}_l V_k (x ,  \theta), \quad 1\le k,l\le d_R. 
\end{equation}
%
%
We write $\trueparam = (\truebeta, \truegamma, \truesigma) \in \Theta$ for the (assumed unique) true value of $\theta = (\beta, \gamma, \sigma)$. 
We denote iterated stochastic integrals w.r.t.~Brownian paths as:   
\vspace{-0.4 \baselineskip} 
\begin{align}
\label{eq:II}
I_\alpha (t) =  \int_{0}^{t}\cdots \int_{0}^{t_2} dB_{\alpha_1, t_1} \cdots dB_{\alpha_l, t_l}, \quad t >  0, \ \  \alpha \in \{0, 1, \ldots, d_R\}^l, \  l \ge 0. 
\end{align} 
We will sometimes write $\mathbb{P}_{\theta}$, $\mathbb{E}_{\theta}$ to emphasise the involvement of $\theta$ in calculations. Similarly, we will write $\mathbb{P}_{\theta^\dagger}$, $\mathbb{E}_{\theta^\dagger}$
to stress, when needed, that derivations are under the true parameter value,
and write 
$\textstyle{\probconv, \;  \xrightarrow{\mathcal{L}_{\theta^{\dagger}}}}
$
to express convergence in probability and in distribution, respectively,  under $\theta^{\dagger}$. 
For a multi-index $\alpha=(\alpha_1,\ldots,\alpha_l) \in \{1, \ldots, d\}^{l}, \, l \ge 1$ 
and a sufficiently smooth $f:\mathbb{R}^{d} \to \mathbb{R}$, 
we write
$\textstyle{\partial^\xi_\alpha f(\xi) := {\partial^{l}f(\xi)}/{\partial \xi_{\alpha_1}\cdots\partial \xi_{\alpha_l}}}, \; \xi \in \mathbb{R}^d$.  
\section{Sampling schemes for elliptic \& hypo-elliptic diffusions} 
\label{sec:scheme}
We propose explicit sampling schemes for the SDE (\ref{eq:model}). 
We introduce the conditions so that the law of $X_t$, $t>0$, admits a smooth Lebesgue density in Section \ref{sec:basic_con}. Then, we present the sampling schemes in Section \ref{sec:sampling_scheme}. Hereafter, we make use of the notation $X_t^x$, $t > 0$, when needed to emphasise the initial state $X_0 = x \in \mathbb{R}^d$. 
\subsection{Basic assumptions for diffusion class} 
\label{sec:basic_con}
We introduce basic conditions to characterise the SDE (\ref{eq:model}) we consider in this work. 
\vspace{-0.5\baselineskip} 
\begin{enumerate}
\renewcommand{\theenumi}{H\arabic{enumi}}
\renewcommand{\labelenumi}{(\theenumi)}
\item \label{assump:param_space}
$\Theta$ is a compact subset of $\mathbb{R}^{d_\theta}$. For each $x \in \mathbb{R}^d$, and any multi-index $\alpha \in \{1, \ldots, d\}^l$, $l\ge 0$,
the function 
$\theta \mapsto  \partial^x_{\alpha} V_j^i (x,  \theta), \;  0\le j \le  {d_R}, \; 1 \le i \le d,$
is continuous.
\item \label{assump:coeff}
For each $\theta \in \Theta$, $V_j (\cdot , \theta) \in C_b^{\infty} (\mathbb{R}^d ; \mathbb{R}^d)$, $0 \le j \le d_R$.
\item \label{assump:hypo1}
The matrix $a_R = a_R(x  ;  \sigma)$ defined in (\ref{eq:VAR}) is positive-definite for all $(x,\theta) \in \mathbb{R}^d\times \Theta$. 
\item \label{assump:hypo2}
For any $\theta \in \Theta$ and any $x\in \mathbb{R}^{d}$, the set of the $2d_R$ vectors:
\vspace{-0.2\baselineskip} 
\begin{align*}
\big\{  V_1 (x,  \theta), \ldots, V_{d_R} (x ,  \theta), 
[\widetilde{V}_0, V_1] (x ,  \theta), \ldots,  [\widetilde{V}_0, V_{d_R}](x ,  \theta)   \big\}
\vspace{-0.2cm} 
\end{align*} 
spans $\mathbb{R}^d$, 
where $\widetilde{V}_0 = \widetilde{V}_0 (x ,  \theta)$ is defined in (\ref{eq:VO}) and  
$[\widetilde{V}_0, V_j] (x , \theta)$ is defined as  (\ref{eq:commutator0}).
\end{enumerate}
\vspace{-0.5\baselineskip}
(\ref{assump:hypo2}) is relevant for hypo-elliptic SDEs and is stronger than  H\"ormander's condition \citep{nual:06}, thus the law of $X_t^x$ is absolutely continuous w.r.t.~the Lebesgue measure for any $t>0$, $(x, \theta) \in \mathbb{R}^d \times \Theta$. H\"ormander's condition allows for iterated Lie brackets of order larger than one (used here), e.g.~of order two $[\widetilde{V}_0, [\widetilde{V}_0,V_k]]$, $1\le k\le d_R$, or above, 
to obtain vectors spanning $\mathbb{R}^{d}$. We require  (\ref{assump:hypo2}) so that certain discretisation schemes arising in our methodology also have a Lebesgue density.
Note that (\ref{assump:coeff}), (\ref{assump:hypo2}) combined imply that 
the density of $X_t^x$, $t>0$, 
is infinitely differentiable \citep[Theorems 2.3.2, 2.3.3]{nual:06}. 
%
%
\subsection{Approximate sampling scheme}
\label{sec:sampling_scheme}
We propose approximate sampling schemes for elliptic/hypo-elliptic SDE (\ref{eq:model}) that satisfy the following three Criteria: 
\vspace{-0.5\baselineskip} 
\begin{enumerate}[leftmargin=0.8cm]
\item[(i)]   
The scheme is  explicit. 
%
%
\item[(ii)]   
The scheme has local weak third order accuracy, i.e.~for test functions $\varphi \in C_p^{\infty} (\mathbb{R}^d ;  \mathbb{R})$ there exists a constant $C=C(x , \theta) >0$ such that: 
$$
\big | \mathbb{E}_{\theta} [\varphi ( X_\Delta^x )]  -  \mathbb{E}_{\theta}  [\varphi ( \bar{X}_\Delta^x )]  \big| \leq C \Delta^3, 
  \quad \Delta \geq 0, 
$$
%
where $\bar{X}_\Delta^x$ denotes the approximated SDE position after a single step of size $\Delta > 0$ with an initial value $\bar{X}_0^x = x \in \mathbb{R}^d$. 
\item[(iii)] The distribution of $\bar{X}_{\Delta}^x$ admits a Lebesgue density. 
\end{enumerate}
\vspace{-0.5\baselineskip}
Criterion (ii) implies that the sampling scheme is a weak second order approximation of $\{X_t\}$ in the following sense.  Let $T >0$, $n\ge 1$, and consider the partition $0 = t_0 < t_1  < \cdots <  t_{n-1} < t_n = T$ with $t_i - t_{i-1} = T/n$,  $1\le i\le n$. Then, for $\varphi \in C_p^{\infty} (\mathbb{R}^d ; \mathbb{R})$, there exists a constant $C=C(x, T, \theta)>0$ such that  
$
  \big| \mathbb{E}_{\theta} 
  [
  \varphi (X_T^{x})]  -  \mathbb{E}_{\theta} [\varphi (\bar{X}_T^{(n), x}) ] 
  \big|
  \leq {C}/{n^2},
$
where $\bar{X}_T^{(n), x}$ is obtained after $n$ iterations of a single step of the scheme, starting from the initial state $x \in \mathbb{R}^d$. %
%
Criterion (iii) is required so that a well-defined contrast function can be obtained and estimates of the parameters be produced. 
%
%
%
%
%
\subsubsection{Sampling scheme for elliptic diffusions} 
\label{sec:review_milstein}
Before writing down our scheme satisfying Criteria (i)-(iii) above for elliptic diffusions, we present the ideas underlying the construction of a weak second order discretisation scheme (equivalently, a local weak third order scheme) following the moment matching techniques in \cite[Section~2.1.2]{mil:21}. 
We apply an It\^o-Taylor expansion 
\citep{kloe:92} for SDE model~(\ref{eq:model}) with $d_S=0$ and obtain that $X_{e,\Delta}^x = \hat{X}_{e,\Delta}^x + \rho_{e,\Delta} (x, \theta)$, where
%
\begin{align} 
\begin{aligned} \label{eq:expansion}
 \hat{X}_{e,\Delta}^x
 = &\; x 
   + V_{R, 0} (x ,  \theta) \Delta  
   + \sum_{1 \le k \le d_R} V_{R, k} (x ,  \theta) B_{k,\Delta} 
   + \sum_{0 \leq k_1, k_2 \leq d_R} \hat{V}_{k_1} {V}_{R, k_2} (x ,  \theta) I_{(k_1, k_2)} (\Delta), 
\end{aligned}
\end{align}
and $\rho_{e,\Delta} (x, \theta)$ is a residual term involving stochastic iterated integrals of order three and above. Then, it can be shown  \citep{mil:21} that for test functions $\varphi \in C_p^\infty (\mathbb{R}^{d_R}; \mathbb{R})$, $\hat{X}_\Delta^x$ has local weak third order accuracy, though it is not an explicit scheme due to the presence of  $I_{(k_1, k_2)}(\Delta)$, $1 \le k_1, k_2 \le d_R$, $k_1 \neq k_2$, in its expression. Thus, we aim at replacing $I_{(k_1, k_2)} (\Delta)$ in (\ref{eq:expansion}) by some tractable random variables $\{\xi_{k_1 k_2, \Delta} \}_{k_1, k_2}$ so that a scheme 
using $\xi_{k_1 k_2, \Delta}$ still achieves  local weak third order convergence. Lemma 2.1.5 in \cite{mil:21} states 
the local weak third order accuracy is preserved if the random variables $\{\xi_{k_1 k_2, \Delta} \}_{k_1, k_2}$ have finite moments up to 6th order and satisfy a collection of moment conditions.

Thus, we write the sampling scheme for elliptic diffusions as: 
\begin{align} 
\label{eq:scheme_ell}
\bar{X}_{e,\Delta}^x 
 = &\; x
     +  V_{R, 0} (x ,  \theta)  \Delta 
     + \sum_{1 \le k \le d_R} V_{R, k} (x ,  \theta) B_{k,\Delta}  
     +  \sum_{\substack{0 \leq k_1, k_2 \leq {d_R}}} \hat{V}_{k_1}  V_{R, k_2}  (x , \theta)\,\xi_{k_1 k_2,\Delta},   
\end{align} 
with the following specification of the random variables $\xi_{k_1 k_2,\Delta}$, $0 \leq k_1, k_2 \leq d_R$, satisfying the moment conditions given in Lemma 2.1.5 \cite{mil:21}: for $1 \le k_1, k_2 \le d_R$,
%
%
\begin{align}
\begin{aligned}
 \label{eq:rv}
 & \xi_{00, \Delta} = \tfrac{\Delta^2}{2}, \qquad 
 \xi_{k_10, \Delta} =  \xi_{0 k_1,\Delta} = \tfrac{1}{2} B_{{k_1}, \Delta} \Delta,  \qquad 
  \xi_{k_1 k_1, \Delta} = \tfrac{1}{2} B_{k_1, \Delta} B_{k_1, \Delta}  -  \tfrac{1}{2} \Delta;  \\
 & \qquad \xi_{k_1k_2, t} = \tfrac{1}{2} B_{k_1, \Delta} B_{k_2, \Delta}  +   \tfrac{1}{2} B_{k_1, \Delta} \widetilde{B}_{k_2, \Delta} \cdot \mathbf{1}_{k_1 < k_2} - \tfrac{1}{2} B_{k_2, \Delta}  \widetilde{B}_{k_1, \Delta} \cdot \mathbf{1}_{k_1 > k_2}. 
%
%
\end{aligned}
\end{align} 
Here, $\widetilde{B} := (\widetilde{B}_{2,t},  \ldots, \widetilde{B}_{d_R, t})_{t \geq 0}$ denotes a $(d_R-1)$-dimensional standard Brownian motion,  independent of $B$. 
%
%
Our  scheme is marginally different from the one appearing in  \cite{mil:21}, with the former scheme using $(2d_R-1)$ Gaussian variates, rather than $2d_R$ ones involved in the latter. The scheme is explicit as it only involves Gaussian variates -- thus Criterion~(i) stated above is satisfied. 
Criterion~(ii) is satisfied, given the above discussion. Criterion~(iii) is clearly satisfied under (\ref{assump:hypo1}).
\begin{rem} \label{rem:third_criteria}
Scheme (\ref{eq:scheme_ell}) under the choice of variates in (\ref{eq:rv}) can also be used in a hypo-elliptic setting as an explicit weak second order scheme. However, such a specification 
fails to meet Criterion (iii). Some more details are given in Remark \ref{rem:degeneracy}. Hence, in this work, we propose separate sampling schemes for the elliptic/hypo-elliptic classes. 
\end{rem}
In (\ref{eq:scheme_ell}) we introduced the extra subscript `$e$' for sampling scheme $\bar{X}_{e,\Delta}^x$ to stress that it involves the elliptic case (not the hypo-elliptic one). We adopt a similar convention for mathematical expressions, when needed, in the rest of the paper.  
\subsubsection{Sampling scheme for hypo-elliptic diffusions}  \label{sec:scheme_hypo}
We treat hypo-elliptic SDEs, thus $d_S \geq 1$.
We make use of an It\^o-Taylor expansion 
where for the rough component we retain all second order integrals and for the smooth one  we retain the second order and third order integrals for which the outside integrator is $dt$.
That is, we have $X_{\Delta}^x = \hat{X}_{\Delta}^x + \rho_{\Delta} (x, \theta)$ where:
%
\begin{align*} 
\hat{X}_{R,\Delta}^x  
& =  x_{R}
 + V_{R, 0} (x, \theta) \Delta 
 + \sum_{1 \le k \le d_R} 
 V_{R,k}  (x,  \theta)\,B_{k, \Delta}  
 + \sum_{\substack{0 \leq k_1, k_2 \leq {d_R}}} \hat{V}_{k_1}  V_{R, k_2}  (x ,  \theta) \,I_{(k_1,k_2)}(\Delta); \nonumber \\ 
\hat{X}_{S,\Delta}^x  & = x_{S}  
   + V_{S, 0} (x ,  \theta)\, \Delta 
   + \sum_{0 \le k \le d_R} \hat{V}_k V_{S,0} (x , \theta) \, I_{(k,0)}( \Delta) 
   + 
   \sum_{\substack{0 \le k_1, k_2 \le {d_R} \\ \mrm{s.t.} \,k_1 = k_2 \neq 0}} \hat{V}_{k_1} \hat{V}_{k_2} V_{S,0} (x , \theta)\, I_{(k_1, k_2,0)}(\Delta). 
%
%
\end{align*} 

%
%
As with the elliptic case, we next replace the non-explicit iterated integrals $I_{(k_1, k_2)} (\Delta)$, $I_{(k_1, k_2, 0)}(\Delta)$, $1 \le k_1, k_2 \le d_R$, $k_1 \neq k_2$, 
with explicit variates, based on moment conditions that ensure that the resulted scheme remains a weak second order one. The tools in \cite{mil:21} cover only the setting of double integrals, thus we need to carry out an  extension of such methodology in the presence of triple integrals. 
%
The extension (see the proof of Proposition \ref{prop:weak_app_hypo} stated below for details) gives rise to the sampling scheme
%
$\bar{X}_\Delta^x 
=  
[ ( \bar{X}_{R,\Delta}^{x} )^\top, (\bar{X}_{S,\Delta}^{x})^\top ]^\top$, $x = [x_R^\top, x_S^\top]^\top
\in \mathbb{R}^d$ for hypo-elliptic SDEs determined as follows:
\vspace{-0.5\baselineskip}  
\begin{align} 
\begin{aligned} \label{eq:scheme_hypo1}
\bar{X}_{R,\Delta}^x  
& =  x_{R}
 + V_{R, 0} (x, \theta) \Delta 
 + \sum_{1 \le k \le d_R} V_{R,k}  (x,  \theta)\,B_{k, \Delta}  
 + \sum_{\substack{0 \leq k_1, k_2 \leq {d_R}}} \hat{V}_{k_1}  V_{R, k_2}  (x ,  \theta) \,\zeta_{k_1 k_2,\Delta}; \\ 
\bar{X}_{S,\Delta}^x  & = x_{S}  
   + V_{S, 0} (x ,  \theta)\, \Delta 
   + \sum_{0 \le k \le d_R} \hat{V}_k V_{S,0} (x , \theta) \, \zeta_{k 0, \Delta}  
   + \sum_{\substack{0 \le k_1, k_2 \le {d_R} \\ \mrm{s.t.} \,k_1 = k_2 \neq 0}} 
   \hat{V}_{k_1} \hat{V}_{k_2} V_{S,0} (x , \theta)\,\eta_{k_1 k_2,\Delta}. 
\end{aligned}
\end{align} 
Variables $\zeta_{k_1 k_2, \Delta}$  are such that $\zeta_{00 ,\Delta} =  {\Delta^2}/{2}$ and for $1\le k_1,k_2\le d_R$:
%
%
\vspace{-0.3 \baselineskip}
\begin{align}
\label{eq:zetas} 
\zeta_{0 k_1, \Delta} = I_{(0, k_1)} (\Delta), \quad 
\zeta_{k_1 0, \Delta} =  
I_{(k_1, 0)} (\Delta), \quad 
\zeta_{k_1 k_2, \Delta} =  \xi_{k_1 k_2, \Delta}, 
\end{align}
for $\xi$'s as determined earlier in (\ref{eq:rv}) for the elliptic case.
Variables $\eta_{k_1 k_2 , t}$ are required to satisfy the following moment conditions, with $1 \leq k_1, k_2, k_3, k_4 \leq {d_R}$: 
{
\begin{align*} 
 & \mathbb{E} [\eta_{k_1 k_2 , \Delta}] = \mathbb{E} [I_{(k_1, k_2, 0)}(\Delta)] = 0, \quad
   \mathbb{E} [\eta_{k_1 0, \Delta}] = \mathbb{E} [I_{(k_1, 0, 0)}(\Delta)] = 0;  \\
 & \mathbb{E} [\eta_{0 k_1, \Delta}] = \mathbb{E} [I_{(0, k_1, 0)}(\Delta)] = 0,  \quad  
   \mathbb{E} [\eta_{k_1 k_2 , \Delta} B_{k_3,\Delta}] = \mathbb{E} [I_{(k_1, k_2, 0)}(\Delta) B_{k_3, \Delta}] = 0;   \\
 & \mathbb{E} [\eta_{k_1 0, \Delta} B_{k_2,\Delta}] 
 = \mathbb{E} [\eta_{0k_1,\Delta} B_{k_2,\Delta}]  
 = \mathbb{E} [I_{(k_1, 0, 0)}(\Delta) B_{k_2, \Delta}] 
 = \mathbb{E} [I_{(0, k_1, 0)}(\Delta) B_{k_2, \Delta}] 
 = \tfrac{\Delta^3}{6} \times  \mathbf{1}_{k_1 = k_2}; \\
%
%
 & \mathbb{E} [\eta_{k_1 0, \Delta} \zeta_{k_2 0, \Delta}] = \mathbb{E} [I_{(k_1, 0, 0)}(\Delta) I_{(k_2,0)}(\Delta)] = \tfrac{\Delta^4}{8} \times\mathbf{1}_{k_1 = k_2}; \\
 & \mathbb{E} [\eta_{0 k_1, \Delta} \zeta_{k_2 0, \Delta}] = \mathbb{E} [I_{(0, k_1, 0)}(\Delta) I_{(k_2,0)}(\Delta)] = \tfrac{\Delta^4}{6}\times \mathbf{1}_{k_1 = k_2}; \\     
 & \mathbb{E} [\eta_{k_1 k_2 , \Delta} \zeta_{k_3 k_4 , \Delta} ] 
    = \mathbb{E} [I_{(k_1, k_2, 0)}(\Delta) I_{(k_3, k_4)} (\Delta)] = \tfrac{\Delta^3}{6} \times\mathbf{1}_{k_1 = k_3, k_2 = k_4}; \\ 
 & \mathbb{E} [\eta_{k_1 k_2 , \Delta} \eta_{k_3 k_4, \Delta} ] 
    = \mathbb{E} [I_{(k_1, k_2, 0)}(\Delta) I_{(k_3, k_4, 0)}(\Delta)  ] = \tfrac{\Delta^4}{12}\times \mathbf{1}_{k_1 = k_3, k_2 = k_4}.
\end{align*}}
A particular choice for the $\eta$'s that we adopt for the rest of the paper is the one below:
\begin{align*}
 \eta_{k_1  0 , \Delta} & =  \tfrac{\Delta}{2} \zeta_{k_1 0, \Delta} - \tfrac{\Delta^2}{12} B_{k_1,\Delta},
\ \ 
 \eta_{0  k_1, \Delta}  =  \Delta \zeta_{k_1 0, \Delta} - \tfrac{\Delta^2}{3} B_{k_1,\Delta}, 
\ \ 
\eta_{k_1 k_2, \Delta} = \tfrac{1}{3} \zeta_{k_1 k_2, \Delta}\Delta  - \widetilde{\eta}_{k_1 k_2,\Delta}\Delta. 
\end{align*}
We have set: 
\begin{align*}  
\widetilde{\eta}_{k_1 k_2, \Delta}  
= &\;
\begin{cases}
\tfrac{1}{6 \sqrt{2}} ( \widetilde{B}_{k,\Delta} \widetilde{B}_{k, \Delta}  -   \Delta ),  & k_1 = k_2 = k;   \\[0.2cm]
\tfrac{1}{6 \sqrt{2}} ( \widetilde{B}_{k_1, \Delta} \widetilde{B}_{k_2,\Delta}  
+ \widetilde{B}_{k_1,\Delta} {W}_{k_2, \Delta}\cdot \mathbf{1}_{k_1 < k_2} 
 -   \widetilde{B}_{k_2, \Delta} {W}_{k_1, \Delta}
\cdot \mathbf{1}_{k_1 > k_2})  & k_1 \ne k_2; \\[0.2cm] 
\end{cases} 
\end{align*}
where $\widetilde{B} := (\widetilde{B}_{1, t},  \ldots, \widetilde{B}_{{d_R}, t})_{t \geq 0}$ and ${W} = ({W}_{2, t},  \ldots,  {W}_{{d_R}, t})_{t \geq 0}$ are standard Brownian motions, mutually independent and independent of $B$. 
\begin{rem}
   Generation of the Gaussian  variables $\zeta_{0 k, \Delta}$ and $\zeta_{k0, \Delta} = B_{k, \Delta}\, \Delta - \zeta_{0 k, \Delta}$ in~(\ref{eq:zetas}) must take under consideration the dependency structure amongst  $\{B_{k, \Delta}\}$, $\{\zeta_{0 k, \Delta}\}$. That is, we have: 
   %
   $
  \mathbb{E} [B_{k_1, \Delta} \zeta_{0 k_2, \Delta}] 
  = {\Delta^2}/{2} \cdot \mathbf{1}_{k_1 = k_2}, \
  \mathbb{E} [\zeta_{0 k_1, \Delta} \zeta_{0 k_2, \Delta}]
  = {\Delta^3}/{3} \cdot \mathbf{1}_{k_1 = k_2} 
   $
   for $1\le k_1, k_2 \le  d_R$. Thus, one can generate $\{ B_{k, \Delta} \}$, 
   $\{ \zeta_{0 k, \Delta} \}$ and $\{ \zeta_{k 0, \Delta} \}$ as follows: 
   \begin{align*}
       B_{k, \Delta} = \Delta^{1/2}\, Z_k, \quad 
       \zeta_{0 k, \Delta} {=}  
       \tfrac{\Delta^{3/2}}{2}
       \big( Z_k + \tfrac{1}{\sqrt{3}} \widetilde{Z}_k \big), \quad 
       \zeta_{k 0, \Delta} {=} 
       \Delta^{3/2} Z_k   - \zeta_{0 k, \Delta}, \qquad 1\le k \le d_R,
   \end{align*}
   where $Z_k, \widetilde{Z}_k \sim \mathscr{N}(0,1)$, $1\le k \le d_R$, are i.i.d.~random variables. 
\end{rem}

We have developed scheme (\ref{eq:scheme_hypo1}) that is explicit (Criterion (i)) and, as stated in the proposition below, it has local weak third order accuracy (Criterion (ii)), with the proof given in Section \ref{appendix:weak_app_hypo} of Supplementary Material. Criterion (iii) is discussed after the following proposition.
\begin{prop} 
\label{prop:weak_app_hypo}
Let $d_S \geq 1$ and consider $\varphi \in C_p^{\infty} (\mathbb{R}^d ; \mathbb{R})$. Under conditions (\ref{assump:param_space})--(\ref{assump:coeff}), for any $(x, \theta) \in \mathbb{R}^d \times \Theta$, there exist constants $C>0$, $q \ge 1$ such that:
\begin{align*}
  \big| \mathbb{E}_{\theta} [\varphi (X_\Delta^x)] - \mathbb{E}_{\theta} [\varphi (\bar{X}_\Delta^x)] \big| \leq C (1 + |x|^q) \Delta^3. 
\end{align*}
%
%
\end{prop} 
\noindent  
Finally, w.r.t~Criterion (iii) we notice that scheme (\ref{eq:scheme_hypo1}) contains in its specification (i.e.~once some terms are removed from (\ref{eq:scheme_hypo1}))  the \emph{local Gaussian scheme} $\widetilde{X}_\Delta^{\mathrm{LG}, x} 
= \bigl[ ( \widetilde{X}_{R,\Delta}^x )^{\, \top}, 
( \widetilde{X}_{S,\Delta}^x )^{\, \top} \bigr]^\top $ given as:
\begin{align}
\begin{aligned}  \label{eq:LG} 
 \widetilde{X}_{\, R,\Delta}^x & =  x_R + V_{R, 0} (x, \beta) \Delta 
  + \sum_{1 \le k \le d_R} V_{R,k} (x, \sigma) B_{k,\Delta}; \\
  \widetilde{X}_{S,\Delta}^x  & =  x_S + 
  V_{S,0 } (x, \gamma) \Delta 
  + \hat{V}_0 V_{S,0} (x, \theta)  \tfrac{\Delta^2}{2} + \sum_{1 \le k \le d_R} \hat{V}_k V_{S,0} (x, \theta) I_{(k,0)}(\Delta),  
\end{aligned}
\end{align}
where the covariance matrix is given as:
\begin{align}
\label{eq:Var}
 \Sigma(\Delta, x; \theta)  =  
%
 \left[ \begin{array}{cc} 
\sum_{1 \le k \le d_R} V_{R,k} (x,  \sigma) V_{R,k}  (x,  \sigma)^{\top} \Delta
& \quad 
\sum_{1 \le k \le d_R} V_{R,k}(x, \sigma) \hat{V}_k V_{S,0} (x, \theta)^{\top} \tfrac{\Delta^2}{2} \\
\sum_{1 \le k \le d_R}  \hat{V}_k V_{S,0} (x, \theta) V_{R,k}(x, \sigma)^{\top}  \tfrac{\Delta^2}{2}
& \quad \sum_{1 \le k \le d_R} \hat{V}_k V_{S,0} (x,  \theta)  \hat{V}_k V_{S,0} (x,  \theta)^{\top} \tfrac{\Delta^3}{3} 
\end{array} \right]. 
\end{align} 
 $\Sigma (\Delta, x ; \theta)$ is positive definite for any $(\Delta, x, \theta) \in (0, \infty) \times \mathbb{R}^d \times \Theta$ under Assumptions (\ref{assump:hypo1})--(\ref{assump:hypo2}), thus Criterion (iii) is satisfied.
%

\begin{rem} 
\label{rem:new1}
Proposition \ref{prop:weak_app_hypo} holds 
for scheme (\ref{eq:scheme_hypo1}) with or without retainment of the triple integral term in the smooth component
%
$
\textstyle 
I_3 
:= \sum_{(k_1,k_2)\in \{0, \ldots, {d_R}\}^2 \setminus \{0,0\} } \hat{V}_{k_1} \hat{V}_{k_2} V_{S,0} (x ,  \theta)\,I_{(k_1, k_2,0)}(\Delta).
$
However, our analysis of  statistical methodology based on the above scheme will illustrate that inclusion of term $I_3$ (or, more precisely, of an explicit substitute for $I_3$) is necessary for obtaining a CLT for the MLE in the complete observation regime under the improved rate $\Delta_n = o(n^{-1/3})$. 
\end{rem}
\begin{rem} \label{rem:degeneracy}
A  difference versus the scheme for elliptic SDEs in (\ref{eq:scheme_ell}) is that 
in (\ref{eq:scheme_hypo1}) variables 
$
\textstyle{I_{(k,0)}(\Delta)}$, $\textstyle{I_{(0, k)}(\Delta)}$, $1 \le k \le {d_R}$,
are kept and are not replaced by $\textstyle{\tfrac{1}{2}B_{k,\Delta} \Delta}$. A replacement of $\textstyle{I_{(k,0)}(\Delta)}$  
by $\textstyle{\tfrac{1}{2}B_{k,\Delta} \Delta}$ 
would lead to a degenerate covariance matrix in place of $\Sigma(\Delta,x;\theta)$ and 
a violation of Criterion (iii). 
\end{rem}

\section{Small time density expansion} \label{sec:density_scheme}
%
%
%
We denote the transition densities of the given SDE~(\ref{eq:model})
and the sampling scheme $\bar{X}_\Delta^x $ by
%
\vspace{-0.5 \baselineskip}   
$$
y \mapsto p_\Delta^X (x, y; \theta)= \mathbb{P}_\theta[{X}_\Delta^x \in dy]/dy, \qquad  y \mapsto p^{\bar{X}}_\Delta (x, y; \theta) = \mathbb{P}_\theta [\bar{X}_\Delta^x 
\in dy] / dy 
$$
%
%
respectively, for $(\Delta, x, \theta) \in (0, \infty) \times \mathbb{R}^d \times \Theta$. The (one step) sampling scheme $\bar{X}_\Delta^x$ in (\ref{eq:scheme_ell}) or (\ref{eq:scheme_hypo1}), for elliptic or hypo-elliptic SDEs respectively, does not  admit, in general, a closed-form density function. 
To obtain a proxy for the density, we work with a  
small time expansion of the density of  $\bar{X}_\Delta^x $ given below in Lemma \ref{lemma:AE_scheme}.
%
%
%
The expansion will be called upon in Section \ref{sec:main} to provide  a method for statistical inference that will possess advantageous characteristics compared with
existing methods that use constrast functions induced by conditionally Gaussian discretisation schemes, e.g.~the local Gaussian scheme  (\ref{eq:LG}) in the hypo-elliptic case or the Euler-Maruyama scheme in the elliptic case.
%
 We proceed by assuming a hypo-elliptic setting with $d_S \geq 1$, but our analysis also covers the elliptic case  $d_S=0$. 
 \subsection{Background material} \label{sec:back}
 Before giving the small time  expansion formula for $p^{\bar{X}}_\Delta (x, y; \theta)$ we sketch its derivation via use of Malliavin calculus, a differential calculus on Wiener space. Full rigorous arguments are given in Section \ref{appendix:density_scheme} of Supplementary Material. We refer interested readers to, e.g., \cite{watanabe1987analysis, ike:14}.

We consider  Wiener functionals, $F = F(\omega) : \Omega \to \mathbb{R}^d$, 
 where for current  Section~\ref{sec:back},   $(\Omega, \mathcal{F},\{\mathcal{F}_{t}\}, \mathbb{P})$ is the probability space giving rise to the $3\times d_R-1$ Brownian motions involved in the definition of sampling scheme $\bar{X}_\Delta$ in (\ref{eq:scheme_hypo1}). 
 %
 %
 Malliavin calculus allows for $\delta_y (F)$, with  $\delta_y$ the Dirac measure at $y \in \mathbb{R}^d$, to be well-defined as an element of a Sobolev space of Wiener functionals, provided that $F$ satisfies regularity and non-degeneracy conditions, in which case such a Wiener functional is referred to as \textit{smooth and non-degenerate in the Malliavin sense}. The non-degeneracy is a sufficient condition for the existence of Lebesgue density for the law of $F$, and in the case of a Gaussian variate, non-degeneracy is equivalent to the positive definiteness of the covariance matrix. It is shown that the Lebesgue density, $p^F(y)$, of $F$ coincides with $\mathbb{E}\,[\,\delta_y (F)\,]$, i.e.~the \emph{generalised expectation} of $\delta_y (F)$, if $F$ is smooth and non-generate in the Malliavin sense. 
 We consider an $F=F^{\varepsilon}$, for small $\varepsilon \in (0,1)$,  of the form 
\vspace{-0.3 \baselineskip}   
\begin{align}
\label{eq:MF}
    F^{\varepsilon} 
    = f_0 + \varepsilon f_1 + \varepsilon^2 f_2 
    + \cdots + \varepsilon^{J}  f_{J}     
\end{align}
for Wiener functionals $f_i = f_i (\omega )$, $0 \le i \le J$, $J \ge 1$.
 We assume that $F^{\varepsilon}$ is smooth and non-degenerate in the Malliavin sense for any $\varepsilon \in (0,1)$. Then, Theorem 9.4 in \cite{ike:14} gives an asymptotic expansion of the density, $p^{F^{\varepsilon}}(y)$, of $F^\varepsilon$ as follows. For every $n\ge 0$, there exist Wiener functionals $\phi_k=\phi_k(\omega,y)$, $0\le k\le n-1$, and $\mathbb{R}$-valued  $r_n (\cdot)$ such that  
 %
 \begin{align} 
 \label{eq:expansion_malliavin}
    p^{F^{\varepsilon}}(y)
    = \mathbb{E}\,[\,\delta_y (F^{\varepsilon})\,]
    = \sum_{0 \le k \le n-1}
       \varepsilon^k \cdot
       \mathbb{E}\,[\,\phi_{k}(\cdot, y)\,] + \varepsilon^n \cdot r_{n} (y,\varepsilon),
 \end{align}
 with $|r_{n}(y,\varepsilon)| \leq C$ for constant $C>0$ independent of $\varepsilon \in (0,1)$, $y \in \mathbb{R}^d$, and the $\phi_k$'s are given via a formal Taylor expansion of $\varepsilon \to \delta_y (F^\varepsilon)$, e.g.,  $\phi_0 = \delta_y (f_0)$, $\phi_1 = \langle \partial_z \delta_y (z)|_{z=f_0}, f_1\rangle$.
\subsection{Density expansion for sampling scheme} \label{sec:density_expansion}
We develop a small time expansion for
the density $p^{\bar{X}}_\Delta (x, y; \theta)$ of $\bar{X}_{\Delta}^x $ in (\ref{eq:scheme_hypo1}) via reference
to (\ref{eq:MF}), (\ref{eq:expansion_malliavin}). We make use of the multi-index notation:
$
\| \alpha \|:= l + (\mathrm{\# \; of \; zeros \; in \; }\alpha), \, | \alpha | := l, 
$
for $\alpha \in \{0,1, \ldots, d\}^l, \, l\ge 0$.
Based on the expression for $\bar{X}_\Delta^x $ in (\ref{eq:scheme_hypo1}), we define  
$
\bar{X}_\Delta^{\varepsilon, x} 
 =  \big[\,(\bar{X}_{R,\Delta}^{\varepsilon, x})^{\top},\,  
(\bar{X}_{S,\Delta}^{\varepsilon, x})^{\top}\,\big]^{\top}
$ as follows:
\begin{align}
\bar{X}_{R,\Delta}^{\varepsilon, x}  
   & = x_{R}
     + \varepsilon^2 V_{R, 0} (x, \beta) \Delta
     + \varepsilon \sum_{1 \le k \le d_R} V_{R, k} (x ,  \sigma ) \,B_{k,\Delta}  
     + \sum_{\substack{0 \leq k_1, k_2 \leq d_R}}
       \varepsilon^{\|(k_1, k_2) \|} \hat{V}_{k_1}  V_{R,k_2} (x , \theta)\, 
         \zeta_{k_1 k_2, \Delta}; \nonumber \\
\bar{X}_{S, \Delta}^{\varepsilon, x}  
  & =  x_{S} + \varepsilon^2 V_{S,0} (x ,  \gamma )\,\Delta 
     +  \sum_{0 \leq k \leq d_R} \varepsilon^{\|(k,0)\|}\hat{V}_k V_{S,0}(x,  \theta)\,\zeta_{k 0, \Delta} 
     \label{eq:scheme_eps} \\
   & \qquad \qquad  \qquad \qquad \qquad \qquad  
   + 
   \sum_{\substack{0 \le k_1, k_2 \le {d_R} \\ \mrm{s.t.} \,k_1 = k_2 \neq 0}}  
   \varepsilon^{\|(k_1, k_2, 0)\|} 
   \hat{V}_{k_1} \hat{V}_{k_2} V_{S,0} (x ,  \theta)\,\eta_{k_1 k_2, \Delta}.  \nonumber 
\end{align}
 The rationale in the consideration of the above process is that, first, the latter is connected with   
 scheme $\bar{X}_{\Delta}^x $ in (\ref{eq:scheme_hypo1}) 
 via the equality in distributions
$\mathbb{P}_\theta \big[\bar{X}_{\Delta}^x \in dy \big] 
= 
\mathbb{P}_\theta \big[\bar{X}_1^{\sqrt{\Delta}, x} \in dy \big]
$
%
and, second, it is straightforward to perform the standardisation defined below upon process (\ref{eq:scheme_eps}). 
We cannot apply expansion (\ref{eq:expansion_malliavin}) on the density of $\bar{X}_1^{\varepsilon, x}$ as $\{\bar{X}_1^{\varepsilon, x} \}_{\varepsilon}$ is degenerate in the Malliavin sense as $\varepsilon \to 0$, e.g.~notice that $\bar{X}_1^{0 ,x} = x$ is deterministic. 
Instead, we  introduce 
$\widetilde{Y}^{\varepsilon}$
via an appropriate standarisation of  $\bar{X}_1^{\varepsilon, x}$ based upon the local Gaussian scheme $\widetilde{X}_\Delta^{\mathrm{LG}, x} 
= \bigl[ ( \widetilde{X}_{R,\Delta}^x )^{\, \top}, 
( \widetilde{X}_{S,\Delta}^x )^{\, \top} \bigr]^\top $ in (\ref{eq:LG}). Thus, the small time expansion ultimately obtained will be centred around the density of the local Gaussian scheme. 
In particular, we set  
\begin{align} %
m_{x, \theta,\varepsilon} (y)
:=
\left[ \begin{array}{c} m_{R, x, \theta,\varepsilon} (y_R) \\
m_{S,x, \theta,\varepsilon} (y_S)
\end{array}
\right] 
= \left[ 
\begin{array}{c}
\frac{{y_{R} - x_{R} - \varepsilon^2 V_{R,0}(x , \beta)}}{\varepsilon} 
\\
\frac{y_{S} - x_{S} - \varepsilon^2 V_{S,0}(x ,  \gamma) - \tfrac{\varepsilon^4}{2} \hat{V}_0 V_{S,0}(x , \theta) }{\varepsilon^3}
\end{array}
\right],
\label{eq:muLG}
\end{align}
and, define 
\begin{align}
\label{eq:many}
\widetilde{Y}^{\varepsilon} = 
\bigl[  
\bigl(\widetilde{Y}_{R}^{\varepsilon}\bigr)^\top, 
\bigl(\widetilde{Y}_{S}^{\varepsilon}\bigr)^\top
\bigr]^\top=
m_{x, \theta,\varepsilon} (\bar{X}^{\varepsilon,x}_{1})
 = \widetilde{Y}^{(0)} + \sum_{1 \le l \le 3} \varepsilon^l \cdot \widetilde{Y}^{(l)},
\end{align}
where one can obtain the $\mathcal{O}(1)$-term in (\ref{eq:many}),  
$\widetilde{Y}^{0} = [(\widetilde{Y}_{R}^{(0)})^\top, (\widetilde{Y}_{S}^{(0)})^\top ]^\top $, as
\begin{align*}
\widetilde{Y}_{R}^{(0)} = \sum_{1 \le k \le d_R}V_{R,k} (x,  \sigma) B_{k,1},  \qquad  
\widetilde{Y}_{S}^{(0)} = \sum_{1 \le k \le d_R} \hat{V}_k V_{S,0} (x,  \theta) I_{(k,0)} (1),  
\end{align*}
with the remaining terms 
$\widetilde{Y}^{(l)}$, $1\le l \le 3$, 
explicitly derived via (\ref{eq:scheme_eps})--(\ref{eq:many}).
$\widetilde{Y}^{(0)}$ follows a Gaussian distribution and is non-degenerate under \mbox{(\ref{assump:hypo1})--(\ref{assump:hypo2})}, as the  covariance matrix $\Sigma (1, x ;  \theta)$ in (\ref{eq:Var}) is positive definite. In brief, the above lead to non-degeneracy in the Malliavin sense of $\{\widetilde{Y}^{\varepsilon}\}_{\varepsilon}$, uniformly in $\varepsilon \in (0,1)$. 
Simple change of variables with $\varepsilon = \sqrt{\Delta}$ yields: 
\vspace{-0.3 \baselineskip}
\begin{align} 
\label{eq:change_rv}
p_\Delta^{\bar{X}} (x, y; \theta) 
= \mathbb{P}_\theta \big[\bar{X}_1^{\sqrt{\Delta}, x} \in dy \big]/dy
= \frac{1}{\sqrt{\Delta^{d_R + 3 d_S}}} \, p^{\widetilde{Y}^{{\sqrt{\Delta}}}} (\xi ; \theta) |_{\xi = m_{x, \theta, {\sqrt{\Delta}}}(y)}. 
\\[-0.8cm] \nonumber 
\end{align}
Thus, application of expansion (\ref{eq:expansion_malliavin}) for the density of $\widetilde{Y}^{\varepsilon}$, $p^{\widetilde{Y}^{\varepsilon}} (\xi ; \theta) 
= \mathbb{E} \,[\,\delta_\xi (\widetilde{Y}^{\varepsilon}) \,]$, produces the small time expansion of $p_\Delta^{\bar{X}} (x, y; \theta)$. 
Due to starting off from (\ref{eq:muLG})
the rightmost side of (\ref{eq:change_rv})  expands around the density of the local Gaussian scheme (\ref{eq:LG}), that is: 
\vspace{-0.5 \baselineskip}   
\begin{align}
\label{eq:exp2}
    \frac{1}{\sqrt{\Delta^{d_R + 3 d_S}}}\times 
      p^{\widetilde{Y}^{(0)}} (\xi ; \theta) \big|_{\xi = m_{x,\theta, \sqrt{\Delta}}(y)}
=p^{\widetilde{X}^{\mathrm{LG}}}_\Delta (x, y; \theta),
\end{align}
where $\xi \mapsto p^{\widetilde{Y}^{(0)}}(\xi ; \theta)$ denotes the density of the probability law of $\widetilde{Y}^{(0)}$ and 
$y\mapsto p^{\widetilde{X}^{\mathrm{LG}}}_\Delta (x, y; \theta)$ denotes the density of the local Gaussian scheme 
$
\widetilde{X}_\Delta^{\mathrm{LG}, x} 
$
defined in (\ref{eq:LG}) and given as:  
\begin{align} \label{eq:lg_density}
p^{\widetilde{X}^{\mathrm{LG}}}_\Delta (x, y; \theta)  
 = \tfrac{1}{\sqrt{(2 \pi)^d | \Sigma (\Delta, x ; \theta) |}} \exp \Big( - \tfrac{1}{2} \big(y - \mu(\Delta,x;\theta))^\top 
\Sigma^{-1}(\Delta, x; \theta) (y - \mu(\Delta,x;\theta)\big) \Big), 
\end{align}
for mean vector:
\begin{align}
\label{eq:Mean}
\mu(\Delta, x; \theta) = 
\left[  \begin{array}{c} 
    \mu_R(\Delta, x; \beta) \\
    \mu_S(\Delta, x; \theta)
\end{array} \right] 
= \left[ \begin{array}{c}
x_R + V_{R,0} (x ,  \beta)\Delta  \\
 x_S + V_{S,0} (x ,  \gamma)\Delta + \hat{V}_0 V_{S,0} (x , \theta) \tfrac{\Delta^2}{2}  
\end{array} \right], 
\end{align}
and covariance matrix $\Sigma (\Delta, x ; \theta)$ given in (\ref{eq:Var}). 
%
%
%
%

To state the asymptotic expansion of $p^{\bar{X}}_\Delta (x, y ; \theta)$, we introduce a class of Hermite polynomials based on the density of $\widetilde{Y}^{(0)}$. To simplify the notation we henceforth write $\widetilde{Y}=\widetilde{Y}^{(0)}$. Thus, we define: 
\begin{align} \label{eq:hermite_Y}
\mathcal{H}_{\alpha}^{\widetilde{Y}} (\xi; \theta) 
 := (-1)^{{|\alpha|}} \, {\partial^\xi_\alpha \, 
  p^{\widetilde{Y}} (\xi; \theta)}/{p^{\widetilde{Y}} ( \xi; \theta)}, \qquad \xi \in \mathbb{R}^d, \; \theta \in \Theta.
\end{align}
%
Also, we write 
%
$\mathcal{H}_{\alpha} (\Delta, x, y ; \theta) 
= \mathcal{H}_{\alpha}^{\widetilde{Y}} (\xi; \theta) |_{\xi = m_{x, \theta,\sqrt{\Delta}}(y)}, \, 
(\Delta, x, y, \theta) \in (0, \infty) \times 
\mathbb{R}^d \times \mathbb{R}^d \times \Theta. 
$
%
We then have the following result whose proof is provided in Section \ref{appendix:density_scheme} of Supplementary Material. 
\begin{lemma} \label{lemma:AE_scheme} 
Let $x, y \in \mathbb{R}^d$, $\Delta>0$, $\theta \in \Theta$, and assume that conditions 
(\ref{assump:param_space})--(\ref{assump:hypo2}) hold. Then,
for any integer $J \geq 4$, the transition density $y \mapsto p^{\bar{X}}_\Delta (x, y ; \theta)$ admits the following representation: 
\begin{align} \label{eq:density_expansion} 
\begin{aligned} 
p^{\bar{X}}_\Delta (x, y ; \theta)  
=   p^{\widetilde{X}^{\mathrm{LG}}}_\Delta (x, y ; \theta)  \Bigl\{ 1 + \sum_{1 \le l \le J-1} 
  \Delta^{{l}/{2}} \,  
  \Psi_l (\Delta, x, y; \theta)
 \Bigr\}  
+ \frac{\Delta^{{J}/{2}}}{\sqrt{\Delta^{d_R + 3 d_S}}} R^{J} (x, y ; \theta),  
\end{aligned} 
\end{align}
where 
%
%
$R^J (x, y ; \theta)$ is the residual term satisfying 
$\textstyle \sup_{x, y \in \mathbb{R}^d, \, \theta \in \Theta} | R^J (x, y ; \theta) |  < C $ 
for some constant $C > 0$, and $\Psi_l (\Delta, x, y; \theta)$, $l\ge 1$, have the general form:
\begin{align} 
\Psi_l (\Delta, x, y; \theta) = \sum_{1 \le k \le \nu(l)}
\sum_{\alpha \in \{1, \ldots, d\}^k} v_\alpha (x, \theta) \times \mathcal{H}_{\alpha} (\Delta, x, y ; \theta) , 
\end{align}
for some positive integer $\nu(l)$, where $v_\alpha : \mathbb{R}^d \times \Theta \to \mathbb{R}$ is explicitly given as a summation of products of $V_j (x, \theta)$, $0\le j \le d_R$, and their partial derivatives. For the first two $\Psi$-terms in the expansion, we get: 
\begin{align*} 
\Psi_1 &(\Delta, x, y; \theta)
= 
\hspace{-0.2cm} 
\sum_{\substack{1 \le j_1, j_2 \le d_R \\ 
1 \le i_1 \le d}}  
\biggl\{ \hat{V}_{j_1} V_{j_2}^{i_1} (x, \theta) 
\cdot \mathbf{1}_{1 \le i_1 \le d_R}
+ \tfrac{1}{3} \hat{V}_{j_1} \hat{V}_{j_2} V_0^{i_1} (x, \theta) \cdot \mathbf{1}_{d_R + 1 \le i_1 \le d} \biggr\} 
\cdot \widetilde{\Psi}^{i_1}_{j_1 j_2} (\Delta, x, y; \theta),  
\end{align*}
with 
\begin{align*}
& \widetilde{\Psi}^{i_1}_{j_1 j_2} (\Delta, x, y; \theta) 
=  \tfrac{1}{2}  \sum_{1 \le i_2, i_3 \le d_R}  
V_{j_1}^{i_2}(x , \theta) 
V_{j_2}^{i_3}(x , \theta)  \, 
\mathcal{H}_{(i_1,i_2,i_3)} (\Delta, x, y; \theta)  \\ 
&  \quad +  
\sum_{\substack{d_R + 1 \leq  i_2 \leq d \\ 1 \leq i_3 \leq d_R}}  
\Bigl\{ 
\tfrac{1}{3} \hat{V}_{j_1} V_0^{i_2} (x , \theta) 
V_{j_2}^{i_3} (x ,  \theta) 
+ \tfrac{1}{6} 
V_{j_1}^{i_3}(x , \theta) 
\hat{V}_{j_2} V_0^{i_2} (x , \theta) 
\Bigr\} 
\, \mathcal{H}_{(i_1,i_2,i_3)} (\Delta, x, y; \theta) 
\\
%
%
& \quad
+  \tfrac{1}{8}  \sum_{d_R + 1 \le i_2, i_3 \le d} 
\hat{V}_{j_1} V_0^{i_2} (x ,  \theta) 
\hat{V}_{j_2} V_0^{i_3} (x , \theta)
\mathcal{H}_{(i_1, i_2, i_3)}
(\Delta, x, y; \theta). 
\end{align*}
Also, 
$
\Psi_2 (\Delta, x, y; \theta)   = \Phi_{2} (\Delta, x, y; \theta)  
+  \widetilde{\Phi}_{2} (\Delta, x, y; \theta) 
$
with 
\begin{align} \label{eq:Phi_2}
 & \Phi_{2} (\Delta, x, y; \theta) 
= \tfrac{1}{2} \sum_{1 \le i_1, i_2 \le d_R} 
\sum_{1 \le k \le d_R} 
 \Bigl( \hat{V}_{k} V_{0}^{i_1} (x,  \theta)
 + \hat{V}_{0} V_{k}^{i_1} (x, \theta) \Bigr) V_k^{i_2}( x ,  \sigma )   
 \mathcal{H}_{(i_1,i_2)} (\Delta, x, y; \theta)  
\nonumber  \\
& \qquad  +  \tfrac{1}{4} \sum_{1 \le i_1, i_2 \le d_R} \sum_{1 \le k_1, k_2 \le d_R}   
\hat{V}_{k_1} V_{k_2}^{i_1}(x ,  \theta) \hat{V}_{k_1} V_{k_2}^{i_2}(x ,  \theta)  \mathcal{H}_{(i_1,i_2)} (\Delta, x, y; \theta)  \nonumber  \\ 
& \qquad +  \sum_{\substack{1 \leq i_1 \leq d_R \\ d_R + 1 \leq i_2 \leq d}}  \sum_{1 \le k \le d_R} 
  \Bigl( \tfrac{1}{3} \hat{V}_{k} V_{0}^{i_1} (x ,  \theta)
 + \tfrac{1}{6} \hat{V}_{0} V_{k}^{i_1} (x ,  \theta) \Bigr) 
 \hat{V}_k V_0^{i_2}(x , \theta)  
\mathcal{H}_{(i_1,i_2)} (\Delta, x, y; \theta)  \nonumber 
\\
& \qquad  + \tfrac{1}{6} \sum_{\substack{1 \leq i_1 \leq d_R \\ d_R + 1 \leq i_2 \leq d}} 
\sum_{1 \le k \le d_R} V_k^{i_1} (x ,  \theta)
\Bigl(\hat{V}_0 \hat{V}_k V_0^{i_2} (x , \theta) + \hat{V}_k \hat{V}_0 V_0^{i_2} (x ,  \theta) \Bigr)
\mathcal{H}_{(i_1,i_2)} (\Delta, x, y; \theta)   
\nonumber \\
& \qquad  +  \tfrac{1}{6} \sum_{\substack{1 \leq i_1 \leq d_R \\ d_R + 1 \leq  i_2 \leq d}} 
\sum_{1 \le  k_1, k_2  \le d_R} 
\hat{V}_{k_1} V_{k_2}^{i_1}(x, \theta) 
\hat{V}_{k_1} \hat{V}_{k_2} V_0^{i_2}(x, \theta)
\mathcal{H}_{(i_1,i_2)} (\Delta, x, y; \theta)  \nonumber \\
& \qquad + \sum_{d_R + 1 \le i_1, i_2 \le d}  
 \sum_{1 \le k \le d_R} 
 \hat{V}_k V_0^{i_1}  (x, \theta)
  \Bigl( \tfrac{1}{6} \hat{V}_0 \hat{V}_k V_0^{i_2} (x,  \theta) 
 +   \tfrac{1}{8}\hat{V}_k \hat{V}_0 V_0^{i_2} (x,  \theta) \Bigr) 
 \mathcal{H}_{(i_1,i_2)} (\Delta, x, y; \theta) 
\nonumber  \\
& \qquad  + \tfrac{1}{24} \sum_{d_R + 1 \le i_1, i_2 \le d} \sum_{1 \le k_1, k_2 \le d_R} 
  \hat{V}_{k_1} \hat{V}_{k_2} V_0^{i_1}(x,  \theta) \hat{V}_{k_1} \hat{V}_{k_2} V_0^{i_2}(x ,  \theta) 
  \mathcal{H}_{(i_1,i_2)} (\Delta, x, y; \theta);   
\end{align} 
 $\widetilde{\Phi}_2 (\Delta, x, y; \theta)$ 
involves higher order Hermite polynomials and is given in the form of
\begin{align*}
 \sum_{l=4,6} \sum_{\substack{\alpha \in \{1, \ldots, d \}^l}}  w_{\alpha} (x, \theta) \times  \mathcal{H}_{\alpha} (\Delta, x, y; \theta),
\end{align*}
where $w_\alpha : \mathbb{R}^d \times \Theta \to \mathbb{R}$ is of the same structure as the one described above for $v_\alpha$.  
\end{lemma} 
%
%
%
\begin{rem}
The Hermite polynomials have explicit expressions. E.g., for $1 \le i_1, i_2 \le d$:
\begin{align}
& \mathcal{H}_{(i_1)}^{\widetilde{Y}} (\xi; \theta) |_{\xi = m_{x, \theta,\sqrt{\Delta}}(y)}  
 = \sum_{1 \le i \le d}
 {\Sigma}_{i_1 i}^{-1} (1, x; \theta) 
 m^i_{x, \theta,\sqrt{\Delta}}(y); 
 \nonumber \\ 
& \mathcal{H}_{(i_1, i_2)}^{\widetilde{Y}} (\xi; \theta) |_{\xi = m_{x, \theta,\sqrt{\Delta}}(y)}
 = \big\{ \mathcal{H}_{(i_1)}^{\widetilde{Y}} (\xi; \theta)
 \mathcal{H}_{(i_2)}^{\widetilde{Y}} (\xi; \theta)\big\} |_{\xi = m_{x, \theta,\sqrt{\Delta}}(y)}
 -  {\Sigma}_{i_1 i_2}^{-1} (1, x; \theta). \label{eq:hermite}
\end{align} 
%
%
%
%
\end{rem}  
\begin{rem} \label{rem:cf_expn} 
We compare our small time density expansion (\ref{eq:density_expansion}) with some closed-form density expansions that have appeared in the literature.
\begin{itemize}
    \item[(i)]  \cite{ait:02, ait:08} provides Hermite series expansions for the density $p^X_\Delta (x, y ; \theta)$ of 
    elliptic diffusions. 
    However, the domains for $x, y \in \mathbb{R}^d$ should be restricted to a compact set so that the related remainder term is bounded and convergence to the true density can be justified.
    \item[(ii)]  \cite{li:13} develops a closed-form expansion formula for $p^X_\Delta (x, y ; \theta)$ in the multivariate elliptic case using Malliavin-Watanabe calculus that enables one to obtain a remainder term that is uniformly bounded without imposing strong restrictions on the domain of $x, y \in \mathbb{R}^d$. 
    %
    Our approach also uses Malliavin-Watanabe calculus to obtain a closed-form expansion for 
    {$\textstyle{p^{\bar{X}}_\Delta (x, y ; \theta)}$} 
    with controllable remainder terms in a hypo-elliptic setting. 
\end{itemize}
\end{rem} 
\section{Analytic results for statistical inference}
\label{sec:main}
Making use of the  small time density expansion 
in Lemma \ref{lemma:AE_scheme}, 
we provide analytic results about statistical inference procedures under high and low-frequency observation regimes. 
We do not require all correction terms $\textstyle{ \Delta^{{l}/{2}} \Psi_l (\Delta,x,y ; \theta)}$, $l = 1, \ldots$, appearing in the density expansion, and aim at obtaining statistical benefits with use of such terms when necessary. 
Our results below identify the parts of the expansion that lead to parameter estimates of improved performance, for each of the two observation regimes. 
Interestingly, such parts differ between the two regimes.  
%
%
%
%
%
%
%
\subsection{High-frequency observation regime} \label{sec:highfreq}
\subsubsection{Contrast estimator}
We consider the complete observation regime (\ref{eq:data}). Such a data setting has been studied in the recent works focused on hypo-elliptic SDEs of \cite{dit:19, glot:21}, and in numerous earlier studies for the elliptic case, see e.g.~\cite{kessler1997estimation, uchi:12}.
%
We introduce the following notation:
\vspace{-0.5 \baselineskip}  
\begin{align}
\label{eq:S1}
\Sigma_1(x;\theta) &:= \Sigma(1, x; \theta), 
\end{align}   
where recall that $\Sigma$ is the covariance matrix of local Gaussian scheme defined in (\ref{eq:Var}). We consider the likelihood of the complete data in  (\ref{eq:data}). 
We work with the following proxy of $-2\times \textrm{log-likelihood}$, with 
$\mu=\mu(\Delta, x; \theta)$, $\Sigma=\Sigma(\Delta, x; \gamma,\sigma)$, $\Sigma_1=\Sigma_1(x;\gamma,\sigma)$ defined in (\ref{eq:Mean}), (\ref{eq:Var}), (\ref{eq:S1}), respectively: 
%
\begin{align}
\begin{aligned}
\ell_{n, \Delta} (\theta) 
&:=  \sum_{1 \le m \le n}  \big( X_{t_m} - \mu (\Delta, X_{t_{m-1}}; \theta)  \big)^{\top} 
\Sigma^{-1} (\Delta,X_{t_{m-1}} ; \gamma,\sigma) \big( X_{t_m} - \mu (\Delta, X_{t_{m-1}}; \theta)  \big) \\
& \qquad + \sum_{1 \le m \le n} \log |\Sigma_1(X_{t_{m-1}}; \gamma,\sigma) | 
 - 2 \Delta \sum_{1 \le m \le n}   \Phi_2 (\Delta, \sample{X}{m-1}, \sample{X}{m}; \theta),
 \end{aligned}
 \label{eq:contrast}
\end{align} 
where $\Phi_2$ is given in \eqref{eq:Phi_2}, namely, $\Phi_2$ is a term with second order Hermite polynomials appearing in the expression for $\Psi_2$ in the density expansion \eqref{eq:density_expansion}. 
We discuss the effect of the correction term $\Phi_2$ in the obtained CLT in Remark \ref{rem:vs_LocalGaussian} later in the paper. In the elliptic case, i.e.~$d=d_R$, $d_\gamma = 0$, the corresponding contrast function is as in (\ref{eq:contrast}) but with functions $\mu$, $\Sigma_1$, $\Sigma^{-1}$,  $\Phi_2$ replaced by $\mu_R$, $a_R$, $a_R^{-1}/\Delta$, $\Phi_{e, 2}$, respectively.    
%
We define the contrast estimator 
%
$\textstyle{ \hat{\theta}_n: = (\hat{\beta}_{n}, \hat{\gamma}_{n}, \hat{\sigma}_{n})  = \arg\min_{\theta\in\Theta} \ell_{n, \Delta} (\theta)}
$
%
%
and write the estimator as 
$\textstyle{\hat{\theta}_{e,n} = (\hat{\beta}_{n}, \hat{\sigma}_{n})}$ in the setting of elliptic diffusions.
\begin{rem} 
We state a property of term $\Phi_2$ defined in (\ref{eq:Phi_2}). For simplicity, we write 
$ 
\textstyle 
\Phi_2 (\Delta, x, y; \theta) 
= \sum_{1 \le i_1, i_2 \le d} G_{i_1 i_2} (x; \theta) \times \mathcal{H}_{(i_1, i_2)} (\Delta, x, y; \theta), \, x, y \in \mathbb{R}^d 
$ 
for appropriate functions $G_{i_1i_2}$ obtain via (\ref{eq:Phi_2}).
Terms $G_{i_1i_2}(x;\theta)$ also arise when considering the covariance of the normalised vector 
$m (\Delta, x, y; \theta) = m_{x, \theta, \sqrt{\Delta}} (y)$. 
%
That is, one obtains: for $1 \le i_1, i_2 \le d$, 
%
%
\begin{align} 
\mathbb{E}_\theta 
&\bigl[
m^{i_1} (\Delta, x, X_\Delta^x; \theta) 
\, 
m^{i_2} (\Delta, x, X_\Delta^x; \theta)  
\bigr] 
= \Sigma_{1, i_1 i_2} (x; \theta) 
+ 2 \Delta \times G_{i_1 i_2} (x; \theta) 
+ \widetilde{G}_{i_1 i_2} (\Delta^2, x; \theta),
\label{eq:expansion_cov}
\end{align}
for $\widetilde{G}_{i_1 i_2} (\cdot, \cdot; \theta): [0, \infty) \times  \mathbb{R}^d \to \mathbb{R}$ such that 
$|\widetilde{G}_{i_1 i_2} (h, x; \theta)| \le C h $ for some $C >0$, and $\widetilde{G}_{i_1 i_2} (h, \cdot; \theta) \in C_b^\infty (\mathbb{R}^d ; \mathbb{R})$ for all $(h, \theta) \in [0, \infty) \times \Theta$ under (\ref{assump:param_space})--(\ref{assump:coeff}). 
%
%
The above connection between $\Phi_2$ and the covariance expression (\ref{eq:expansion_cov}) plays a key role to obtain a CLT under the weaker condition $\Delta_n = o (n^{-1/3})$.
\end{rem}
%
%
\subsubsection{Conditions for high-frequency regime}
To state the main result, we introduce the following additional conditions.
\vspace{-0.4 \baselineskip}  
\begin{enumerate}
\setcounter{enumi}{4}
\renewcommand{\theenumi}{H\arabic{enumi}}
\renewcommand{\labelenumi}{(\theenumi)} 
\item  \label{assump:bd_deriv}
For each $x \in \mathbb{R}^d$, $\theta \mapsto V_j (x ,  \theta)$ is three times differentiable. 
For any $\alpha \in \{1, \ldots,  d_\theta\}^l$, with $l\in\{1,2\}$, the functions  
$
x \mapsto  
\partial^{\theta}_{\alpha}
V_j^i (x ,  \theta), 
\;  0 \le j \le d_R, \; 1 \le i \le d, 
$
have bounded derivatives of every order uniformly in $\theta \in \Theta$. 
\item \label{assump:lipschitz} 
For $1  \le  i \le d_S$, there exists a function 
$B^i\in C_{b}^{\infty} (\mathbb{R}^{d} ; \mathbb{R})$ 
such that for any $x \in \mathbb{R}^d$ and $\gamma', \gamma \in \Theta_\gamma$, 
\begin{align*}
     | V_{S,0}^{i} (x,  \gamma') - V_{S,0}^{i} (x, \gamma) | \leq B^i (x) |\gamma' - \gamma|.
 \end{align*}
%
%
\item \label{assump:moments}
The diffusion process $\{X_t\}_{t \geq 0}$ in (\ref{eq:model}) is ergodic under $\theta = \trueparam$, with  invariant distribution $\truedist$ on $\mathbb{R}^d$. Furthermore, all moments of $\truedist$  are finite. 
\item \label{assump:finite_moment}
It holds that for all $p \ge 1$, 
$\textstyle \sup_{t > 0} \mathbb{E}_{\trueparam} [|X_t|^p] < \infty$. 
\item \label{assump:ident}
The true parameters lie in the interior of $\Theta$. If it holds 
\begin{align*}
V_{R, 0} (x , \beta) = V_{R, 0} (x , \truebeta), \ \  
V_{S, 0} (x , \gamma) = V_{S, 0} (x , \truegamma), \ \ 
V (x , \sigma) = V (x , \truesigma), 
\end{align*}
for $x$ in a set of probability 1 under  $\nu_{\trueparam}$, then 
$\beta = \truebeta, \,  \gamma = \truegamma, \, \sigma = \truesigma$. 
\end{enumerate} 
\subsubsection{Asymptotic properties of the contrast estimator}
We can now prove that the estimator $\hat{\theta}_n$ has the following asymptotic properties.  
\begin{theorem}[Consistency] 
\label{consistency}
Under conditions (\ref{assump:param_space})--(\ref{assump:ident}), it holds that if $n \to \infty$, $\Delta_n \to 0$ and $n \Delta_n \to \infty$, then
$ 
\hat{\theta}_{n}  \xrightarrow{\mathbb{P}_{\theta^\dagger}} \theta^\dagger.
$
\end{theorem}  
%
%
%
\begin{theorem}[Asymptotic normality] \label{thm:asymp_norm}
Under conditions (\ref{assump:param_space})--(\ref{assump:ident}), it holds that if $n \to \infty$, $\Delta_n \to 0$, $n \Delta_n \to \infty$ and $\Delta_n=o(n^{-1/3})$, then:  
\begin{enumerate}
\item[(a)] if $d = d_R$, i.e.~in the setting of elliptic diffusions, 
\begin{align*} 
%
\Bigl[
 \sqrt{n \Delta}\, (\hat{\beta}_n - \beta^\dagger)^\top, 
 \sqrt{n} \,(\hat{\sigma}_n - \sigma^\dagger)^\top
\Bigr]^\top 
%
\xrightarrow{\mathcal{L}_{\theta^\dagger}} \mathscr{N} \bigl( 0, I_e (\theta^\dagger)^{-1} \bigr),
\end{align*}
\item[(b)] if $d > d_R$, i.e.~in the setting of hypo-elliptic diffusions,
\begin{align*} 
\Bigl[ \sqrt{n \Delta} \,(\hat{\beta}_n - \beta^\dagger)^\top, \,
 \sqrt{\tfrac{n}{\Delta}} \,(\hat{\gamma}_n - \gamma^\dagger)^\top, \, 
 \sqrt{n}\, (\hat{\sigma}_n - \sigma^\dagger)^\top 
\Bigr]^\top 
 \xrightarrow{\mathcal{L}_{\theta^\dagger}} \mathscr{N}\bigl( 0, I (\theta^\dagger)^{-1} \bigr). 
\end{align*}
\end{enumerate}
$
I_e (\trueparam) = \mrm{diag} \bigl[ 
\big( I_{e, ij}^{\, \beta}(\trueparam) \big)_{1 \leq i, j \leq d_{\beta}}, 
\bigl( I_{e, ij}^{\, \sigma}(\trueparam) \bigr)_{1 \leq i, j \leq d_{\sigma}} \bigr]
$ is the asymptotic precision matrix 
%
with block matrix elements:
\begin{align*}
& I_{e, ij}^{\, \beta}(\trueparam) 
  =  \int 
   \big(\, \partial_{\beta_i}  V_{R,0} (x,  \truebeta)^{\top}
  a_R^{-1} (x,  \sigma^\dagger) 
  \, \partial_{\beta_j}  V_{R,0} (x, \truebeta) \,\big) 
  \,\nu_{\trueparam} (dx); \\
 & I_{e, ij}^{\sigma}(\trueparam)   
 =\tfrac{1}{2}  \int
  \mathrm{tr} \bigl(\,
  \partial_{\sigma_i} a_R (x, \sigma^\dagger) \, 
  a_R^{-1} (x,  \sigma^\dagger)
  \,\partial_{\sigma_j} a_R (x,  \sigma^\dagger) 
  \,a_R^{-1} (x,  \sigma^\dagger) \,\bigr)
  \,\nu_{\trueparam} (dx). 
\end{align*}
Similarly, the asymptotic precision matrix $I(\trueparam)$ has the block-diagonal structure:
\begin{align}
\label{eq:precision}
I^{}(\trueparam) 
= 
\mrm{diag} \bigl[
\big( I_{ij}^{\, \beta}(\trueparam) \big)_{1 \leq i, j \leq d_{\beta}}, 
\big( I_{ij}^{\, \gamma}(\trueparam) \big)_{1 \leq i, j \leq d_{\gamma}}, 
\bigl( I_{ij}^{\, \sigma}(\trueparam) \bigr)_{1 \leq i, j \leq d_{\sigma}} 
\bigr], 
\end{align}
with block matrix elements:
\begin{align*}
 &I_{ij}^{\, \beta}(\trueparam) 
  =  \int 
   \big(\,
   \partial_{\beta_{i}} V_{R,0} (x , \truebeta)^{\top}
  {\Sigma}_{1, RR}^{-1} (x; \sigma^\dagger) \, 
  \partial_{\beta_{j}} V_{R,0} (x , \truebeta) \,\big)\,
  \nu_{\trueparam} (dx); \\
 & I_{ij}^{\, \gamma}(\trueparam) 
  = 4 \int 
   \big(\, 
   \partial_{\gamma_{i}} V_{S,0} (x , \truegamma)^{\top}
   \, 
   {\Sigma}_{1, SS}^{-1} (x; \gamma^\dagger, \sigma^\dagger)
   \,
   \partial_{\gamma_{j}} V_{S,0} (x , \truegamma) 
   \,\big)
   \,\nu_{\trueparam} (dx); \\
 &I_{ij}^{\sigma}(\trueparam)   =\tfrac{1}{2}  \int
  \mathrm{tr} \bigl(\,\partial_{\sigma_i}{\Sigma}_1 (x;  \gamma^\dagger, \sigma^\dagger)\, {\Sigma}_1^{-1}(x;   \gamma^\dagger, \sigma^\dagger)  \,\partial_{\sigma_j}{\Sigma}_1 (x;  \gamma^\dagger, \sigma^\dagger) \,{\Sigma}_1^{-1}(x;  \gamma^\dagger, \sigma^\dagger)\,\bigr)\,\nu_{\trueparam} (dx).
\end{align*} 
\end{theorem}
\noindent The proofs of Theorems \ref{consistency} and \ref{thm:asymp_norm} are given in Sections \ref{pf_consistency} and \ref{sec:pf_asymp_norm}, respectively, in Supplementary Material.
\begin{rem} \label{rem:vs_LocalGaussian}
\cite{glot:21} obtain an MLE using the transition density  of the local Gaussian scheme in (\ref{eq:LG}), that is 
$
\widetilde{\theta}_n 
= \mathrm{argmin}_{\theta \in \Theta} \widetilde{\ell}_{n ,\Delta} (\theta) 
$
where $\widetilde{\ell}_{n ,\Delta} (\theta)$ is given by~(\ref{eq:contrast}) without the last term $\textstyle -2 \Delta  \sum_{m = 1}^n \Phi_2 (\Delta, \sample{X}{m-1}, \sample{X}{m} ; \theta)$. 
%
%
They prove asymptotic normality for  $\widetilde{\theta}_n$ under the conditions $n \to \infty$, $ \Delta_{n} \to 0$, $n \Delta_n \to \infty$ and $\Delta_n= o(n^{-1/2})$.  
The asymptotic precision they obtain is identical to the one we find here in  (\ref{eq:precision}). Thus, addition of the term $\textstyle{- 2 \Delta  \sum_{m=1}^n  \Phi_2 (\Delta, \sample{X}{m-1}, \sample{X}{m}; \theta})$ allows for a weaker experimental design condition under which a CLT holds, even if it does not alter the asymptotic variance. 
%
\begin{rem} \label{rem:score}
We explain that the proposed estimator $\hat{\theta}_{n}$ contributes to an improvement mainly in the estimation of diffusion parameter $\sigma$ when compared with the existing estimator $\widetilde{\theta}_n$. This is clarified by observing the role of the weaker condition $\Delta_n = o(n^{-1/3})$ for the CLT. 
To prove the CLT, we show:
\begin{align}  \label{eq:score_conv}
\sum_{1 \le m  \le n} \mathbb{E}_{\trueparam} [ s (\Delta_n, \sample{X}{m-1}, \sample{X}{m}; \trueparam)  | \mathcal{F}_{t_{m-1}}] \probconv 0,
\\[-0.8cm] \nonumber 
\end{align}
as \limit, with the aid of $\Delta_n = o (n^{-1/3})$, where $s=s(\cdot;\cdot)$ is determined from the scaled score function as: 
$$ 
\bigl[
    \tfrac{1}{\sqrt{n \Delta_n}} 
    \partial_\beta^\top \, 
    \ell_{n, \Delta} (\trueparam), \, 
    \sqrt{\tfrac{\Delta_n}{n}} \partial_\gamma^\top \, 
    \ell_{n, \Delta} (\trueparam), \,  
    \tfrac{1}{\sqrt{n}} \partial_\sigma \, 
    \ell_{n, \Delta} (\trueparam) 
\bigr]^\top  
\equiv \sum_{1 \le m \le n} s (\Delta_n, \sample{X}{m-1}, \sample{X}{m}; \trueparam).  
$$
Since the contrast function (\ref{eq:contrast}) splits into a log-Gaussian part (related to the local Gaussian scheme) and a correction part (involving $\Phi_2$), we write the score function as: 
$
s (\Delta_n, \sample{X}{m-1}, \sample{X}{m-1}; \trueparam) \equiv 
s_{\mrm{LG}} (\Delta_n, \sample{X}{m-1}, \sample{X}{m-1}; \trueparam) 
+ s_{\Phi_2} (\Delta_n, \sample{X}{m-1}, \sample{X}{m-1}; \trueparam).  
$
%
Then, the expansion formula (\ref{eq:expansion_cov}) yields: 
\begin{align}
& \sum_{1 \le m \le n} \mathbb{E}_{\trueparam} 
[ 
{s}^{\, i_1}_{\mrm{LG}} (\Delta_n, \sample{X}{m-1}, \sample{X}{m}; \trueparam) | \mathcal{F}_{t_{m-1}} 
]
= \tfrac{1}{n} \sum_{1 \le m \le n} {g}^{i_1} \bigl(
\sqrt{n \Delta_n^3}, \sample{X}{m-1}; 
\trueparam \bigr), \quad   
1 \le i_1 \le d_{\beta} + d_\gamma;  \nonumber \\ 
& \sum_{1 \le m \le n} \mathbb{E}_{\trueparam} 
[ 
{s}^{\, i_2}_{\mrm{LG}} (\Delta_n, \sample{X}{m-1}, \sample{X}{m}; \trueparam) | \mathcal{F}_{t_{m-1}} 
]  = \tfrac{\sqrt{n \Delta_n^2}}{n} \sum_{1 \le m \le n}  
\mrm{tr} \bigl( 2 \cdot G (\sample{X}{m-1}; \trueparam) \partial_{\theta_{i_2}} \Sigma_{1}^{-1} (\sample{X}{m-1}; \truegamma, \truesigma)  \bigr) \nonumber \\ 
& \qquad + \tfrac{1}{n} \sum_{1 \le m \le n} 
{g}^{i_2} \bigl(
\sqrt{n \Delta_n^4}, \sample{X}{m-1}; 
\trueparam \bigr), \quad  d_{\beta} + d_\gamma + 1 \le i_2 \le d_\theta, \label{eq:score_lg} 
\end{align}
where $g^{i} (\cdot, \cdot; \trueparam): [0,\infty) \times \mathbb{R}^d \to \mathbb{R}, \, 1 \le i \le d,$
has the same regularity property as $\widetilde{G}_{i_1 i_2} (\cdot, \cdot; \trueparam)$ 
in (\ref{eq:expansion_cov}). The above two equations correspond to the score functions of the drift parameters $(\beta, \gamma) \in \Theta_\beta \times \Theta_\gamma$ and the diffusion parameter $\sigma \in \Theta_\sigma$, respectively. Also for the correction part, we show that 
\begin{align*}
& \sum_{1 \le m \le n} \mathbb{E}_{\trueparam} 
[ 
{s}^{\, i_1}_{\Phi_2} (\Delta_n, \sample{X}{m-1}, \sample{X}{m}; \trueparam) | \mathcal{F}_{t_{m-1}} 
]
= \tfrac{1}{n} \sum_{1 \le m \le n} \widetilde{g}^{i_1} \bigl(
\sqrt{n \Delta_n^3}, \sample{X}{m-1}; 
\trueparam \bigr), \quad 
1 \le i_1 \le d_{\beta} + d_\gamma; \nonumber \\
& \sum_{1 \le m \le n} \mathbb{E}_{\trueparam} 
[ 
{s}^{\, i_2}_{\Phi_2} (\Delta_n, \sample{X}{m-1}, \sample{X}{m}; \trueparam) | \mathcal{F}_{t_{m-1}} 
]  = - \tfrac{\sqrt{n \Delta_n^2}}{n} 
\hspace{-0.2cm} \sum_{1 \le m \le n} 
\mrm{tr} \bigl( 2 \cdot G (\sample{X}{m-1}; \trueparam) \partial_{\theta_{i_2}} \Sigma_{1}^{-1} (\sample{X}{m-1}; \truegamma, \truesigma)  \bigr) \nonumber \\ 
& \qquad + \tfrac{1}{n} \sum_{1 \le m \le n} \widetilde{g}^{i_2} \bigl(\sqrt{n \Delta_n^4}, \sample{X}{m-1}; 
\trueparam \bigr), \quad  d_{\beta} + d_\gamma + 1 \le i_2 \le d_\theta,  
\end{align*} 
where $\widetilde{g}^i$ has the same structure as $g^i$. Thus, inclusion of the $\Phi_2$-term in our contrast function $\ell_{n, \Delta} (\theta)$ in  \eqref{eq:contrast}, results in the cancellation of a quantity corresponding to the 1st term of size $\textstyle{\mathcal{O} (\sqrt{n \Delta_n^2})}$ on the right side of (\ref{eq:score_lg}), and then convergence (\ref{eq:score_conv}) holds under $\Delta_n = o (n^{-1/3})$. For details, see the proof of Lemma \ref{lemma:conv_normal} in Section \ref{appendix:pf_conv_normal} in Supplementary Material. Note that the required condition for $\Delta_n$ is weakened for the score w.r.t. diffusion parameter, while the condition for the score w.r.t. drift parameters remains the same after inclusion of the correction part (see $g^{i_1}$ and $\widetilde{g}^{i_1}$). 
Thus, it is expected that the new estimator performs better than the existing estimator $\widetilde{\theta}_n$ in the estimation of diffusion parameter.
\end{rem} 
\end{rem} 
\subsubsection{Contrast estimator for stochastic damping Hamiltonian systems}
\label{sec:stochastic_damping}
We focus on an important sub-class of hypo-elliptic diffusions used in applications, namely the \emph{stochastic damping Hamiltonian systems}, and write down in detail the form of the proposed contrast function for such a family of models. 
An SDE in the above class can be written as follows: 
%
\begin{align}
\begin{aligned} \label{eq:sdhs} 
d X_{R,t} & = - (c_\beta (X_{S,t}) X_{R,t} + g_\beta (X_{S,t}))
dt  + \Xi (\sigma)  \, d B_t, \quad  X_{R, 0} = x_{R, 0} \in \mathbb{R}^{\bar{d}};  \\ 
d X_{S, t} & = X_{R, t} dt, \quad X_{S, 0} = x_{S, 0} \in \mathbb{R}^{\bar{d}},  
\end{aligned}
\end{align}
where $d_S = d_R = \bar{d} \geq 1$, $g_\beta: \mathbb{R}^{\bar{d}} \to \mathbb{R}^{\bar{d}}$, $c_\beta: \mathbb{R}^{\bar{d}} \to \mathbb{R}^{\bar{d} \times \bar{d}}$ are smooth functions depending on some parameter $\beta \in \Theta_\beta$, and $\Xi (\sigma) = \mrm{diag} (\sigma)$ for $\sigma = (\sigma_1, \ldots, \sigma_{\bar{d}}) \in \Theta_\sigma$ with $\sigma_j > 0, \, 1 \le j \le \bar{d}$. 
System (\ref{eq:sdhs}) includes many models used in applications, including, e.g., the Langevin equation \citep{pavl:14}.
%

For $X_t = [X_{R,t}^\top, X_{S,t}^\top]^\top \in \mathbb{R}^{2 \bar{d}}$ given by (\ref{eq:sdhs}), the local Gaussian density (\ref{eq:lg_density}) is specified via the following mean vector $\mu (\Delta, x; \theta) \in \mathbb{R}^{2 \bar{d}}$ and covariance matrix $\Sigma (\Delta, x ; \theta) \in \mathbb{R}^{
2 \bar{d} \times 2 \bar{d}}$: 
\begin{align*} 
\mu (\Delta, x; \theta) 
= 
\begin{bmatrix}
 x_R - (c_\beta (x_S) x_R + g_\beta (x_S)) \cdot \Delta \\
 x_S + x_R \cdot \Delta - (c_\beta (x_S) x_R + g_\beta (x_S)) \cdot {\Delta^2}/{2}
\end{bmatrix}, \ \  
\Sigma (\Delta, x; \theta) 
=
\begin{bmatrix}
  \Delta &  {\Delta^2}/{2} \\
  {\Delta^2}/{2} & {\Delta^3}/{3}
\end{bmatrix}
\otimes \Xi (\sigma) \Xi (\sigma)^\top. 
\end{align*}
The diffusion matrix $\Xi$ is independent of the state $x \in \mathbb{R}^{2 \bar{d}}$ and the drift function is linear 
in~$X_{R,t}$ given $X_{S,t}$. Thus, most of the terms in the definition of $\Phi_2$ appearing in the contrast (\ref{eq:contrast}) are equal to $0$, e.g., 
%
$\textstyle 
\hat{V}_{j_1} V_{k_1}^{i_1} (x, \theta) = 0, \    
\hat{V}_{j_2} \hat{V}_{k_2} V_0^{i_2} (x, \theta) = 0$ for  
$ 1 \le i_1, k_1, k_2 \le \bar{d}, \,  
0 \le j_1, j_2 \le \bar{d}, \, \bar{d} + 1 \le i_2 \le 2 \bar{d}. 
$
%
%
Thus, the correction term $\Phi_2$ in the case of system (\ref{eq:sdhs}) takes up a simple form. For $x = (x_R, x_S), y \in \mathbb{R}^{2 \bar{d}}, \Delta > 0$ and  $\theta = (\beta, \sigma) \in \Theta$, we get:
\begin{align} \label{eq:Phi_2_sdhs}
\Phi_2 (\Delta, x, y ; \theta) 
& = - \sum_{1 \le i_1, i_2 \le \bar{d}}  c_{\beta, i_1 i_2} (x_S) \times (\sigma_{i_2})^2 
 \times 
 \Bigl\{ \tfrac{1}{2} \mathcal{H}_{(i_1 i_2)} (\Delta, x, y; \theta) 
+ \tfrac{1}{3} \mathcal{H}_{(i_1 i_2 + \bar{d})} 
(\Delta, x, y; \theta) \nonumber  \\ 
& \qquad \qquad 
+ \tfrac{1}{6} \mathcal{H}_{(i_1 + \bar{d}, i_2)} 
(\Delta, x, y; \theta)  
+ \tfrac{1}{8} \mathcal{H}_{(i_1 + \bar{d}, i_2 + \bar{d})}(\Delta, x, y; \theta) 
\Bigr\},  
\end{align} 
%
%
where $\textstyle \mathcal{H}_{(i_1 i_2)} (\Delta, x, y; \theta)$ is the second order Hermite polynomial defined in (\ref{eq:hermite}). 
%
%
\subsection{Low-frequency observation regime} \label{sec:lowfreq}
We consider the scenario of low-frequency observations, i.e.~the time interval among successive observations, $\Delta = t_{i} - t_{i-1}$, $1\le i \le n$, is now assumed fixed, and large enough so that approximation schemes must be combined with a \emph{data augmentation} approach. Thus, we move onto a Bayesian inference setting. 
Let $\delta = \delta_M := \Delta / M, \;  M\ge 1,$
be the user-induced step-size after imputation of $(M-1)$ data points amongst a pair of observations.
Given a transition density scheme, say 
$\bar{p}_{\delta} (x,y;\theta)$, 
the true (intractable) transition density is approximated as:
\begin{align*}
p_{\Delta}^X (X_{t_{i-1}}, X_{t_i}; \theta) 
& \approx 
\int_{\mathbb{R}^{d \times (M-1)}}
\big\{\prod_{1 \le j \le M} 
\bar{p}_{\delta} (x_{j-1},x_{j};\theta) 
\big\} \, dx_{1}\cdots dx_{M-1}, 
\quad  x_0= \sample{X}{i-1},
\; x_M = \sample{X}{i}.
%
\end{align*} 
%
%
In the case of elliptic diffusions, i.e., $d_S = 0, \, d_R = d$, \cite{gobet:08} showed that the bias induced by the Euler-Maruyama (EM) scheme is of size $\mathcal{O}(M^{-1})$. 
%

In Section \ref{sec:density_1}, \ref{sec:density_2} we develop two-types of explicit transition density schemes achieving local weak third order convergence via appropriate choice of higher order correction terms in the density expansion formula (\ref{eq:density_expansion}). 
We then illustrate in Section \ref{sec:disc_density} that the discretisation bias of the developed transition schemes is of size $\mathcal{O} (M^{-2})$, for the class of elliptic SDEs.   
\subsubsection{Local weak third order transition density scheme -- version I} \label{sec:density_1}
Making use of the small time expansion of the transition density $y \mapsto p^{\bar{X}}_\Delta (x, y ; \theta)$ in Lemma \ref{lemma:AE_scheme}, we obtain 
the following key result whose proof is given in Section \ref{appendix:local_wa_ver1} of Supplementary Material. 

%
%
\begin{prop}[Density expansion for local third order weak approximation] \label{prop:density_scheme}
Let $x, y \in \mathbb{R}^d$, $\Delta>0$, $\theta \in \Theta$, and assume that conditions 
(\ref{assump:param_space})--(\ref{assump:hypo2}) hold. Then, for any integer $J \geq 4$, the transition density $y \mapsto p^{\bar{X}}_\Delta (x, y ; \theta)$, admits the following representation: 
\vspace{-0.4 \baselineskip}  
\begin{align*} 
\begin{aligned} 
p^{\bar{X}}_\Delta (x, y ; \theta)  
  =   p^{\widetilde{X}^{\mathrm{LG}}}_\Delta (x, y ; \theta)  \bigl\{ 1 + \Psi^{\mrm{weak}} (\Delta, x, y; \theta)
+  R_1^J (\Delta, x, y ; \theta) \bigr\}  
   + \frac{\Delta^{{J}/{2}}}{\sqrt{\Delta^{d_R + 3 d_S}}}  R_2^J (x, y ; \theta),  
\end{aligned} 
\end{align*}
with residual terms $R_1^J (\Delta, x, y ; \theta)$ and $R_2^J (x, y ; \theta)$. In particular:
\begin{itemize}[leftmargin=0.2cm]
\item[(i)] We have the expression:
\begin{align}
\label{eq:Phi1}
\Psi^{\mrm{weak}} (\Delta, x, y; \theta) =  \sqrt{\Delta} \cdot  \Psi_1^{\mrm{weak}} (\Delta, x, y; \theta) + 
\Delta \cdot \Psi_2^{\mrm{weak}} (\Delta, x, y; \theta) 
+ \sqrt{\Delta^3} \cdot \Psi_3^{\mrm{weak}} (\Delta, x, y; \theta),
\end{align}
for the individual terms: \\[-0.6cm]
\begin{align}
\Psi_1^{\mrm{weak}} (\Delta, x, y; \theta) & =  \tfrac{1}{2}  \sum_{1 \le i_1, i_2, i_3 \le d_R} \sum_{1 \le k_1, k_2 \le d_R} \hspace{-0.4cm} 
   \hat{V}_{k_1} V_{k_2}^{i_1} (x , \theta) 
   V_{k_1}^{i_2}(x , \theta) V_{k_2}^{i_3}(x , \theta) \, \mathcal{H}_{(i_1,i_2,i_3)} (\Delta, x, y; \theta); 
   \label{eq:psi1_w} \\
\Psi_2^{\mrm{weak}} (\Delta, x, y; \theta) 
&  =  \tfrac{1}{2} \sum_{1 \le i_1, i_2 \le d_R} \sum_{1 \le k \le d_R} 
  \Bigl( \hat{V}_{k} V_{0}^{i_1} (x,  \theta)
 + \hat{V}_{0} V_{k}^{i_1} (x, \theta) \Bigr) V_k^{i_2}( x ,  \theta)  \mathcal{H}_{(i_1,i_2)} (\Delta, x, y; \theta) \nonumber \\
& \qquad  +  \tfrac{1}{4} \sum_{1 \le i_1, i_2 \le d_R} \sum_{1 \le k_1, k_2 \le d_R}   
\hat{V}_{k_1} V_{k_2}^{i_1}(x ,  \theta) \hat{V}_{k_1} V_{k_2}^{i_2}(x ,  \theta) \mathcal{H}_{(i_1,i_2)} (\Delta, x, y; \theta);  \label{eq:psi2_w} 
\\
\Psi_3^{\mrm{weak}} (\Delta, x, y; \theta)  
& = \tfrac{1}{2}  \sum_{1 \le i \le d_R} \hat{V}_{0} V_{0}^{i}(x,  \theta)  
  \mathcal{H}_{(i)} (\Delta, x, y; \theta).
\label{eq:psi3_w}
\end{align}
\item[(ii)] For any $\varphi \in C_p^{\infty} (\mathbb{R}^d ; \mathbb{R})$, there exist constants $C > 0$, $q \ge 1$ such that: 
\begin{align} \label{eq:weak_remainder}
\Bigl| \, \int_{\mathbb{R}^d} 
 \varphi (y) p_\Delta^{\widetilde{X}^{\mrm{LG}}} (x, y; \theta) R_1^J (\Delta, x, y ; \theta)
 dy \, \Bigr| \le C (1 + |x|^q) \Delta^3. 
\end{align}
\item[(iii)] $\textstyle{\sup_{x, y \in \mathbb{R}^d,  \theta \in \Theta}} | R_2^J (x, y ; \theta)|<C$, for a constant $C>0$.
\item[(iv)]   $\textstyle{ \int_{\mathbb{R}^d}  \Psi^{\mrm{weak}} (\Delta, x, y; \theta)  p^{\widetilde{X}^{\mathrm{LG}}}_\Delta (x, y ; \theta)   dy = 0}$. 
\end{itemize} 
Thus, from (ii) and (iii) with $J \ge 6$, we have that for any $\varphi \in C_p^{\infty} (\mathbb{R}^d ; \mathbb{R})$ and $(\Delta, x, \theta) \in (0, \infty) \times \mathbb{R}^d \times \Theta$, there exist constants $C>0$, $q \ge 1$ such that: 
\begin{align*} 
  \Big|\,
 \mathbb{E}_{\theta} [\varphi (\bar{X}_{\Delta}^x)]  -   \int_{\mathbb{R}^d} \varphi (y)  p^{\widetilde{X}^{\mathrm{LG}}}_\Delta (x, y ; \theta)  
 \bigl\{ 1 + \Psi^{\mrm{weak}} (\Delta, x, y; \theta) \bigr\}  dy \,\Big| 
\; \leq \; C (1 + |x|^q) \Delta^3. 
\end{align*} 
Using Proposition \ref{prop:weak_app_hypo}, the above gives immediately that, for any $\varphi \in C_p^{\infty} (\mathbb{R}^d ; \mathbb{R})$ and $(\Delta, x, \theta) \in (0, \infty) \times \mathbb{R}^d \times \Theta$, there exist constants $C>0$, $q\ge 1$ so that: 
\begin{align*}
\begin{aligned}
 \Big|\, \mathbb{E}_{\theta} [\varphi (X_{\Delta}^x)] 
 - \int_{\mathbb{R}^d} \varphi (y) 
 \, p^{\widetilde{X}^{\mathrm{LG}}}_\Delta  (x, y ; \theta) \bigl\{ 1 + \Psi^{\mrm{weak}} (\Delta, x, y; \theta) \bigr\} dy \, \Big|  
\; \leq \; C (1 + |x|^q) \Delta^3. 
\end{aligned}
\end{align*} 
\end{prop} 
In the case of elliptic diffusions ($d = d_R$, $d_S = 0$), Proposition \ref{prop:density_scheme} is interpreted as follows: 
\begin{align*} 
\begin{aligned}
  p^{\bar{X}}_{e,\Delta} (x,  y ; \theta) 
= p^{\widetilde{X}^{\mathrm{EM}}}_{e, \Delta} (x,  y ; \theta)\, 
 \big\{ 1 + \Psi_e^{\mrm{weak}} (\Delta, x, y; \theta)  
   +  R_{e,1}^J (\Delta, x, y ; \theta)    \bigr\}
   + \frac{\Delta^{{J}/{2}}}{\sqrt{\Delta^{d_R}}} R_{e,2}^{J} (x, y ; \theta),  
\end{aligned}
\end{align*}
where $y \mapsto p^{\widetilde{X}^{\mathrm{EM}}}_{e,\Delta} (x, y ; \theta)$
denotes the density of the one-step Euler-Maruyama scheme given $\theta \in \Theta$ and the start point $x \in \mathbb{R}^{d_R}$, and terms $R_{e,1}^J (\Delta, x, y ; \theta)$, $R_{e,2}^J (x, y ; \theta)$ have the same properties as the corresponding ones in Proposition~\ref{prop:density_scheme}. 
The correction part $\Psi_e^{\mrm{weak}} (\Delta, x, y ; \theta)$ 
corresponds to $\Psi^{\mrm{weak}} (\Delta, x, y ; \theta)$ with $d = d_R$ and the Hermite polynomials replaced with  $\textstyle{\mathcal{H}_{\alpha}^{\widetilde{Y}_{e}} (\xi;\theta)}|_{\xi = \textstyle{m_{x, \theta,\sqrt{\Delta}}(y)}}$ whose definition has been adjusted in an apparent way to conform to the setting $d_S=0$.  

In light of Proposition \ref{prop:density_scheme}, 
we propose the use of the following transition density scheme for the hypo-elliptic case:
\begin{align} 
\label{eq:density_formula}
\bar{p}^{\,\mathrm{I}}_\Delta (x, y; \theta) := p^{\widetilde{X}^{\mathrm{LG}}}_\Delta (x, y ; \theta) 
 \big\{ 1 + \Psi^{\mrm{weak}} (\Delta, x, y; \theta) \big\},
\end{align}  
for $\Psi^{\mrm{weak}} (\Delta, x, y; \theta)$ as given in (\ref{eq:Phi1}),
with the following corresponding scheme for the elliptic case:
\begin{align}
\label{eq:density_formula_e}
\bar{p}^{\,\mathrm{I}}_{e,\Delta} (x, y; \theta) := 
p^{\widetilde{X}^{\mathrm{EM}}}_{e,\Delta} (x, y ; \theta) 
\big\{ 1 + \Psi_e^{\mrm{weak}} (\Delta, x, y; \theta)   \big\}. 
\end{align}
It is worth mentioning that the quantity 
$
\bar{p}^{\,\mathrm{I}}_{e,\Delta} (x, y; \theta)
$
has also appeared in results by \cite{iguchi:21-2, igu:22}, where a weak high order  approximation was developed for elliptic SDEs via a different approach.  
We also note that  
$ \bar{p}^{\,\mathrm{I}}_\Delta (x, y; \theta)$ 
has the following properties:
\vspace{-0.3 \baselineskip}  
\begin{itemize}
\item[(i)] The normalising constant is equal to 1, i.e. it holds that 
\begin{align*}
 \int_{\mathbb{R}^d} \bar{p}^{\,\mathrm{I}}_\Delta (x, y; \theta) dy =  \int_{\mathbb{R}^d} p^{\widetilde{X}^{\mathrm{LG}}}_\Delta (x, y ; \theta)  \bigl\{1 + \Psi^{\mrm{weak}} (\Delta, x, y ; \theta) \bigr\} dy  = 1. 
\end{align*}
\item[(ii)] $\Psi^{\mrm{weak}} (\Delta,x , y ; \theta)$, thus, also, 
$ \bar{p}^{\,\mathrm{I}}_\Delta (x, y; \theta)$, can, in general, take negative values, as do other closed-form expansions proposed by \cite{ait:08} and \cite{li:13}.
\end{itemize}  
\subsubsection{Local weak third order transition density scheme -- version II} \label{sec:density_2}
%
%
We construct an alternative transition density scheme that provides a proper, everywhere positive, density function. See also \cite{stra:10} for a related consideration in the context of the density expansion for elliptic SDEs put forward by \cite{ait:02}. 
We consider the truncated Taylor expansion of $\log (1 + z)$ and  introduce, for $z \in \mathbb{R}$, 
\begin{align}
\label{eq:kappaf}
 K(z) := \sum_{1 \le l \le 6} (-1)^{l+1} \tfrac{z^l}{l}.  
\end{align} 
For $\Delta > 0$, $x, y \in \mathbb{R}^d$, $\theta \in \Theta$, we define the unnormalised non-negative approximation: 
\begin{align} 
\label{eq:new_tran_density}
 \bar{p}^{\,\mathrm{II}}_\Delta (x, y ; \theta)
  := p^{\widetilde{X}^{\mathrm{LG}}}_\Delta (x, y; \theta) 
  \exp \bigl( K( \Psi^{\mrm{weak}} (\Delta, x, y; \theta))   \bigr). 
\end{align} 
The mapping $y\mapsto \bar{p}^{\,\mathrm{II}}_\Delta (x, y ; \theta)$ is integrable as the highest order term of $K(z)$ in (\ref{eq:kappaf}) is $-z^6$, so $\exp(K(z))$ is bounded. Notice that the normalising constant 
$
 \textstyle{ \bar{Z} (\Delta, x ; \theta):
 =\int_{\mathbb{R}^{d}} \bar{p}^{\,\mathrm{II}}_\Delta (x, y ; \theta)dy} 
$
is no longer $1$. Nevertheless, the unnormalised non-negative approximation 
$\bar{p}^{\,\mathrm{II}}_\Delta (x ,y ;\theta)$
delivers a local weak third order approximation for $X_{\Delta}^x$, in the following sense with its proof provided in Section \ref{appendix:diff_density_weak} of Supplementary Material. 
\begin{prop} \label{prop:local_WA_ver2}
For any $\varphi \in C_p^\infty (\mathbb{R}^d ; \mathbb{R})$ and $(\Delta, x, y) \in (0, \infty) \times \mathbb{R}^d \times \Theta$, there exist constants $C>0$, $q \ge 1$, such that:
\vspace{-0.5 \baselineskip}   
\begin{align} 
\label{eq:diff_density_weak}
 \Big| 
    \int_{\mathbb{R}^d}  
    \varphi (y) \bar{p}^{\,\mathrm{I}}_\Delta (x, y; \theta)  dy 
    - \int_{\mathbb{R}^d}  \varphi (y) 
     \bar{p}^{\,\mathrm{II}}_\Delta (x, y ; \theta) dy\,  \Big| 
\leq C (1 + |x|^q) \Delta^3, 
\end{align}
and then, from Proposition \ref{prop:density_scheme} (iv), 
\vspace{-0.5 \baselineskip}   
\begin{align*}
 \Big|\, \mathbb{E}_{\theta} [ \varphi (X_{\Delta}^x) ]
     -  \int_{\mathbb{R}^d}
        \varphi (y) \bar{p}^{\,\mathrm{II}}_\Delta (x, y ; \theta) dy  
    \,\Big|   \leq C ( 1 +|x|^q) \Delta^3, 
\end{align*}
for some constants $C>0$, $q \ge 1$. 
\end{prop}
\noindent 
The transition density scheme $\bar{p}^{\,\mathrm{II}}_\Delta (x, y ; \theta)$  in  (\ref{eq:new_tran_density}) has the  properties: 
\vspace{-0.5 \baselineskip}
\begin{itemize}
\item[(i)] It takes non-negative values for all $x, y \in \mathbb{R}^d$ and $\theta \in \Theta$. 
\item[(ii)] From Proposition \ref{prop:local_WA_ver2}, it  provides a local weak third order approximation for the solution of the SDE (\ref{eq:model}), for test functions $\varphi \in C_p^{\infty} (\mathbb{R}^d ; \mathbb{R})$. The discrepancy between the normalising constant $\bar{Z} (\Delta, x; \theta)$ and $1$ is of size $\mathcal{O} (\Delta^{3})$, by setting $\varphi (y)=1$ in (\ref{eq:diff_density_weak}). 
\end{itemize} 
\subsubsection{Discretisation based upon the transition density schemes} \label{sec:disc_density}
%
We finally illustrate the discretisation bias induced from the proposed transition density schemes for elliptic SDEs, i.e., $d_S = 0, d_R =d $. For the first-type transition density scheme in (\ref{eq:density_formula_e}), \cite{iguchi:21-2} showed that the discretisation bias is of size  $\mathcal{O}(M^{-2})$, i.e., there exist constants $C, c > 0$, $q\geq d_R/2$ and a non-decreasing $h(\cdot)$ such that for any $x, y \in \mathbb{R}^{d_R}$, 
\vspace{-0.4 \baselineskip}  
\begin{align} \label{eq:bd_conv_ver1}
    \Big|\, 
    p_{e,\Delta}^X (x, y;\theta)  - \bar{p}^{\,\mathrm{I},(M)}_{e,\Delta} (x, y;\theta)\, 
    \Big|
    \leq \frac{C}{M^2} \frac{h(\Delta)}{\Delta^q} 
    e^{- c\,\tfrac{|y - x|^2}{\Delta} },
\end{align}
\begin{rem}
The bound (\ref{eq:bd_conv_ver1}) was derived without reference to unknown SDE parameters, thus are not uniform in  $\theta$. Nevertheless, under compactness of $\Theta$, see (\ref{assump:param_space}), they can be readily 
adapted to be uniform also in  $\theta\in \Theta$.
\end{rem}
\noindent We have the following estimate for the second-type transition density scheme, which is a main new result obtained in this section, with a proof given in Section \ref{sec:pf_density_approx_second} of Supplementary Material.  
\begin{theorem} \label{thm:density_approx_second}
Let $\Delta> 0$. Assume $d_S = 0$ and that conditions (\ref{assump:param_space})--(\ref{assump:hypo1}) hold. Then, there exist constants $C, c > 0$, $q \geq d_R/2$ and a non-decreasing function $h(\cdot)$ such that for any $x, y \in \mathbb{R}^{d_R}$, $\theta \in \Theta$:
\begin{align} 
\label{eq:density_approx_second}
 \Big|\,
   p_{e,\Delta}^X (x , y; \theta) 
  - \bar{p}^{\,\mathrm{II},(M)}_{e,\Delta} (x , y; \theta) 
 \,\Big| 
 \leq \frac{C}{M^2}  \frac{h(\Delta)}{\Delta^{q}}
 e^{ - c\,\tfrac{|y-x|^2}{\Delta} }. 
\end{align}
\end{theorem} 
\section{Numerical experiments} \label{sec:sim_study}
We provide numerical experiments related to the analytic results in Section \ref{sec:main}. 
Note that our software implementation (\url{https://github.com/matt-graham/simsde}) of both the proposed contrast function and the sampling schemes is fully general.
The user is only required to specify functions for evaluating the drift and diffusion coefficient terms for given a process state and model parameters, with functions for evaluating the contrast estimator or simulating from the sampling scheme then automatically generated. This is achieved by symbolically computing the terms in the relevant estimator or sampling scheme expression using the Python-based computer algebra system SymPy (\url{https://www.sympy.org/}). The  generated functions can also optionally use the numerical primitives defined in the high-performance numerical computing framework JAX (\url{https://jax.readthedocs.io/}); this allows the derivatives of the generated functions to be computed using JAX’s automatic differentiation support.
We exploit the above to automatically generate efficient functions for computing the gradient of the contrast function in the optimisation-based simulation study in Section \ref{sec:sim_high} and the gradient of the posterior density in the Bayesian inference numerical experiments in Section \ref{sec:sim_low}.
\subsection{High-frequency observation regime}
\label{sec:sim_high}
We consider the stochastic Jansen--Rit neural mass model that describes the evolution of neural population in a local cortical circuit, that is, the interaction of the main pyramidal cells with the excitatory and inhibitory interneurons. See \cite{abl:17} for more details about the model. The model is defined as the following $6$-dimensional hypo-elliptic SDE driven by a $3$-dimensional Brownian motion ($d = 6$, $d_S = d_R =3$): 
\begin{align}
 \begin{aligned} \label{eq:jrnmm}
  d X_{R,t} & = \bigl( - \Gamma_\theta^2  X_{S,t} -2 \Gamma_\theta X_{R,t} + G (X_{S,t}, \theta ) \bigr) dt
  + \Sigma_\theta d B_{t}, \quad  X_{R,0} = x_{R, 0} \in \mathbb{R}^3;  \\ 
  d X_{S,t} & = X_{R,t} dt, \quad  X_{S,0} = x_{S, 0} \in \mathbb{R}^3, 
 \end{aligned}
\end{align}
for parameter specified as $\theta = (\beta, \sigma)$ with 
%
$ \beta  = (A, B, C, \mu, \nu_0, a, b, r, \nu_{\mathrm{max}}) \in  \mathbb{R}^5   
   \times (0, \infty)^4, 
\, 
\sigma = (\sigma_1, \sigma_2, \sigma_3)  \in (0, \infty)^3.
$ 
%
%
We have set $\Gamma_\theta := \mathrm{diag} (a, a, b)$, $\Sigma_\theta := \mathrm{diag} (\sigma_1, \sigma_2, \sigma_3)$ and 
\begin{align*} 
 G(x_S, \theta) := \bigl[
   A a \times \mathcal{S}_\theta  (x_S^2 - x_S^3), \, 
   A a \bigl(\mu +  C_1 \times \mathcal{S}_\theta  (C x_S^1) \bigr), \,
   B b \, C_2 \times \mathcal{S}_\theta  (C_3 \, x_S^1) 
 \bigr]^\top,  \quad x_S = (x_S^1, x_S^2, x_S^3),  
\end{align*}
where $C_1 = 0.8 \times C$, $C_2 = C_3 = 0.25 \times  C$ and $\mathcal{S}_\theta : \mathbb{R} \to [0, \nu_{\mathrm{\max}}]$ is a sigmoid function defined as
$   \mathcal{S}_\theta (z) := {\nu_{\mathrm{\max}}}/{\{1 + \exp (r (\nu_0 - z))\}}, \, z \in \mathbb{R}$.  
%
Following the numerical experiment in \cite{buc:20}, we fix a part of the parameter vector $\theta$ to
$A = 3.25, \, B = 22, \, \nu_0 = 6, \, 
a = 100, \, b = 50, \, \nu_{\mathrm{max}} = 5, \, 
r = 0.56, \, 
\sigma_1 = 0.01, \, \sigma_3 = 1.
$
Then, we estimate the parameter $(C, \mu, \sigma_2)$ from observations of all coordinates of model (\ref{eq:jrnmm}) using the new contrast and the local Gaussian estimators. We set the true parameter values  to 
$(C^\dagger, \mu^\dagger , \sigma_2^\dagger) = (135.0, 220.0, 2000.0)$, 
and generate synthetic datasets $Y_{\mathrm{JR}}$ from the local Gaussian scheme (\ref{eq:LG}) with discretisation step $10^{-4}$ on the time interval $[0,100]$. Then, we check the performance of the new contrast estimator 
versus the local Gaussian one 
in the following three scenarios for $n$ and $\Delta_n$ by subsampling from the synthetic datasets $Y_{\mathrm{JR}}$: 
\vspace{-0.5 \baselineskip}
\begin{description}
\item[JR-1.] $(n , \Delta_n) = (1.25 \times 10^4, 0.008)$ with the time interval of observations $T = 100$.
\item[JR-2.] $(n , \Delta_n) = (2.5 \times 10^4, 0.004)$ with  $T = 100$. \ \ \ \textbf{JR-3.} $(n , \Delta_n) = (5 \times 10^4, 0.002)$ with $T = 100$. 
\end{description}
\vspace{-0.5 \baselineskip} 
%
Since model \eqref{eq:jrnmm} belongs in the class of stochastic damping Hamiltonian system (\ref{eq:sdhs}) and matrix $\Gamma_\theta$ is diagonal, the correction term $\Phi_2$ in the proposed contrast function is given as follows:
\begin{align*}
\Phi_2 (\Delta, x, y ; \theta) 
& = - \sum_{1 \le i \le 3} 2 \Gamma_{\theta, ii} (\sigma_i)^2 \times \Bigl\{ 
\tfrac{1}{2} \mathcal{H}_{(i, i)}  (\Delta, x, y; \theta) 
 + \tfrac{1}{3} 
 \mathcal{H}_{(i, i + 3)} (\Delta, x, y; \theta)   \\ 
& \qquad \qquad  + \tfrac{1}{6} \mathcal{H}_{(i + 3, i)} (\Delta, x, y; \theta) 
 + \tfrac{1}{8} \mathcal{H}_{(i + 3, i + 3 )} (\Delta, x, y; \theta)  \Bigr\},  
\end{align*} 
with $\Delta > 0$,  $x, y \in \mathbb{R}^d$, $\theta \in \Theta$. 
As we noted in Section \ref{sec:stochastic_damping}, 
partial derivatives of the non-linear function $G(X_{S,t}, \theta)$ are not required in the computation of the new contrast function. For the minimisation of the contrast, we use the Adam optimiser \citep{kigma:15} with the following algorithmic specifications: (step-size) = $0.01$, (exponential decay rate for the first moment estimates) = $0.9$, (exponential decay rate for the second moment estimates) = $0.999$, (additive term for numerical stability) = $1 \times 10^{-8}$ and (number of iterations) = $20,000$. In Table \ref{table:est_JRNMM}, we summarise the mean and standard error from $50$ replicates of parameter estimates under scenario \textbf{JR-3}. 
The table indicates that both the local Gaussian and the new contrast estimate accurately the drift parameters $(C, \mu)$, but the new contrast produces better results for the diffusion parameter $\sigma_2$ as explained in Remark \ref{rem:score}. Thus, we focus on the estimation of the diffusion parameter $\sigma_2$ and compare the proposed contrast estimator with, first, the local Gaussian estimator and, second, the estimator based on the quadratic variation of $X_{R, t}^2$, i.e., 
$\textstyle 
\hat{\sigma}_{2, n}^{\mrm{QV}} :=  
\sqrt{\tfrac{1}{T} \sum_{1 \le m \le n} 
\bigl( X_{R, m \Delta_n}^2 -  X_{R, (m-1) \Delta_n }^2  \bigr)^2 }.
$ 
We summarise in Table \ref{table:vs_qv} the mean and standard deviation of (true value of $\sigma_2$) - (estimator of $\sigma_2$) from $50$ replicates of estimates under scenario \textbf{JR-1,2,3}. Also, in Figure \ref{fig:rmes_jrnmm} we plot the root mean squared errors (RMSEs) from $50$ replicates of estimates of $\sigma_2$. Note that the new contrast estimator for $\sigma_2$ gives faster convergence to the true value $\sigma_2^\dagger = 2000$, in agreement with a relatively larger $\Delta_n$ (as a function of $n$) permitted in the CLT for the new contrast -- recall that for the CLT, the new contrast requires $\Delta_n = o (n^{-1/3})$ versus $\Delta_n = o (n^{-1/2})$ required by the local Gaussian contrast. We also observe that the new contrast estimator outperforms $\hat{\sigma}_{2, n}^{\mrm{QV}}$. 
This is arguably not a coincidence for the particular experiment, as the contrast estimator makes use of observations from all co-ordinates (not only of $X_R^2$) via a likelihood-based contrast function, 
and MLEs are well-understood to possess optimal asymptotic properties, e.g.~due to the Cram\'er-Rao bound. 
%
\begin{table}
 \caption{\small Mean and standard error (in brackets) from $50$ replicates of parameter estimates in the Jansen-Rit neural mass model (\ref{eq:jrnmm}), under scenario \textbf{JR-3}.}
 \label{table:est_JRNMM}
 \centering 
  \begin{tabular}{cccc}
  \hline  
  Parameter & True value & {Local Gaussian contrast} & {New contrast} 
  \\
  \hline 
  \\[-0.3cm]
  $C$ & 135.0 & 134.80 (0.0062) & 134.80 (0.0062)  \\ 
  $\mu$ & 220.0 & 220.84 (0.6167) & 220.84 (0.6176) \\ 
  $\sigma_2 $ & 2000.0 & 
  1843.49 (4.3520) & 1989.00 (4.750) 
  \\
  \hline
  \end{tabular}
\end{table}
\begin{table}
 \caption{\small Mean and standard deviation (in brackets) of (true value of $\sigma_2$) - (estimator of $\sigma_2$)
 from $50$ replicates of three different estimates in the Jansen-Rit neural mass model (\ref{eq:jrnmm}), under scenarios \textbf{JR-1,2,3}.}
 \label{table:vs_qv}
 \centering 
  \begin{tabular}{cccc}
  \hline  
  Scenario & Quadratic variation & {Local Gaussian contrast} & {New contrast} 
  \\
  \hline 
  \\[-0.3cm]  
  \textbf{JR-1} & 460.34 (11.327) & 314.99 (8.913) & -105.16 (12.419)  \\ 
  \textbf{JR-2} & 246.26 (8.677) & 266.46 (6.695) & 39.46 (7.450) \\ 
  \textbf{JR-3} & 121.42 (6.588) & 156.51 (4.352) & 11.00 (4.750) 
  \\
  \hline
  \end{tabular}
\end{table} 
\begin{figure} 
\caption{
\small{Root mean squared errors (RMSE) for estimators for $\sigma_2$ (50 replicates) in the Jansen-Rit neural mass model (\ref{eq:jrnmm}) under scenarios \textbf{JR-1,2,3}.}} 
\label{fig:rmes_jrnmm} 
\centering
\includegraphics[keepaspectratio, width=7.3cm]{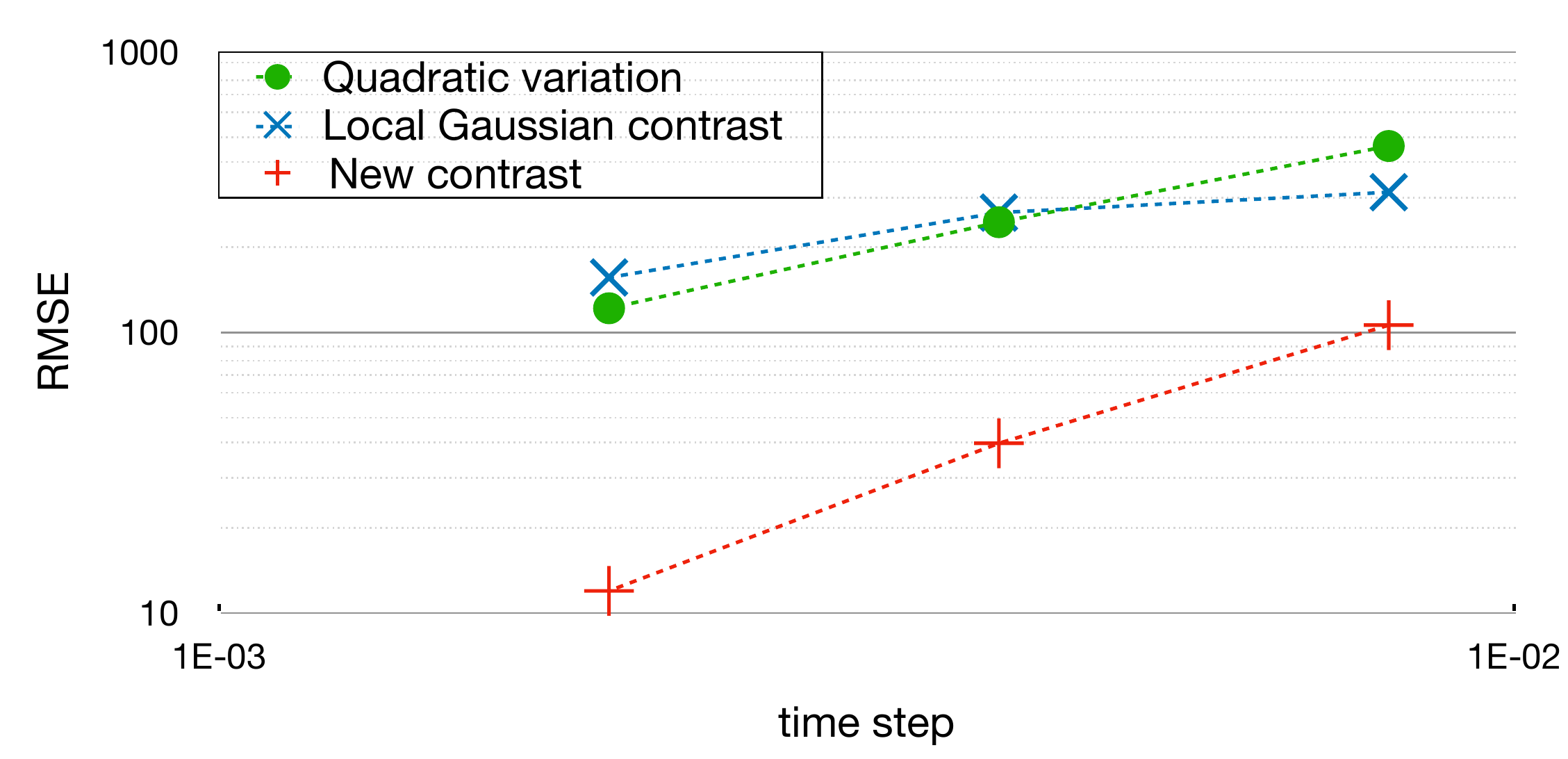}
\end{figure}
\subsection{Low-frequency observation regime}
\label{sec:sim_low}
We perform Bayesian inference, in a framework requiring data augmentation,  for a susceptible-infected-recovered (SIR) model with a time-varying contact rate. The model is specified as the following 3-dimensional SDE:
\begin{align} \label{eq:sir}
\begin{aligned}
 \left[ \begin{array}{c} 
    dS_t \\
    dI_t \\
    dC_t
\end{array} \right] 
 = \left[ \!\! \begin{array}{c}  
    - \tfrac{C_t S_t I_t }{N} \\
    \tfrac{C_t S_t I_t }{N} - \lambda I_t \\
    \bigl( \alpha  (\beta - \log C_t) + \tfrac{\sigma^2}{2}  \bigr) C_t 
\end{array} \!\! 
  \right] dt  
 +  
  \left[ \!\!  \begin{array}{ccc}  
  \sqrt{\tfrac{C_t S_t I_t }{N}} & 0 & 0 \\
  - \sqrt{\tfrac{C_t S_t I_t }{N}} & \sqrt{\lambda I_t} & 0 \\
  0 & 0 & \sigma \\
  \end{array}
  \!\! \right] \, 
  \left[ \begin{array}{c}
   d B_{1,t} \\
   d B_{2,t} \\
   d B_{3,t} \\
   \end{array}
  \right],
\end{aligned}
\end{align}
where $S$, $I$ account for the number of susceptible and infected individuals respectively, $C$ for the contact rate, $N$ is the population size and $\theta = (\alpha, \beta, \sigma, \lambda)$  the parameter vector.
Motivated by an implementation in  \cite{graham2022manifold}, we consider the following setting:
\begin{itemize}[leftmargin=0.5cm]
    \item Observations $Y = \{ Y_t \}_{t=0, \ldots, T}$  correspond to the number of daily infected individuals over the period of $14$ days, with time interval $\Delta = 1$ (day), and 
    they are assumed to be measured with additive Gaussian noise, as $Y_t = I_t + 
    \sigma_y \, \varepsilon_t$,  $0 \leq t \leq 13$, where $ \varepsilon_t \sim \mathscr{N}(0,1)$ and $\sigma_y > 0$ is constant.  
    \item For the parameter vector $\theta$ and the initial condition $C_0$, priors are determined via $\theta = g_\theta (u)$ and $C_0 = g_0 (u_0)$, where $u$, $u_0$  are vectors of independent standard normal random variables and $g_\theta$, $g_0$ are some tractable functions.
    \item Adopting a Bayesian data augmentation approach, 
    a  numerical scheme is called upon to impute instances of model (\ref{eq:sir}) 
    at times separated by discretisation step $\delta$. The MCMC method is based on a non-centred imputation approach.  That is, latent variables correspond to the Brownian increments driving the SDE,
    and they are collected in a vector of (apriori) independent standard normal random variables $v$.
\end{itemize}
In the above setting, the posterior law of the $d_q$-dimensional vector $q = [u^\top, u_0^\top, v^\top]^\top$ given the data $Y$ has a tractable density w.r.t.~the Lebesgue measure (for more details, see  Section 11 in the Supporting Material of \cite{graham2022manifold}). Thus, one can apply standard Hamiltonian Monte Carlo (HMC) to sample from the posterior with energy function $H : \mathbb{R}^{d_q} \times \mathbb{R}^{d_q} \to \mathbb{R}$ given as $H(q, p) = \ell(q) + \tfrac{1}{2} p^\top M^{-1} p$, where $\ell(q)$ is the negative log-posterior density and $M$ a diagonal mass matrix. We will run HMC on the posterior induced both by the Euler-Maruyama (EM) scheme and the proposed weak second order scheme, and we will show (numerically) that the involved bias is smaller in the case of the new scheme. 

The details of the design of the experiment are as follows. We use the observations taken from \cite{anon:78} as data $Y$. We fix $N = 763$, $S_0 = 762$, $I_0 = 0$, $\sigma_y = 5$. We assign priors $\log \alpha \sim \mathscr{N}(0, 1)$, $\beta \sim \mathscr{N} (0,1)$, $\log \sigma \sim \mathscr{N}(-3, 1)$, $\log \lambda \sim \mathscr{N} (0,1)$ and $\log C_0 \sim \mathscr{N} (0,1)$. We do not treat $\sigma_y$ as unknown here due to the widely varying posterior scales induced by non-constant $\sigma_y$ reported in \cite{graham2022manifold}.
The (data-imputing) numerical schemes are applied to the log-transformation $X_t = [\log S_t, \log I_t, \log C_t]^\top$ to ensure positiveness of $S$, $I$, and $C$, and avoid numerical issues with the square-root terms  in the diffusion coefficient. We use a dynamic integration-time HMC implementation (\cite{beta:17}) with a dual-averaging algorithm (\cite{hoff:14}) to adapt the integrator step-size. For the 
time-discretisation of the Hamiltonian dynamics we use the leapfrog integrator with St\"ormer-Verlet splitting. We set the mass matrix $M$ to identity. 
For each of the following three choices of numerical schemes,  we run four HMC chains of 1,500 iterations with the first 500 iterations used as an adaptive warm-up phase: 
\vspace{-0.5 \baselineskip}  
\begin{description}
   \item[SIR-baseline:] 
   $\bigl(\mathrm{Numerical \, scheme}, \delta_M \bigr) = 
   \bigl(\mathrm{EM},  0.001 \bigr).$ 
   \item[SIR-EM:] $ \bigl(\mathrm{Numerical \, scheme}, \delta_M \bigr) = 
   \bigl(\mathrm{EM},  0.05 \bigr). $
   \item[SIR-Weak2nd:] $ \bigl(\mathrm{Numerical \, scheme}, \delta_M \bigr) = 
   \bigl(\mathrm{Weak \, second \, order \, scheme},  0.05 \bigr). $
\end{description}
\vspace{-0.5 \baselineskip} 
Results from the HMC implementations are summarised in Figure \ref{fig:hmc_sir} and Table \ref{table:hmc_sir}. Figure~\ref{fig:hmc_sir} shows estimated pair-wise and marginal posteriors for parameters  $\alpha$, $\beta$,  $\sigma$,  $\lambda$ and $\log C_0$. Results for \textbf{SIR-EM} and \textbf{SIR-Weak2nd} are shown in separate plots, and each plot superimposes corresponding results from \textbf{SIR-baseline}, this latter scheme treated as providing the `correct' posterior quantities due to the use of very small  $\delta_M=0.001$.
Table \ref{table:hmc_sir} shows summary statistics that monitor the performance of the HMC algorithm, in particular, bulk effective sample size (ESS), tail ESS and improved $\hat{R}$ with rank-normalisation and folding (see \cite{veh:21} for analytical definitions). These are computed from the non-warm-up steps of the four HMC chains. Note that for all parameters and choices of numerical schemes, $\hat{R}$ is smaller than $1.01$  and ESS is larger than $400$, as \cite{veh:21} recommend. One can thus be reasonably confident that the estimated posteriors shown in Figure~\ref{fig:hmc_sir} are reliable representations of the true ones.  
In agreement with the analytical theory in this work, Figure \ref{fig:hmc_sir} illustrates that for all parameters the estimated  posteriors obtained via the weak second order scheme capture more accurately the `correct' baseline posteriors than the corresponding ones obtained via the EM scheme.  
\begin{figure} 
    \begin{tabular}{cc}
      \begin{minipage}[t]{0.49\hsize}
        \centering
        \includegraphics[keepaspectratio, scale=0.22]{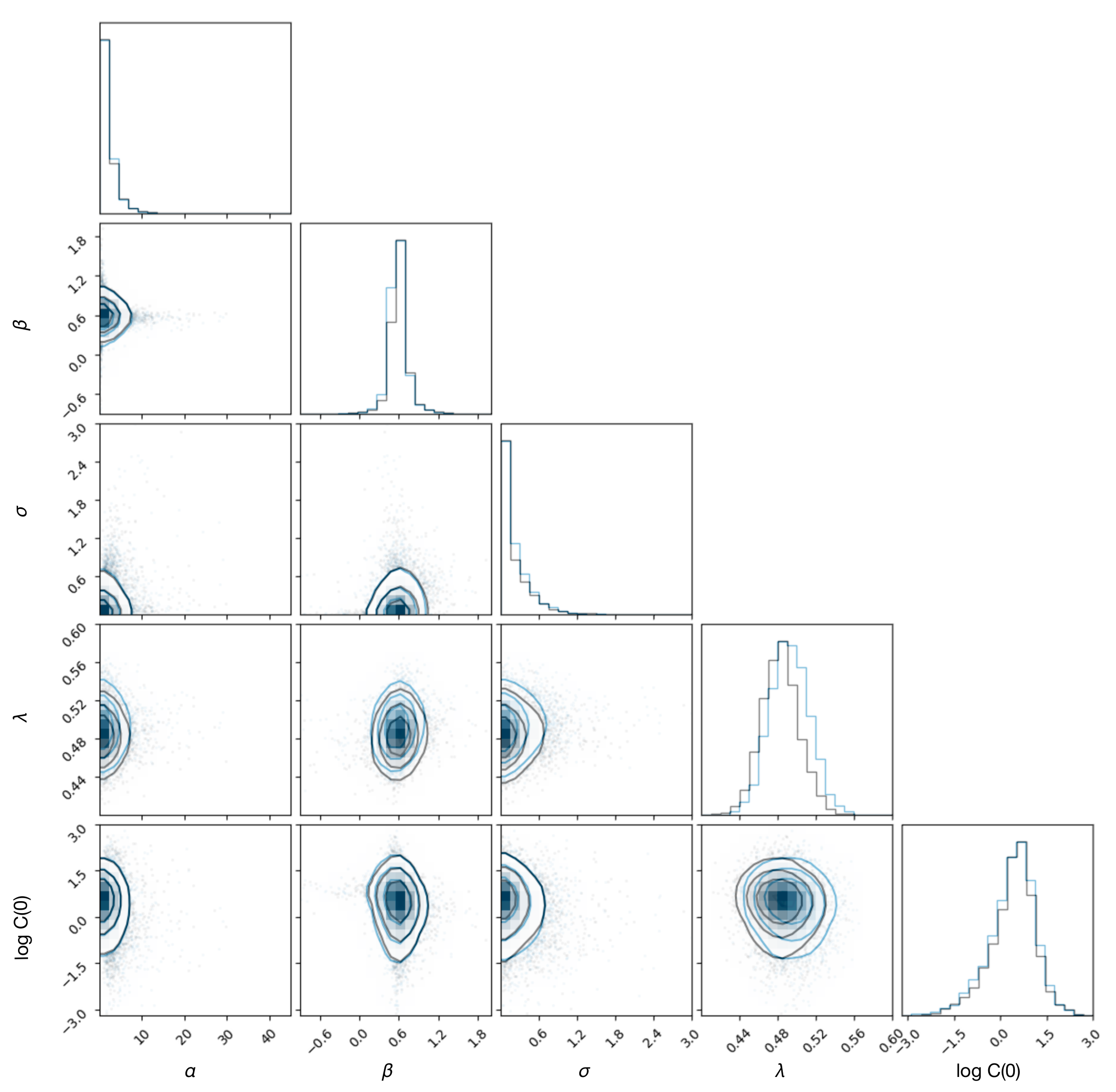}
        \subcaption{Euler-Maruyama scheme}
        \label{fig:hmc_em}
      \end{minipage} 
      \begin{minipage}[t]{0.49\hsize}
        \centering
        \includegraphics[keepaspectratio, scale=0.22]{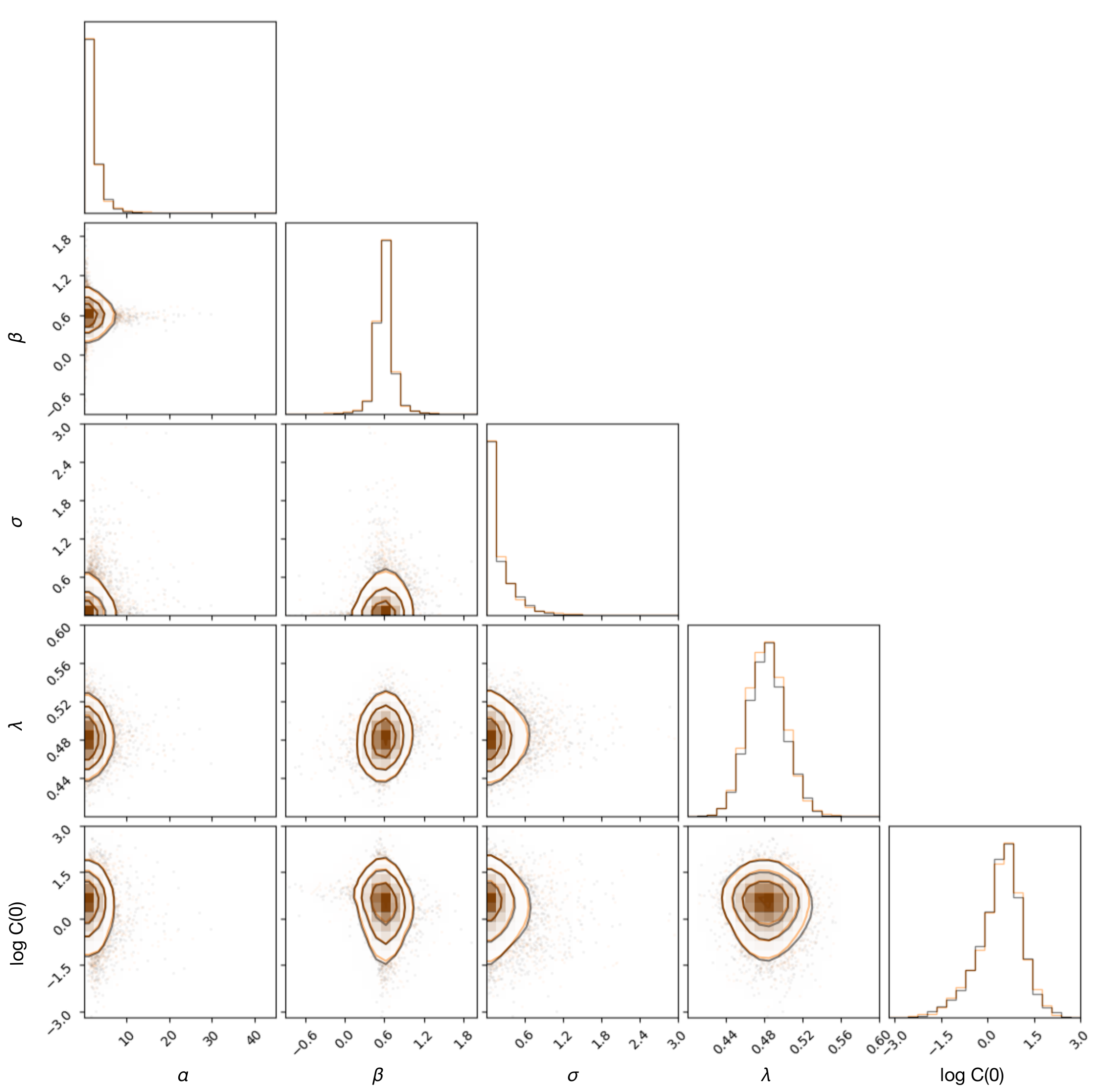}
        \subcaption{Weak second order scheme}
        \label{fig:hmc_2nd}
      \end{minipage} 
    \end{tabular}
     \caption{
     \small{
     Posterior Estimates for SIR Model. 
     The blue (left panel) and orange (right panel) histograms and contour plots are obtained by HMC that uses the EM scheme and the weak second order scheme, respectively, both with the same discretisation step $\delta_M = 0.05$.
     The black histograms and contour plots superimposed in both plots show the correct quantities, as obtained from HMC that uses the EM scheme with very small $\delta_M = 10^{-3}$. 
     } } 
     \label{fig:hmc_sir}
\end{figure}

\begin{table}
 \caption{\small Summary statistics of HMC for the SIR model.}
 \label{table:hmc_sir}
 \centering 
  \begin{tabular}{clccc}
\hline
Parameter & Scenario  & bulk ESS & tail ESS & 
improved $\hat{R}$
\\
\hline
\\[-0.3cm] 
                           & {SIR-baseline} & 4052.0    & 2811.0    & 1.0    \\
                           & {SIR-EM}        & 3874.0    & 3131.0    & 1.0   \\  
\multirow{-3}{*}{$\alpha$} & {SIR-Weak2nd}  & 3876.0    & 2813.0    & 1.0    
\\
\hline
\\[-0.3cm] 
                           & {SIR-baseline} & 4272.0    & 2934.0    & 1.0    \\
                           & {SIR-EM}        & 4058.0    & 2878.0    & 1.0    \\
\multirow{-3}{*}{$\beta$}  & {SIR-Weak2nd}  & 3880.0    & 2969.0  & 1.0  
\\
\hline
\\[-0.3cm] 
                           & {SIR-baseline} & 1720.0    & 2836.0    & 1.0    \\
                           & {SIR-EM}  & 1525.0    & 2634.0    & 1.0    \\
\multirow{-3}{*}{$\sigma$}  & {SIR-Weak2nd}  & 1781.0    & 2513.0    & 1.0   
\\
\hline
\\[-0.3cm]  
                           & {SIR-baseline} & 6596.0    & 2952.0    & 1.0    \\
                           & {SIR-EM}        & 7063.0    & 3052.0    & 1.0    \\
\multirow{-3}{*}{$\lambda$}   & {SIR-Weak2nd}  & 7024.0    & 2941.0    & 1.0  
\\
\hline
\\[-0.3cm] 
                           & {SIR-baseline} & 4555.0    & 2889.0    & 1.0    \\
                           & {SIR-EM}        & 5036.0    & 3092.0    & 1.0    \\
\multirow{-3}{*}{$\log C_0$} & {SIR-Weak2nd}  & 4695.0    & 3053.0    & 1.0 
\\
\hline
\\[-0.3cm]  
\end{tabular}
\end{table}
\section{Conclusions} 
\label{sec:conclusions}
This work begins by putting forward weak second order sampling schemes for elliptic and hypo-elliptic SDEs. 
Then, we develop a small time density expansion of the scheme as a proxy for the intractable SDE transition density. Via appropriate choice of the higher order expansion terms, 
we have provided analytical results both: in a high-frequency classical setting, showcasing the advantageous rate of $\Delta_n=o(n^{-1/3})$, achieved for hypo-elliptic models; and in a low-frequency Bayesian data augmentation setting, where we have deduced two local weak third order density schemes (\ref{eq:density_formula}) and (\ref{eq:new_tran_density}) and shown that the induced bias by the schemes is of size $\mathcal{O} (M^{-2})$ when covering the fixed time interval $\Delta > 0$ with inner time-steps $\delta = \Delta / M, \, M > 0$ specified by user. 

We iterate here that there is the flexibility to apply MCMC methods based only on the sampling schemes without reference to transition densities, via, e.g., non-centred model parameterisation approaches (see e.g.~\cite{papa:07, besk:15}), particle-filtering based MCMC methods \citep{andrieu:10}, and recent manifold-based algorithms for the case of observations without/with noise \citep{graham2022manifold}. Via the derivation of approximate transition densities, one is still given the option to use the wealth of data augmentation methods for diffusion models that require such a density expression.

\subsection*{Acknowledgements}
We thank two referees and the Associate Editor for their comments that led to major improvements in the content of the paper. 
YI is supported by the Additional Funding Programme for Mathematical Sciences, delivered by EPSRC (EP/V521917/1) and the Heilbronn Institute for Mathematical Research. 

\appendix 
\renewcommand{\thelemma}{\Alph{section}.\arabic{lemma}}

\section*{Supplementary Material}  
The supplementary material contains technical proofs omitted in the main text. The organisation of this supplementary material is as follows:
\begin{description}
\item[Section \ref{sec:pre}:]
 We recall the notation and conditions used in the main text. 
\item[Section \ref{sec:pf_main}:]
 Proofs of the main results Theorems \ref{consistency}, \ref{thm:asymp_norm} given in the main text.  
\item[Section \ref{sec:pf_technical_high_freq}:]
 Proof of technical results (Lemma \ref{lemma:conv_contrast_gamma}--\ref{lemma:conv_cov}) required by Theorems \ref{consistency}, \ref{thm:asymp_norm}.  
\item[Section \ref{appendix:weak_app_hypo}:] Proof of Proposition \ref{prop:weak_app_hypo} given in the main text. 
%
\item[Section \ref{appendix:density_scheme}:] Proof of Lemma \ref{lemma:AE_scheme} given in the main text. 
%
\item[Section \ref{appendix:local_wa_ver1}:] Proof of Proposition \ref{prop:density_scheme} given in the main text.
\item[Section \ref{appendix:diff_density_weak}:] Proof of Proposition \ref{prop:local_WA_ver2} given in the main text. 
%
\item[Section \ref{sec:pf_density_approx_second}:] Proof of Theorem \ref{thm:density_approx_second} given in the main text. 
\item[Section \ref{sec:fn}:] 
Additional numerical experiment under high frequency observations regime.  
\end{description} 
\section{Notation}  \label{sec:pre}
\noindent We set: 
\begin{align*}
V_0 (\cdot, \beta,\gamma)
=\big[\,V_{R,0}(\cdot, \beta)^{\top}, V_{S,0}(\cdot, \gamma)^{\top}\,\big]^{\top}, \quad 
V_k (\cdot, \sigma) 
=\big[\,V_{R,k}(\cdot, \sigma)^{\top}, \mathbf{0}_{d_S}^\top \,\big]^{\top},
\quad  1 \le k \le d_R. 
\end{align*}
%
\noindent We set 
$V_R= V_R(x , \sigma) =  \big[ V_{R,1}(x, \sigma),\ldots, V_{R, d_R}(x, \sigma) \big]\in \mathbb{R}^{d_R\times d_R}$, and define:
\begin{align}
\label{eq:VAR}
a_R(x , \sigma) := V_R(x, \sigma) V_R(x,  \sigma)^{\top}\in \mathbb{R}^{d_R\times d_R}.
\end{align}
%
We consider the following differential operators, acting upon  maps $\mathbb{R}^{n}\to \mathbb{R}$, $n\ge 1$,
\begin{gather*}
\partial_{u}  :=  \big[ \tfrac{\partial }{\partial u^{1}}, \ldots,  \tfrac{\partial }{\partial u^{n} } \big]^{\top}; 
\quad  
\partial^{2}_{u}   = \partial_{u} \partial_{u}^{\top} \equiv \big(\tfrac{\partial^{2}}{\partial u^{i}\partial u^{j} }\big)_{i,j=1}^{n}. \\[-0.5cm]
\end{gather*}
If needed, we use superscripts to specify co-ordinates of a vector.
We denote by $\mathscr{S}$ the space of functions 
$f : [0, \infty) \times \mathbb{R}^d \times \Theta \to \mathbb{R}$ so that there exists a constant $C>0$ such that $| f (h, x, \theta) | \leq C h$ for any $(h, x, \theta) \in [0, \infty) \times \mathbb{R}^d \times \Theta$,
and $\theta \mapsto f (h, x, \theta)$ is continuous on $\Theta$ for all  $(h ,x) \in [0, \infty) \times \mathbb{R}^d$.
We define:  
\begin{align*}
 \nu_i (\Delta; \theta) :=
 \begin{bmatrix}
    \frac{X_{R,t_i} - \mu_{R, \Delta} (\sample{X}{i-1}; \theta)}{\sqrt{\Delta}} \\
    \frac{X_{S,t_i} - \mu_{S, \Delta} (\sample{X}{i-1}; \theta)}{\sqrt{\Delta^3}} \\
 \end{bmatrix},  \ \   
\end{align*}
 for $1 \le i \le n$, where we have set: 
\begin{align}
 \mu_{\Delta} (x ; \theta) 
 \equiv  
 \begin{bmatrix}
 \mu_{R, \Delta} (x; \theta)  \\ 
 \mu_{S, \Delta} (x ; \theta)    
 \end{bmatrix} 
 \equiv 
 \begin{bmatrix}
x_{R}  + V_{R, 0} (x, \beta) \Delta  \\[0.1cm] 
x_{S}  + V_{S, 0} (x, \gamma) \Delta + \hat{V}_0 V_{S,0}  (x,  \theta) \tfrac{\Delta^2}{2} 
 \end{bmatrix},   \nonumber 
\end{align}
for $\theta = (\beta, \gamma, \sigma) \in \Theta$ and $x = [x_R^\top, x_S^\top ]^{\top} \in \mathbb{R}^d$. We notice that 
\begin{align} 
 &\mathbb{E} \bigl[ \nu_i^j (\Delta; \theta)  | \mathcal{F}_{t_{i-1}} \bigr]
  = \mathcal{O} ( \Delta^{\frac{3}{2}}), \ \ 1 \le j\le d,  \label{eq:mean_nu} 
 \end{align}
and for $1 \le k_1 \le  d_{\beta}$, $d_\beta + 1 \le  k_2 \le d_\beta + d_{\gamma}$, $d_\beta + d_{\gamma} + 1 \le k_3 \le d_\theta$: 
\begin{align} 
& \deriv{\theta}{k_1} \nu_i(\Delta; \theta)  = 
\begin{bmatrix}
  -  \deriv{\theta}{k_1} V_{R, 0} (\sample{X}{i-1}, \beta) \sqrt{\Delta} \\[0.1cm] 
  -  \deriv{\theta}{k_1} \hat{V}_0 V_{S, 0} (\sample{X}{i-1}, \theta)    \tfrac{\sqrt{\Delta}}{2}
\end{bmatrix};  \label{eq:nu_deriv_beta} \\[0.1cm]   
&  \deriv{\theta}{k_2} \nu_i(\Delta; \theta)  = 
\begin{bmatrix}
  \mathbf{0}_{d_R} \\[0.2cm]  
 - \deriv{\theta}{k_2} V_{S,0} (\sample{X}{i-1}, \gamma) \tfrac{1}{\sqrt{\Delta}} 
 -  \deriv{\theta}{k_2} \hat{V}_0 V_{S, 0} (\sample{X}{i-1}, \theta)    \tfrac{\sqrt{\Delta}}{2}
\end{bmatrix};  
\nonumber
\\[0.1cm]    
&  \deriv{\theta}{k_3} \nu_i(\Delta; \theta)  = 
\begin{bmatrix}
  \mathbf{0}_{d_R} \\[0.2cm]  
 -  \deriv{\theta}{k_3} \hat{V}_0 V_{S, 0} (\sample{X}{i-1}, \theta) \tfrac{\sqrt{\Delta}}{2}
\end{bmatrix}.  
\nonumber
\end{align}
For $\Delta > 0$, $x \in \mathbb{R}^d$ and $(\gamma, \sigma) \in \Theta_{\gamma} \times \Theta_{\sigma}$, we set:
\begin{align*}
    {\Sigma}_\Delta  (x; \gamma, \sigma) 
    \equiv  {\Sigma} (\Delta, x; \sigma) 
    \equiv  
    \begin{bmatrix}
    {\Sigma}_{\Delta, RR}  (x;  \sigma)  
    & {\Sigma}_{\Delta, RS}  (x;  \gamma, \sigma)  \\[0.1cm]  
    {\Sigma}_{\Delta, SR}  (x;  \gamma, \sigma) 
    & {\Sigma}_{\Delta, SS}  
    (x; \gamma, \sigma) 
    \end{bmatrix}, 
\end{align*} 
where each block matrix is determined as: 
\begin{align*}
\begin{aligned}
    & {\Sigma}_{\Delta, RR} (x;  \sigma)  
    = \Delta \, a_R (x, \sigma), \ \
    {\Sigma}_{\Delta, RS}  (x;  \gamma, \sigma)  
    = \tfrac{\Delta^2}{2}  \, a_R (x, \sigma) 
    \bigl( \partial_{x_R}^\top V_{S,0} (x, \gamma)  \bigr)^\top;  \\[0.2cm]
    & {\Sigma}_{\Delta, SR}  
    (x; \gamma, \sigma)  
    = {\Sigma}_{\Delta, RS} 
    (x;  \gamma, \sigma)^\top,
    \ \ 
    {\Sigma}_{\Delta, SS}  (x;  \gamma, \sigma)  
    = \tfrac{\Delta^3}{3} 
    \partial_{x_R}^\top V_{S,0} (x, \gamma) 
    \, a_R (x, \sigma) \, 
    \bigl( \partial_{x_R}^\top V_{S,0} (x, \gamma)  \bigr)^\top. \vspace{0.2cm}
\end{aligned}
\end{align*}
We note that under conditions (\ref{assump:hypo1})--(\ref{assump:hypo2}), matrices  $\Sigma_\Delta (x;  \gamma, \sigma)$, $\Sigma_{\Delta, RR} (x; \sigma) $ and $\Sigma_{\Delta, SS} (x; \gamma, \sigma)$ are positive definite for any $\Delta > 0$, $x \in \mathbb{R}^d$, $(\gamma, \sigma) \in \Theta_\gamma \times \Theta_\sigma$. Similarly, we use the following block expression for the inverse of matrix ${\Sigma}_1 (x ; \gamma, \sigma)$: 
\begin{align*}
{\Sigma}_1 ^{-1}  (x; \gamma, \sigma)
\equiv 
\begin{bmatrix}
{\Lambda}_{1, RR}  (x;  \gamma, \sigma)  & \Lambda_{1,RS}  (x; \gamma, \sigma)  \\[0.1cm]  
\Lambda_{1, SR}  (x; \gamma, \sigma)  & \Lambda_{1, SS}  (x; \gamma, \sigma) 
\end{bmatrix}.
\end{align*} 
In particular, we have: 
\begin{align}
  \Lambda_{1, SS} (x; \gamma, \sigma) = 4 \cdot \Sigma_{1, SS}^{-1} (x; \gamma, \sigma).
\end{align}
%
We recall the definition of $\Phi_2 (\Delta, \sample{X}{i-1}, \sample{X}{i}; \theta)$ and write 
\begin{align*}
\Phi_2 (\Delta, \sample{X}{i-1}, \sample{X}{i}; \theta) 
=  \sum_{i_1, i_2 = 1}^d {G}_{i_1i_2} (\sample{X}{i-1}; \theta) \mathcal{H}_{(i_1, i_2)}(\Delta, \sample{X}{i-1}, \sample{X}{i} ;  \theta), 
\end{align*} 
where $\mathcal{H}_{(i_1, i_2)}(\Delta, \sample{X}{i-1}, \sample{X}{i} ;  \theta)$ is the Hermite polynomial given as: 
%
\begin{align} 
\label{eq:hermite_second}
  \mathcal{H}_{(i_1, i_2)} (\Delta, \sample{X}{i-1}, \sample{X}{i}; \theta) 
 & =  \sum_{j_1, j_2 = 1}^d {\Sigma}_{1, i_1 j_1}^{-1} (\sample{X}{i-1}; \gamma, \sigma) {\Sigma}_{1, i_2 j_2}^{-1} (\sample{X}{i-1}; \gamma, \sigma) \nu_{i}^{j_1} (\Delta; \theta)  \nu_{i}^{j_2} (\Delta; \theta)  \nonumber \\  
& \qquad\qquad \qquad - {\Sigma}_{1, i_1 i_2}^{-1} (\sample{X}{i-1}; \gamma, \sigma),  
\end{align}
and the matrix $G(x; \theta) = (G_{i_1i_2}(x; \theta))_{1  \leq i_1, i_2  \leq d}$, $x \in \mathbb{R}^d, \; \theta \in \Theta$, is given by 
\begin{align} \label{eq:G_matrix}
  G (x ; \theta) = 
  \begin{bmatrix}
  G_{RR} (x ; \theta) & G_{RS} (x ; \theta)  \\
  G_{RS} (x ; \theta)^{\top} & G_{SS} (x ; \theta) 
  \end{bmatrix}, 
\end{align}
with
\begin{align*}
 G_{RR} (x ; \theta)
 & =  \tfrac{1}{2}  \sum_{k = 1}^{d_R}
  \bigl( \hat{V}_{k} V_{R, 0} (x,  \theta)
 + \hat{V}_{0} V_{R,k}(x, \theta) \bigr) V_{R,k} (x, \theta)^{\top}   
 + \tfrac{1}{4} \sum_{k_1, k_2 = 1}^{d_R}   
  \hat{V}_{k_1} V_{R, k_2} (x, \theta) \hat{V}_{k_1} V_{R, k_2}  (x,  \theta)^\top;  \\[0.2cm]
 G_{RS} (x ; \theta)
 & = \tfrac{1}{2}  \sum_{k = 1}^{d_R} 
  \Bigl( 
   \tfrac{1}{3} \hat{V}_{k} V_{R,0} (x , \theta)
 + \tfrac{1}{6} \hat{V}_{0} V_{R, k} (x , \theta)
  \Bigr)
   \hat{V}_k V_{S,0} (x ,  \theta)^{\top}
 + \tfrac{1}{12} \sum_{k  = 1}^{d_R}  V_{R, k} (x ,  \theta)  
\bigl( \hat{V}_{k} \hat{V}_{0}  V_{S,0} (x ,  \theta)  + \hat{V}_{0} \hat{V}_{k} V_{S,0} (x,   \theta)  \bigr)^{\top}
\nonumber \\
& \qquad \qquad+ \tfrac{1}{12} \sum_{k_1, k_2 = 1}^{d_R} 
\hat{V}_{k_1} V_{R, k_2} (x ,  \theta)  
\hat{V}_{k_1} \hat{V}_{k_2} V_{S,0}  (x ,  \theta)^{\top}; 
\nonumber \\[0.2cm]
G_{SS} (x ; \theta)
& =   \sum_{k=1}^{d_R} \hat{V}_k V_{S, 0} (x,  \theta) 
  \Bigl(\tfrac{1}{6} \hat{V}_{0} \hat{V}_{k} V_{S,0} (x , \theta) + \tfrac{1}{8} \hat{V}_{k} \hat{V}_{0}  V_{S,0} (x,  \theta)  
  \Bigr)^{\top}  
  +  \tfrac{1}{24}  \sum_{k_1, k_2 = 1}^{d_R} 
  \hat{V}_{k_1} \hat{V}_{k_2} V_{S,0}  (x , \theta) 
  \hat{V}_{k_1} \hat{V}_{k_2} V_{S,0} (x ,  \theta)^{\top}. 
  \nonumber 
\end{align*} 
We recall the notation for multiple stochastic integrals: 
\begin{align}
I_\alpha (t) =  \int_{0}^{t}\cdots \int_{0}^{t_2} dB_{\alpha_1, t_1} \cdots dB_{\alpha_l, t_l}, \ \ t >  0,   \nonumber 
\end{align} 
where $\alpha \in \{0, 1, \ldots, d_R \}^l, \; l \in \mathbb{N}$. Finally, we define $|\alpha| = l$ for a multi-index $\alpha \in \{1, \ldots, d \}^l$, $l \in \mathbb{N}$. 
\section{Proofs of Theorem \ref{consistency} \& \ref{thm:asymp_norm}} 
\label{sec:pf_main}
\subsection{Proof of Theorem \ref{consistency}} \label{pf_consistency}
We prove consistency of the estimator $\est$ in the complete observation regime (\ref{eq:data}). A number of required technical results are collected in Section \ref{sec:pf_lemma_consistency}. To show consistency, a careful consideration of the structure of the covariance matrix of the model is needed, due to hypo-ellipticity. 
Then, we proceed with our proof using a similar main strategy as in \cite{melnykova2020parametric} and \cite{glot:21}, i.e.:
 
%
\begin{description}
 \item[Step 1.] We prove consistency for the estimator $\hat{\gamma}_n$. Using the consistency, we show that, if $n \to \infty, n \Delta \to \infty$ and $\Delta \to 0$, then 
\begin{align} \label{eq:gamma_rate}
  \tfrac{1}{\sqrt{\Delta}} (\hat{\gamma}_n - \gamma^\dagger) \probconv 0.
\end{align}
 \item[Step 2.] Given the above convergence (\ref{eq:gamma_rate}) and condition (\ref{assump:lipschitz}), we prove consistency for $(\hat{\beta}_n, \hat{\sigma}_n)$.   
\end{description}

\subsubsection{Step 1.}
%
\noindent  Recall that 
$\theta = (\beta, \gamma, \sigma)$, and that the contrast function $\ell_{n, \Delta} (\theta)$ is defined so that it is \emph{minimized} at $\theta=\est$. 
We have the following result whose proof is given in Section \ref{sec:conv_contrast_gamma}:
%
\begin{lemma} 
\label{lemma:conv_contrast_gamma}
Assume that conditions (\ref{assump:coeff})--(\ref{assump:bd_deriv}),  (\ref{assump:moments}) and (\ref{assump:finite_moment}) hold.
If $n \to \infty$, $\Delta \to 0$ and $n \Delta \to \infty$, then  
\begin{align*}
\tfrac{\Delta}{n} \ell_{n, \Delta} (\theta) 
\probconv   
 \int
  \bigl( V_{S,0} (x ,  \truegamma) -  V_{S,0}  (x ,  \gamma) \bigr)^{\top} 
  {\Lambda}_{1,SS} (x ;  \gamma, \sigma) 
  \bigl( V_{S,0} (x , \truegamma) -  V_{S, 0} (x ,  \gamma) \bigr) \,\truedist(dx), 
\end{align*}
uniformly in $\theta = (\beta, \gamma, \sigma) \in \Theta$. 
\end{lemma}
\noindent 
Lemma \ref{lemma:conv_contrast_gamma}, together with conditions (\ref{assump:param_space}) and (\ref{assump:ident}), gives directly that, as $n \to \infty$, $\Delta \to 0$, then $\tfrac{\Delta}{n} \ell_{n, \Delta} ({\beta}, \gamma^\dagger, {\sigma})\rightarrow 0$, in probability, uniformly in  $\theta \in \Theta$, thus
\begin{align} 
\label{eq:convergence_gamma}
 & \tfrac{\Delta}{n} \ell_{n, \Delta} ({\beta}, \gamma, {\sigma}) 
 - \tfrac{\Delta}{n} \ell_{n, \Delta} ({\beta}, \gamma^\dagger, {\sigma})  \\[0.2cm]
& \quad \xrightarrow{\mathbb{P}_{\theta^\dagger}} 
  \int
  \bigl( V_{S,0} (x ,  \truegamma) -  V_{S,0}  (x ,  \gamma) \bigr)^{\top} 
  \Lambda_{1,SS} (x; \gamma, \sigma) \bigl( V_{S,0} (x ,  \truegamma) -  V_{S, 0} (x ,  \gamma) \bigr)\,\truedist(dx),
  \nonumber
\end{align} 
uniformly in $\theta \in \Theta$. We will now show that (\ref{eq:convergence_gamma}) implies $\hat{\gamma}_n  \xrightarrow{\mathbb{P}_{\theta^\dagger}} \gamma^\dagger$. From the identifiability condition (\ref{assump:ident}) and the definition of the contrast estimator, it holds that for any $\varepsilon > 0$, 
\begin{align} \label{eq:estgamm_bd}
    \mathbb{P}_{\trueparam} 
    \left( |\estgamma - \truegamma | > \varepsilon \right) 
    \leq  
    \mathbb{P}_{\trueparam} 
    \left( \tfrac{\Delta}{n} \ell_{n, \Delta} (\hat{\theta}_n) < \tfrac{\Delta}{n} \ell_{n, \Delta} ( \hat{\beta}_n, \truegamma, \hat{\sigma}_n ) \right). 
\end{align}
Now we note that the right hand side of (\ref{eq:convergence_gamma})  must be no smaller than $0$ due to the non-negative definiteness of ${\Sigma}_1(x; \theta)$ for every $x \in \mathbb{R}^d$, $\theta \in \Theta$, and the convergence (\ref{eq:convergence_gamma}) holds uniformly in $\theta \in \Theta$. Hence, from condition (\ref{assump:param_space}), i.e, the compactness of the parameter space $\Theta$, we have that if $n \to \infty$, $n \Delta \to \infty$ and $\Delta \to 0$, then 
$ \textstyle{ \mathbb{P}_{\trueparam} 
\left( \tfrac{\Delta}{n} \ell_{n, \Delta} (\hat{\theta}_n) < \tfrac{\Delta}{n} \ell_{n, \Delta} ( \hat{\beta}_n, \truegamma, \hat{\sigma}_n ) \right) } \to 0, 
$
which leads to the consistency of $\estgamma$ from inequality (\ref{eq:estgamm_bd}).

Next, we show that convergence (\ref{eq:gamma_rate}) holds. The Taylor expansion of $\partial_{\gamma} \contrast{\hat{\theta}_n}, \; \hat{\theta}_n = (\hat\beta_n, \hat{\gamma}_n, \hat{\sigma}_n)$
around $\partial_{\gamma} 
\contrast{\hat{\beta}_n, \truegamma, \hat{\sigma}_n}$ gives 
%
%
\begin{align}  
\label{eq:taylor_gamma}
     A_\gamma \bigl( \hat{\beta}_n, \hat\sigma_{n} \bigr)
     = B_\gamma (\hat\theta_n) \times \tfrac{1}{\sqrt\Delta}
     (\estgamma - \truegamma), 
\end{align}
where we have defined for $\theta = (\beta, \gamma, \sigma) \in \Theta$, 
\begin{align*}
    A_\gamma \bigl( {\beta}, \sigma \bigr) 
    := -\tfrac{\sqrt\Delta}{n} \partial_{\gamma} 
    \contrast{ \beta, \truegamma, \sigma },  \quad
    B_\gamma ( \theta ) 
    := \tfrac{\Delta}{n} 
      \int_0^1 \partial_\gamma^2 \ell_{n, \Delta} 
      \bigl( {\beta}, \truegamma + \lambda (\gamma - \truegamma), \sigma \bigr) d\lambda.  
\end{align*}
Then, we have the following result whose proof is given in Section \ref{sec:pf_taylor_gamma_conv}:
%
\begin{lemma} \label{lemma:taylor_gamma_conv}
Assume that conditions (\ref{assump:coeff})--(\ref{assump:bd_deriv}),  (\ref{assump:moments}) and (\ref{assump:finite_moment}) hold. If $n \to \infty$, $\Delta \to 0$ and $n \Delta \to \infty$, then  
\begin{align}
& A_\gamma \bigl( {\beta}, \sigma \bigr)  \probconv 0; 
 \label{eq:A_gamma}  \\
& B_\gamma \bigl( \beta, \estgamma, \sigma \bigr)
 \probconv \int 
 \partial_\gamma V_{S,0}^\top (x , \truegamma) \,
 \Lambda_{1, SS} (x ; \truegamma, \sigma)  \,
 \bigl( \partial_\gamma V_{S,0}^\top  (x , \truegamma) \bigr)^\top \nu_{\trueparam} (dx). 
 \label{eq:B_gamma}
\end{align}
uniformly in $(\beta, \sigma) \in \Theta_{\beta} \times \Theta_\sigma$. 
\end{lemma}
Here, we notice that the right hand side of (\ref{eq:B_gamma}) is  positive definite due to the positive definiteness of $\Lambda_{1, SS} (x ; \gamma, \sigma)$ for any $x \in \mathbb{R}^d$ and $(\gamma, \sigma) \in \Theta_{\gamma} \times \Theta_\sigma$. Hence, Lemma \ref{lemma:taylor_gamma_conv} and equation (\ref{eq:taylor_gamma}) give  convergence (\ref{eq:gamma_rate}).  
\subsubsection{Step 2.} 
\noindent For the proof, we define
%
$T(\theta) :=  \tfrac{1}{n} \contrast{\theta}, \ \  
Q(\theta) :=  \tfrac{1}{n \Delta} \contrast{\theta} - \tfrac{1}{n \Delta} \contrast{\truebeta, \gamma, \sigma}$, 
%
where $\theta \in \Theta$. We will first consider the limit of $T(\beta, \estgamma, \sigma)$ to show consistency of $\hat{\sigma}_n$, and then that of $Q(\beta, \estgamma, \sigma)$ to prove consistency of $\hat{\beta}_n$. We note that these results will rely on the consistency of $\estgamma$ and  convergence (\ref{eq:gamma_rate}).
In particular, convergence (\ref{eq:gamma_rate}) and condition (\ref{assump:lipschitz})
will allow for obtaining appropriate limits for $T(\beta, \estgamma, \sigma)$ and $Q(\beta, \estgamma, \sigma)$. 

\subsubsection*{Step 2-(i). Consistency of $\hat{\sigma}_n$} \label{sigma_consistency}
\noindent We have the following result due to convergence (\ref{eq:gamma_rate}).   
\begin{lemma} \label{lemma:conv_contrast_sigma}
Under conditions (\ref{assump:coeff})-(\ref{assump:finite_moment}), it holds that as $n \to \infty$, $\Delta \to 0$, $n \Delta \to \infty$,  
\begin{align} \label{conv_sigma_contrast}
T(\beta, \estgamma, \sigma) 
\xrightarrow{\mathbb{P}_{\theta^\dagger}}   
 \;\int
 \Bigl(   \mathrm{tr} \bigl( {\Sigma}_1^{-1} (x; \gamma^\dagger, \sigma) {\Sigma}_1 (x; \gamma^\dagger, \truesigma ) \bigr)
         + \log |{\Sigma}_1(x; \truegamma, \sigma)|  
 \Bigr)\,  
 \nu_{\theta^\dagger} (dx),  
\end{align}
uniformly in $(\beta, \sigma) \in \Theta_\beta \times \Theta_\sigma$.  
\end{lemma}

\noindent We show the proof of Lemma \ref{lemma:conv_contrast_sigma} in Section \ref{sec:conv_contrast_sigma}. 
Limit (\ref{conv_sigma_contrast}) with conditions (\ref{assump:param_space}), (\ref{assump:ident}) will provide the consistency for~$\hat{\sigma}_n$. In particular, we will show that 
if $n \to \infty$, $n \Delta_n \to \infty$, $\Delta_n \to 0$ then, 
\begin{align} \label{eq:sigma_likelihood}
\mathbb{P}_{\trueparam} 
   \left( 
   T (\hat{\theta}_n ) 
   <  
   T (\hat{\beta}_n, \hat{\gamma}_n, \truesigma) \right) 
   \to 0,
\end{align}
which implies consistency of $\hat{\sigma}_n$ from the same discussion as in the proof of consistency of $\estgamma$. 
We denote by $U(\sigma)$ the right-hand-side of (\ref{conv_sigma_contrast}), and will show that 
$\sigma \mapsto U(\sigma)$ takes the minimum value at $\sigma = \sigma^\dagger$ from the following discussion. We define 
\begin{align*}
  g (\sigma; x) 
  :=   \mathrm{tr} \bigl( {\Sigma}_1^{-1} (x; \truegamma, \sigma ) {\Sigma}_1 (x; \truegamma, \truesigma) \bigr) +      \log |{\Sigma}_1(x;  \truegamma, \sigma )|. 
 \end{align*} 
Taking the derivative of $g (\sigma; x)$ w.r.t.~$\sigma_i$, $1\le i \le d_\sigma$, yields:
\begin{align*}
\partial_{\sigma_i}  g (\sigma; x) 
%
 = &\; \mathrm{tr} \bigl( \partial_{\sigma_i}{\Sigma}_1^{-1} (x; \truegamma , \sigma) 
  \Sigma_1 (x; \truegamma, \truesigma)  \bigr) 
 - \mathrm{tr} \bigl( \partial_{\sigma_i} {\Sigma}_1^{-1} (x;\truegamma , \sigma) {\Sigma}_1 (x; \truegamma , \sigma)  \bigr), 
\end{align*} 
where we used that: 
\begin{align*}
\partial_{\sigma_i} \log |{\Sigma}_1(x;  \truegamma, \sigma )| 
= \mathrm{tr} \bigl( {\Sigma}_1^{-1} (x; \truegamma , \sigma) 
\partial_{\sigma_i}{\Sigma}_1 (x; \truegamma , \sigma)  \bigr)
= - \mathrm{tr} \bigl( \partial_{\sigma_i}{\Sigma}_1^{-1} (x; \truegamma , \sigma)  {\Sigma}_1 (x; \truegamma , \sigma) \bigr),  
\end{align*} 
obtained via consideration of the partial derivative of $\mathrm{tr} \bigl( {\Sigma}_1^{-1} (x; \truegamma , \sigma) {\Sigma}_1 (x; \truegamma , \sigma)   \bigr)\equiv \mathrm{tr}(I_d)$. Hence, we have $\partial_{\sigma} g (\sigma; x) |_{\sigma = \sigma^\dagger} = 0$ 
for all $x \in \mathbb{R}^d$, and then, 
$ 
\partial_{\sigma} U (\sigma) |_{\sigma = \sigma^\dagger} = 0.
$ 
From the compactness of the parameter space (\ref{assump:param_space}), we obtain the limit (\ref{eq:sigma_likelihood}), and deduce the consistency of $\hat{\sigma}_{n}$.
\subsubsection*{Step 2-(ii). Consistency of $\hat{\beta}_n$}  \label{beta_consistency} 
\noindent Due to convergence (\ref{eq:gamma_rate}), and the consistency of $\estgamma$ and $\hat{\sigma}_n$, we have the following result with  proof provided in Section \ref{sec:conv_contrast_beta}. 
%
\begin{lemma}  \label{lemma:conv_contrast_beta}
Under conditions (\ref{assump:coeff})--(\ref{assump:finite_moment}), it holds that as $n \to \infty$, $\Delta \to 0$, $n \Delta \to \infty$, 
\begin{align*} 
  Q (\beta, \estgamma, \hat{\sigma}_n) 
 \probconv  
 \int 
  \bigl( e(x; \trueparam)  -  e (x; \beta, \truegamma, \truesigma ) \bigr)^{\top}  
  {\Sigma}_1^{-1}  (x;\truegamma, \truesigma )
  \bigl( e(x; \trueparam)  -  e (x; \beta, \truegamma, \truesigma) \bigr)\, \nu_{\theta^\dagger} (dx),  
\end{align*} 
uniformly in $\beta \in \Theta_\beta$, where $e (\cdot ; \theta) : \mathbb{R}^d \to \mathbb{R}^d, \; \theta \in \Theta$, is defined as
\begin{align*}
  e (x ;  \theta)  = \big[\,  
     V_{R,0}  (x,  \beta)^{\top},\,  
     \tfrac{1}{2} \hat{V}_0 V_{S, 0} (x, \theta)^{\top}\,\big]^{\top},\quad  x \in \mathbb{R}^d. 
\end{align*}
\end{lemma}

\noindent From this convergence, together with conditions (\ref{assump:param_space}) and (\ref{assump:ident}), we deduce $\hat{\beta}_n \probconv \truebeta$ using the same argument as for the consistency of $\hat{\gamma}_n$.

\subsection{Proof of Theorem \ref{thm:asymp_norm}}  \label{sec:pf_asymp_norm}
We now prove  asymptotic normality for estimator $\est$. Throughout the proof, we consider the hypo-elliptic SDE scenario, $d_R < d$. 

The Taylor expansion of $\partial_{\theta} \ell_{n, \Delta} (\hat{\theta}_n)$, $\hat{\theta}_n = (\hat{\beta}_n, \hat{\gamma}_n, \hat{\sigma}_n )$, around the true parameter vector $\theta^\dagger = (\beta^\dagger, \gamma^\dagger, \sigma^\dagger)$ yields
\begin{align} 
\label{eq:Taylor_likelihood}
  A_{n, \Delta} (\theta^\dagger) = \int_0^1 B_{n, \Delta} \bigl(\theta^\dagger + \lambda (\hat{\theta}_n - \theta^\dagger)  \bigr)\,d \lambda \times C_{n, \Delta} (\hat{\theta}_n), 
\end{align}
where, for $\theta \in \Theta$, we have defined
\begin{align*}
A_{n, \Delta} (\theta) & :=     
 -\Big[\,
 \tfrac{1}{\sqrt{n \Delta}}\,\partial_{\beta} \ell_{n, \Delta} (\theta)^{\top},\, 
 \sqrt{\tfrac{\Delta}{n}}\, \partial_{\gamma} \ell_{n, \Delta} (\theta)^{\top},\, 
 \tfrac{1}{\sqrt{n}}\,\partial_{\sigma}\,\ell_{n, \Delta} (\theta)^{\top} 
 \,\Big]^{\top};  \\
B_{n, \Delta} (\theta) & :=  M_{n, \Delta} \partial^2_{\theta} \ell_{n, \Delta} (\theta) M_{n, \Delta}, \quad  
C_{n, \Delta} (\theta) :=    
 \Big[\,
 \sqrt{n \Delta}\,({\beta} - \beta^\dagger)^{\top},\, 
 \sqrt{\tfrac{n}{\Delta}}\,(\gamma - \gamma^\dagger)^{\top},\, 
 \sqrt{n}\,(\sigma - \sigma^\dagger)^{\top} 
 \,\Big]^{\top}, 
\end{align*}
%
with a scaling matrix 
$ M_{n, \Delta} := \mathrm{diag} (v_{n, \Delta})$,
where $v_{n, \Delta}$ is a $d_{\theta}$-dimensional vector given by 
\begin{align*}
 v_{n, \Delta} 
 = 
 \Bigl[ \; 
 \underbrace{\tfrac{1}{\sqrt{n \Delta}}, \ldots , \tfrac{1}{\sqrt{n \Delta}}}_{d_{\beta} - \mathrm{dim}},  \;   
 \underbrace{\sqrt{\tfrac{\Delta}{n}}, \ldots, \sqrt{\tfrac{\Delta}{n}}}_{d_{\gamma} - \mathrm{dim}},  \; 
 \underbrace{\tfrac{1}{\sqrt{n}}, \ldots, \tfrac{1}{\sqrt{n}}}_{d_{\sigma} - \mathrm{dim}} 
 \; \Bigr]^{\top}, \quad  d_{\theta} = d_{\beta} + d_{\gamma} + d_{\sigma}.  
\end{align*}
%
%
%
We will continue from (\ref{eq:Taylor_likelihood}), using the  the following results.
\begin{lemma} \label{lemma:conv_normal}
Under conditions (\ref{assump:coeff})--(\ref{assump:bd_deriv}),  (\ref{assump:moments}) and (\ref{assump:finite_moment}), it holds that as $n \to \infty$, $\Delta \to 0$, $n \Delta \to \infty$ and $n \Delta^3 \to 0$, 
\begin{align*}
     A_{n, \Delta_n} (\theta^\dagger)  \xrightarrow{\mathcal{L}_{\theta^\dagger}} \mathscr{N} \bigl( 0, 4 I(\theta^\dagger) \bigr).
\end{align*} 
\end{lemma}

\begin{lemma} 
\label{lemma:conv_cov}
Under conditions (\ref{assump:param_space})--(\ref{assump:ident}), it holds that as $n \to \infty$, $\Delta \to 0$ and $n \Delta \to \infty$, 
\begin{align*}
 B_{n, \Delta} 
 \bigl(\theta^\dagger + \lambda (\hat{\theta}_n - \theta^\dagger)  \bigr)
  \xrightarrow{\mathbb{P}_{\theta^\dagger}} 2 I(\theta^\dagger), 
\end{align*}  
uniformly in $\lambda \in (0,1]$. 
\end{lemma}

\noindent We provide the proofs in Section \ref{sec:asymp_norm}. 
Via these results, we immediately obtain the stated asymptotic normality of $\hat{\theta}_n$.  

\section{Proof of Technical Results for Theorem \ref{consistency} and \ref{thm:asymp_norm}}
\label{sec:pf_technical_high_freq}
\subsection{Auxiliary Results}  \label{sec:aux_result}
\noindent Before we prove Lemmas \ref{lemma:conv_contrast_gamma}--\ref{lemma:conv_cov}, we prepare some auxiliary results in this section. 
\begin{lemma} \label{lemma:aux_matrix}
Assume Conditions (\ref{assump:hypo1})-(\ref{assump:hypo2}) holds. For any $x = (x_R, x_S) \in \mathbb{R}^d$, $(\gamma, \sigma) \in \Theta_{\gamma} \times \Theta_{\sigma}$, it holds that 
\begin{align*}
\Lambda_{1, SR} (x; \gamma, \sigma)
+ \tfrac{1}{2} 
\Lambda_{1, SS} (x; \gamma, \sigma ) 
\partial^\top_{x_R} V_{S,0} (x, \gamma) 
= \mathbf{0}_{d_S \times d_R}. 
\end{align*}
\end{lemma}
\begin{proof}
We have  
\begin{align*}
& \Lambda_{1, SR} (x; \gamma, \sigma) 
+ \tfrac{1}{2} \Lambda_{1, SS} (x; \gamma, \sigma ) 
\partial^\top_{x_R} V_{S,0} (x, \gamma)  \nonumber \\[0.2cm] 
& \quad = 
\Lambda_{1, SR} (x; \gamma, \sigma)
\Sigma_{1, RR} (x; \sigma)
\Sigma_{1, RR}^{-1} (x; \sigma) 
+ \tfrac{1}{2} \Lambda_{1, SS} (x; \gamma, \sigma ) 
\partial^\top_{x_R} V_{S,0} (x, \gamma) 
\Sigma_{1, RR} (x; \sigma)
\Sigma_{1, RR}^{-1} (x; \sigma) \nonumber \\[0.2cm]
& \quad = \Bigl( 
\Lambda_{1, SR}
(x; \gamma, \sigma)
\Sigma_{1, RR}(x; \sigma) 
+ \Lambda_{1, SS} (x; \gamma, \sigma )
\Sigma_{1, SR}(x; \gamma, \sigma) \Bigr) 
\Sigma_{1, RR}^{-1} (x; \sigma) \nonumber \\[0.2cm] 
& \quad= \mathbf{0}_{d_S \times d_R}, 
\end{align*}
where we have used the equations: 
\begin{align*}
& \Sigma_{1, SR}(x; \gamma, \sigma) 
\equiv \tfrac{1}{2} \partial^\top_{x_R} V_{S,0} (x, \gamma) 
\Sigma_{1, RR}(x; \sigma);  \\[0.2cm]
& \Lambda_{1, SR} (x; \gamma, \sigma)
\Sigma_{1, RR} (x; \sigma) 
+ \Lambda_{1, SS} (x; \gamma, \sigma )
\Sigma_{1, SR}(x; \gamma, \sigma) = \mathbf{0}_{d_S \times d_R}. 
\end{align*}
\end{proof}
\begin{lemma} \label{Genon_Catalot}
Let $Y_{t_i}$, $U$ be random variables, with $Y_{t_i}$ being $\mathcal{F}_{t_i}$-measurable. The following two convergences imply $\textstyle{\sum_{i=1}^n Y_{t_{i}} \probconv U}$:
\begin{align*}
  \sum_{i=1}^n \trueE [Y_{t_{i}} | \filtration{i-1}]  \probconv  U,  \ \ \ 
  \sum_{i=1}^n \trueE\big[\,\bigl(Y_{t_{i}}\bigr)^2\,|\,\filtration{i-1}\,\big]  \probconv  0. 
\end{align*} 
\end{lemma} 
\begin{proof}
See Lemma 9 in \cite{genon:93} . 
\end{proof} 

\begin{lemma} \label{ergodic_thm}
Let $f : \mathbb{R}^d \times \Theta \to \mathbb{R}$ be differentiable with derivatives of polynomial growth in $x \in \mathbb{R}^d$ uniformly in $\theta \in \Theta$. Assume that conditions 
(\ref{assump:coeff}), (\ref{assump:bd_deriv}), (\ref{assump:moments}), (\ref{assump:finite_moment}) hold. Then, as $n \to \infty$, $\Delta \to 0$ and $n \Delta \to \infty$,   
\begin{align*}
   \frac{1}{n} \sum_{i=1}^n f (\sample{X}{i-1}, \theta) \probconv 
   \int  f (x,  \theta) \truedist (dx),
\end{align*}
uniformly in $\theta \in \Theta$. 
\end{lemma}
\begin{proof}
This is a multivariate version of Lemma 8 in \cite{kessler1997estimation}, so we omit the proof.  
\end{proof}

\begin{lemma} \label{weak_Taylor_expansion}
Let $1 \le i \le  n$, $1 \leq j_1, j_2 \leq d$ and $\Delta > 0, \, \theta \in \Theta$. Under condition (\ref{assump:coeff}), 
it holds that: 
\begin{align*} 
 \mathbb{E}_{\theta} \Bigl[ \nu_i^{j_1} (\Delta ; \theta)  \nu_i^{j_2} (\Delta ; \theta)  | \filtration{i-1} \Bigr]  
& = {\Sigma}_{1, j_1 j_2} (\sample{X}{i-1}; \theta)
+ \Xi_{j_1 j_2}  (\sample{X}{i-1}; \theta) \Delta  
  + \mathcal{E} (\Delta^2, \sample{X}{i-1}, \theta),  
\end{align*} 
for $\mathcal{E}\in\mathscr{S}$
and a $d \times d$ matrix $\Xi = \Xi(x, \theta)$, such that for $1 \leq l_1, l_2 \leq d_R$, $d_R + 1 \leq l_3, l_4 \leq d$,  
\begin{align*}
&  {\Xi}_{l_1 l_2} (x ; \theta)  
 = \tfrac{1}{2} \sum_{k=1}^{d_R} 
  \big\{ \hat{V}_0 V_k^{l_1}(x , \theta)  V_k^{l_2}(x , \theta) 
 +  \hat{V}_0 V_k^{l_2}(x , \theta) V_k^{l_1}(x , \theta)  \bigr\}    \nonumber \\
& \quad  + \tfrac{1}{2} \sum_{k=1}^{d_R}  
\bigl\{ \hat{V}_k V_0^{l_1} (x , \theta) V_k^{l_2} (x , \theta)
+  \hat{V}_k V_0^{l_2} (x , \theta)V_k^{l_1} (x , \theta) \bigl\}  
\; + \; \tfrac{1}{2} \sum_{k_1, k_2 = 1}^{d_R} 
  \hat{V}_{k_1} V_{k_2}^{l_1} (x , \theta) 
  \hat{V}_{k_1} V_{k_2}^{l_2}(x , \theta);   \nonumber  \\[0.2cm]
& {\Xi}_{l_1 l_3} (x ; \theta)  
=    \sum_{k = 1}^{d_R} 
\Bigl( 
  \tfrac{1}{3} \hat{V}_{k_1}V_{0}^{l_1}(x,  \theta) 
 + \tfrac{1}{6} \hat{V}_{0}V_{k_1}^{l_1}(x,  \theta) 
 \Bigr) \hat{V}_{k_2}V_{0}^{l_3}(x,  \theta) \nonumber \\
 &  \quad + \tfrac{1}{6} \sum_{k = 1}^{d_R} 
 V_{k}^{l_1} (x ,  \theta) 
 \bigl(
  \hat{V}_{k} \hat{V}_0 V_0^{l_3} (x ,  \theta) 
  + \hat{V}_{0} \hat{V}_{k} V_0^{l_3} (x,   \theta) 
 \bigr) 
 + \tfrac{1}{6} \sum_{k_1, k_2 = 1}^{d_R}  
 \hat{V}_{k_1} V_{k_2}^{l_1} (x, \theta) 
 \hat{V}_{k_1} \hat{V}_{k_2} V_{0}^{l_3}(x,  \theta) ; \\[0.3cm] 
& {\Xi}_{l_3 l_1} (x ; \theta) = {\Xi}_{l_1 l_3} (x ; \theta); \nonumber \\[0.2cm]  
& {\Xi}_{l_3 l_4}  (x ;  \theta)
  =  \sum_{k = 1}^{d_R} 
\hat{V}_{k} V_0^{l_3} (x , \theta) 
 \Bigl(
   \tfrac{1}{6} \hat{V}_{0} \hat{V}_{k} V_0^{l_4} (x, \theta)  
   + \tfrac{1}{8} \hat{V}_{k} \hat{V}_0 V_0^{l_4} (x ,  \theta) 
 \Bigr) \\ 
& \quad + \sum_{k = 1}^{d_R}
\Bigl(\tfrac{1}{6} \hat{V}_{0} \hat{V}_{k} V_0^{l_3} (x,  \theta)
 + \tfrac{1}{8} \hat{V}_{k} \hat{V}_0 V_0^{l_3} (x,  \theta) 
\Bigr) 
\hat{V}_{k} V_0^{l_4} (x ,  \theta)  
 \; + \; \tfrac{1}{12} \sum_{k_1, k_2 = 1}^{d_R} 
 \hat{V}_{k_1} \hat{V}_{k_2} V_0^{l_3} (x, \theta) \hat{V}_{k_1} \hat{V}_{k_2} V_{0}^{l_4}(x ,  \theta). 
\end{align*} 
%

%
\end{lemma} 

\begin{proof}
The stochastic It\^o-Taylor expansion of 
$X_{\Delta} = \bigl[ X_{R, \Delta}^\top, X_{S, \Delta}^\top \bigr]^\top$ gives as: 
\begin{align} 
 X_{R,\Delta} 
&  = x_{R}
 + \sum_{k=0}^{d_R} V_{R,k}(x , \theta) B_{k, \Delta} 
 + \sum_{k_1, k_2 = 0}^{d_R} 
  \hat{V}_{k_1} V_{R, k_2}(x , \theta) I_{(k_1, k_2)}(\Delta) 
 + \rho_{R} (\Delta , \theta);  \label{eq:expansion_R}  \\ 
X_{S,\Delta} 
&  = x_{S}
 + V_{S,0}(x ,  \theta) \Delta  
 + \sum_{k = 0}^{d_R} 
   \hat{V}_{k} V_{S, 0}(x , \theta) I_{(k,0)}(\Delta)   
 + \sum_{k_1, k_2 = 0}^{d_R}  
   \hat{V}_{k_1} \hat{V}_{k_2}  V_{S, 0}(x ,  \theta) I_{(k_1, k_2, 0)}(\Delta) 
 + \rho_{S} (\Delta , \theta), \nonumber  
\end{align}
where 
$\rho_{R}(\Delta, \theta) := \rho_{R, 1} (\Delta , \theta)  +  \rho_{R, 2} (\Delta , \theta)$ 
and
$ \rho_{S} (t , \theta) := \rho_{S, 1} (\Delta , \theta)  +  \rho_{S, 2} (\Delta , \theta)$ 
with 
\begin{align*}
\rho_{R, 1} (\Delta , \theta) 
 & = \sum_{k=3,4,5}\,\, \sum_{\alpha \in \{0,1, \ldots, d_R\}^k} \hat{V}_{\alpha_1} \cdots \hat{V}_{\alpha_{k-1}} V_{R, \alpha_k} (x,  \theta) I_{\alpha} (\Delta); \\
\rho_{S, 1} (\Delta , \theta) 
& = \sum_{k= 3,4}\,\, \sum_{\alpha \in \{0,1, \ldots, d_R \}^k }
  \hat{V}_{\alpha_1} \cdots \hat{V}_{\alpha_{k}} V_{S, 0} (x,  \theta) 
I_{(\alpha_1, \ldots, \alpha_k, 0)} (\Delta).     
\end{align*}
Under condition (\ref{assump:coeff}), there exist constants $c_1, c_2 > 0$ so that 
$
 \bigl| \mathbb{E}_{\theta} [\rho_{R, 2} (\Delta , \theta) ] \bigr| 
  \leq c_1 \Delta^3
$
and 
\\ 
$
\bigl| \mathbb{E}_{\theta} [\rho_{S, 2} (\Delta , \theta) ] \bigr| \leq c_2 \Delta^3
$. 
From expansion (\ref{eq:expansion_R}) we have for $1 \leq j_1, j_2 \leq d_R$,  
\begin{align*}
\mathbb{E}_{\theta} \Bigl[ \nu_i^{j_1} (\Delta ; \theta)  \nu_i^{j_2} (\Delta ; \theta)  | \filtration{i-1} \Bigr]
& = {\Sigma}_{1, j_1 j_2} (\sample{X}{i-1}; \theta)  
+ \tfrac{\Delta}{2} \sum_{k = 1}^{d_R} V_{k}^{j_1} (\sample{X}{i-1}, \theta) 
 \bigl( \hat{V}_{0}V_{k}^{j_2}(\sample{X}{i-1}, \theta) + \hat{V}_{k}V_0^{j_2}(\sample{X}{i-1}, \theta) \bigr) \\ 
& \quad + 
\tfrac{\Delta}{2} \sum_{k = 1}^{d_R} 
\bigl(\hat{V}_{0}V_{k}^{j_1}(\sample{X}{i-1}, \theta) + \hat{V}_{k}V_0^{j_1}(\sample{X}{i-1}, \theta) \bigr) V_{k}^{j_2} (\sample{X}{i-1}, \theta)  \\ 
& \quad + 
\tfrac{\Delta}{2} \sum_{k_1, k_2 = 1}^{d_R} 
  \hat{V}_{k_1}V_{k_2}^{j_1}(\sample{X}{i-1}, \theta) \hat{V}_{k_1}V_{k_2}^{j_2}(\sample{X}{i-1}, \theta) 
+ \mathcal{E}_1 (\Delta^2, \sample{X}{i-1}, \theta),     
\end{align*}
for $\mathcal{E}_1\in\mathscr{S}$, where we have used that  
\begin{align*}
 & \mathbb{E} [B_{k_1, \Delta} B_{k_2, \Delta}] 
  = \Delta \times \mathbf{1}_{k_1 = k_2}, \ \ 
   \mathbb{E} [B_{k_1, \Delta} I_{(k_2,0)}(\Delta)] 
  = \mathbb{E} [B_{k_1, \Delta} I_{(0, k_2)}(\Delta)] 
  = \tfrac{\Delta^2}{2} \times  \mathbf{1}_{k_1 = k_2};   \\[0.2cm]
& \mathbb{E} [I_{(k_1, k_2)}(\Delta) I_{(k_3, k_4)} (\Delta)] 
  = \tfrac{\Delta^2}{2} \times  \mathbf{1}_{k_1 = k_3, k_2 = k_4},  \ \   
  \mathbb{E} [I_{(k_1, 0)}(\Delta) I_{(k_2, k_3)}(\Delta)] = 0, \\[-0.4cm]
\end{align*}
for $1 \leq k_1, k_2, k_3, k_4 \leq d_R$, and that 
\begin{align*} 
\mathbb{E}_{\theta} [\bigl( X_{R, t_i}^{j_1}  - \mu_{R, \Delta}^{j_1} (\sample{X}{i-1}; \theta) \bigr) 
\rho_{R}^{j_2} (\Delta ,  \theta) | \mathcal{F}_{t_{i-1}}  ] 
  = \mathcal{E}_2 (\Delta^3, \sample{X}{i-1} ,  \theta), \ \ 1 \leq j_1, j_2 \leq d_R,  \\[-0.4cm]
\end{align*}
for $\mathcal{E}_2\in\mathscr{S}$.
Similarly, for $1 \leq j_1 \leq d_R, \; d_R + 1 \leq j_2 \leq d$, we have
\begin{align*}
\mathbb{E}_{\theta} 
  \Bigl[ \nu_i^{j_1} (\Delta ; \theta)  \nu_i^{j_2} (\Delta ; \theta)  | \filtration{i-1} \Bigr]
& = {\Sigma}_{1, j_1 j_2} (\sample{X}{i-1}; \theta) 
+ \Delta \sum_{k  = 1}^{d_R} 
\Bigl( \tfrac{1}{3} \hat{V}_{k}V_{0}^{j_1}(\sample{X}{i-1} ,  \theta) 
+ \tfrac{1}{6} \hat{V}_{0}V_{k}^{j_1}(\sample{X}{i-1},  \theta) \Bigr) \hat{V}_{k}V_{0}^{j_2}(\sample{X}{i-1} ,  \theta)  \\ 
& \quad 
 + \tfrac{\Delta}{6} \sum_{k = 1}^{d_R} V_{k}^{j_1} (\sample{X}{i-1},  \theta) 
 \bigl(
  \hat{V}_{k} \hat{V}_0 V_0^{j_2} (\sample{X}{i-1},  \theta) 
  + \hat{V}_{0} \hat{V}_{k} V_0^{j_2} (\sample{X}{i-1} ,  \theta)  
 \bigr) \\ 
& \quad 
+ \tfrac{\Delta}{6} \sum_{k_1, k_2 = 1}^{d_R}
 \hat{V}_{k_1} V_{k_2}^{j_1}(\sample{X}{i-1},  \theta) 
 \hat{V}_{k_1} \hat{V}_{k_2} V_{0}^{j_2}(\sample{X}{i-1},  \theta)  
+ \mathcal{E}_3 (\Delta^2, \sample{X}{i-1} , \theta),   
\end{align*} 
for $\mathcal{E}_3\in\mathscr{S}$, where we used
\begin{align*} 
& \mathbb{E} [B_{k_1, \Delta} I_{(k_2, 0)}(\Delta)] 
  = \tfrac{\Delta^2}{2} \times \mathbf{1}_{k_1 = k_2};  \nonumber \\[0.2cm]  
& \mathbb{E} [B_{k_1, \Delta} I_{(k_2, 0, 0)}(\Delta)] 
 = \mathbb{E} [B_{k_1, \Delta} I_{(0, k_2, 0)}(\Delta)]
 = \tfrac{\Delta^3}{6} \times \mathbf{1}_{k_1 = k_2}, \ \ 
 \mathbb{E} [B_{k_1, \Delta} I_{(k_2, k_3, 0)}(\Delta)] 
 = 0; \\[0.2cm]
& \mathbb{E} [I_{(k_1, 0)}(\Delta) I_{(k_2, 0)}(\Delta)] 
  = \tfrac{\Delta^3}{3}  \times \mathbf{1}_{k_1 = k_2},  \ \ 
\mathbb{E}  [I_{(0, k_1)}(\Delta) I_{(k_2, 0)}(\Delta)] = \tfrac{\Delta^3}{6} \times \mathbf{1}_{k_1 = k_2};   \\[0.2cm]
& \mathbb{E} [I_{(k_1, k_2)}(\Delta) I_{(k_3, k_4, 0)}(\Delta)] 
 = \tfrac{\Delta^3}{6} \times  \mathbf{1}_{k_1 = k_3, k_2 = k_4},\\[-0.4cm]
\end{align*} 
for $1 \le k_1, k_2, k_3, k_4 \le d_R$, and 
\begin{align*}
& \mathbb{E}_{\theta} [\bigl( X_{R, t_i}^{j_1}  - \mu_{R, \Delta}^{j_1} (\sample{X}{i-1}; \theta) \bigr) 
\rho_{S}^{j_2} (\Delta, \theta) | \filtration{i-1}  ] 
  = \mathcal{E}_4 (\Delta^4, \sample{X}{i-1}, \theta);  \\[0.2cm]
& \mathbb{E}_{\theta} [\rho_{R}^{j_1} (\Delta, \theta) \bigl( X_{S, t_i}^{j_2}  - \mu_{S, \Delta}^{j_2} (\sample{X}{i-1}; \theta) \bigr)  | \filtration{i-1}   ]  
  = \mathcal{E}_5 (\Delta^4, \sample{X}{i-1}, \theta), \\[-0.4cm]
\end{align*}
for $1 \leq j_1 \leq d_R, \; 1 \leq j_2 \leq d_S$, where $\mathcal{E}_4, \mathcal{E}_5 \in \mathscr{S}$. Finally, for $d_R + 1 \leq j_1, j_2 \leq d$, 
\begin{align*}
 \mathbb{E}_{\theta} \Bigl[ \nu_i^{j_1} (\Delta ; \theta)  \nu_i^{j_2} (\Delta ; \theta)  | \filtration{i-1} \Bigr]
& = {\Sigma}_{1, j_1 j_2} (\sample{X}{i-1}; \theta)  
 + \Delta \sum_{k = 1}^{d_R} 
\hat{V}_{k} V_0^{j_1} (\sample{X}{i-1},  \theta) 
  \Bigl(
  \tfrac{1}{8} \hat{V}_{k} \hat{V}_0 V_0^{j_2} (\sample{X}{i-1},  \theta) 
  +   \tfrac{1}{6} \hat{V}_{0} \hat{V}_{k} V_0^{j_2} (\sample{X}{i-1} , \theta)  \Bigr) \\ 
& \quad   
 + \Delta \sum_{k = 1}^{d_R}
\Bigl( \tfrac{1}{8} \hat{V}_{k} \hat{V}_0 V_0^{j_1} (\sample{X}{i-1} , \theta) 
  + \tfrac{1}{6} \hat{V}_{0} \hat{V}_{k} V_0^{j_1} (\sample{X}{i-1} ,  \theta)  \Bigr) 
\hat{V}_{k} V_0^{j_2} (\sample{X}{i-1} , \theta) 
 \\ 
& \quad 
 + \tfrac{\Delta}{12} \sum_{k_1, k_2 = 1}^{d_R}  
 \hat{V}_{k_1} V_{k_2} V_0^{j_1} (\sample{X}{i-1} ,  \theta) \hat{V}_{k_1} \hat{V}_{k_2} V_{0}^{j_2}(\sample{X}{i-1} , \theta)  
+ \mathcal{E}_6 (\Delta^2, \sample{X}{i-1} ,  \theta),   
\end{align*} 
for $\mathcal{E}_6\in\mathscr{S}$, where we used 
\begin{align*}
& \mathbb{E} [I_{(k_1, 0)}(\Delta)I_{(k_2, 0)}(\Delta)]  
= \tfrac{\Delta^3}{3} \times \mathbf{1}_{k_1 = k_2}, \ \ 
 \mathbb{E} [I_{(k_1, 0)}(\Delta) I_{(k_2, 0, 0)} (\Delta)] 
  =  \tfrac{\Delta^4}{8} \times  \mathbf{1}_{k_1 = k_2} ;   \\[0.2cm]
& \mathbb{E} [I_{(k_1, 0)}(\Delta) I_{(0, k_2, 0)} (\Delta)] 
=  \tfrac{\Delta^4}{6}  \times \mathbf{1}_{k_1 = k_2},   \ \   
   \mathbb{E} [I_{(k_1, k_2, 0)}(\Delta) I_{(k_3, k_4, 0)}(\Delta)] 
   = \tfrac{\Delta^4}{12} \times \mathbf{1}_{k_1 = k_3, k_2 = k_4}, \\[-0.4cm]
\end{align*}
and 
\begin{align*}
& \mathbb{E}_{\theta} [\bigl( X_{S, t_i}^{j_1}  - \mu_{S, \Delta}^{j_1} (\sample{X}{i-1}; \theta) \bigr) 
\rho_{S}^{j_2} (\Delta, \theta) | \filtration{i-1} ] 
  = \mathcal{E}_7 (\Delta^5, \sample{X}{i-1} ,  \theta),    \ \ 1 \leq j_1, j_2 \leq d_S, \\[-0.4cm]
\end{align*}
%
for $\mathcal{E}_7\in\mathscr{S}$. The proof is now complete. 
\end{proof}
\begin{lemma} \label{canonical_conv}
Let $1 \leq j_1, j_2 \leq d$ and let $f : \mathbb{R}^d \times \Theta \to \mathbb{R}$ satisfy the regularity assumption of Lemma \ref{ergodic_thm}. Under conditions (\ref{assump:coeff}), (\ref{assump:bd_deriv}), (\ref{assump:moments}) and (\ref{assump:finite_moment}), it holds that as $n \to \infty$, $\Delta \to 0$, $n \Delta \to \infty$, 
\begin{align}
& \tfrac{1}{n} \sum_{i=1}^n 
 f (\sample{X}{i-1}, \theta)  \nu_i^{j_1} (\Delta, \trueparam) \nu_i^{j_2} (\Delta, \trueparam) 
   \probconv \int f (x, \theta) {\Sigma}_{1, j_1 j_2} (x; \truegamma, \truesigma) \truedist (dx);  \label{canonical_conv11} \\
 & \tfrac{1}{n \sqrt{\Delta}} \sum_{i=1}^n   f (\sample{X}{i-1}, \theta)  \nu_i^{j_1} (\Delta, \trueparam)   \probconv 0,  
 \label{canonical_conv12} 
\end{align}
uniformly in $\theta \in \Theta$. 
\end{lemma}

\begin{proof}
First, we will prove (\ref{canonical_conv11}). 
Let 
$$ \xi_i^{j_1 j_2}  (\theta) := \tfrac{1}{n} f (\sample{X}{i-1}, \theta)  \nu_i^{j_1} (\Delta, \trueparam) \nu_i^{j_2} (\Delta, \trueparam), \ \  1 \le i \le n, \; 1 \leq j_1, j_2 \leq d.$$ 
We have 
\begin{align*}
\sum_{i=1}^n  \trueE \bigl[ \xi_i^{j_1 j_2} (\theta)  | \filtration{i-1} \bigr] 
 = &\; \tfrac{1}{n} \sum_{i=1}^n  f (\sample{X}{i-1}, \theta) 
 {\Sigma}_{1, j_1j_2} (\sample{X}{i-1}; \truegamma, \truesigma) 
  + \tfrac{1}{n} \sum_{i=1}^n \mathcal{E}_1 (\Delta, \sample{X}{i-1},  \theta);   \\ 
\sum_{i=1}^n  \trueE \bigl[ (\xi_i^{j_1 j_2} (\theta) )^2  | \filtration{i-1} \bigr] 
  = &\; \tfrac{1}{n^2} \sum_{i=1}^n \mathcal{E}_2 (1, \sample{X}{i-1}, \theta).
\end{align*}
for $\mathcal{E}_1, \mathcal{E}_2 \in \mathscr{S}$.
Hence, it follows from Lemma \ref{ergodic_thm} that 
\begin{align*}
\sum_{i=1}^n  \trueE \bigl[ \xi_i^{j_1 j_2}  (\theta) | \filtration{i-1} \bigr] 
&\; \probconv  \int f (x, \theta) 
{\Sigma}_{1, j_1 j_2} (x; \truegamma, \truesigma) \truedist (dx);   \nonumber \\ 
\sum_{i=1}^n  \trueE \bigl[ (\xi_i^{j_1 j_2} (\theta))^2  | \filtration{i-1} \bigr] &\; \probconv 0. 
\end{align*}
Then, Lemma \ref{Genon_Catalot} yields convergence (\ref{canonical_conv11}). Next, we turn to (\ref{canonical_conv12}). Let 
$$
\eta_i^{j_1}  (\theta) 
:= \tfrac{1}{n \sqrt{\Delta}} f (\sample{X}{i-1}, \theta)  \nu_i^{j_1} (\Delta, \trueparam), 
\qquad  1 \le i \le  n, \quad 1 \le j_1 \le  d. 
$$ 
Equation (\ref{eq:mean_nu}) and Lemma \ref{ergodic_thm} give that, as $n \to \infty$, $\Delta \to 0$ and $n \Delta \to \infty$,  
\begin{align}
\sum_{i=1}^n  \trueE \bigl[\eta_i^{j_1}  (\theta)  | \filtration{i-1} \bigr] 
 = &\; \tfrac{1}{n}  \sum_{i=1}^n  \mathcal{E}_3 (\Delta, \sample{X}{i-1}, \theta)  \probconv 0; \nonumber \\ 
 \sum_{i=1}^n  \trueE \bigl[ ( \eta_i^{j_1}  (\theta) )^2  | \filtration{i-1} \bigr] 
  = &\; \tfrac{1}{n \Delta} \times \tfrac{1}{n}  \sum_{i=1}^n  \mathcal{E}_4 (1, \sample{X}{i-1}, \theta)  \probconv 0, \nonumber  
\end{align} 
where $\mathcal{E}_3, \, \mathcal{E}_4 \in \mathscr{S}$. 
Finally, we obtain convergence (\ref{canonical_conv12}) due to Lemma \ref{Genon_Catalot}. 
\end{proof}  
%
%
\begin{lemma} \label{hermite_expectation}
Let 
$1 \leq i \leq n$, 
$1 \leq j_1, j_2 \leq d$, 
$1 \leq k_1 \leq d_{\beta}$, 
$d_{\beta} + 1 \leq k_2 \leq d_{\beta} + d_{\gamma}$,
$d_{\beta} + d_{\gamma} + 1 \leq k_3 \leq d_{\theta}$. 
Under condition (\ref{assump:coeff}), it holds that for any $\Delta >0$, $\theta \in \Theta$, 
\begin{align}
  \mathbb{E}_{\theta} [\mathcal{H}_{(j_1, j_2)} ( \Delta, X_{t_{i-1}}, X_{t_{i}} ; \theta)| \filtration{i-1} ] 
 & =  \mathcal{E}_1 (\Delta, X_{t_{i-1}}, \theta); \label{hermite_1}  \\[0.2cm]
 \mathbb{E}_{\theta}[ \deriv{\theta}{k_1} \mathcal{H}_{(j_1, j_2)} ( \Delta, X_{t_{i-1}}, X_{t_{i}} ; \theta)| \mathcal{F}_{t_{i-1}}]  
 & =  \mathcal{E}_{2} (\Delta^2, X_{t_{i-1}},\theta); \label{hermite_2}  
 \\[0.2cm]
 \mathbb{E}_{\theta}  [ \deriv{\theta}{k_2} 
 \mathcal{H}_{(j_1, j_2)} ( \Delta, X_{t_{i-1}}, X_{t_{i}} ; \theta)  | \mathcal{F}_{t_{i-1}}]  
 &=  \deriv{\theta}{k_2} {\Sigma}^{-1}_{1, j_1 j_2} (X_{t_{i-1}}; \gamma, \sigma) 
 +  \mathcal{E}_{3}  (\Delta, X_{t_{i-1}},  \theta);   \label{hermite_3}  \\[0.2cm]
 \mathbb{E}_{\theta}  [ \deriv{\theta}{k_3} \mathcal{H}_{(j_1, j_2)} ( \Delta, X_{t_{i-1}}, X_{t_{i}} ; \theta) 
| \mathcal{F}_{t_{i-1}}] 
 &=   \deriv{\theta}{k_3} {\Sigma}^{-1}_{1, j_1, j_2} (X_{t_{i-1}}; \gamma, \sigma) 
 +  \mathcal{E}_{4} (\Delta, X_{t_{i-1}}, \theta),   \label{hermite_4}
\end{align} 
where $\mathcal{E}_1, \mathcal{E}_2, \mathcal{E}_3, \mathcal{E}_4  \in \mathscr{S}$. 
\end{lemma} 

\begin{proof}
We first check (\ref{hermite_1}). Using Lemma \ref{weak_Taylor_expansion} with (\ref{eq:hermite_second}), we immediately have, for $\mathcal{E}_5  \in \mathscr{S}$,
\begin{align*}
& \mathbb{E}_{\theta} [ \mathcal{H}_{(j_1, j_2)} ( \Delta, X_{t_{i-1}}, X_{t_{i}} ; \theta)| \mathcal{F}_{t_{i-1}}]   \\ 
& =  \sum_{i_1, i_2 = 1}^d  
 {\Sigma}_{1, j_1 i_1}^{-1} (\sample{X}{i-1}; \gamma, \sigma) 
 {\Sigma}_{1, j_2 i_2}^{-1} (\sample{X}{i-1}; \gamma, \sigma) 
 {\Sigma}_{1, i_1 i_2} (\sample{X}{i-1}; \gamma, \sigma) 
 + \mathcal{E}_5 (\Delta, \sample{X}{i-1}, \theta)  
 - {\Sigma}_{1, j_1 j_2}^{-1} (\sample{X}{i-1}; \gamma, \sigma)  \nonumber \\[0.3cm] 
%
& = \mathcal{E}_1 (\Delta, \sample{X}{i-1}, \theta). 
\end{align*} 
Next, we turn to (\ref{hermite_2}). We have 
\begin{align*} 
& \mathbb{E}_{\theta}  [ \deriv{\theta}{k_1} \mathcal{H}_{(j_1, j_2)} (\Delta, \sample{X}{i-1}, \sample{X}{i}; \theta) | \filtration{i-1}]   \\[0.2cm]
& = \sum_{i_1 ,i_2 = 1}^{d}   
  {\Sigma}_{1, j_1 i_1}^{-1} (\sample{X}{i-1}; \gamma, \sigma) 
  {\Sigma}_{1, j_2 i_2}^{-1} (\sample{X}{i-1}; \gamma, \sigma) 
  \mathbb{E}_{\theta} [\deriv{\theta}{k_1} 
  \nu_i^{i_1}(\Delta; \theta)  \nu_i^{i_2}(\Delta; \theta) | \filtration{i-1}]      \\
& \quad  +  \sum_{i_1 ,i_2 = 1}^{d}   
  {\Sigma}_{1, j_1 i_1}^{-1} (\sample{X}{i-1}; \gamma, \sigma) 
  {\Sigma}_{1,j_2 i_2}^{-1} (\sample{X}{i-1}; \gamma, \sigma) 
  \mathbb{E}_{\theta} [ \nu_i^{i_1}(\Delta; \theta)  
  \deriv{\theta}{k_1} \nu_i^{i_2}(\Delta; \theta) | \filtration{i-1} ]  \\[0.3cm]
& =  \mathcal{E}_2  (\Delta^2, \sample{X}{i-1}, \theta), 
\end{align*}
where we used (\ref{eq:mean_nu}) and (\ref{eq:nu_deriv_beta}). Finally, let show (\ref{hermite_4}). The formula (\ref{hermite_3}) is derived in the similar way to (\ref{hermite_4}) and we will omit the derivation. From Lemma \ref{weak_Taylor_expansion}, we have 
\begin{align*}
& \mathbb{E}_{\theta} 
 \bigl[  \deriv{\theta}{k_3} \mathcal{H}_{(j_1, j_2)} (\Delta, \sample{X}{i-1}, \sample{X}{i}; \theta) | \filtration{i-1} 
 \bigr] \nonumber \\ 
& = \sum_{i_1, i_2 = 1}^d 
 \deriv{\theta}{k_3}
 \left( 
   {\Sigma}_{1, j_1 i_1}^{-1} (\sample{X}{i-1}; \gamma, \sigma) 
   {\Sigma}_{1, j_2 i_2}^{-1} (\sample{X}{i-1}; \gamma, \sigma) 
 \right)  {\Sigma}_{1, i_1i_2} (\sample{X}{i-1}; \gamma, \sigma)  
 - \deriv{\theta}{k_3} {\Sigma}_{1, j_1 j_2}^{-1} (\sample{X}{i-1}; \gamma, \sigma) 
+  \mathcal{E}_4 (\Delta, \sample{X}{i-1}, \theta)  \\
& =  \deriv{\theta}{k_3} {\Sigma}_{1, j_1 j_2}^{-1} (\sample{X}{i-1}; \gamma, \sigma)  + \mathcal{E}_4  (\Delta, \sample{X}{i-1}, \theta).  
\end{align*}
The proof is now complete. 
\end{proof} 

\subsection{Proof of Technical Results for Theorem 4.3} \label{sec:pf_lemma_consistency} 
  
\subsubsection{Proof of Lemma \ref{lemma:conv_contrast_gamma}} \label{sec:conv_contrast_gamma}
\noindent We define $\textstyle{\frac{\Delta}{n} \ell_{n, \Delta} (\theta)} = \Gamma_1 (\theta) + \Gamma_2 (\theta) + \Gamma_3 (\theta)$, where 
\begin{align*}
   \Gamma_1 (\theta) & =  
   \tfrac{\Delta}{n} \sum_{i=1}^n  \bigl( X_{t_i} -  \mu_{\Delta} (\sample{X}{i-1} ; \theta) \bigr)^{\top} \Sigma^{-1} (\Delta, \sample{X}{i-1} ; \gamma, \sigma) \bigl( X_{t_i} - \mu_{\Delta} (\sample{X}{i-1} ; \theta)  \bigr) ;   \nonumber \\ 
  \Gamma_2 (\theta) & =  \tfrac{\Delta}{n} \sum_{i=1}^n 
   \log |{\Sigma}_1 (\sample{X}{i-1}; \gamma, \sigma)|,     
  \ \ \ 
 \Gamma_3 (\theta) =  - \tfrac{\Delta^2}{n} \sum_{i=1}^n \Phi_2 (\Delta, X_{t_{i-1}}, X_{t_{i}}; \theta).     
\end{align*}
We will show that, as $n \to \infty$, $\Delta \to 0$, $n \Delta \to \infty$, 
\begin{align*}
 & \Gamma_1 (\theta) \xrightarrow{\mathbb{P}_{\trueparam}}   
  \int
  \bigl( V_{S,0} (x ,  \truegamma) -  V_{S,0} (x ,  \gamma) \bigr)^{\top}
  \Lambda_{1, SS} (x; \gamma, \sigma)
  \bigl( V_{S,0} (x,  \truegamma) -  V_{S, 0} (x , \gamma) \bigr)  \truedist(dx); \nonumber  \\ 
 &  \Gamma_2 (\theta) 
    \xrightarrow{\mathbb{P}_{\trueparam}} 0, 
 \ \ \  \Gamma_3 (\theta)  
    \xrightarrow{\mathbb{P}_{\trueparam}} 0,  \nonumber
\\[-0.4cm]
\end{align*}
uniformly in $\theta = (\beta, \gamma, \sigma) \in \Theta$. 

From Lemma \ref{ergodic_thm}, we immediately have 
$ 
\Gamma_2 (\theta) \xrightarrow{\mathbb{P}_{\trueparam}} 0, 
$ 
uniformly in $\theta \in \Theta$. We next turn to $\Gamma_1 (\theta)$. We introduce 
$\Gamma_1 (\theta) = \Gamma_{1,1} (\theta)  +   \Gamma_{1,2} (\theta)  +  \Gamma_{1,3} (\theta)$  
with 
{\small \begin{align*}
%
%
\Gamma_{1,1} (\theta) & 
  =  
    \tfrac{\Delta}{n} \sum_{i=1}^n  
   \nu_i(\Delta ; \trueparam)^{\top} 
   {\Sigma}_1^{-1} (\sample{X}{i-1}; \gamma, \sigma) 
   \nu_i(\Delta ;  \trueparam);  \\ 
\Gamma_{1,2} (\theta) & 
 = \tfrac{2 \Delta}{n} \sum_{i=1}^n  \bigl( X_{t_i} - \mu_{\Delta}  
  (\sample{X}{i-1}; \trueparam)  \bigr)^{\top} 
  \Sigma^{-1} (\Delta, \sample{X}{i-1} ; \gamma, \sigma) 
  \bigl( \mu_{\Delta} ( \sample{X}{i-1}; \trueparam) -  \mu_{\Delta} (\sample{X}{i-1}; \theta) \bigr);  \\ 
\Gamma_{1,3} (\theta) 
& =  \tfrac{\Delta}{n}  \sum_{i=1}^n \bigl( \mu_{\Delta} ( \sample{X}{i-1}; \trueparam) 
 -  \mu_{\Delta} (\sample{X}{i-1}; \theta) \bigr)^{\top}  
 \Sigma^{-1} (\Delta, \sample{X}{i-1} ; \gamma, \sigma)
  \bigl( \mu_{\Delta} ( \sample{X}{i-1}; \trueparam) -  \mu_{\Delta} (\sample{X}{i-1}; \theta) \bigr). 
\end{align*} }
%
From Lemmas \ref{Genon_Catalot}, \ref{ergodic_thm} and \ref{canonical_conv}, we immediately have $\Gamma_{1,k} (\theta) \probconv 0, \;  k = 1,2.$ Since we have
\begin{align*} 
 \Gamma_{1,3} (\theta)   
& = \tfrac{1}{n} \sum_{i=1}^n   
\bigl( V_{S,0} (x ,  \truegamma) - V_{S,0}  (x ,  \gamma) \bigr)^{\top} 
\Lambda_{1, SS} (x; \gamma, \sigma)  
\bigl( V_{S,0} (x ,  \truegamma) - V_{S,0}  (x ,  \gamma) \bigr) \Bigl|_{x= \sample{X}{i-1}} 
\nonumber \\
& \qquad \qquad + \tfrac{1}{n} \sum_{i=1}^n  \mathcal{E} (\Delta, \sample{X}{i-1}, \theta),
\end{align*}
for $\mathcal{E}  \in \mathscr{S}$, Lemmas \ref{Genon_Catalot} and \ref{ergodic_thm} yield 
\begin{align*}
   \Gamma_{1,3} (\theta)   \probconv 
   \int 
    \bigl( V_{S,0} (x,  \truegamma) -  V_{S,0} (x, \gamma) \bigr)^{\top}
    {\Lambda}_{1, SS}
    (x; \gamma, \sigma) 
    \bigl( V_{S,0} (x,  \truegamma) -  V_{S,0} (x ,  \gamma) \bigr)  \truedist(dx),  
\end{align*} 
uniformly in $\theta \in \Theta$. Thus, we have the convergence of $\Gamma_1(\theta)$ in probability. Similarly, for $\Gamma_3 (\theta)$, due to term $\Delta$ in front of the definition of $\Gamma_3 (\theta)$, we have $\Gamma_3 (\theta) \rightarrow 0$, in probability under $\mathbb{P}_{\theta^\dagger}$, 
uniformly in $\theta \in \Theta$.  The proof is now complete.
\subsubsection{Proof of Lemma \ref{lemma:taylor_gamma_conv}}
\label{sec:pf_taylor_gamma_conv}
\noindent Convergence (\ref{eq:A_gamma}) is deduced from the proof of Lemma \ref{lemma:conv_cov} in Section \ref{pf_conv_cov}, so here we only provide the proof of convergence (\ref{eq:B_gamma}). 
We have for any $(\beta, \sigma) \in \Theta_{\beta} \times \Theta_{\gamma}$, 
\begin{align*}
    \tfrac{\sqrt\Delta}{n} 
    \partial_{\gamma_k} \contrast{\beta, \truegamma, \sigma} 
    = A_{\gamma_k}^{(1)} (\beta, \sigma)
     + A_{\gamma_k}^{(2)} (\beta, \sigma), 
     \ \ 1 \le k \le d_\gamma, 
\end{align*}
where we have set: 
\begin{align*}
    & A_{\gamma_k}^{(1)} (\beta, \sigma)
    \equiv \tfrac{\sqrt\Delta}{n} 
    \sum_{i = 1}^n \, \biggl\{
    - 2 \sum_{\substack{1 \le j_1 \le d \\ d_R +1 \le j_2 \le d}} \nu_{i}^{j_1} (\Delta; \beta, \truegamma, \sigma)
    \Sigma^{-1}_{1, j_1 j_2} (\sample{X}{i-1}; \truegamma, \sigma) \tfrac{\partial_{\gamma_k} V_0^{j_2} (\sample{X}{i-1}, \truegamma)}{\sqrt\Delta} \nonumber \\[0.3cm]
    & \quad - \sqrt\Delta 
    \sum_{\substack{1 \le j_1 \le d \\ d_R +1 \le j_2 \le d}} \nu_{i}^{j_1} (\Delta; \beta, \truegamma, \sigma)
    \Sigma^{-1}_{1, j_1 j_2} (\sample{X}{i-1}; \truegamma, \sigma) {\partial_{\gamma_k} \hat{V}_0 V_0^{j_2} (\sample{X}{i-1}, \beta, \truegamma, \sigma )} \nonumber \\[0.1cm]
    & \qquad + 
    \nu_i (\Delta, \beta, \truegamma, \sigma)^\top
    \partial_{\gamma_k} \Sigma^{-1}_{1} (\sample{X}{i-1}; \truegamma, \sigma)
    \nu_i (\Delta, \beta, \truegamma, \sigma) 
    +  \partial_{\gamma_k} 
    \log |\Sigma_1 (\sample{X}{i-1}; \truegamma, \sigma)| \biggr \};  \\[0.5cm] 
     & A_{\gamma_k}^{(2)} (\beta, \sigma)
    \equiv - \tfrac{2 \sqrt{\Delta^3}}{n} \sum_{i=1}^n
    \partial_{\gamma_k} \Phi_2 (\Delta, \sample{X}{i-1}, \sample{X}{i} ;  \beta, \truegamma, \sigma).
\end{align*}
Since we have 
\begin{align*}
    \nu_i (\Delta, \beta, \truegamma, \sigma)
    = \nu_i ( \Delta, \trueparam) 
     + \sqrt\Delta 
     \begin{bmatrix}
      V_{R,0} (\sample{X}{i-1}, \truebeta) 
      - V_{R,0} (\sample{X}{i-1}, \beta) \\[0.1cm] 
      \tfrac{1}{2} \hat{V}_0 V_{S,0} (\sample{X}{i-1}, \trueparam) 
      - \tfrac{1}{2} \hat{V}_0 V_{S,0} (\sample{X}{i-1}, \beta, \truegamma, \sigma) 
     \end{bmatrix}, 
\end{align*}
we immediately obtain from Lemmas \ref{ergodic_thm}, \ref{canonical_conv} that, as $n \to \infty$, $n \Delta_n \to \infty$ and $\Delta_n \to 0$, then $A_{\gamma_k}^{(1)} (\beta, \sigma) \rightarrow 0$, in probability under $\mathbb{P}_{\theta^\dagger}$,  uniformly in $\beta $ and $\sigma$, for all $1 \le k \le d_\gamma$. Similarly, due to the presence of $\sqrt{\Delta^3}$, 
we obtain $A_{\gamma_k}^{(2)} (\beta, \sigma) \rightarrow 0$, in probability under $\mathbb{P}_{\theta^\dagger}$, uniformly in $\beta $ and $\sigma$. We now obtain convergence (\ref{eq:A_gamma}) and the proof is complete. 
\subsubsection{Proof of Lemma \ref{lemma:conv_contrast_gamma}}
\label{sec:conv_contrast_sigma} 
\noindent We write $\textstyle{ T(\theta) = \sum_{l=1}^3 T_l (\theta)}, \; \theta \in \Theta$, where 
 \begin{align}
T_1 (\theta) & :=  
 \tfrac{1}{n} \sum_{i=1}^n \bigl(X_{t_i} - \mu_{\Delta} (\sample{X}{i-1}; \theta) \bigr)^{\top}
 \Sigma^{-1} (\Delta, \sample{X}{i-1} ; \gamma, \sigma) \bigl(X_{t_i} - \mu_{\Delta} (\sample{X}{i-1}; \theta) \bigr);  \nonumber \\
T_2 (\theta) & :=  \tfrac{1}{n} \sum_{i=1}^n 
 \log |{\Sigma}_1 (\sample{X}{i-1}; \gamma, \sigma)|,  \ \ \ 
T_3 (\theta) :=  - \tfrac{2 \Delta}{n} \sum_{i=1}^n \Phi_2 (\Delta, \sample{X}{i-1}, \sample{X}{i}; \theta). \nonumber 
\end{align}
We will show that as $n \to \infty$, $\Delta \to 0$, $n \Delta \to \infty$, 
\begin{align}
  T_{1} (\beta, \hat{\gamma}_n, \sigma)  
   & \xrightarrow{\mathbb{P}_{\trueparam}} 
 \int 
  \mathrm{tr} \bigl( {\Sigma}^{-1}_1 (x; \truegamma, \sigma) 
  {\Sigma}_1 (x; \truegamma, \truesigma) \bigr) 
  \nu_{\trueparam}(dx);  \label{T1_conv} \\ 
 T_2  (\beta, \hat{\gamma}_n, \sigma)  
 & \xrightarrow{\mathbb{P}_{\trueparam}} 
  \int  \log | {\Sigma}_1  (x;  \truegamma, \sigma) | \nu_{\trueparam}(dx) ;  \label{T2_conv}\\ 
  T_3  (\beta, \hat{\gamma}_n, \sigma) 
 & \xrightarrow{\mathbb{P}_{\trueparam}} 0,  \label{T3_conv}
\end{align} 
uniformly in $(\beta, \sigma) \in \Theta_\beta \times \Theta_\sigma$. 
\\

\noindent
{\bf Proof of (\ref{T1_conv}) and (\ref{T2_conv}).}
Via Lemma \ref{ergodic_thm}, we immediately obtain convergence (\ref{T2_conv}). We will check convergence (\ref{T1_conv}). We have 
\begin{align} \label{T1_decomp}
T_1 (\theta) = T_{1,1} (\theta) + T_{1,2} (\theta) + T_{1,3} (\theta),
\end{align}
where $\textstyle{T_{1, i} (\theta) = \frac{1}{\Delta}\Gamma_{1, i} (\theta), \; i = 1,2,3}$. 
From Lemma \ref{canonical_conv} and the continuous mapping theorem, as $n \to \infty$, $\Delta \to 0$, $n \Delta \to \infty$, 
\begin{align*}
T_{1,1} ( \beta, \hat{\gamma}_n, \sigma )  
& \probconv   \sum_{j_1, j_2 = 1}^d \int  
{\Sigma}_{1, j_1 j_2}^{-1}  (x; \truegamma, \sigma) 
{\Sigma}_{1, j_1 j_2} (x; \truegamma, \truesigma \, ) \truedist (dx)  
=  \int \mathrm{tr} \bigl( 
{\Sigma}^{-1}_1  (x; \truegamma, \sigma) 
{\Sigma}_1  (x; \truegamma, \truesigma ) \bigr) \truedist (dx),
\end{align*} 
uniformly in $(\beta, \sigma) \in \Theta_{\beta} \times \Theta_{\sigma}$. We look at the convergence of term $ T_{1,2} (\beta, \hat{\gamma}_n, \sigma )$. We have 
\begin{align*}
& T_{1,2} (\beta, \hat{\gamma}_n, \sigma )  
= \tfrac{1}{n} \sum_{i=1}^n \, \biggl\{ {2 \sqrt{\Delta}}
  \sum_{\substack{1 \leq j_1 \leq d \\ 1 \leq j_2 \leq d_R}}     \nu_i^{j_1} (\Delta; \trueparam)
 {\Sigma}^{-1}_{1, j_1 j_2} (x;  \estgamma, \sigma) 
  \bigl( V_0^{j_2} (x ,  \truebeta) - V_0^{j_2} (x , \beta) \bigr)    \nonumber \\ 
& \quad +  {2}
\sum_{\substack{1 \leq j_1 \leq d \\ d_R + 1 \leq j_2 \leq d}}   \nu_i^{j_1} (\Delta; \trueparam) 
{\Sigma}^{-1}_{1, j_1 j_2} (x;  \estgamma, \sigma) 
  \frac{ V_0^{j_2} (x ,  \truegamma) - V_0^{j_2} (x ,  \estgamma)  }{\sqrt{\Delta}}    \\[0.2cm]
& \quad  +  {\sqrt{\Delta}} \sum_{\substack{1 \leq j_1 \leq d \\ d_R + 1 \leq j_2 \leq d}}    \nu_i^{j_1} (\Delta; \trueparam) 
 {\Sigma}^{-1}_{1, j_1 j_2} (x;  \estgamma, \sigma)  
 \bigl( 
  \hat{V}_0 V_0^{j_2} (x,  \trueparam) - 
  \hat{V}_0 V_0^{j_2} \left(x , (\beta, \estgamma, \sigma ) \right) 
  \bigr) \biggr\} \bigr|_{x=\sample{X}{i-1}} .   \nonumber 
\end{align*} 
Using Lemma \ref{canonical_conv}, we have $T_{1,2} (\beta, \hat{\gamma}_n, \sigma)  \probconv 0$, uniformly in $(\beta, \sigma) \in \Theta_{\beta} \times \Theta_{\sigma}$. Next, we turn to term $T_{1,3} (\beta, \hat{\gamma}_n, \sigma )$. We have 
$ \textstyle{ T_{1,3} (\beta, \hat{\gamma}_n, \sigma ) 
= \sum_{l=1}^3 T_{1,3,l}(\beta, \hat{\gamma}_n, \sigma )} $, where
\begin{align}
T_{1,3,1} (\beta, \hat{\gamma}_n, \sigma)  
  & =  \tfrac{2}{n} \sum_{i=1}^n \sum_{\substack{1 \leq j_1 \leq d_R \\ d_R + 1 \leq j_2 \leq d}}
  {\Sigma}_{1, j_1 j_2}^{-1} (\sample{X}{i-1};  \estgamma, \sigma)  
  \bigl( V_0^{j_1} (\sample{X}{i-1} , \truebeta)  - V_0^{j_1} (\sample{X}{i-1} , \beta)\bigr)   \nonumber \\[0.1cm]
 & \qquad \qquad \times \bigl( V_0^{j_2} (\sample{X}{i-1} , \truegamma)  - V_0^{j_2} (\sample{X}{i-1} , \estgamma)\bigr);   \nonumber \\[0.3cm]
  T_{1,3,2}(\beta, \hat{\gamma}_n, \sigma ) 
  & =  \tfrac{1}{n} \sum_{i=1}^n \sum_{j_1, j_2 = d_R + 1}^d    
  {\Sigma}_{1, j_1 j_2}^{-1} (\sample{X}{i-1};  \estgamma, \sigma ) 
\frac{\bigl( V_0^{j_1} (\sample{X}{i-1} ,  \truegamma) 
 - V_0^{j_1} (\sample{X}{i-1} , \estgamma)\bigr)}{\sqrt{\Delta}}  \nonumber \\[0.1cm]
& \qquad \qquad \times 
\frac{\bigl( V_0^{j_2} (\sample{X}{i-1} , \truegamma)  
- V_0^{j_2} (\sample{X}{i-1} , \estgamma) \bigr)}{\sqrt{\Delta}}; \nonumber \\[0.3cm]
T_{1,3,3} (\beta, \hat{\gamma}_n, \sigma ) 
& =  \tfrac{1}{n} \sum_{i=1}^n \mathcal{E} (\Delta, \sample{X}{i-1}, (\beta, \hat{\gamma}_n, \sigma )). 
\nonumber  
\end{align}  
for $\mathcal{E}  \in \mathscr{S}$. Due to Lemma \ref{ergodic_thm} and the consistency of $\estgamma$, we directly have
$ T_{1,3,k} (\beta, \hat{\gamma}_n, \sigma )  \rightarrow 0$, for $k = 1, 3$, in probability under $\mathbb{P}_{\theta^\dagger}$, uniformly in $(\beta, \sigma) \in \Theta_\beta \times \Theta_\sigma$. We also have $T_{1,3,2}(\beta, \hat{\gamma}_n, \sigma )  \rightarrow 0$, in probability under $\mathbb{P}_{\theta^\dagger}$, due to the result in the following Lemma.
%
\begin{lemma} \label{canonical_conv2}
Let $f : \mathbb{R}^d \times \Theta \to \mathbb{R}$ be a function satisfying the regularity assumption of Lemma \ref{ergodic_thm} and $d_R + 1 \leq j_1, j_2 \leq d$. Under conditions (\ref{assump:coeff}) and (\ref{assump:bd_deriv})--(\ref{assump:finite_moment}), it holds that, as $n \to \infty, \; \Delta \to 0, \; n \Delta \to \infty$, then 
%
\begin{align} 
& \tfrac{1}{n} \sum_{i=1}^n 
 f (\sample{X}{i-1} , \theta) 
 \frac{V_0^{j_1}(\sample{X}{i-1}, \truegamma) 
 - V_0^{j_1}(\sample{X}{i-1}, \truegamma + \lambda(\estgamma - \truegamma ) )}{\sqrt{\Delta}}  \probconv 0;  \label{eq:conv_1}  \\[0.1cm]  
& \tfrac{1}{n} \sum_{i=1}^n  \biggl\{  f (\sample{X}{i-1} , \theta) 
 \frac{V_0^{j_1}(\sample{X}{i-1}, \truegamma) 
 - V_0^{j_1}(\sample{X}{i-1}, \truegamma + \lambda(\estgamma - \truegamma ) )}{\sqrt{\Delta}} \nonumber \\
& \qquad \qquad \qquad \qquad  \times \frac{V_0^{j_2}(\sample{X}{i-1}, \truegamma) - V_0^{j_2}(\sample{X}{i-1}, \truegamma + \lambda(\estgamma - \truegamma ) )}{\sqrt{\Delta}} \biggr\} \probconv 0,  \label{eq:conv_2}
\end{align}
uniformly in $\theta \in \Theta$ and $\lambda \in [0,1]$. 
\end{lemma}
\noindent We will give the proof of Lemma \ref{canonical_conv2} in Section \ref{pf_canonical_conv2}. 
\\

\noindent
{\bf Proof of (\ref{T3_conv}).}
From (\ref{eq:hermite_second}), we have $$\textstyle{ \mathcal{H}_{(i_1, i_2)} (\Delta, \sample{X}{i-1}, \sample{X}{i}; \theta) = \sum_{l = 1}^5 \eta_{(i_1, i_2)}^{(l)} (\Delta, \sample{X}{i-1}, \sample{X}{i}; \theta)},$$ where 
%
\begin{align}
&  \eta_{(i_1, i_2)}^{(1)} (\Delta, \sample{X}{i-1}, \sample{X}{i}; \theta) 
 = \sum_{j_1, j_2 =1}^d 
 {\Sigma}^{-1}_{1, i_1 j_1} (\sample{X}{i-1}; \gamma, \sigma )
 {\Sigma}^{-1}_{1, i_2 j_2} (\sample{X}{i-1};\gamma, \sigma ) 
 \nu_i^{j_1} (\Delta; \trueparam)
 \nu_i^{j_2} (\Delta; \trueparam);  \nonumber  \\[0.2cm] 
& \eta_{(i_1, i_2)}^{(2)} (\Delta, \sample{X}{i-1}, \sample{X}{i}; \theta) 
 = \sum_{j_1, j_2 =1}^d 
{\Sigma}^{-1}_{1, i_1 j_1} (\sample{X}{i-1}; \gamma, \sigma )
{\Sigma}^{-1}_{1, i_2 j_2} (\sample{X}{i-1}; \gamma, \sigma )  \nu_i^{j_1} (\Delta; \trueparam)  
\cdot \frac{\mu_{i-1}^{j_2} (\Delta; \trueparam) - \mu_{i-1}^{j_2} (\Delta; \theta )}{\Delta^{\frac{1}{2} + \mathbf{1}_{\{j_2 \geq d_R + 1\}}}};  \nonumber \\[0.2cm] 
%
 & \eta_{(i_1, i_2)}^{(3)} (\Delta, \sample{X}{i-1}, \sample{X}{i}; \theta) 
 = \sum_{j_1, j_2 =1}^d 
 {\Sigma}^{-1}_{1, i_1 j_1} (\sample{X}{i-1}; \gamma, \sigma )
 {\Sigma}^{-1}_{1, i_2 j_2} (\sample{X}{i-1}; \gamma, \sigma )  \nu_i^{j_2} (\Delta; \trueparam) 
\cdot \frac{\mu_{i-1}^{j_1} (\Delta; \trueparam) - \mu_{i-1}^{j_1} (\Delta; \theta )}{\Delta^{\frac{1}{2} + \mathbf{1}_{\{j_1 \geq d_R + 1\}}}};  \nonumber \\[0.1cm] 
%
 & \eta_{(i_1, i_2)}^{(4)} (\Delta, \sample{X}{i-1}, \sample{X}{i}; \theta) 
 = \sum_{j_1, j_2 =1}^d 
 {\Sigma}^{-1}_{1, i_1 j_1} (\sample{X}{i-1}; \gamma, \sigma)
 {\Sigma}^{-1}_{1, i_2 j_2} (\sample{X}{i-1}; \gamma, \sigma)  \nonumber \\ 
 & \qquad  \qquad  \qquad  \qquad  \qquad  \qquad  \qquad  \qquad  
  \times \frac{\mu_{i-1}^{j_1} (\Delta; \trueparam) - \mu_{i-1}^{j_1} (\Delta; \theta  )}{\Delta^{\frac{1}{2} + \mathbf{1}_{\{j_1 \geq d_R + 1\}}}}
 \frac{\mu_{i-1}^{j_2} (\Delta; \trueparam) - \mu_{i-1}^{j_2} (\Delta; \theta  )}{\Delta^{\frac{1}{2} + \mathbf{1}_{\{j_2 \geq d_R + 1\}}}};  
 \nonumber \\[0.3cm] 
& \eta_{(i_1, i_2)}^{(5)} (\Delta, \sample{X}{i-1}, \sample{X}{i}; \theta)  
 = - {\Sigma}^{-1}_{1, i_1 i_2} (\sample{X}{i-1}; \gamma, \sigma ). \nonumber  
\end{align}
Lemmas \ref{ergodic_thm}, \ref{canonical_conv}, \ref{canonical_conv2} and the consistency of $\estgamma$ give that, as $n \to \infty, \, \Delta \to 0, \, n \Delta \to \infty$,  
\begin{align*}
   \tfrac{\Delta}{n} \sum_{l = 1}^5  
    \eta_{(i_1, i_2)}^{(l)} \Bigl( 
    \Delta, \sample{X}{i-1}, \sample{X}{i}; \beta, \hat{\gamma}_n, \sigma 
    \Bigr) \probconv 0, 
\end{align*}
uniformly in $(\beta, \sigma) \in \Theta_\beta \times \Theta_\sigma$, thus also $T_{3} (\beta, \hat{\gamma}_n, \sigma ) \probconv 0$. The proof is now complete.  

\subsubsection{Proof of Lemma \ref{lemma:conv_contrast_beta}} \label{sec:conv_contrast_beta} 
\noindent We write $\textstyle{Q (\theta) = \sum_{l=1}^4 Q_l (\theta)}$, where
\begin{align*}
& Q_1 \bigl( \theta \bigr) 
 =  \tfrac{2}{n \Delta} \sum_{i=1}^n  
 \bigl( X_{t_i} - \mu_{\Delta} (\sample{X}{i-1}; \trueparam )  \bigr)^{\top} 
 \Sigma^{-1} (\Delta, \sample{X}{i-1}; \gamma, \sigma) \\
  & \qquad \qquad  \qquad \qquad     \times\bigl( \mu_{\Delta} (\sample{X}{i-1}; \truebeta, \gamma, \sigma)  - \mu_{\Delta} (\sample{X}{i-1}; \theta) \bigr);   \\[0.3cm]
& Q_2 \bigl( \theta \bigr) 
 =  \tfrac{1}{n \Delta} \sum_{i=1}^n \bigl( 
 \mu_{\Delta} (\sample{X}{i-1}; \truebeta, \gamma, \sigma)  
   - \mu_{\Delta} (\sample{X}{i-1}; \theta) \bigr)^{\top}  
   \Sigma^{-1} (\Delta, \sample{X}{i-1} ; \gamma, \sigma )   \\[0.1cm]
 & \qquad \qquad  \qquad \qquad     \times  \bigl( 
 \mu_{\Delta} (\sample{X}{i-1}; \truebeta, \gamma,  \sigma)  - 
  \mu_{\Delta} (\sample{X}{i-1}; \theta ) \bigr); \nonumber \\[0.3cm]
& Q_3 \bigl( \theta \bigr) 
  = \tfrac{2}{n \Delta} \sum_{i=1}^n \bigl( \mu_{\Delta} (\sample{X}{i-1}; \truebeta, \gamma, \sigma) 
     - \mu_{\Delta} (\sample{X}{i-1}; \theta) \bigr)^{\top}  
     \Sigma^{-1} (\Delta, \sample{X}{i-1}; \gamma, \sigma)  \\[0.1cm]
 &\qquad \qquad  \qquad \qquad  \times  \bigl( \mu_{\Delta} (\sample{X}{i-1}; \trueparam)  - \mu_{\Delta} (\sample{X}{i-1}; \truebeta, \gamma, \sigma) \bigr);   \\[0.3cm]
& Q_4 \bigl( \theta \bigr)  
 = - \tfrac{2}{n} \sum_{i=1}^n  
 \Bigl\{ \Phi_2 (\Delta, \sample{X}{i-1}, \sample{X}{i};  \theta ) - \Phi_2 (\Delta, \sample{X}{i-1}, \sample{X}{i};  \truebeta, \gamma, \sigma)  \Bigr\}. \nonumber
\end{align*} 
We will show that, as $n \to \infty,$ $\Delta \to 0$, $n \Delta \to \infty$, 
\begin{align}
& Q_1 \bigl( \beta, \hat{\gamma}_n , \sigma \bigr)
 \xrightarrow{\mathbb{P}_{\trueparam}} 0;  \label{Q1_conv} \\[0.3cm]
&  Q_2 \bigl( \beta, \hat{\gamma}_n , \sigma\bigr)  \xrightarrow{\mathbb{P}_{\trueparam}} 
\int \bigl( e  (x; \truebeta, \truegamma, \sigma )  
  -  e (x; \beta, \truegamma, \sigma ) \bigr)^{\top} 
 {\Sigma}^{-1}_1 (x; \truegamma, \sigma ) 
 \nonumber \\[0.1cm] 
& \qquad \qquad \qquad \qquad \qquad \qquad  \times
\bigl( e (x; \truebeta, \truegamma, \sigma ) 
-  e (x; \beta, \truegamma, \sigma ) \bigr) \nu_{\trueparam} (dx);  \label{Q2_conv}  \\[0.3cm]
& Q_3 \bigl( \beta, \hat{\gamma}_n , \sigma \bigr)  \probconv
2  \int \bigl( e (x; \truebeta, \truegamma, \sigma) 
 -  e (x; \beta, \truegamma, \sigma)    \bigr)^{\top} 
 {\Sigma}^{-1}_1 (x; \truegamma, \sigma)  \nonumber \\
& \qquad \qquad \qquad \qquad \qquad \qquad  \times 
\bigl( e (x; \trueparam) - e (x; \truebeta, \truegamma, \sigma) \bigr) \truedist (dx);   \label{Q3_conv}  \\[0.3cm]
& Q_4 \bigl( \beta, \hat{\gamma}_n , \sigma \bigr)  
\probconv
2 \sum_{\substack{1 \leq i_1, i_2, j_1, j_2  \leq d }}  \int 
\Bigl\{ 
  {G}_{i_1i_2}(x; \truebeta, \truegamma, \sigma ) 
- {G}_{i_1i_2}(x; \beta, \truegamma, \sigma ) 
\Bigr\} \nonumber \\
& \quad  \times  \Bigl\{ 
  {\Sigma}^{-1}_{1, i_1 j_1} (x; \truegamma, \sigma) 
 {\Sigma}^{-1}_{1, i_2 j_2} (x; \truegamma, \sigma ) 
 {\Sigma}_{1, j_1 j_2}  (x; \truegamma, \truesigma) 
 - {\Sigma}^{-1}_{1, i_1 i_2} (x; \truegamma, \sigma)  \Bigr\} \truedist (dx), \label{Q4_conv}
\end{align} 
uniformly in $(\beta, \sigma) \in \Theta_\beta \times \Theta_\sigma$, where for each $\theta \in \Theta$, $e (\cdot ; \theta) : \mathbb{R}^d \to \mathbb{R}^d$ is given by 
\begin{align*}
 e (x; \theta)  = \big[\,  
     V_{R,0}  (x,  \beta)^{\top},\,  
     \tfrac{1}{2} \hat{V}_0 V_{S, 0} (x , \theta)^{\top}\,\big]^{\top},\quad  x \in \mathbb{R}^d;  
\end{align*}
recall here that the $d \times d$ matrix $G = G (x; \theta)$ is defined in (\ref{eq:G_matrix}). Upon showing the above convergences, the proof of Lemma 5.4 will be complete since we have $Q_3  (\beta, \hat{\gamma}_n , \hat{\sigma}_n)  \rightarrow 0$, $Q_4  (\beta, \hat{\gamma}_n , \hat{\sigma}_n ) \rightarrow 0$, in probability under $\mathbb{P}_{\theta^\dagger}$, due to the consistency of $\hat{\sigma}_n$ and 
\begin{align*}
  \sum_{j_1, j_2 = 1}^d  {\Sigma}^{-1}_{1, i_1 j_1} (x; \gamma, \sigma) 
 {\Sigma}^{-1}_{1, i_2 j_2} (x; \gamma, \sigma ) 
 {\Sigma}_{1, j_1 j_2}  (x; \gamma, \sigma) 
 - {\Sigma}^{-1}_{1, i_1 i_2} (x; \gamma, \sigma) 
 = 0,  
\end{align*}
for any $x \in \mathbb{R}^d, \;  (\gamma, \sigma) \in \Theta_\gamma \times \Theta_\sigma$. 
\\ 

\noindent 
{\bf Proof of (\ref{Q1_conv}) and (\ref{Q2_conv}).}
From Lemma \ref{canonical_conv}, we immediately have that, as $n \to \infty$, $\Delta \to 0$ and $n \Delta \to \infty$, $Q_1  (\beta, \hat{\gamma}_n , \sigma) \rightarrow 0$, in probability under $\mathbb{P}_{\theta^\dagger}$. 
%
%
Thus, Lemma \ref{ergodic_thm} and the consistency of $\hat{\gamma}_n$ yield the convergence (\ref{Q2_conv}). 
\\

\noindent
{\bf Proof of (\ref{Q3_conv}).}
We have $Q_3  (\beta, \hat{\gamma}_n , \sigma) = Q_{3,1}  (\beta, \hat{\gamma}_n , \sigma)
 + Q_{3,2}  (\beta, \hat{\gamma}_n , \sigma)$, where 
\begin{align*}
& Q_{3,1}  (\beta, \hat{\gamma}_n , \sigma) 
= \tfrac{2}{n \sqrt\Delta} 
 \sum_{i=1}^n  \sum_{\substack{1 \leq j_1 \leq d \\ d_R + 1 \leq j_2 \leq d}} {\Sigma}^{-1}_{1, j_1 j_2} (\sample{X}{i-1}; \hat{\gamma}_n, \sigma) \nonumber \\
&  \quad \times  \bigl( e_{j_1} (\sample{X}{i-1}; \truebeta, \estgamma, \sigma)
- e_{j_1} (\sample{X}{i-1}; \beta, \estgamma, \sigma)    \bigr) 
\frac{\bigl( V_0^{j_2} (\sample{X}{i-1}, \truegamma) - V_0^{j_2} (\sample{X}{i-1} ,  \estgamma ) \bigr)}{\sqrt{\Delta}};  \\[0.3cm] 
& Q_{3,2} (\beta, \hat{\gamma}_n , \sigma)   
=  \tfrac{2}{n} \sum_{i=1}^n  \sum_{j_1, j_2 = 1}^d 
{\Sigma}^{-1}_{1, j_1 j_2} (\sample{X}{i-1}; \hat{\gamma}_n, \sigma )  \nonumber \\
& \quad \times  \bigl( e_{j_1} (\sample{X}{i-1}; \truebeta, \estgamma, \sigma) - e_{j_1} (\sample{X}{i-1}; \beta, \estgamma, \sigma)    \bigr)
\bigl( e_{j_2} (\sample{X}{i-1}; \trueparam) 
  - e_{j_2} (\sample{X}{i-1}; \truebeta, \estgamma, \sigma) \bigr).
  \\[-0.4cm]
\end{align*} 
Notice that
\begin{align*} 
e(\sample{X}{i-1}; \truebeta, \estgamma, \sigma)
- e(\sample{X}{i-1}; \beta, \estgamma, \sigma) = 
\begin{bmatrix}
V_{R, 0} (\sample{X}{i-1}, \truebeta ) 
- V_{R, 0} (\sample{X}{i-1}, \beta )  \\[0.2cm]  
\tfrac{1}{2} \partial_{x_R}^\top  V_{S,0} (x, \hat{\gamma}_n ) 
\bigl( V_{R, 0} (\sample{X}{i-1}, \truebeta ) 
- V_{R, 0} (\sample{X}{i-1}, \beta ) \bigr) 
\end{bmatrix}.  
\end{align*}
Hence, $Q_{3,1} (\beta, \hat{\gamma}_n , \sigma)= 0$ from Lemma \ref{lemma:aux_matrix}. 
Furthermore, from Lemma \ref{ergodic_thm}, we directly have 
\begin{align*}
& Q_{3,2}  (\beta, \hat{\gamma}_n , \sigma)  
\probconv
2  \int \bigl( e (x; \truebeta, \truegamma, \sigma) -  e (x; \beta, \truegamma, \sigma) \bigr)^{\top} 
{\Sigma}^{-1}_1 (x; \truegamma, \sigma)  \nonumber \\
& \qquad \qquad \qquad \qquad \qquad \qquad  \times \bigl( e (x; \trueparam) - e (x; \truebeta, \truegamma, \sigma) \bigr) \truedist (dx).  
\end{align*}
We have now obtained the convergence (\ref{Q3_conv}). 
\\

\noindent
{\bf Proof of (\ref{Q4_conv}).}
Recall $\textstyle{ \mathcal{H}_{(i_1, i_2)} (\Delta, \sample{X}{i-1}, \sample{X}{i}; \theta) = \sum_{l = 1}^5 \eta_{(i_1, i_2)}^{(l)} (\Delta, \sample{X}{i-1}, \sample{X}{i}; \theta)}$, see Section \ref{sec:conv_contrast_sigma}. We have 
$\textstyle{ -\tfrac{2}{n} \sum_{i = 1}^n \Phi_2 (\Delta, \sample{X}{i-1}, \sample{X}{i} ;  \beta, \estgamma, \sigma) 
= \sum_{k=1}^5 Q_{4,k}(\beta, \estgamma, \sigma)}
$, 
where 
\begin{align*}
Q_{4,k}(\beta, \estgamma, \sigma)  
 = - \tfrac{2}{n} \sum_{i=1}^n  \sum_{i_1, i_2 = 1}^d  {G}_{i_1i_2}(\sample{X}{i-1}; \beta, \hat{\gamma}_n , \sigma ) 
 \eta_{(i_1, i_2)}^{(k)} (\Delta, \sample{X}{i-1}, \sample{X}{i}; \beta, \estgamma, \sigma ). 
\end{align*}
%
We will check the convergence of each term $Q_{4,k} (\beta, \estgamma, \sigma)$. From Lemmas \ref{ergodic_thm}, \ref{canonical_conv} and the consistency of $\hat{\gamma}_n$, we obtain 
\begin{align*}
&  Q_{4,1} (\beta, \estgamma, \sigma)  \probconv  
-2  \sum_{i_1, i_2, j_1, j_2 = 1}^d \int 
{G}_{i_1i_2}(x; \beta, \truegamma, \sigma ) 
{\Sigma}^{-1}_{1, i_1 j_1} (x; \truegamma, \sigma)  
{\Sigma}^{-1}_{1, i_2 j_2} (x; \truegamma, \sigma ) 
{\Sigma}_{1, j_1 j_2}  (x; \truegamma, \truesigma) \truedist (dx); \nonumber \\[0.3cm] 
& Q_{4,2}  (\beta, \estgamma, \sigma)  
 \probconv  0,  \ \ \ 
 Q_{4,3}  (\beta, \estgamma, \sigma)  
  \probconv  0, \ \ \ 
%
Q_{4,5}  ((\beta, \estgamma, \sigma))   
  \probconv  
 2  \int 
 \mathrm{tr}
 \left( 
   {G} (x; \beta, \truegamma, \sigma ) 
   {\Sigma}^{-1}_1  (x; \truegamma, \sigma) \right) 
  \truedist (dx). 
\end{align*}
For term $Q_{4,4}  (\beta, \estgamma, \sigma)$, we have 
\begin{align*}
Q_{4,4}  (\beta, \estgamma, \sigma)   
& = - \tfrac{2}{n} \sum_{i=1}^n  
 \sum_{\substack{1 \le i_1, i_2 \le d \\ d_R + 1 \le j_1, j_2 \le d}}    
 \Biggl\{  
 {G}_{i_1i_2}(\sample{X}{i-1}; \beta, \hat{\gamma}_n , \sigma) 
 {\Sigma}^{-1}_{1, i_1 j_1} (\sample{X}{i-1}; \hat{\gamma}_n, \sigma)  
 {\Sigma}^{-1}_{1, i_2 j_2} (\sample{X}{i-1}; \hat{\gamma}_n, \sigma)   
 \nonumber \\
& \qquad  \times  \f{V_0^{j_1} (\sample{X}{i-1} , \truegamma) - V_0^{j_1} (\sample{X}{i-1} ,  \estgamma)  }{\sqrt{\Delta}}
\f{V_0^{j_2} (\sample{X}{i-1} , \truegamma) - 
 V_0^{j_2} (\sample{X}{i-1} ,  \estgamma) }{\sqrt{\Delta}}  \Biggr\} 
\nonumber \\[0.2cm] 
& \quad - \tfrac{4}{n} \sum_{i=1}^n  
\sum_{\substack{1 \le i_1, i_2 \le d, \\ 1 \leq j_1 \leq d_R \\ d_R + 1 \leq j_2 \leq d}} 
 \Biggl\{ 
{G}_{i_1i_2}(\sample{X}{i-1}; \beta, \hat{\gamma}_n , \sigma) 
 {\Sigma}^{-1}_{1, i_1 j_1} (\sample{X}{i-1},  \hat{\gamma}_n, \sigma) 
 {\Sigma}^{-1}_{1, i_2 j_2} (\sample{X}{i-1},  \hat{\gamma}_n, \sigma)  
 \nonumber \\
& \qquad  \times 
 \bigl( {V_0^{j_1} (\sample{X}{i-1}, \truebeta) - V_0^{j_1} (\sample{X}{i-1},  \beta)  } \bigr)
\bigl( V_0^{j_2} (\sample{X}{i-1}, \truegamma) - V_0^{j_2} (\sample{X}{i-1},  \estgamma) \bigr) \Biggr\} 
\nonumber \\[0.2cm]
& \quad - \tfrac{2}{n}  \sum_{i=1}^n   
\sum_{\substack{1 \le i_1, i_2 \le d \\ d_R + 1 \le j_1, j_2 \le d}}   
 \Biggl\{ {G}_{i_1, i_2}(\sample{X}{i-1}; \beta, \hat{\gamma}_n , \sigma)
 {\Sigma}^{-1}_{1, i_1 j_1} (\sample{X}{i-1}; \hat{\gamma}_n, \sigma) 
 {\Sigma}^{-1}_{1,i_2 j_2} (\sample{X}{i-1}; \hat{\gamma}_n, \sigma) \nonumber \\
& \qquad  \times 
\bigl( {\hat{V}_0 V_0^{j_1} (\sample{X}{i-1}, \trueparam) 
 - \hat{V}_0 V_0^{j_1} (\sample{X}{i-1},  (\beta, \estgamma, \sigma))  } \bigr)
\bigl( V_0^{j_2} (\sample{X}{i-1}, \truegamma) 
- V_0^{j_2} (\sample{X}{i-1}, \estgamma) \bigr) \Biggr\}   \nonumber \\[0.2cm] 
& \quad +  \tfrac{1}{n} \sum_{i=1}^n  \mathcal{E} (\Delta, \sample{X}{i-1}, (\beta, \estgamma, \sigma)).  \nonumber 
\end{align*} 
for $\mathcal{E}  \in \mathscr{S}$.
Due to Lemmas \ref{ergodic_thm}, \ref{canonical_conv2} and the consistency of $\hat{\gamma}_n$, we immediately have that, as $n \to \infty, \; \Delta \to 0, \; n \Delta \to \infty$, $Q_{4,4}  (\beta, \estgamma, \sigma)
 \rightarrow  0$, in probability under $\mathbb{P}_{\theta^\dagger}$, uniformly in $(\beta, \sigma) \in \Theta_\beta \times \Theta_\sigma$. Finally, as $n \to \infty$, $\Delta \to 0$, $n \Delta \to \infty$,
\begin{align}
& -\tfrac{2}{n} \sum_{i = 1}^n \Phi_2 (\Delta, \sample{X}{i-1}, \sample{X}{i} ;  \beta, \estgamma, \sigma)  
\probconv \nonumber \\  
&-2 \sum_{i_1, i_2, j_1, j_2 = 1}^d  \int 
{G}_{i_1 i_2}(x; \beta, \truegamma, \sigma )  \label{conv_weight} 
\\[0.2cm]
&  \qquad\qquad \times 
 \Bigl\{
 {\Sigma}^{-1}_{1,i_1 j_1} (x; \truegamma, \sigma) 
 {\Sigma}^{-1}_{1,i_2 j_2} (x; \truegamma, \sigma ) 
 {\Sigma}_{1, j_1 j_2}  (x; \truegamma, \truesigma) 
 - {\Sigma}^{-1}_{1, i_1 i_2} (x; \truegamma, \sigma)   
 \Bigr\} \,  \truedist (dx),   \nonumber
\end{align}
uniformly in $(\beta, \sigma) \in \Theta_{\beta} \times \Theta_{\sigma}$. From (\ref{conv_weight}), we obtain the convergence (\ref{Q4_conv}). 


\subsubsection{Proof of Lemma \ref{canonical_conv2}} \label{pf_canonical_conv2}
\noindent  
We will prove only convergence (\ref{eq:conv_2}) since (\ref{eq:conv_1}) is deduced in a similar manner. Define for $\theta \in \Theta$ and $d_R + 1 \le j_1, j_2 \le d$,
\begin{align*}
 \zeta_{\lambda} (\theta) & 
 :=   \tfrac{1}{n} \sum_{i=1}^n  \biggl\{ f (\sample{X}{i-1}, \theta)
 \frac{V_0^{j_1}(\sample{X}{i-1} ,  \truegamma) - V_0^{j_1}(\sample{X}{i-1} , \truegamma + \lambda(\estgamma - \truegamma ) )}{\sqrt{\Delta}} \nonumber \\
& \qquad \qquad \qquad \qquad  \times \frac{V_0^{j_2}(\sample{X}{i-1} , \truegamma) - V_0^{j_2}(\sample{X}{i-1} , \truegamma + \lambda(\estgamma - \truegamma ) )}{\sqrt{\Delta}} \biggr\} .  
\end{align*}
It holds that, as $n \to \infty, \; \Delta \to 0, \; n \Delta \to \infty$, then 
\begin{align*}
  | \zeta_{\lambda} (\theta)| 
  \leq  {\lambda^2} \left| \tfrac{1}{\sqrt\Delta} (\estgamma - \truegamma) \right|^2 
   \tfrac{1}{n}  \sum_{i=1}^n | f (\sample{X}{i-1}, \theta) 
   B^{j_1} (\sample{X}{i-1}) B^{j_2} (\sample{X}{i-1})|
   \probconv 0,
\end{align*}
uniformly in $\theta \in \Theta, \; \lambda \in (0,1]$, where we used the condition (\ref{assump:lipschitz}), Lemma \ref{ergodic_thm} and the convergence (\ref{eq:gamma_rate}) in the main text. 
Now, we obtain $\zeta_\lambda (\theta) \rightarrow 0$, in probability under $\mathbb{P}_{\theta^\dagger}$, uniformly in $\theta \in \Theta$ and the proof is complete.  
\subsection{Proof of Technical Results for Theorem \ref{thm:asymp_norm}} \label{sec:asymp_norm}
\subsubsection{Proof of Lemma \ref{lemma:conv_normal}}  \label{appendix:pf_conv_normal} 
\noindent From the definition of the contrast $\ell_{n, \Delta}(\theta), \; \theta \in \Theta$, we have 
\begin{align*}
  - M_{n, \Delta}^{jj} \partial_{\theta_j} \ell_{n, \Delta} (\theta)  
  =  - M_{n, \Delta}^{jj} \sum_{i=1}^n L_{i}^j (\theta),   
\end{align*}
for $1 \le j  \le  d_{\theta}$, where $L_{i}^j (\theta)$ is explicitly given later. Let $1 \leq j_1, j_2, j_3, j_4 \leq d_{\theta}$. Then, we will show the following results:
\begin{description}
\item[a.] As $n \to \infty$, $\Delta \to 0$, $n \Delta \to \infty$ and $n \Delta^3 \to 0$, 
\begin{align}
 &   M_{n, \Delta}^{j_1 j_1} \sum_{i=1}^n \trueE [L_{i}^{j_1} (\trueparam) | \mathcal{F}_{t_{i-1}}] 
 \xrightarrow{\mathbb{P}_{\trueparam}} 0.  \label{eq:normal_conv_1}  
\end{align}
\item[b.] As $n \to \infty$, $\Delta \to 0$ and $n \Delta \to \infty$,
\end{description}
\begin{align}
 & M_{n, \Delta}^{j_1 j_1} M_{n, \Delta}^{j_2 j_2} 
 \sum_{i=1}^n \trueE [L_{i}^{j_1} (\trueparam) L_{i}^{j_2} (\trueparam) | \mathcal{F}_{t_{i-1}}] 
 \xrightarrow{\mathbb{P}_{\trueparam}} 4 I_{j_1 j_2} (\trueparam);   \label{eq:normal_conv_3} \\    
 &  M_{n, \Delta}^{j_1 j_1} M_{n, \Delta}^{j_2 j_2} M_{n, \Delta}^{j_3 j_3} M_{n, \Delta}^{j_4 j_4} \sum_{i=1}^n \trueE [L_{i}^{j_1} (\trueparam) L_{i}^{j_2} (\trueparam) L_{i}^{j_3} (\trueparam) 
 L_{i}^{j_4} (\trueparam) | \mathcal{F}_{t_{i-1}}] 
 \xrightarrow{\mathbb{P}_{\trueparam}} 0.   \label{eq:normal_conv_4}    
\end{align} 
Upon showing the above results, the proof of Lemma \ref{lemma:conv_normal} is complete due to Theorems 3.2 and 3.4 in \cite{hall:14}. We consider the partial derivatives of $$\ell_{n, \Delta} (\trueparam) = \ell_{n, \Delta}^{(1)} (\trueparam) +  \ell_{n, \Delta}^{(2)} (\trueparam),$$ where 
\begin{align} 
 & \ell_{n, \Delta}^{(1)} (\theta) 
 =   \sum_{i=1}^n  \nu_{i}(\Delta; \theta)^{\top} 
{\Sigma}^{-1}_1  (X_{t_{i-1}}; \gamma, \sigma) \nu_{i}(\Delta; \theta)  
+ \sum_{i=1}^n \log | {\Sigma}_1 (\sample{X}{i-1}; \gamma, \sigma) |; \label{eq:l1}   \\
 & \ell_{n, \Delta}^{(2)} (\theta)
=  - 2  \Delta \sum_{i=1}^n \Phi_2 (\Delta, X_{t_{i-1}}, X_{t_i}; \theta). \label{eq:l2}
\end{align} 
We have for $1 \leq j \leq d_{\theta}$: 
\begin{align*}
\partial_{\theta_j} \ell_{n, \Delta} (\theta) |_{\theta = \trueparam} 
 &  =    \Bigl\{ \partial_{\theta_j} \ell_{n, \Delta}^{(1)} (\theta) + \partial_{\theta_j} \ell_{n, \Delta}^{(2)} (\theta) \Bigr\} \Big|_{\theta = \trueparam}
   = \sum_{i=1}^n \bigl\{ L_i^{j, (1)}(\trueparam) + L_i^{j, (2)}(\trueparam) \bigr\}, 
\end{align*} 
where we have set, for $1 \leq j_1 \leq d_\beta, \; d_\beta + 1 \leq j_2 \leq d_\beta + d_{\gamma}, \; d_\beta + d_{\gamma}  + 1 \leq j_3 \leq d_{\theta}$, 
\begin{align}
& L_i^{j_1, (1)}(\trueparam) 
  =  - 2 \sqrt{\Delta}
   \bigl(
   \partial_{\theta_{j_1}} \mu_1  (\sample{X}{i-1}; \beta) 
   \bigr)^{\top}
   {\Sigma}^{-1}_1 (\sample{X}{i-1} ; \gamma, \sigma) 
    \nu_i (\Delta; \theta) \bigr|_{\theta = \trueparam};  
    \label{L_1_1} \\[0.3cm]
&  L_i^{j_2, (1)} (\trueparam)
 = \biggl\{   - 2 
 \sum_{\substack{d_R + 1 \le i_1 \le d \\ 1 \le i_2 \le d}} {\Sigma}^{-1}_{1, i_1 i_2} (\sample{X}{i-1} ; \gamma, \sigma) 
 \frac{\partial_{\theta_{j_2}} V_0^{i_1} ( \sample{X}{i-1}, \gamma)}{\sqrt{\Delta}} \, 
 \nu_i^{i_2} ( \Delta; \theta )  \nonumber   \\[0.1cm]
& \qquad  - \sqrt{\Delta} 
 \sum_{\substack{d_R + 1 \le i_1 \le d \\ 1 \le i_2 \le d}} 
 {\Sigma}^{-1}_{1, i_1 i_2} (\sample{X}{i-1} ; \gamma, \sigma)  
 \partial_{\theta_{j_2}} \hat{V}_0 V_0^{i_1} (\sample{X}{i-1}, \theta ) \, 
 \nu_i^{i_2} (\Delta; \theta )   \nonumber   \\[0.1cm]
& \qquad  + 
  \nu_i (\Delta; \theta )^\top 
  \partial_{{\theta}_{j_2}} {\Sigma}^{-1}_1  (\sample{X}{i-1}; \gamma, \sigma)  \,  
  \nu_i (\Delta; \theta ) 
 +  \partial_{\theta_{j_2}} \log |{\Sigma}_1 (\sample{X}{i-1}; \gamma, \sigma)|  \biggr\} 
 \Bigr|_{\theta = \trueparam } ;  \label{L_1_2} \\[0.3cm]    
&  L_i^{j_3, (1)} (\trueparam)  
 = \biggl\{  
 \nu_i (\Delta; \theta)^{\top} \partial_{\theta_{j_3}} {\Sigma}^{-1}_1  (\sample{X}{i-1}; \gamma, \sigma)  
    \nu_i (\Delta; \theta) 
 +  \partial_{\theta_{j_3}}  \log | {\Sigma}_1 (\sample{X}{i-1}; \gamma, \sigma)|
   \nonumber \\[0.1cm]
& \quad \quad   
 - 2 \sqrt{\Delta}   \, 
    \partial_{\theta_{j_3}} \mu_{1} (\sample{X}{i-1}; \theta)^\top
   {\Sigma}^{-1}_1 (\sample{X}{i-1} ; \gamma, \sigma ) 
   \, \nu_i (\Delta; \theta) 
 \biggr\}  \Bigr|_{\theta = \trueparam}, \label{L_1_3}
\end{align} 
and
\begin{align} \label{L_2}
L_i^{j, (2)}(\trueparam)
= - 2 \Delta \times \partial_{\theta_{j}} \Phi_2 (\Delta, X_{t_{i-1}}, X_{t_{i}} ; \theta) \bigr|_{\theta = \trueparam}, 
\ \ 1 \leq j \leq d_{\theta}.  
\end{align} 

\noindent
{\bf Proof of (\ref{eq:normal_conv_1}).} Using Lemmas \ref{weak_Taylor_expansion}, \ref{hermite_expectation}, we have for $1 \leq j_1 \leq d_{\beta}, \; d_{\beta} + 1 \leq j_2  \leq d_{\beta} + d_{\gamma}$, 
\begin{align*}
& \trueE [ L_i^{j_1, (1)}(\trueparam)  | \filtration{i-1}] 
 =  \mathcal{E}^{j_1, (1)} (\Delta^2, \sample{X}{i-1}, \trueparam); \\ 
 &\trueE [ L_i^{j_1, (2)}(\trueparam)  | \filtration{i-1}] 
 = \mathcal{E}^{j_1, (2)} (\Delta^3, \sample{X}{i-1}, \trueparam); \\
& \trueE [ L_i^{j_2, (1)}(\trueparam)  | \filtration{i-1}] 
%
  =   \mathcal{E}^{{j_2, (1)}} (\Delta, \sample{X}{i-1},  \trueparam); \\   
 &\trueE [ L_i^{j_2, (2)}(\trueparam) | \filtration{i-1}] 
 =  \mathcal{E}^{{j_2, (2)}} ( \Delta^2, \sample{X}{i-1},  \trueparam), 
\end{align*} 
for some $\mathcal{E}^{{j_{i_1}, (i_2)}} \in \mathscr{S}$,  $i_1, i_2 = 1,2$. Notice that in the above computation, we used (\ref{eq:mean_nu}) and
\begin{align} 
 \partial_{\theta_{j_2}} \log |{\Sigma}_1 (\sample{X}{i-1}; \gamma, \sigma)|
 & = 
 \mathrm{tr}  
  \Bigl( {\Sigma}^{-1}_1   (\sample{X}{i-1}; \gamma , \sigma) 
   \partial_{\theta_{j_2}} {\Sigma}_1   (\sample{X}{i-1}; \gamma, \sigma)  \Bigr) \nonumber  \\ 
 & = - \mathrm{tr} 
  \Bigl(  \partial_{\theta_{j_2}} 
  {\Sigma}^{-1}_1  (X_{t_{i-1}}; \gamma, \sigma) 
  {\Sigma}_1   (\sample{X}{i-1}; \gamma, \sigma) \Bigr).  \label{eq:deriv_logdet}
%
\end{align}
%
It follows from Lemmas \ref{ergodic_thm}, \ref{canonical_conv} that
as $n \to \infty$, $\Delta \to 0$, $n \Delta \to \infty$ and $n \Delta^3 \to 0$,  
\begin{align*}
   \tfrac{1}{\sqrt{n \Delta}} \sum_{i=1}^n \trueE [L_i^{j_1, (1)} + L_i^{j_1, (2)}  | \filtration{i-1}]  
   \probconv &\;  0; \\ 
   \tfrac{\sqrt{\Delta}}{\sqrt{n}} \sum_{i=1}^n 
   \trueE [L_i^{j_2, (1)} + L_i^{j_2,  (2)}  | \filtration{i-1}]  
   \probconv &\; 0.   
\end{align*} 
Now, let $d_\beta + d_\gamma + 1 \leq j_3 \leq d_\theta$. 
Lemmas \ref{weak_Taylor_expansion}, \ref{hermite_expectation} give 
\begin{align}
& \trueE [ L_i^{j_3,  (1)} (\trueparam)  | \filtration{i-1}] 
  = \Bigl\{ \mathrm{tr} 
 \bigl( \partial_{\theta_{j_3}} 
 {\Sigma}^{-1}_1 (X_{t_{i-1}}; \gamma, \sigma) {\Sigma}_1 (X_{t_{i-1}}; \gamma, \sigma)   \bigr)
 + \partial_{\theta_{j_3}} \log |{\Sigma}_1 (\sample{X}{i-1}; \gamma, \sigma)| \nonumber \\[0.1cm]
& \qquad + 
\Delta\,\mathrm{tr} 
 \bigl( \Xi (\sample{X}{i-1}; \theta )  \partial_{\theta_{j_3}} {\Sigma}^{-1}_1  (\sample{X}{i-1}; \gamma, \sigma ) \bigr)  \Bigr\} \Bigr|_{\theta = \trueparam}
 + \mathcal{E}^{j_3, (1)} (\Delta^2, \sample{X}{i-1}, \trueparam)  \nonumber \\[0.2cm]
& \quad  = 2 \Delta \mathrm{tr} 
\bigl( G  (\sample{X}{i-1}; \theta)  
\partial_{\theta_{j_3}} {\Sigma}^{-1}_1 (\sample{X}{i-1}; \gamma, \sigma)  
\bigr) |_{\theta = \trueparam} 
+ \mathcal{E}^{j_3, (1)} (\Delta^2, \sample{X}{i-1}, \trueparam), 
 \label{L_sigma_1}    
\end{align}
for $\mathcal{E}^{j_3, (1)} \in \mathscr{S}$, where we used (\ref{eq:deriv_logdet}). Furthermore, we have from Lemma \ref{hermite_expectation}
\begin{align}  \label{L_sigma_2} 
& \trueE [ L_i^{j_3, (2)} (\trueparam)  | \filtration{i-1}] 
 =  - 2  \Delta \times  \trueE [ \partial_{\theta_{j_3}} 
   \Phi_2 (\Delta, X_{t_{i-1}}, X_{t_{i}}; \theta) |_{\theta = \trueparam} ]  \nonumber \\
& \qquad = -  2 \Delta  \mathrm{tr} \, \bigl( 
G (\sample{X}{i-1}; \theta)   
\partial_{\theta_{j_3}} 
{\Sigma}^{-1}_1 (\sample{X}{i-1}; \gamma, \sigma ) \bigr) |_{\theta = \trueparam}
 + \mathcal{E}^{j_3, (2)} (\Delta^2, \sample{X}{i-1}, \trueparam), 
\end{align} 
for $\mathcal{E}^{j_3, (2)} \in \mathscr{S}$. Hence, it follows from (\ref{L_sigma_1}) a (\ref{L_sigma_2}) that, as $n \to \infty$, $\Delta \to 0$, $n \Delta \to \infty$ and $n \Delta^4 \to 0$, 
\begin{align*}
 \tfrac{1}{\sqrt{n}} \sum_{i=1}^n \trueE [ L_i^{j_3, (1)} (\trueparam)
 + L_i^{j_3,  (2)} (\trueparam)   | \filtration{i-1}] 
 = \tfrac{1}{n} \sum_{i=1}^n \mathcal{E} (\sqrt{ n \Delta^4}, \sample{X}{i-1} , \trueparam) \probconv 0.
\end{align*}  
where $\mathcal{E} \in \mathscr{S}$. \vspace{0.2cm}

\noindent
{\bf Proof of (\ref{eq:normal_conv_3}).}
We check (\ref{eq:normal_conv_3}) for the following two cases:
(a)  $1 \le j_1, j_2 \le d_\beta$; 
(b)  $d_{\beta} + d_{\gamma} + 1 \le j_1, j_2 \le d_\theta$,  
since the proof for other cases are identical. Using (\ref{L_1_1}), (\ref{L_2}) and Lemmas \ref{lemma:aux_matrix}, \ref{ergodic_thm}, \ref{canonical_conv}, we obtain the following convergence for $1 \leq j_1, j_2 \leq d_{\beta}$: 
\begin{align} \label{eq:cov_beta}
& \tfrac{1}{n \Delta} \sum_{i=1}^n \trueE \bigl[L_i^{j_1} (\trueparam) L_i^{j_2} (\trueparam)   | \filtration{i-1} \bigr]  \\
& = \tfrac{4}{n} \sum_{i=1}^n \sum_{i_1, i_2, i_3, i_4 = 1}^d 
\Bigl\{ {\Sigma}^{-1}_{1, i_1i_3} (\sample{X}{i-1}; \gamma, \sigma) 
\partial_{\theta_{j_1}} \mu_1^{i_1}  (\sample{X}{i-1}; \theta )
{\Sigma}_{1, i_3 i_4} (\sample{X}{i-1};\gamma, \sigma) \nonumber \\ 
& \qquad \times {\Sigma}^{-1}_{1, i_2i_4}(\sample{X}{i-1}; \gamma, \sigma) 
\partial_{\theta_{j_2}} \mu_1^{i_2}  (\sample{X}{i-1} ; \theta) 
\Bigr\}  \Bigl|_{\theta= \trueparam} \nonumber \,  
+ \tfrac{1}{n} \sum_{i=1}^n \mathcal{E}(\Delta, \sample{X}{i-1}, \trueparam) \nonumber \\
%
%
& =  \tfrac{4}{n} \sum_{i=1}^n 
\bigl( \partial_{\theta_{j_1}} 
\mu_1(\sample{X}{i-1}; \theta) \bigr)^\top
{\Sigma}^{-1}_{1} (\sample{X}{i-1}; \gamma, \sigma) 
\partial_{\theta_{j_2}} \mu_1 
(\sample{X}{i-1}; \theta) 
\bigl|_{\theta = \trueparam }
+ \tfrac{1}{n} \sum_{i=1}^n \mathcal{E}(\Delta, \sample{X}{i-1}, \trueparam)  \nonumber \\
& = 
\tfrac{4}{n} 
\sum_{i=1}^n 
\bigl( \partial_{\theta_{j_1}} V_{R, 0} (
\sample{X}{i-1}, \beta)  \bigr)^\top 
\Sigma_{1,RR}^{-1} (\sample{X}{i-1}; \sigma)
\partial_{\theta_{j_2}} V_{R, 0} 
( \sample{X}{i-1}, \beta)   \bigl|_{\theta = \trueparam } 
+ \tfrac{1}{n} \sum_{i=1}^n \mathcal{E}(\Delta, \sample{X}{i-1}, \trueparam) 
\nonumber \\ 
& \probconv 4 I_{j_1 j_2} (\trueparam), \nonumber 
\end{align}
as $n \to \infty$, $\Delta \to 0$, $n \Delta \to \infty$, for $\mathcal{E} \in \mathscr{S}$. In the third equation of (\ref{eq:cov_beta}), we have applied Lemma \ref{lemma:aux_matrix} to obtain: 
\begin{align*}
& \tfrac{4}{n} \sum_{i=1}^n 
\bigl( \partial_{\theta_{j_1}} 
\mu_1(\sample{X}{i-1}; \theta) \bigr)^\top
{\Sigma}^{-1}_{1} (\sample{X}{i-1}; \gamma, \sigma) 
\partial_{\theta_{j_2}} \mu_1 
(\sample{X}{i-1}; \theta)  \\
& = \tfrac{4}{n} \sum_{i=1}^n 
\bigl( \partial_{\theta_{j_1}} 
V_{R, 0} (\sample{X}{i-1}, \beta) \bigr)^\top
\Lambda_{1, RR} (\sample{X}{i-1}; \gamma, \sigma) 
\partial_{\theta_{j_2}} V_{R, 0} (\sample{X}{i-1}, \beta)  \nonumber  \\ 
& \; + \tfrac{2}{n} \sum_{i=1}^n 
\bigl( \partial_{\theta_{j_1}} 
V_{R, 0} (\sample{X}{i-1}, \beta) \bigr)^\top
\Lambda_{1, RS} (\sample{X}{i-1}; \gamma, \sigma) 
\partial_{x_R}^{\top} V_{S, 0} (\sample{X}{i-1}, \gamma) 
\partial_{\theta_{j_2}}  V_{R, 0} (\sample{X}{i-1}, \beta )  \nonumber \\ 
& = \tfrac{4}{n} \sum_{i=1}^n 
\bigl( \partial_{\theta_{j_1}} 
V_{R, 0} (\sample{X}{i-1}, \beta) \bigr)^\top
\Lambda_{1, RR} (\sample{X}{i-1}; \gamma, \sigma) 
\Sigma_{1, RR} (\sample{X}{i-1}; \sigma)
\Sigma_{1, RR}^{-1} (\sample{X}{i-1}; \sigma)
\partial_{\theta_{j_2}} V_{R, 0} (\sample{X}{i-1}, \beta)  \nonumber  \\ 
& \; 
+ \tfrac{4}{n} \sum_{i=1}^n 
\bigl( \partial_{\theta_{j_1}} 
V_{R, 0} (\sample{X}{i-1}, \beta) \bigr)^\top
\Lambda_{1, RS} (\sample{X}{i-1}; \gamma, \sigma) 
\Sigma_{1, RS} (\sample{X}{i-1}; \gamma, \sigma)
\Sigma_{1, RR}^{-1} (\sample{X}{i-1}; \sigma)
\partial_{\theta_{j_2}}  V_{R, 0} (\sample{X}{i-1}, \beta )  \nonumber \\  
& = \tfrac{4}{n} \sum_{i=1}^n 
\bigl( \partial_{\theta_{j_1}} 
V_{R, 0} (\sample{X}{i-1}, \beta) \bigr)^\top
\Sigma_{1, RR}^{-1} (\sample{X}{i-1}; \sigma)
\partial_{\theta_{j_2}}  
V_{R, 0} (\sample{X}{i-1}, \beta). 
\end{align*}
We next consider case (b), $d_{\beta} + d_{\gamma} + 1 \leq j_1, j_2 \leq d_{\theta}$. 
(\ref{L_1_3}) and (\ref{L_2}) yield -- we omit the argument $(\sample{X}{i-1}; \gamma, \sigma)$  from the elements of the $\Sigma$ matrix below, and instead apply the substitution $(\sample{X}{i-1}; \gamma, \sigma)\leftrightarrow (i)$  to reduce the size of the expressions: 
 \begin{align*}
&  \tfrac{1}{n} \sum_{i=1}^n 
\trueE \bigl[L_i^{j_1} (\trueparam) L_i^{j_2} (\trueparam) | \filtration{i-1} \bigr] \nonumber \\
& = \tfrac{1}{n} \sum_{i=1}^n \Biggl\{  \sum_{i_1, i_2, i_3, i_4 = 1}^d  \partial_{\theta_{j_1}} {\Sigma}^{-1}_{1, i_1 i_2}(i)\, \partial_{\theta_{j_2}} {\Sigma}^{-1}_{1, i_3 i_4}(i) 
    \nonumber \\
& \qquad \qquad \qquad\qquad \times \Bigl( {\Sigma}_{1, i_1 i_2}(i) 
{\Sigma}_{1, i_3 i_4}(i) + {\Sigma}_{1, i_1 i_3}(i)  {\Sigma}_{1, i_2 i_4}(i) 
+ {\Sigma}_{1, i_1 i_4}(i)  {\Sigma}_{1, i_2 i_3}(i)  \Bigr) 
\nonumber \\[0.2cm]
& \qquad\qquad +  \sum_{i_1, i_2 = 1}^d 
  \partial_{\theta_{j_1}} 
  {\Sigma}^{-1}_{1, i_1 i_2}(i) 
  {\Sigma}^{-1}_{1, i_1 i_2}(i)  
  \partial_{\theta_{j_2}}  \log | {\Sigma}_1(i) |     \nonumber \\
& \qquad\qquad  +   \sum_{i_1, i_2 = 1}^d 
 \partial_{\theta_{j_2}} {\Sigma}^{-1}_{1, i_1 i_2}(i)  
 {\Sigma}^{-1}_{1, i_1 i_2}(i) 
 \partial_{\theta_{j_1}} 
  \log | {\Sigma}_1(i)   |  \nonumber \\
& \qquad\qquad  +  \partial_{\theta_{j_1}}  \log | {\Sigma}_1(i)  | 
\; \partial_{\theta_{j_2}}  \log | {\Sigma}_1(i)   |     \Biggr\}   \Biggr|_{\theta = \trueparam } 
+ \tfrac{1}{n} \sum_{i=1}^n 
\mathcal{E} (\Delta, \sample{X}{i-1}, \trueparam)  \nonumber \\[0.3cm] 
& = \tfrac{1}{n} \sum_{i=1}^n \biggr\{  \sum_{i_1, i_2, i_3, i_4 = 1}^d 
  \partial_{\theta_{j_1}} {\Sigma}^{-1}_{1, i_1 i_2}(i) 
  \partial_{\theta_{j_2}} {\Sigma}^{-1}_{1, i_3 i_4}(i)  
  {\Sigma}_{1, i_1 i_3} (i)   
  {\Sigma}_{1, i_2 i_4} (i)     \nonumber \\ 
& \qquad \qquad\qquad+  \sum_{i_1, i_2, i_3, i_4 = 1}^d 
 \partial_{\theta_{j_1}} {\Sigma}^{-1}_{1, i_1 i_2}(i)   
 \partial_{\theta_{j_2}} {\Sigma}^{-1}_{1, i_3 i_4}(i)   
 {\Sigma}_{1, i_1 i_4}(i)  
 {\Sigma}_{1, i_2 i_3}(i)   \biggr\}   \biggr|_{\theta = \trueparam }  \nonumber \\[0.1cm] 
& \qquad\qquad + \tfrac{1}{n} \sum_{i=1}^n \mathcal{E} (\Delta, \sample{X}{i-1}, \trueparam),   \nonumber  \\[0.3cm]
& =   \tfrac{2}{n} \sum_{i=1}^n 
 \mathrm{tr} \Bigl( \partial_{\theta_{j_1}} 
 {\Sigma}_1  (\sample{X}{i-1}; \theta) 
 {\Sigma}^{-1}_1  (\sample{X}{i-1}; \theta) 
\partial_{\theta_{j_2}} {\Sigma}_1  (\sample{X}{i-1}; \theta)
{\Sigma}^{-1}_1 (\sample{X}{i-1}; \theta ) \Bigr)  \Bigr|_{\theta = \trueparam }
 \\ & \qquad\qquad+ \tfrac{1}{n} \sum_{i=1}^n \mathcal{E} (\Delta, \sample{X}{i-1},  \trueparam),  \nonumber 
\end{align*} 
for $\mathcal{E} \in \mathscr{S}$,
where in the second equality we used (\ref{eq:deriv_logdet}) and in the last equality we used:
\begin{align*}
\partial_{\theta_j}
{\Sigma}^{-1}_{1, i_1 i_2}(i) 
=  -\sum_{l_1, l_2 = 1}^d {\Sigma}^{-1}_{1, i_1 l_1}(i)
\partial_{\theta_j} {\Sigma}_{1, l_1 l_2}(i)
{\Sigma}^{-1}_{1, l_2 i_2} (i)
\end{align*} 
Finally, Lemma \ref{ergodic_thm} yields that as $n \to \infty$, $\Delta \to 0$ and $n \Delta \to \infty$,
\begin{align*}
 \tfrac{1}{n} \sum_{i=1}^n \trueE \bigl[L_i^{j_1} (\trueparam) L_i^{j_2} (\trueparam) | \filtration{i-1} \bigr]
 \probconv4 I_{j_1 j_2 } (\trueparam),  \qquad   d_\beta + d_\gamma + 1 \leq j_1, j_2 \leq d_\theta. 
\end{align*} 

\noindent
{\bf Proof of (\ref{eq:normal_conv_4}).}
From (\ref{L_1_1}), (\ref{L_1_2}), (\ref{L_1_3}) and (\ref{L_2}), we deduce that, as $n \to \infty$, $\Delta \to 0$, $n \Delta \to \infty$,
\begin{align}
& M_{n, \Delta}^{j_1 j_1} M_{n, \Delta}^{j_2 j_2} M_{n, \Delta}^{j_3 j_3} M_{n, \Delta}^{j_4 j_4} \sum_{i=1}^n \trueE [L_{i}^{j_1} (\trueparam) L_{i}^{j_2} (\trueparam) L_{i}^{j_3} (\trueparam) 
 L_{i}^{j_4} (\trueparam) | \mathcal{F}_{t_{i-1}}]  \nonumber \\
& \qquad\qquad\qquad =  \frac{1}{n^2} \sum_{i=1}^n 
\mathcal{E} (1, \sample{X}{i-1}, \trueparam)  \probconv 0, 
\end{align} 
where $\mathcal{E} \in \mathscr{S}$.
\subsubsection{Proof of Lemma \ref{lemma:conv_cov}}  \label{pf_conv_cov} 
\noindent Recall $\textstyle{\ell_{n, \Delta} (\theta) = \ell_{n, \Delta}^{(1)} (\theta)  + \ell_{n, \Delta}^{(2)} (\theta)}$, where the definitions of $\ell_{n, \Delta}^{(i)} (\theta)$, $i = 1,2$, are given in (\ref{eq:l1}) and (\ref{eq:l2}). We divide the function $B_{n, \Delta} (\theta)$,  $\theta \in \Theta$, into two parts:
\begin{align*}
  B_{n, \Delta} (\theta)
  = M_{n, \Delta} \,  \partial_\theta^2 \ell_{n, \Delta}^{(1)} (\theta) \, M_{n, \Delta} 
  + M_{n, \Delta} \, \partial_\theta^2 \ell_{n, \Delta}^{(2)} (\theta) \, M_{n, \Delta} \equiv B_{n, \Delta}^{(1)} (\theta) + B_{n, \Delta}^{(2)} (\theta),
\end{align*} 
where we recall $M_{n, \Delta} = \mathrm{diag} (v_{n, \Delta})$ with  
\begin{align*}
 v_{n, \Delta} = 
 \Bigl[ 
\underbrace{\tfrac{1}{\sqrt{n \Delta}}, \ldots , \tfrac{1}{\sqrt{n \Delta}}}_{d_{\beta} - \mathrm{dim}}, 
\underbrace{\sqrt{\tfrac{\Delta}{n}}, \ldots, \sqrt{\tfrac{\Delta}{n}}}_{d_{\gamma} - \mathrm{dim}}, 
\underbrace{\tfrac{1}{\sqrt{n}}, \ldots, \tfrac{1}{\sqrt{n}}}_{d_{\sigma} - \mathrm{dim}} \Bigr]^{\top}, \qquad  d_{\theta} = d_{\beta} + d_{\gamma} + d_{\sigma}.
\end{align*}
We will show that, as $n \to \infty$, $\Delta \to 0$ and $n \Delta \to \infty$, 
\begin{align}
 &  
 B_{n, \Delta}^{(1)} \bigl(\trueparam + \lambda (\hat{\theta}_n - \trueparam)  \bigr)  \probconv 2 I (\trueparam);   \label{B_conv1} \\
 & B_{n, \Delta}^{(2)} \bigl(\trueparam + \lambda (\hat{\theta}_n - \trueparam)  \bigr)
  \xrightarrow{\mathbb{P}_{\trueparam}} 0, \label{B_conv2}
\end{align} 
uniformly in $\lambda \in (0,1]$. 
\\

\noindent
{\bf Proof of (\ref{B_conv1}).} 
We will compute the second  derivatives of $\contrastG{\theta}$, $\theta \in \Theta$. 
Let $1 \leq k_1, k_2 \leq d_{\beta}, \; 
d_{\beta} + 1 \leq k_3, k_4 \leq d_{\beta} + d_{\gamma}, 
\; d_{\beta} + d_{\gamma} + 1 \leq k_5, k_6 \leq d_{\theta}$. Using the expression of (\ref{T1_decomp}) and Lemma \ref{lemma:aux_matrix}, we have:  
\begin{align*}
& \tfrac{1}{n \Delta} \partial_{\theta_{k_1}} \partial_{\theta_{k_2}} \contrastG{\theta} 
 = \tfrac{1}{n} \sum_{i=1}^n  \,  \biggl\{ 2 
 \bigl( \partial_{\theta_{k_1}}  \mu_{1}  (x; \theta) \bigr)^{\top}  
 {\Sigma}^{-1}_1  (x; \gamma, \sigma ) 
 \partial_{\theta_{k_2}} \mu_{1}  (x; \theta)   \nonumber \\[0.2cm]
 & \qquad -\tfrac{2}{\sqrt{\Delta}} 
  \nu_{i} (\Delta; \trueparam )^{\top}
  {\Sigma}^{-1}_1 (x; \gamma, \sigma )  
  \partial_{\theta_{k_1}} \partial_{\theta_{k_2}} \mu_1 (x; \theta) \nonumber \\[0.2cm]
%
 & \qquad  - 2
  \sum_{\substack{1 \le j_1 \le d_R \\ 1 \le j_2 \le d}}  
 {\Sigma}^{-1}_{1, j_1 j_2} (x; \gamma, \sigma) 
 \bigl( V_{0}^{j_1} (x, \truebeta) - V_{0}^{j_1} (x, \beta) \bigr)
 \partial_{\theta_{k_1}} 
 \partial_{\theta_{k_2}} 
 \mu^{j_2}_{1}  (x; \theta)  \biggr\} \biggr|_{x = \sample{X}{i-1}};   \nonumber  \\[0.3cm] 
%
%
& \tfrac{\Delta}{n} \partial_{\theta_{k_3}} \partial_{\theta_{k_4}} \contrastG{\theta}
= \tfrac{1}{n} \sum_{i=1}^n \, \biggl\{ 2
\bigl( 
\partial_{\theta_{k_3}}  V_{S,0}  (x, \gamma)  \bigr)^{\top}
\Lambda_{1, SS}  (x; \gamma, \sigma )  
\partial_{\theta_{k_4}} V_{S,0}  (x, \gamma)   \nonumber \\[0.2cm] 
& \qquad + {\Delta} \, \nu_i (\Delta; \trueparam)^{\top}   
  \partial_{\theta_{k_3}}
  \partial_{\theta_{k_4}} 
  {\Sigma}^{-1}_1  (x ; \gamma, \sigma) \, 
  \nu_i  (\Delta; \trueparam)   \nonumber \\[0.2cm] 
& \qquad + 
 \bigl(V_{S,0}  (x,  \truegamma)  - V_{S,0} (x,  \gamma) \bigr)^{\top}  
 \partial_{\theta_{k_3}} 
 \partial_{\theta_{k_4}} \Lambda_{1, SS}   (x; \theta) \, 
 \bigl(V_{S,0}  (x, \truegamma)  
 - V_{S,0}  ( x, \gamma) \bigr) 
 \nonumber \\[0.2cm] 
& \qquad +    \sum_{j = d_R + 1}^{d} 
 \mathcal{E}_j^{(k_3, k_4)} (1, x, \theta)  \bigl(V_{0}^{j} (x, \truegamma)  - V_{0}^{j} (x, \gamma ) \bigr) 
 + 
 \mathcal{E}^{(k_3, k_4)} (\Delta, x, \theta)  \biggr\} \biggr|_{x = \sample{X}{i-1}};  
 \\[0.3cm]
%
%
& \tfrac{1}{n} \partial_{\theta_{k_5}}  
\partial_{\theta_{k_6}}  \contrastG{\theta} 
 = \tfrac{1}{n}  \sum_{i=1}^n   \, \biggl\{  
  \nu_i (\Delta ; \trueparam)^{\top}
  \partial_{\theta_{k_5}} 
  \partial_{\theta_{k_6}}  
  {\Sigma}^{-1}_1 (x; \gamma, \sigma) 
  \nu_i (\Delta ; \trueparam)   \nonumber \\[0.2cm]
& \qquad + \partial_{\theta_{k_5}}
   \partial_{\theta_{k_6}}
   \log | {\Sigma}_1 (x; \gamma, \sigma)| 
   + \tfrac{1}{\sqrt{\Delta}}  \sum_{j_1, j_2 = 1}^d  
  \nu_i^{j_1} (\Delta ; \trueparam)  
  \, \mathcal{E}_{j_1 j_2}^{(k_5, k_6)} (1 , x, \theta) \nonumber \\[0.2cm] 
& \qquad + 
 \tfrac{1}{\Delta} \bigl( 
  V_{S,0} (x, \truegamma) - V_{S,0} (x, \gamma) \bigr)^{\top}
 \partial_{\theta_{k_5}}
 \partial_{\theta_{k_6}} 
  \Lambda_{1, SS}  (x ; \gamma, \sigma) 
\bigl( V_{S,0} (x, \truegamma) - V_{S,0} (x, \gamma)  \bigr)  
\nonumber \\[0.2cm]
& \qquad  - 
 \bigl(
 {V_{S,0} (x, \truegamma) - V_{S,0} (x, \gamma)} \bigr)^{\top}
 \partial_{\theta_{k_5}}   \Lambda_{1, SS} (x; \gamma, \sigma) 
 \partial_{\theta_{k_6}} 
 \hat{V}_0 V_{S,0} (x, \theta) 
 \nonumber \\[0.2cm]
& \qquad - 
 \bigl( 
 {V_{S,0} (x, \truegamma) - V_{S,0} (x, \gamma)} \bigr)^{\top} 
 \Lambda_{1,SS}  (x; \gamma, \sigma)
 \partial_{\theta_{k_5}}
 \partial_{\theta_{k_6}} 
 \hat{V}_0 V_{S,0}  (x, \theta)  
 +   \mathcal{E}^{(k_5, k_6)}  (\sqrt{\Delta}, x, \theta) \biggr\} \biggr|_{x = \sample{X}{i-1}};   \\[0.2cm] 
%
%
&  \tfrac{1}{n} \partial_{\theta_{k_1}}
\partial_{\theta_{k_3}} \contrastG{\theta} 
= \tfrac{1}{n} \sum_{i=1}^n \,  \biggl\{ 
 {2}  
 \sum_{\substack{ 1 \le j_1 \le  d  \\ d_R +1 \le j_2 \le d }}   
 \partial_{\theta_{k_1}} \mu_{1}^{j_1} (x; \theta)  
 \Sigma^{-1}_{1, j_1 j_2} (x ; \gamma, \sigma)
 \partial_{\theta_3} V_{0}^{j_2} (x , \gamma) 
\nonumber \\[0.2cm]
& \qquad +  {2} \sum_{j_1, j_2 = 1}^d 
  \partial_{\theta_{k_3}} {\Sigma}^{-1}_{1, j_1 j_2} (x; \gamma, \sigma) 
  \, \partial_{\theta_{k_1}} \mu^{j_1}_{1}  (x;  \theta) 
  \, \frac{\mu_{\Delta}^{j_2} (x; \trueparam) - \mu_{\Delta}^{j_2} (x; \theta)}{\Delta^{\mathbf{1}_{\{j_2 \geq d_R + 1 \}}}} 
  \nonumber \\[0.2cm] 
& \qquad +  {2} \sum_{j_1, j_2 = 1}^d 
  {\Sigma}^{-1}_{1, j_1 j_2} (x; \gamma, \sigma) 
   \partial_{\theta_{k_1}}  \partial_{\theta_{k_3}} 
    \mu^{j_1}_{1} (x; \theta) 
    \frac{\mu_{\Delta}^{j_2} (x; \trueparam) - \mu_{\Delta}^{j_2} (x; \theta)}{\Delta^{\mathbf{1}_{\{j_2 \geq d_R +1 \}}}}  \nonumber \\[0.2cm]
& \qquad +
  \sum_{j_1, j_2 = 1}^d  
  \nu^{j_1}_i (\Delta; \trueparam) 
  \mathcal{E}_{j_1 j_2}^{(k_1, k_3)} (\sqrt{\Delta}, x,  \theta) 
  \biggr\}  \biggr|_{x = \sample{X}{i-1}} ;  \\[0.3cm]
%
%
& \tfrac{1}{n \sqrt{\Delta}}  
\partial_{\theta_{k_1}} 
\partial_{\theta_{k_5}} \contrastG{\theta}  
 = \tfrac{1}{n} \sum_{i=1}^n \, \biggl\{ - 2  \sum_{j_1, j_2 = 1}^d  
  \partial_{\theta_{k_5}} {\Sigma}^{-1}_{1, j_1 j_2} (x; \gamma, \sigma) 
  \partial_{\theta_{k_1}} \mu^{j_1}_{1} (x; \theta) 
 \frac{\mu_{\Delta}^{j_2} (x; \trueparam) - \mu_{\Delta}^{j_2} (x; \theta)}{\Delta^{ \frac{1}{2} + \mathbf{1}_{\{j_2 \geq d_R +1 \}}}}  \nonumber  \\[0.2cm] 
& \qquad - 2 \sum_{j_1, j_2 = 1}^d  
 {\Sigma}^{-1}_{1, j_1 j_2} (x; \gamma, \sigma) 
 \partial_{\theta_{k_1}} \partial_{\theta_{k_5}} \mu^{j_1}_{1} (x; \theta) 
 \frac{\mu_{\Delta}^{j_2} (x; \trueparam) - \mu_{\Delta}^{j_2} (x; \theta)}{\Delta^{ \frac{1}{2} + \mathbf{1}_{\{j_2 \geq d_R +1 \}}}}   \nonumber  \\[0.2cm]  
& \qquad +  {\sqrt{\Delta}}
 \bigl(
  \partial_{\theta_{k_1}} \mu_{1}  (x; \theta) \bigr)^{\top} 
  {\Sigma}^{-1}_1  (x; \gamma, \sigma)  
  \partial_{\theta_{k_5}}  \mu_{1} (x; \theta)  
  +   \sum_{j_1, j_2= 1}^d 
\nu_i^{j_1} (\Delta; \trueparam)  \mathcal{E}_{j_1 j_2}^{(k_1, k_5)}  (1, x, \theta) 
\biggr\} \Bigl|_{x = \sample{X}{i-1}};   \\[0.2cm] 
%
%
& \tfrac{\sqrt{\Delta}}{n}  
\partial_{\theta_{k_3}} \partial_{\theta_{k_5}}  \contrastG{\theta}  
 = \tfrac{1}{n} \sum_{i=1}^n  \, \biggl\{  
  {\sqrt{\Delta}} \, 
  \nu_i (\Delta ; \trueparam)^{\top}
  \partial_{\theta_{k_3}}
  \partial_{\theta_{k_5}}  
  {\Sigma}^{-1}_1  (x; \gamma, \sigma)
  \nu_i (\Delta ; \trueparam)     \nonumber \\[0.2cm]  
 & \qquad + \tfrac{1}{\sqrt{\Delta}}   
 \bigl( V_{S,0}  (x, \truegamma) - V_{S,0} (x, \gamma) \bigr)^{\top} 
 \partial_{\theta_{k_3}}
 \partial_{\theta_{k_5}}   
 \Lambda_{1, SS}  (x; \gamma, \sigma) 
 \bigl( V_{S,0}  (x, \truegamma) - V_{S,0} (x, \gamma) \bigr)      \nonumber \\[0.2cm] 
 & \qquad +   \sum_{j_1, j_2 = d_R + 1}^d 
\frac{V_0^{j_1} (x , \truegamma) - V_0^{j_1} (x, \gamma) }{\sqrt{\Delta}} 
 \mathcal{E}_{j_1 j_2}^{(k_3, k_5)}  (1, x, \theta) 
 + \sum_{j_1, j_2 = 1}^d  \nu_i^{j_1} (\Delta ; \trueparam) \,  \widetilde{\mathcal{E}}_{j_1 j_2}^{\, (k_3, k_5)} (1, x, \theta)  \nonumber \\[0.2cm]  
& \qquad  
 + {\mathcal{E}}^{ (k_3, k_5)} ( \sqrt{\Delta}, x, \theta) \biggr\} \biggr|_{x = \sample{X}{i-1}},     
\end{align*}
for $\theta = (\beta, \gamma, \sigma) \in \Theta$,  with $\mathcal{E}^{(m_1, m_2)},  \mathcal{E}^{(m_1, m_2)}_{j_1 j_2}, \widetilde{\mathcal{E}}^{\, (m_1, m_2)}_{j_1 j_2}  \in \mathscr{S}$, for all $ 1 \leq j_1, j_2 \leq d, \, 1 \le m_1, m_2 \le d_\theta$. 
Due to Lemmas \ref{lemma:aux_matrix}, \ref{ergodic_thm},  \ref{canonical_conv}, \ref{canonical_conv2} and the consistency of the estimator $\hat{\theta}_n$, we obtain the following convergences. As $n \to \infty, \; \Delta \to 0$, $n \Delta \to \infty$, 
\begin{align} 
\tfrac{1}{n \Delta} 
\partial_{\theta_{k_1}} \partial_{\theta_{k_2}} \contrastG{\theta}|_{\theta =  \trueparam + \lambda (\hat{\theta}_n - \trueparam)} 
& \probconv 2  \int 
  \bigl( \partial_{\theta_{k_1}} \mu_1  (x; \theta) \bigr)^{\top} 
  {\Sigma}^{-1}_1 (x; \gamma, \sigma )  
  \partial_{\theta_{k_2}}  \mu_1  (x; \theta) \Bigl|_{\theta = \trueparam}
   \,  \truedist (dx)  \nonumber \\
& 
= 2 \int 
\bigl( \partial_{\theta_{k_1}} V_{R, 0} (x, \beta) \bigr)^\top
\Sigma_{1, RR}^{-1} (x; \sigma) \partial_{\theta_{k_2}} V_{R, 0} (x, \beta) 
\Bigl|_{\theta = \trueparam} 
\truedist(dx) \nonumber  \\[0.1cm] 
%
\tfrac{\Delta}{n} 
\partial_{\theta_{k_3}} \partial_{\theta_{k_4}} 
 \contrastG{\theta}|_{\theta =  \trueparam + \lambda (\hat{\theta}_n - \trueparam)} 
& \probconv 2 
  \int 
  \bigl( \partial_{\theta_{k_3}} V_{S,0} (x, \gamma) \bigr)^{\top}
  \Lambda_{1,SS} (x; \gamma, \sigma)  
  \partial_{\theta_{k_4}} V_{S,0} (x, \gamma) 
 \Bigl|_{\theta = \trueparam} \,  \truedist (dx);   \nonumber \\[0.3cm]
%
\tfrac{1}{n} 
\partial_{\theta_{k_1}} \partial_{\theta_{k_3}} 
\contrastG{\theta}|_{\theta =  \trueparam + \lambda (\hat{\theta}_n - \trueparam)} 
&  \probconv 2 
\sum_{\substack{1 \le j_1 \le d \\ d_R +1 \le j_2 \le d}} \int 
\partial_{\theta_{k_1}} \mu_1^{j_1}  (x; \theta) 
{\Sigma}^{-1}_{1, j_1j_2}  (x; \gamma, \sigma)  
\partial_{\theta_{k_3}}  V_0^{j_2} (x, \gamma)
\Bigl|_{\theta = \trueparam} \, \truedist (dx)
= 0; \nonumber \\[0.1cm]   
%
\Bigl\{ \tfrac{1}{n \sqrt{\Delta}}  \partial_{\theta_{k_1}} \partial_{\theta_{k_5}} 
\contrastG{\theta}
& + \tfrac{\sqrt{\Delta}}{n}  
\partial_{\theta_{k_3}}  \partial_{\theta_{k_5}}  \contrastG{\theta} \Bigr\}  |_{\theta =  \trueparam + \lambda (\hat{\theta}_n - \trueparam)}
\probconv 0;  \nonumber \\[0.3cm]
%
 \partial_{\theta_{k_5}}  \partial_{\theta_{k_6}}  
\contrastG{\theta}|_{\theta =  \trueparam + \lambda (\hat{\theta}_n - \trueparam)} 
& \probconv 
%
\int
 \Bigl\{ 
 \mathrm{tr} \Bigl( \partial_{\theta_{k_5}}  \partial_{\theta_{k_6}}   {\Sigma}^{-1}_{1} (x; \gamma, \sigma) 
 {\Sigma}_{1} (x; \gamma, \sigma) \Bigr)
 +  \partial_{\theta_{k_5}}  \partial_{\theta_{k_6}} 
 \log |{\Sigma}_1 (x; \gamma, \sigma)| \Bigr\}  
 \Bigl|_{\theta = \trueparam} \,  \truedist (dx) \nonumber \\
 & \qquad \quad  =  \int 
\mathrm{tr} \Bigl(  
\partial_{\theta_{k_5}}  {\Sigma}_1  (x; \gamma, \sigma) 
\, {\Sigma}^{-1}_1 (x; \gamma, \sigma) 
\,  \partial_{\theta_{k_6}}  {\Sigma}_1 (x; \gamma, \sigma) 
\,  {\Sigma}^{-1}_1 (x; \gamma, \sigma) 
 \Bigr) 
 \Bigl|_{\theta = \trueparam} \, \truedist (dx),  \nonumber 
\end{align} 
uniformly in $\lambda \in [0,1]$, where we used 
\begin{align*}
& \partial_{\theta_{k_5}}  \partial_{\theta_{k_6}}  
\log |{\Sigma}_1 (x; \gamma, \sigma) | 
 =
-  \mathrm{tr} \bigl( \partial_{\theta_{k_5}}  \partial_{\theta_{k_6}}   {\Sigma}^{-1}_{1} (x; \gamma, \sigma) 
 {\Sigma}_{1} (x; \gamma, \sigma) \bigr)
- \mathrm{tr} 
  \Bigl( \partial_{\theta_{k_5}}  {\Sigma}_1 (x;  \gamma, \sigma) 
 \partial_{\theta_{k_6}}   {\Sigma}_1^{-1} (x; \gamma, \sigma)   \Bigr)
\end{align*} 
and
\begin{align*}
%
 & \mathrm{tr} 
  \Bigl( \partial_{\theta_{k_5}}  {\Sigma}_1 (x;  \gamma, \sigma) 
 \partial_{\theta_{k_6}}   {\Sigma}_1^{-1} (x; \gamma, \sigma)   \Bigr)  
 = -  \mathrm{tr} 
 \Bigl(
 \partial_{\theta_{k_5}}  {\Sigma}_1 (x;  \gamma, \sigma)
 \,  
 {\Sigma}^{-1}_1 (x; \gamma, \sigma) 
 \, 
 \partial_{\theta_{k_6}}  {\Sigma}_1  (x; \gamma, \sigma)
 \, 
 {\Sigma}^{-1}_1  (x; \gamma, \sigma)   \Bigr). 
\end{align*} 
%
Hence, we obtain the convergence (\ref{B_conv1}). 
\\

\noindent
{\bf Proof of (\ref{B_conv2}).} 
Let $1 \le l_1, l_2 \le d_\theta$. We have $\textstyle{M_{n, \Delta}^{l_1 l_1} M_{n, \Delta}^{l_2 l_2} \partial_{\theta_{l_1}}  \partial_{\theta_{l_2}} 
 \contrastW{\theta} = \sum_{l=1}^4 \Upsilon_l (\theta)}$, where we have set:
\begin{align}
& \Upsilon_1 (\theta)  =  - 2 \Delta M_{n, \Delta}^{l_1 l_1} M_{n, \Delta}^{l_2 l_2}   \sum_{i=1}^n \sum_{i_1, i_2 = 1}^{d} 
 \partial_{\theta_{l_1}} \partial_{\theta_{l_2}}  {G}_{i_1 i_2}  (\sample{X}{i-1}; \theta) \mathcal{H}_{(i_1, i_2)} (\Delta, \sample{X}{i-1}, \sample{X}{i}; \theta);   \nonumber \\
& \Upsilon_2 (\theta) =  - 2 \Delta  M_{n, \Delta}^{l_1 l_1} M_{n, \Delta}^{l_2 l_2} 
  \sum_{i=1}^n \sum_{i_1, i_2 = 1}^{d} 
   \partial_{\theta_{l_1}} {G}_{i_1 i_2}  (\sample{X}{i-1}; \theta) 
   \partial_{\theta_{l_2}} \mathcal{H}_{(i_1, i_2)}  (\Delta, \sample{X}{i-1}, \sample{X}{i}; \theta); \nonumber \\
& \Upsilon_3 (\theta) =   -2 \Delta  M_{n, \Delta}^{l_1 l_1} M_{n, \Delta}^{l_2 l_2}  \sum_{i=1}^n \sum_{i_1, i_2 = 1}^{d}   
\partial_{\theta_{l_2}} {G}_{i_1 i_2}  (\sample{X}{i-1}; \theta)  \partial_{\theta_{l_1}}  \mathcal{H}_{(i_1, i_2)} (\Delta, \sample{X}{i-1}, \sample{X}{i}; \theta); \nonumber \\
& \Upsilon_4 (\theta) =  -2 \Delta  M_{n, \Delta}^{l_1 l_1} M_{n, \Delta}^{l_2 l_2}  \sum_{i=1}^n \sum_{i_1, i_2 = 1}^{d}  
 {G}_{i_1 i_2} (\sample{X}{i-1}; \theta) 
   \partial_{\theta_{l_1}} \partial_{\theta_{l_2}}   \mathcal{H}_{(i_1, i_2)} (\Delta, \sample{X}{i-1}, \sample{X}{i}; \theta). \nonumber 
\end{align}
Notice that each term of $\Upsilon_l (\theta)$ has a factor $\Delta$ in the front. We obtain  
\begin{align*}
 \Upsilon_1 \bigl( \hat{\theta}_n - \lambda (\trueparam - \hat{\theta}_n) \bigr) \probconv 0, 
\end{align*} 
uniformly in $\lambda \in (0,1]$ from the same discussion we have shown the convergence (\ref{conv_weight}). We turn to terms $\Upsilon_2 \bigl( \theta \bigr)$ and $\Upsilon_3 \bigl( \theta \bigr)$. Recall $\textstyle{ \mathcal{H}_{(i_1, i_2)} (\Delta, \sample{X}{i-1}, \sample{X}{i}; \theta) = \sum_{l = 1}^5 \eta_{(i_1, i_2)}^{(l)} (\Delta, \sample{X}{i-1}, \sample{X}{i}; \theta)}$ in Section \ref{sec:conv_contrast_sigma}.  For $1 \leq k_1 \leq d_{\beta}$, $ d_{\beta} + 1 \leq k_2 \leq d_{\beta} + d_{\gamma}$, $d_{\beta} + d_{\gamma} + 1 \leq k_3 \leq d_{\theta}$, we have 
\begin{align*}
& \Delta\,  \partial_{\theta_{k_1}} \mathcal{H}_{(i_1, i_2)} (\Delta, \sample{X}{i-1}, \sample{X}{i}; \theta)
=  \sum_{j =1}^d  \mathcal{E}_{j}^{(k_1)} (\sqrt{\Delta^3}, \sample{X}{i-1}, \theta) \nu_i^j (\Delta ; \theta), \nonumber \\[0.3cm]
%
%
& \Delta\, \partial_{\theta_{k_2}} \mathcal{H}_{(i_1, i_2)} (\Delta, \sample{X}{i-1}, \sample{X}{i}; \theta)  
= \biggl\{  \sum_{j_1, j_2 = 1}^d \mathcal{E}_{j_1 j_2}^{(k_2)} (\Delta, x,  \theta) \nu_i^{j_1} (\Delta; \trueparam) \nu_i^{j_2} (\Delta;\trueparam) \nonumber \\
& \quad + \sum_{j=1}^d 
 \mathcal{E}_j^{(k_2)} (\sqrt\Delta, x, \theta) \nu_i^j (\Delta; \trueparam)  
  + \sum_{j=1}^d 
 \widetilde{\mathcal{E}}_j^{\, (k_2)} (\sqrt{\Delta}, x, \theta) 
 \frac{V_0^{j} (x,\truegamma) - V_0^{j} (x, \gamma) }{\sqrt\Delta}  \nonumber \\
& \quad + \sum_{j_1, j_2 =1}^d \widetilde{\mathcal{E}}_{j_1 j_2}^{\, (k_2)}  (\Delta, x,  \theta) 
\frac{V_0^{j_1} (x , \truegamma) - V_0^{j_1} (x, \gamma) }{\sqrt\Delta}
\frac{V_0^{j_2} (x, \truegamma) - V_0^{j_2} (x, \gamma) }{\sqrt\Delta}   + \mathcal{E}^{(k_2)} (\Delta, x, \theta) \biggr\} \Bigr|_{x = \sample{X}{i-1}} ;  \nonumber \\[0.3cm] 
& \Delta\, \partial_{\theta_{k_3}} \mathcal{H}_{(i_1, i_2)} (\Delta, \sample{X}{i-1}, \sample{X}{i}; \theta)  
= \biggl\{  \sum_{j_1, j_2 = 1}^d \mathcal{E}_{j_1 j_2}^{(k_3)} (\Delta, x,  \theta) 
\nu_i^{j_1} (\Delta; \trueparam) \nu_i^{j_2} (\Delta; \trueparam) 
\nonumber \\
& + \sum_{j_1, j_2 = d_R + 1}^d \widetilde{\mathcal{E}}_{j_1 j_2}^{\, (k_3)}  (\Delta, x,  \theta) 
  \frac{V_0^{j_1} (x , \truegamma) - V_0^{j_1} (x , \gamma) }{\sqrt\Delta}
  \frac{V_0^{j_2} (x , \truegamma) - V_0^{j_2} (x , \gamma) }{\sqrt\Delta}   \nonumber \\
& + \sum_{j = d_R + 1}^d {\mathcal{E}}_j^{(k_3)} (\Delta, x, \theta) \frac{V_0^{j} (x ,\truegamma) - V_0^{j} (x , \gamma) }{\sqrt\Delta}  
+ \sum_{j=1}^d \widetilde{\mathcal{E}}_j^{\, (k_3)} (\sqrt\Delta, x, \theta) \nu_i^j (\Delta;\trueparam) 
  + \mathcal{E}^{(k_3)} (\Delta, x,  \theta) \biggr\} \Bigr|_{x = \sample{X}{i-1}},   \nonumber  
\end{align*}
where $\mathcal{E}^{(m)}, \, \mathcal{E}^{(m)}_{j_1}, \,  \mathcal{E}^{(m)}_{j_1 j_2}, \, \widetilde{\mathcal{E}}^{\, (m)}_{j_1}, \,  \widetilde{\mathcal{E}}^{(m)}_{j_1 j_2} \in \mathscr{S}$, for $1 \le j_1, j_2 \le d$ and $1 \le m \le d_\theta$. Applying Lemmas \ref{ergodic_thm}, \ref{canonical_conv} and \ref{canonical_conv2}, we obtain 
\begin{align*}
    \Upsilon_k \bigl( \hat{\theta}_n - \lambda (\trueparam - \hat{\theta}_n) \bigr) \probconv 0, \qquad  k = 2,3,
\end{align*} 
uniformly in $\lambda \in (0,1]$. 
Finally, we consider the term $\Upsilon_4 (\theta)$. We compute the second order partial derivatives of $\mathcal{H}_{(i_1, i_2)} (\Delta, \sample{X}{i-1}, \sample{X}{i}; \theta)$ with respect to the parameter $\theta \in \Theta $ as follows. For $1 \leq k_1, k_2 \leq d_{\beta}, \; d_{\beta} + 1 \leq k_3, k_4 \leq d_{\beta} + d_{\gamma}, \; d_{\beta} + d_{\gamma} + 1 \leq k_5,k_6 \leq d_{\theta}$, 
\begin{align}
& \Delta \,\partial_{\theta_{k_1}} \partial_{\theta_{k_2}} \mathcal{H}_{(i_1, i_2)} (\Delta, \sample{X}{i-1}, \sample{X}{i}; \theta) 
 =  \biggl\{  
 \sum_{j=1}^d \mathcal{E}_j^{(k_1, k_2)}  (\sqrt{\Delta^3}, x,  \theta) \, \nu_i^j(\Delta; \trueparam)  \nonumber \\
& \qquad   + \sum_{j=d_R + 1}^d 
\widetilde{\mathcal{E}}_j^{(k_1, k_2)} (\sqrt{\Delta^3}, x, \theta) 
  \frac{V_0^{j} (x ,  \truegamma) - V_0^{j} (x ,  \gamma)}{\sqrt{\Delta}} 
  + \mathcal{E}^{(k_1, k_2)} (\Delta^2, x, \theta) 
  \biggr\} \biggr|_{x= \sample{X}{i-1}};   \nonumber \\[0.3cm] 
%
& \Delta\, \partial_{\theta_{k_3}} \partial_{\theta_{k_4}} \mathcal{H}_{(i_1, i_2)} (\Delta, \sample{X}{i-1}, \sample{X}{i}; \theta) 
= \biggl\{ 
\sum_{j_1, j_2 = 1}^d 
\mathcal{E}_{j_1 j_2}^{(k_3, k_4)}  (\Delta, x,  \theta)  
 \nu_i^{j_1} (\Delta ; \trueparam) \nu_i^{j_2} (\Delta ; \trueparam)  \nonumber \\
&  \qquad   + \sum_{j=1}^d \mathcal{E}_j^{(k_3, k_4)}  (\sqrt{\Delta}, x,  \theta) \nu_i^j(\Delta; \trueparam) 
     + \mathcal{E}^{(k_3, k_4)} (1, x,  \theta) \biggr\} \biggr|_{x = \sample{X}{i-1}};   \nonumber \\[0.3cm] 
%
& \Delta\, \partial_{\theta_{k_5}} \partial_{\theta_{k_6}} \mathcal{H}_{(i_1, i_2)} (\Delta, \sample{X}{i-1}, \sample{X}{i}; \theta)  = 
 \biggl\{ \sum_{j_1, j_2 = 1}^d \mathcal{E}_{j_1 j_2}^{(k_5, k_6)}  (\Delta, x,  \theta) 
 \nu_i^{j_1} (\Delta ; \trueparam) \nu_i^{j_2} (\Delta ; \trueparam)  \nonumber \\ 
& \qquad  +  \sum_{j_1, j_2 = d_R +1}^d 
\widetilde{\mathcal{E}}_{j_1 j_2}^{\, (k_5, k_6)}  (1, x, \theta)  
  \bigl( V_0^{j_1} (x, \truegamma)  -  V_0^{j_1} (x, \gamma)  \bigr) 
  \bigl( V_0^{j_2} (x, \truegamma)  -  V_0^{j_2} (x, \gamma)  \bigr) \nonumber \\
& \qquad  + \sum_{j=1}^d  \mathcal{E}_j^{(k_5, k_6)}  (\sqrt{\Delta}, x,  \theta) \nu_i^j(\Delta; \trueparam) 
+ \mathcal{E}^{(k_5, k_6)} (\Delta, x, \theta) \biggr\} \biggr|_{x =\sample{X}{i-1}} ;   \nonumber \\[0.3cm]
%
& \Delta\, \partial_{\theta_{k_1}} \partial_{\theta_{k_3}} \mathcal{H}_{(i_1, i_2)} (\Delta, \sample{X}{i-1}, \sample{X}{i}; \theta) 
 =  \sum_{j=1}^d  
  \mathcal{E}_j^{(k_1, k_3)}  ( \sqrt{\Delta^3}, \sample{X}{i-1}, \theta)  \nu_i^j (\Delta ; \trueparam) 
  + \mathcal{E}^{(k_1, k_3)} (\Delta,  \sample{X}{i-1}, \theta); 
  \nonumber \\[0.3cm]
 & \Delta\, \partial_{\theta_{k_1}} \partial_{\theta_{k_5}} \mathcal{H}_{(i_1, i_2)} (\Delta, \sample{X}{i-1}, \sample{X}{i}; \theta) 
 =  \sum_{j=1}^d  
 \mathcal{E}_j^{(k_1, k_5)}  (\sqrt{\Delta^3}, \sample{X}{i-1}, \theta)  \nu_i^j (\Delta ; \trueparam) 
  + \mathcal{E}^{(k_1, k_5)} (\Delta, \sample{X}{i-1}, \theta); \nonumber \\[0.3cm]
%
 & \Delta\, \partial_{\theta_{k_3}} \partial_{\theta_{k_5}}  \mathcal{H}_{(i_1, i_2)} (\Delta, \sample{X}{i-1}, \sample{X}{i}; \theta)  
= \sum_{j_1, j_2 = 1}^d 
\mathcal{E}_{j_1 j_2}^{(k_3, k_5)}  (\Delta, \sample{X}{i-1},  \theta)  \nu_i^{j_1} (\Delta ; \trueparam) \nu_i^{j_2} (\Delta ; \trueparam)  \nonumber \\
& \qquad  +  \sum_{j=1}^d   
\mathcal{E}_j^{(k_3, k_5)}  (\sqrt{\Delta}, \sample{X}{i-1}, \theta)  \nu_i^j (\Delta ; \trueparam) 
  + \mathcal{E}^{(k_3, k_5)} (1, \sample{X}{i-1}, \theta),  \nonumber 
\end{align}  
where $\mathcal{E}^{(m_1, m_2)},
\mathcal{E}_{j_1}^{(m_1, m_2)}, 
\mathcal{E}_{j_1 j_2}^{(m_1, m_2)},
\widetilde{\mathcal{E}}_{j_1}^{\, (m_1, m_2)},
\widetilde{\mathcal{E}}_{j_1 j_2}^{\, (m_1, m_2)} \in \mathscr{S}$ for all 
$1 \le j_1, j_2 \le d$ and all $1 \le m_1, m_2 \le d_\theta$. 
Hence, from Lemmas \ref{ergodic_thm}, \ref{canonical_conv} and the consistency of $\hat{\theta}_n$, we obtain 
\begin{align*}
 \Upsilon_4 \bigl( \hat{\theta}_n - \lambda (\trueparam - \hat{\theta}_n) \bigr) \probconv 0,
\end{align*} 
uniformly in $\lambda \in (0,1]$. Thus, as $n \to \infty, \; \Delta \to 0, \; n \Delta \to \infty$,
\begin{align*} 
 M_{n, \Delta}^{l_1 l_1} M_{n, \Delta}^{l_2 l_2} 
 \partial_{\theta_{l_1}} \partial_{\theta_{l_2}}
 \contrastW{\theta} |_{\theta =  \trueparam + \lambda (\hat{\theta}_n - \trueparam)}  \probconv 0, 
\end{align*}
uniformly in $\lambda \in (0,1]$,  which implies the convergence (\ref{B_conv2}).  

\section{Proof of Proposition \ref{prop:weak_app_hypo}} \label{appendix:weak_app_hypo}
\noindent We define
$ \bar{X}_\Delta^x 
:=  Z_\Delta^x + \bar{Z}_\Delta^x$, 
with 
${Z}_\Delta^x = [Z_{R, \Delta}^\top, Z_{S,\Delta}^\top]^{\top}$ 
and
$\bar{Z}_\Delta^x = [\mathbf{0}_{d_R}^\top, \bar{Z}_{S,\Delta}^\top]^{\top}$ 
 given as: 
\begin{align*}
 Z_{R,\Delta} & =  x_{R} + V_{R, 0} (x, \beta) \Delta +  \sum_{k=1}^{d_R} V_{R, k} (x, \sigma) B_{k, \Delta}
 + \sum_{k_1, k_2 = 0}^{d_R} \hat{V}_{k_1} V_{R, k_2}(x, \theta) 
\zeta_{k_1 k_2, \Delta};  \\ 
 Z_{S,\Delta} & = x_{S} + V_{S, 0} (x, \theta) \Delta  
+ \sum_{k = 0}^{d_R} \hat{V}_{k} V_{S, 0} (x, \theta) \zeta_{k0, \Delta}; 
\\[0.3cm]
\bar{Z}_{S,\Delta} 
& = \sum_{\alpha \in \{ 0,\ldots,d_R \}^2 \setminus \{0,0\} } \hat{V}_{\alpha_1} \hat{V}_{\alpha_2} V_{S, 0} (x,  \theta) \eta_{\alpha_1 \alpha_2, \Delta}. 
\end{align*}
We are able to show that $Z_\Delta^x $ has a third order of accuracy on a single step in the sense of weak approximation for $X_\Delta^x$. In fact, one can check that random variables $\zeta_{k_1 k_2, \Delta}, \; k_1, k_2 = 1, \ldots, d_R$ satisfy the moment conditions given in 
Lemma 2.1.5 in \cite{mil:21} and then, the scheme has local weak third order precision from Theorem 2.1.6 in \cite{mil:21}. Hence, it suffices to show that for $\varphi \in C_p^{\infty} (\mathbb{R}^d)$, there exist constants $C ,q > 0$ such that 
\begin{align} \label{eq:error_bound}
 \Bigl|  \mathbb{E}_{\theta} [ \varphi (\bar{X}_\Delta^x )] 
       - \mathbb{E}_{\theta} [ \varphi ({Z}_\Delta^x )]  \Bigr|
   \leq C (1 + |x|^q) \Delta^3 . 
\end{align} 
This is because once we have the bound (\ref{eq:error_bound}), we obtain: 
\begin{align*}
 \Bigl|  \mathbb{E}_{\theta} [ \varphi ({X}_\Delta^x )] 
       - \mathbb{E}_{\theta} [ \varphi (\bar{X}_\Delta^x)]  \Bigr|
 & \le
 \Bigl|  \mathbb{E}_{\theta} [ \varphi ({X}_\Delta^x )] 
       - \mathbb{E}_{\theta} [ \varphi ({Z}_\Delta^x )]  \Bigr|
+ \Bigl|  \mathbb{E}_{\theta} [ \varphi (\bar{X}_\Delta^x )] 
       - \mathbb{E}_{\theta} [ \varphi ({Z}_\Delta^x )]  \Bigr| \\[0.1cm] 
&  \le C (1 + |x|^q) \Delta^3.  
\end{align*}
Now, we define 
${\Delta}_{\bar{X}} = \bar{X}_\Delta^x  - x \in \mathbb{R}^d$
and 
${\Delta}_Z := {Z}_\Delta^x  - x \in \mathbb{R}^d$. Let $\varphi \in C_p^\infty (\mathbb{R}^d)$. We consider the Taylor expansions of $\varphi (\bar{X}_\Delta^x )$ and $\varphi (Z_\Delta^x )$ around $\varphi (x)$, respectively as: 
\begin{align*} 
 \varphi  ( \bar{X}_\Delta^x ) 
  = \varphi (x) 
  +  \sum_{k=1}^5 \sum_{\beta \in \{1, \ldots, d \}^k} \frac{1}{k!} \partial^{\beta} \varphi (x)
   \prod_{l=1}^k {\Delta}^{\beta_l}_{\bar{X}}  
  +  {R}^{\varphi}_{\bar{X}} (x, \Delta; \theta);   \\
 \varphi  (Z_\Delta^x) 
  = \varphi (x) 
 +  \sum_{k=1}^5 \sum_{\beta \in \{1, \ldots, d \}^k} \frac{1}{k!} \partial^{\beta} \varphi (x)
   \prod_{l=1}^k  {\Delta}^{\beta_l}_{Z}  
   + {R}^{\varphi}_{Z} (x, \Delta; \theta),    
\end{align*}
where ${R}^{\varphi}_{\bar{X}} (x, \Delta; \theta)$, ${R}^{\varphi}_{Z} (x, \Delta; \theta)$ are remainder terms specified as follows. Under  condition (\ref{assump:coeff}), there exist constants $C, q > 0$ such that 
\begin{align*} 
 | \mathbb{E}_{\theta} \bigl[ R^{\varphi}_{\bar{X}} (x, \Delta; \theta) \bigr] | 
 + | \mathbb{E}_{\theta} [ R^{\varphi}_{Z} (x, \Delta; \theta) ] | \leq C  (1 + |x|^q) \Delta^3,       
\end{align*} 
that comes from the polynomial growth of derivatives of $\varphi$ with estimates given as follows. For any $p \ge 1$, there exist constants $C,q >0$ such that
\begin{align*}
    \mathbb{E}_\theta [|\bar{X}_\Delta^x|^p ] \le C (1 + |x|^q) 
\end{align*}
and
\begin{align*}
\Bigl| \mathbb{E}_{\theta} 
\Bigl[ \prod_{l=1}^6 {\Delta}^{\beta_l}_{\bar{X}} 
\Bigr] \Bigr|  = \mathcal{O} (\Delta^3),  \qquad 
\Bigl| \mathbb{E}_{\theta} 
\Bigl[ \prod_{l=1}^6 {\Delta}^{\beta_l}_{Z} 
\Bigr] \Bigr| 
= \mathcal{O} (\Delta^3), 
\end{align*}
for any multi-index $\beta \in \{1, \ldots, d\}^6$. 
Thus, we have  
\begin{align*}
\mathbb{E}_{\theta} [ \varphi  ( \bar{X}_\Delta^x ) ]  
-  \mathbb{E}_{\theta} [ \varphi (Z_\Delta^x) ] 
& =  
 \sum_{k=1}^5 \sum_{\beta \in \{1, \ldots, d \}^k} 
 \frac{1}{k!}  \partial^{\beta} \varphi (x)
 \mathbb{E}_{\theta} 
 \Bigl[\prod_{l=1}^k {\Delta}^{\beta_l}_{\bar{X}}  -  \prod_{l=1}^k {\Delta}^{\beta_l}_{Z} \Bigr] \nonumber \\[0.1cm]
& \qquad +  {R}^{\varphi}_{\bar{X}}  (x, \Delta ; \theta)  - {R}^{\varphi}_{Z} (x, \Delta; \theta). 
\end{align*} 
It follows that
\begin{align}
& \sum_{l = 1}^d  
\mathbb{E}_{\theta} \Bigl[ {\Delta}^{l}_{\bar{X}}  -   {\Delta}^{l}_Z \Bigr] 
=  
\sum_{l = d_R + 1}^{d} \mathbb{E}_{\theta} \Bigl[  \bar{Z}_\Delta^{x, l}  \Bigr] 
= 0 ;  \nonumber \\ 
& \sum_{\beta \in \{1, \ldots, d \}^k} 
 \Bigl| 
 \mathbb{E}_{\theta} \Bigl[\prod_{l=1}^k  {\Delta}^{\beta_l}_{\bar{X}} -  \prod_{l=1}^k {\Delta}^{\beta_l}_Z \Bigr]   
 \Bigr| \leq c \Delta^3, \ \ k = 2, 3, 4, 5,  \nonumber   
\end{align}
for some constant $c > 0$ from the following properties of the random variable $\eta_{k_1 k_2 ,\Delta}$. For $1 \le k_1, k_2, k_3 \le d_R$,
\begin{align*}
& \mathbb{E} [\eta_{k_1 k_2, \Delta}] = 0,   \ \ 
\mathbb{E} [\eta_{k_1 k_2, \Delta} B_{k_3,\Delta}] = 0, \ \ 
| \mathbb{E}[ \eta_{k_1 k_2, \Delta}]| = \mathcal{O} (\Delta^2);    \nonumber \\
& \mathbb{E} [\eta_{k_1 0, \Delta} ] 
= \mathbb{E} [\eta_{0 k_1, \Delta} ]  = 0, \ \ 
\bigl| \mathbb{E} [ ( \eta_{k_1 0, \Delta}  +  \eta_{0 k_1, \Delta}  ) B_{k_3, \Delta}] \bigr| = \mathcal{O} (\Delta^3),  
\end{align*} 
Now, we have the bound (\ref{eq:error_bound}) and the proof is complete. 

\section{Proof of Lemma \ref{lemma:AE_scheme}} 
\label{appendix:density_scheme}
\noindent This section provides the small-time expansion for the proposed approximation scheme, given in Lemma \ref{lemma:AE_scheme}. 
For the derivation, we use Malliavin-Watanabe calculus and consider an asymptotic expansion of the Wiener functional $F^{\varepsilon} : \Omega \to \mathbb{R}^n, \, n \in \mathbb{N}$, that depends on a small parameter $\varepsilon \in (0,1)$. In our analysis, we typically assume $\varepsilon = \sqrt{\Delta}$ and then obtain an expansion with respect to $\sqrt{\Delta}$. Section \ref{appendix:density_scheme} is organised as follows. In Section \ref{intro_malliavin}, we introduce notations and quickly review some key results of Malliavin calculus (for more details see e.g.~\cite{nual:06, watanabe1987analysis}). The proof of Lemma \ref{lemma:AE_scheme} is given in Section \ref{pf_hypo_density}. 
Proofs of technical results required by Lemma \ref{lemma:AE_scheme} are collected in Section \ref{appendix:technical_Malliavin}. 
\subsection{Preliminaries from Malliavin Calculus} \label{intro_malliavin}
\noindent 
\textbf{-- Wiener space}

\noindent Let $(\Omega, \mathcal{F}, \mathbb{P})$ be the $m$-dimensional Wiener space, $m \in \mathbb{N}$, where 
$\Omega = \{ w : [0, \infty) \to \mathbb{R}^m  ; \; w(0) = 0, w \mathrm{\; 
is \; continuous}  \}$, $\mathcal{F}$ is the Borel field over $\Omega$ and $\mathbb{P}$ the Wiener measure. We denote by $B_t = (B_{1, t}, \ldots, B_{m,t}), \, t \geq 0$ the $m$-dimensional Brownian motion under $\mathbb{P}$. 
Let $H = L^2 ([0,\infty) ; \mathbb{R}^m)$ be a Hilbert space equipped with the inner product $\langle \cdot, \cdot \rangle_H$ given as:
$$ 
\langle h_1, h_2 \rangle_H =   \int_0^\infty h_{1,s} \cdot  h_{2, s} \, ds, \qquad   h_1, h_2 \in H.
$$
We denote by $L^p (\Omega ; \mathbb{R}^n)$, $p \in [1, \infty)$, $n \in \mathbb{N}$, the Banach space of Wiener functionals $F: \Omega \to \mathbb{R}^n$ such that $\| F \|_p := \{ \mathbb{E} [|F|^p] \}^{1/p} < \infty$, where $| \cdot |$ is the Euclidean norm. Similarly, $L^p (\Omega ; H)$ is the set of maps $f : \Omega \to H$ such that $\| f \|_H \in L^p (\Omega ; \mathbb{R}) :=
L^p (\Omega)$.
\\
 
\noindent
\textbf{-- Malliavin Derivative Operator} 

\noindent For $h \in H$, we define a Wiener integral $W(h)$ as $\textstyle{W(h):=  \int_0^{\infty} h_s \cdot d B_{s}}$. We introduce a class of Wiener functionals as:  
$$
\mathcal{S}:= \bigl\{F = \varphi \bigl(W(h_1), \ldots, W(h_l) \bigr); \; l \geq 1, \varphi \in C_p^{\infty} (\mathbb{R}^l), \; h_1, \ldots, h_l \in H \bigr\}. 
$$
%
%
The Malliavin derivative of $F \in \mathcal{S}$ is defined as: 
\begin{align*}
   D F  = \sum_{i=1}^l \partial_i \varphi  \bigl( W(h_1), \ldots, W(h_l) \bigr) h_i, \qquad  F \in \mathcal{S},   
\end{align*}
or written as a stochastic process $\{D_t F\}_{t \geq 0}$ with coordinates specified as:  
\begin{align*}
   D_{j, t} F  = \sum_{i=1}^l \partial_i \varphi  \bigl( W(h_1), \ldots, W(h_l) \bigr) h_{i,t}^j, \qquad t  \geq 0, \quad  1 \le j \le m. 
\end{align*} 
For example, one has for $0 < t < T$, 
\begin{align*}
D_{j_1, t} B_{j_2, T} = \mathbf{1}_{t \in [0,T], \; j_1 = j_2}, \quad  
D_{j_1, t} \int_0^T t d B_{j_2, t} 
= t \times \mathbf{1}_{t \in [0,T], \; j_1 = j_2}, \qquad 1 \le j_1 , j_2 \le d.
\end{align*}
The iterated derivative $D^k F$ is also defined as an $H^{\otimes k}$-valued random variable. 
It can be shown that for any $k \in \mathbb{N}$, operator $D^k$ is closable from $\mathcal{S}$ into $L^p (\Omega ;  H^{\otimes k})$ for any $\; p > 1$. We denote by $\mathbb{D}^{k,p}$, $k \in \mathbb{N}, \; p >1$, the completion of $\mathcal{S}$ with respect to the norm 
\begin{align*}
\| F \|_{k, p} 
:= \left\{ \mathbb{E} [|F|^p]  
 + \sum_{j=1}^k \mathbb{E} [ \|D^j F\|^p_{H^{\otimes j}}]   
  \right\}^{1/p}. 
\end{align*}
We define $\textstyle{\mathbb{D}^{\infty} = \bigcap_{k \in \mathbb{N}, \; p > 1}} \mathbb{D}^{k, p}$ and its dual $\textstyle{\mathbb{D}^{- \infty} = \bigcup_{k \in \mathbb{N}, \; p > 1}} \mathbb{D}^{-k, p} $, where $\mathbb{D}^{-k, p}$ is the dual space of $\mathbb{D}^{k, q}$ with $q = p / (p-1)$. The elements in $\mathbb{D}^{\infty}$ are referred to as \textit{smooth Wiener functionals in Malliavin sense}. Since the elements in $\mathbb{D}^{-\infty}$ are not Wiener functionals in the usual sense, they are called {\it generalised Wiener functionals}. For $\Phi \in \mathbb{D}^{-\infty}, \; G \in \mathbb{D}^{\infty}$, the coupling is defined as 
${}_{-\infty} \langle \Phi, G \rangle_{\infty} := \mathbb{E}[ \, \Phi\,  G \, ]$. 
In particular, ${}_{-\infty} \langle \Phi, 1 \rangle_{\infty} = \mathbb{E}[ \Phi]$ is called the {\it generalised expectation} of $\Phi$. 
\\

\noindent
\textbf{-- Divergence Operator and Duality Formula} 

\noindent The divergence operator $\delta : L^2 (\Omega ; H) \to L^2 (\Omega)$ is the adjoint of $D$ characterised by the duality formula: 
\begin{align*}  
   \mathbb{E} [F \delta (u)] 
   =\mathbb{E} \Bigl[ \int_0^\infty D_{s} F \cdot u_s \,   ds  \Bigr],
   \qquad  F \in \mathbb{D}^{1,2}.  
\end{align*}
When $\{u_t \}_{t \ge 0}$ is a square-integrable stochastic process, adapted to the filtration of Brownian motion, $\delta(u)$ coincides with the It\^o-integral. In this case, the duality formula is written as: 
\begin{align} \label{eq:base_duality}
  \mathbb{E} \Bigl[ F \int_0^\infty u_s \cdot dB_s \Bigr]  
 =   \mathbb{E} \Bigl[  \int_0^\infty D_{s} F \cdot u_s \,  ds  \Bigr]. 
\end{align}   
In particular, for 
$F = \varphi (G)$, $\varphi \in C_b^\infty (\mathbb{R}^N)$, where
$G = (G^1, \ldots, G^N)$ with $G^i \in \mathbb{D}^\infty$, formula (\ref{eq:base_duality}) is interpreted as  
\begin{align} \label{eq:duality_formula}
      \mathbb{E} \Bigl[ \varphi (G) \int_0^\infty u_s \cdot dB_s \Bigr]  
      =  \sum_{i = 1}^N \mathbb{E} \Bigl[ \partial_i \varphi (G)  \int_0^\infty D_{s} G^i \cdot u_s \,  ds  \Bigr], 
\end{align}
where we applied the chain rule of the Malliavin derivative.
\\

\noindent
\textbf{-- Non-degeneracy of Wiener Functional and Malliavin Integration by Parts} 

\noindent For a Wiener functional $F = (F^1, \ldots, F^N)$ with $F^i \in \mathbb{D}^{\infty}$, $1 \le i \le N$, the Malliavin covariance is defined as an $N \times N$ non-negative definite matrix with elements $\sigma^{F}_{ij}  = \langle DF^i, DF^j \rangle_H$, $1 \le i,j \le N$. When the matrix $\sigma^F$ is invertible a.s.~and $\|(\det \sigma^F)^{-1} \|_p < \infty$ for every $1 \leq p < \infty$, $F$ is said to be \textit{non-degenerate in Malliavin sense} (henceforth, we simply say non-degenerate). 
Now, we introduce the integration by parts on Wiener space that plays an important role in the proof of Lemma \ref{lemma:AE_scheme} and Theorem \ref{thm:density_approx_second}. Let $G \in \mathbb{D}^{\infty}$, $f \in C_b^{\infty} (\mathbb{R}^N)$ and $F \in (\mathbb{D}^{\infty})^N$ be non-degenerate, with the inverse of Malliavin covariance denoted by $\gamma^F \equiv (\sigma^F)^{-1}$. Then, for any multi-index $\alpha \in \{1, \ldots, N \}^k$, $k \in \mathbb{N}$, there exists $H_\alpha \in \mathbb{D}^\infty$ such that
\begin{align*}
  \mathbb{E} [\partial^{\alpha} f (F) G  ] 
  = \mathbb{E} [ f (F) H_{\alpha} (F, G) ], 
\end{align*}
where $H_\alpha $ is recursively defined as: 
\begin{align*} 
H_{(i)} (F,  G)   = \sum_{j=1}^N \delta 
 \bigl( \gamma^F_{ij} \,  G \, D F^j \bigr), \  1 \le i \le N, \ \ \ 
 H_\alpha (F,  G)   = 
 H_{(\alpha_k)} \bigl(F, H_{(\alpha_1, \ldots, \alpha_{k-1})} (F,  G) \bigr). 
\end{align*} 
In the proof of Lemma \ref{lemma:bd_E1_E2}, the following estimate on the stochastic weight $H_\alpha (F, G)$ (see e.g.~p.102 in \cite{nual:06}) is frequently used. For $\alpha \in \{1, \ldots, N\}^k$, $k \ge 1$ and $1 \le p < q < \infty$, there exist constants $c_1, c_2, \beta, \gamma, \lambda$ and integers $m, n$ such that 
\begin{align} \label{eq:estimate_weight}
 \left\| H_{\alpha} (F, G)  \right\|_p 
 & \le c_1 \,  \| \gamma^F D F \|_{k, \beta}^k  \, \| G \|_{k, q}  \nonumber \\[0.1cm]    
 & \le c_2 \, \| (\det \sigma^F)^{-1} \|_{\gamma}^m 
  \, \| D F \|_{k, \lambda}^n \,  \| G \|_{k, q}. 
\end{align}

\noindent
\textbf{-- Generalised Expectation}

\noindent Let $\mathscr{S}' (\mathbb{R}^N)$ be the space of real tempered distributions on $\mathbb{R}^N$. For $T \in \mathscr{S}' (\mathbb{R}^N)$ and a non-degenerate $F \in (\mathbb{D}^{\infty})^N$, the composition $T \circ F = T(F)$ is well-defined as an element of $\mathbb{D}^{-\infty} $, and the {generalised expectation} is defined as ${}_{-\infty} \langle T(F) , 1 \rangle_{\infty} = \mathbb{E} \bigl[ T (F) \bigr].$ 
In particular, for the Dirac measure $\textstyle{\delta_y}$ at $\textstyle{y \in \mathbb{R}^N}$ ($\textstyle{\delta_y \in \mathscr{S}' (\mathbb{R}^N)}$) and a non-degenerate Wiener functional $F \in (\mathbb{D}^{\infty})^N$, $\delta_y (F) \in \mathbb {D}^{-\infty}$ is well-defined and the probability law of $F$ has a smooth density $p^F(y)$ with respect to the Lebesgue measure $dy$ on $\mathbb{R}^N$. Then, it can be shown that the density is given as the following {generalised expectation}: 
\begin{align*}
p^F (y) =  {}_{-\infty} \langle \delta_y (F), 1 \rangle_{\infty} = \mathbb{E} [\, \delta_y (F) \, ].
\end{align*} 
Also, it holds that for any $f \in C_b^{\infty} (\mathbb{R}^N)$, $G \in \mathbb{D}^{\infty}$ and non-degenerate Wiener functional $F \in (\mathbb{D}^{\infty})^N$, 
\begin{align} \label{eq:expectaion_rep}
 \mathbb{E} [ f (F)  G] 
 & = \int_{\mathbb{R}^N} f (y) \, p^F(y) \, \mathbb{E}[G | F = y] dy \nonumber \\
 & = \int_{\mathbb{R}^N} f (y) {}_{-\infty} 
   \langle \delta_y (F), G \rangle_{\infty} dy 
  = \int_{\mathbb{R}^N} f (y) \,  \mathbb{E} [\, \delta_y (F) G \, ] \, dy, 
\end{align} 
that leads to 
\begin{align} \label{eq:conditional_density}
    \mathbb{E} [ \, \delta_y (F) G \,  ] 
    = p^F(y) \cdot \mathbb{E}[\, G | F = y \,],  
\end{align}
for almost all $y$ such that $p^F (y) > 0$. 
\\ 

\noindent 
\textbf{-- Asymptotic Expansion of Wiener Functional.} 

\noindent Assume that a Wiener functional $F$ depends on a small parameter $\varepsilon \in (0,1)$, \\ i.e.~$F = F^\varepsilon = (F^{\varepsilon,1},  \ldots, F^{\varepsilon, N})$. We consider the composition of a tempered distribution $T \in \mathscr{S}' (\mathbb{R}^N)$ and the Wiener functional 
$F^{\varepsilon}$ with $F^{\varepsilon,i} \in \mathbb{D}^{\infty}, \; 1 \le i \le N$,
and review an asymptotic expansion of the generalised expectation ${}_{-\infty} \langle T (F^{\varepsilon}), 1 \rangle_{\infty} 
= \mathbb{E} [\, T (F^{\varepsilon}) \,]$ 
using Malliavin-Watanabe calculus (\cite{watanabe1987analysis}). Now, suppose that $F^\varepsilon$ has the following asymptotic expansion. For any $J \in \mathbb{N}$,
\begin{align*}
    F^\varepsilon = \sum_{i = 0}^J \varepsilon^i f_i + \widetilde{f}_{J+1,\varepsilon},
\end{align*}
where $f_i \in \mathbb{D}^\infty$, $0 \le i \le J$ and 
\begin{align*}
 \limsup_{\varepsilon \downarrow 0} 
 \frac{\| \widetilde{f}_{J+1,\varepsilon}  \|_{k, p}}{\varepsilon^{J+1}}  < \infty,
\end{align*} 
for any $p \in (1, \infty)$ and $k \in \mathbb{N}$. 
Furthermore, we assume that $F^\varepsilon$ is uniformly non-degenerate with respect to $\varepsilon \in (0,1)$, i.e. 
\begin{align*}
\sup_{\varepsilon \in (0,1)} \| (\det \sigma^{F^{\varepsilon}})^{-1} \|_p < \infty, 
\end{align*}
for all $p \in (1, \infty)$, where we recall that $\sigma^{F^{\varepsilon}}$ is the Malliavin covariance of $F^{\varepsilon}$. According to 
\cite{ike:14}, we have an asymptotic expansion of $\mathbb{E} [\, T (F^{\varepsilon}) \,]$ as follows.  
For any $J \in \mathbb{N}$, it holds 
\begin{align} 
\label{eq:asymp_expn}
    \mathbb{E} [\, T (F^{\varepsilon}) \,] 
    = \sum_{i = 0}^J \varepsilon^i \cdot 
    \mathbb{E} [ \Phi_i ]  + R_{J + 1, \varepsilon}, 
\end{align}
where $\Phi_i \in \mathbb{D}^{-\infty}, \;  0 \le i \le J$ and 
$|R_{J + 1, \varepsilon}| \le C \varepsilon^{J + 1}$ for some positive constant $C >0$. In particular, $\Phi_i \in {\mathbb{D}}^{-\infty}, \; 0 \le i \le J$ are determined by the formal Taylor expansion, e.g.
\begin{align*}
   &  \Phi_0 = T (f_0),  \qquad   \Phi_1 = \sum_{i=1}^N  f_1^i (\partial_i T) (f_0);  \\
   &  \Phi_2 = \sum_{i=1}^N f_2^i (\partial_i T) (f_0) + \frac{1}{2!} \sum_{i_1, i_2=1}^N f_1^{i_1} f_1^{i_2}
       (\partial_{i_1} \partial_{i_2} T) (f_0).    
\end{align*}  
We note that in the case of $T = \delta_y \in \mathscr{S}' (\mathbb{R}^N)$, i.e.~the Dirac measure mass at $y \in \mathbb{R}^N$, (\ref{eq:asymp_expn}) yields the asymptotic expansion for the Lebesgue density of the law of $F^\varepsilon$, i.e.~$p^{F^{\varepsilon}} (y) = \mathbb{E}[ \, \delta_y (F^{\varepsilon}) \, ]$. Via the Malliavin integration by parts, the remainder term $R_{J + 1, \varepsilon}$ is bounded by some constant $C > 0$ independent of $\delta_y$, typically of $y \in \mathbb{R}^N$.
\subsection{Proof of Lemma \ref{lemma:AE_scheme}}  \label{pf_hypo_density}
\subsubsection{Step 1. Asymptotic Expansion}
We first recall the definition of random variables $\bar{X}_\Delta^{\varepsilon, x}$ in (\ref{eq:scheme_eps}) in the main text. We introduce: 
\begin{align} \label{eq:tilde_Y}
\widetilde{Y}^{\varepsilon} = 
m_{x, \theta, \varepsilon} (\bar{X}^{ \varepsilon, x}_{1})
 \equiv \widetilde{Y} + \sum_{l = 1}^3 
 \varepsilon^l \cdot  \widetilde{Y}^{(l)},
\end{align}
where we have set for $y = [y_R^\top, y_S^\top ]^\top \in \mathbb{R}^d$, 
\begin{align*}
m_{x, \theta,\varepsilon} (y):=
\left[ \begin{array}{c}
\frac{{y_{R} - x_{R} - \varepsilon^2 V_{R,0}(x, \beta)}}{\varepsilon}\\
\frac{y_{S} - x_{S} - \varepsilon^2 V_{S,0}(x, \gamma)  - \tfrac{\varepsilon^4}{2} \hat{V}_0 V_{S,0}(x , \theta)}{\varepsilon^3}
\end{array}
\right],
\end{align*}
and $\widetilde{Y}, \, \widetilde{Y}^{(l)}, \, 1 \leq l \leq 3$ are explicitly given as: 
\begin{align*}
 \widetilde{Y}
 & = 
 \begin{bmatrix}
 \widetilde{Y}_{R} \\[0.1cm] 
 \widetilde{Y}_{S}  
 \end{bmatrix}  
  = 
 \begin{bmatrix}
  \sum_{k=1}^{d_R} V_{R,k} (x,  \beta) B_{k, 1} \\[0,1cm]
  \sum_{k=1}^{d_R} \hat{V}_k V_{S,0} (x , \gamma) \zeta_{k 0, 1} 
 \end{bmatrix},  \ \ 
%
 \widetilde{Y}^{(1)}
 = 
 \begin{bmatrix}
  \sum_{k_1, k_2 = 1}^{d_R} 
     \hat{V}_{k_1} V_{R, k_2} (x , \theta) \zeta_{k_1 k_2, 1} \\[0.1cm] 
  \sum_{k_1, k_2 = 1}^{d_R}  \hat{V}_{k_1} \hat{V}_{k_2} V_{S,0} (x , \theta)  \eta_{k_1 k_2, 1} 
 \end{bmatrix}; \\[0.1cm] 
%
%
 \widetilde{Y}^{(2)} 
 & = 
 \begin{bmatrix}
\sum_{k=1}^{d_R}
\bigl\{ 
  \hat{V}_0 V_{R,k} (x,  \theta) \zeta_{0k, 1} 
  + \hat{V}_k V_{R,0} (x , \theta) \zeta_{k0, 1} 
\bigr\} \\[0.2cm] 
\sum_{k=1}^{d_R}
\bigl\{
\hat{V}_0 \hat{V}_k V_{S,0} (x,  \theta) \eta_{0k, 1} 
+ \hat{V}_k \hat{V}_0 V_{S,0} (x, \theta) 
\eta_{k0, 1} 
\bigr\}
 \end{bmatrix}, \  \  
\widetilde{Y}^{(3)} 
= 
\begin{bmatrix}
\tfrac{1}{2} \hat{V}_0 V_{R,0} (x, \theta) \\[0.1cm] 
\mathbf{0}_{d_S}
\end{bmatrix}. 
\end{align*}
Transformation of random variables with $\varepsilon = \sqrt{\Delta}$ yields: 
\begin{align} \label{X_density}
p_\Delta^{\bar{X}} (x, y; \theta) 
= \mathbb{P}_\theta \big[\bar{X}_1^{\sqrt{\Delta}, x} \in dy \big]/dy
= \frac{1}{\sqrt{\Delta^{d_R + 3 d_S}}} \, p^{\widetilde{Y}^{{\sqrt{\Delta}}}} (\xi ; \theta) |_{\xi = m_{x, \theta, {\sqrt{\Delta}}}(y)}, 
\end{align} 
and thus we consider the asymptotic expansion of
$
\textstyle{ p^{\widetilde{Y}^{\varepsilon}} (\xi ;\theta ) = \mathbb{E}_{\theta} [\delta_\xi ( \widetilde{Y}^{\varepsilon})] }
$
with respect to $\varepsilon$.  
The expansion is justified by the following result, with its proof contained in Section \ref{appendix:uniform_nondeg}.
\begin{lemma} \label{lemma:uniform_nondeg}
Under conditions (\ref{assump:param_space})--(\ref{assump:hypo2}), the family of variables $\{\widetilde{Y}^{\varepsilon} \}_{\varepsilon \in (0,1)}$ is uniformly non-degenerate, i.e.~for any $(x, \theta) \in \mathbb{R}^d \times \Theta$, it holds:
\begin{align*}
 \sup_{\varepsilon \in (0,1)} \|  (\det \sigma^{\widetilde{Y}^{\varepsilon}})^{-1} \|_p < \infty,  \ \ 
\end{align*} 
for all $p \in (1, \infty)$,  where $\sigma^{\widetilde{Y}^{\varepsilon}}$ is the Malliavin covariance of $\widetilde{Y}^{\varepsilon}$. 
\end{lemma}
\noindent Using Lemma \ref{lemma:uniform_nondeg} and (\ref{eq:asymp_expn}),  we get the following expansion of 
$ 
p^{\widetilde{Y}^{\varepsilon}} (\xi; \theta) 
$: for $J \ge 3$, 
\begin{align} \label{eq:expansion_GE}
p^{\widetilde{Y}^{\varepsilon}} (\xi; \theta) 
& = \mathbb{E}_{\theta} [ \delta_{\xi} (\widetilde{Y}^{ \varepsilon})] 
= \mathbb{E}_{\theta}[ \delta_{\xi} (\widetilde{Y})]
  +  \sum_{k=1}^{J-1} \varepsilon^k \cdot \mathbb{E}_{\theta} [g_k] 
  + \varepsilon^J \cdot  {R}_J  (x, \xi; \theta) \nonumber \\
& \equiv 
 p^{\widetilde{Y}} (\xi; \theta)   
 + \sum_{k = 1}^{J-1} \varepsilon^k \cdot \Xi_k (x, \xi; \theta) 
 + \varepsilon^J \cdot  R_J (x, \xi ; \theta),    
\end{align}
where $g_k, \, 1 \le k \le J-1$ are explicitly defined via the formal Taylor expansion of $\varepsilon \mapsto  \delta_{\xi} (\widetilde{Y}^{\varepsilon})$ and the term $R_J  (x, \xi ; \theta)$ is given as:
\begin{align} \label{eq:remainder_bd}
 {R}_J (x, \xi; \theta) 
 = \int_0^1 \frac{(1- \lambda)^{J-1}}{(J-1)!} 
\mathbb{E}_{\theta} \Bigl[ \frac{\partial^J}{\partial \eta^J}  \delta_{\xi} (\widetilde{Y}^{\eta}) \Bigr] |_{\eta = \varepsilon \lambda} \, d \lambda.
\end{align}
Note that $\xi \mapsto p^{\widetilde{Y}} (\xi; \theta)$ represents the density of Gaussian random variable $\widetilde{Y}$ whose covariance is non-degenerate under Conditions (\ref{assump:hypo1})--(\ref{assump:hypo2}). In particular, we have 
\begin{align}
   & \Xi_1 (x, \xi ; \theta)  
   =    
  \sum_{i=1}^{d_R} \sum_{j_1, j_2 = 1}^{d_R}  
  \hat{V}_{j_1} V_{j_2}^i (x,  \theta)  
  \mathbb{E}_{\theta} 
  \Bigl[ 
  \partial_i \delta_{\xi} (\widetilde{Y}) \zeta_{j_1 j_2,1} 
  \Bigr]  +  \sum_{i= d_R + 1}^{d} \sum_{j_1, j_2 = 1}^{d_R}  \hat{V}_{j_1} \hat{V}_{j_2} V_0^i (x,  \theta)   \mathbb{E}_{\theta} \Bigl[ \partial_i \delta_{\xi} (\widetilde{Y}) \eta_{j_1 j_2, 1} \Bigr];    \nonumber \\[0.2cm]
  & \Xi_2 (x, \xi ; \theta)
  =  \sum_{i=1}^{d_R} \sum_{j=1}^{d_R}
  \mathbb{E}_{\theta} \Bigl[ \partial_i \delta_{\xi} (\widetilde{Y}) 
  \bigl\{  \hat{V}_0 V_j^i (x, \theta)    \zeta_{0j, 1} +  \hat{V}_j V_0^i (x, \theta) \zeta_{j0, 1}   
  \bigr\}  
  \Bigr]  \nonumber \\ 
  & \quad + \sum_{i=d_R + 1}^{d} \sum_{j=1}^{d_R}
  \mathbb{E}_{\theta} \Bigl[ \partial_i \delta_{\xi} (\widetilde{Y}) 
  \bigl\{ \hat{V}_0 \hat{V}_j V_0^i (x,  \theta)  \eta_{0j, 1} 
   +  \hat{V}_j \hat{V}_0 V_0^i (x,  \theta)  \eta_{j0, 1} 
  \bigr\}
  \Bigr] 
  \nonumber \\ 
  & \quad + \tfrac{1}{2} \sum_{i_1, i_2 =1}^{d_R} 
  \sum_{j_1, j_2, j_3, j_4 = 1}^{d_R}
  \hat{V}_{j_1} V_{j_2}^{i_1}(x, \theta) 
  \hat{V}_{j_3} V_{j_4}^{i_2}(x, \theta) 
  \mathbb{E}_{\theta} 
  \Bigl[ \partial_{i_1} \partial_{i_2} \delta_{\xi} (\widetilde{Y}) 
  \zeta_{j_1 j_2, 1}  \zeta_{j_3 j_4,1} 
  \Bigr]  \nonumber \\
  & \quad + \sum_{\substack{1 \leq i_1 \leq d_R \\ d_R + 1 \leq  i_2 \leq d}} \sum_{j_1, j_2, j_3, j_4 = 1}^{d_R} 
  \hat{V}_{j_1} V_{j_2}^{i_1}(x ,  \theta) 
  \hat{V}_{j_3} \hat{V}_{j_4} V_0^{i_2}(x , \theta)
  \mathbb{E}_{\theta} 
  \Bigl[\partial_{i_1} \partial_{i_2} 
  \delta_{\xi} (\widetilde{Y}) \zeta_{j_1 j_2,1 } \eta_{j_3 j_4,1}
  \Bigr]  \nonumber \\
  & \quad + \tfrac{1}{2} \sum_{i_1, i_2 = d_R + 1}^d 
  \sum_{j_1, j_2, j_3, j_4 = 1}^{d_R} 
  \hat{V}_{j_1} \hat{V}_{j_2} V_0^{i_1}(x ,  \theta) 
  \hat{V}_{j_3} \hat{V}_{j_4} V_0^{i_2}(x ,  \theta)
   \mathbb{E}_{\theta} 
   \Bigl[ \partial_{i_1} \partial_{i_2} \delta_{\xi} (\widetilde{Y})  
   \eta_{j_1 j_2, 1} \eta_{j_3 j_4, 1} 
   \Bigr]. 
\end{align}
%
\subsubsection{Step 2. Expression for Expanded Terms}
\label{sec:step2_density}
We next consider the expression for the expanded terms $\Xi_l (x, \xi ; \theta), \, 1 \le l \le J-1$, and observe that they have the general form: 
\begin{align} \label{eq:Xi_form}
\Xi_l (x, \xi ; \theta) = \sum_{1 \le k \le \nu (l)} \sum_{\alpha \in 
\{1, \ldots,  d \}^k} \nu_\alpha (x,  \theta) \mathcal{H}^{\widetilde{Y}}_\alpha(\xi; \theta) p^{\widetilde{Y}} (\xi; \theta), 
\end{align}
where $\nu(l) \ge 1$ is some integer, $\nu_{\alpha} (\cdot, \theta) \in C_b^{\infty} (\mathbb{R}^d; \mathbb{R}), \,  \theta \in \Theta$ is given by the products of $V_j (\cdot ,  \theta), \; 0 \le j \le d_R$ and theirs partial derivatives, and $\mathcal{H}^{\widetilde{Y}}_\alpha(\xi; \theta)$ is Hermite polynomial defined in (\ref{eq:hermite_Y}) given in the main text. We here provide the representation for $\Xi_l (x, \xi; \theta), \, l = 1, 2$, and then discuss the case $l \ge 3$. To obtain an analytic expression for $\Xi_{l} (x, \xi; \theta), \, l = 1,  2$, we make use of the following result whose proof is found in Section \ref{pf_representation_weight}:  
\begin{lemma} \label{representation_weight}
Let $1 \leq j_1, j_2, j_3, j_4  \leq d_R$ and $1 \leq i_1, i_2 \leq d$. Under conditions (\ref{assump:hypo1})--(\ref{assump:hypo2}), it holds that
\begin{align}
%
%
& \mathbb{E}_{\theta} 
  [ \partial_{i_1} \delta_{\xi} (\widetilde{Y}) \zeta_{0j_1, 1} ]  \nonumber \\
& \quad = \sum_{i_2 = 1}^d 
 \Bigl\{ \tfrac{1}{2}  V_{j_1}^{i_2} (x, \theta) 
 \cdot \mathbf{1}_{1 \le i_2 \le d_R} 
 + \tfrac{1}{6} 
  \hat{V}_{j_1} V_0^{i_2} (x ,  \theta) 
  \cdot \mathbf{1}_{d_R + 1 \le i_2 \le d} \Bigr\} \cdot
\mathcal{H}_{(i_1, i_2)}^{\widetilde{Y}} (\xi; \theta)  p^{\widetilde{Y}} (\xi ; \theta);  \label{weight_2}  \\
& \mathbb{E}_{\theta} [ \partial_{i_1} \delta_{\xi} ( \widetilde{Y} ) \zeta_{j_10,1}] \nonumber \\ 
&\quad  = \sum_{i_2 = 1}^{d} \Bigl\{ \tfrac{1}{2} V_{j_1}^{i_2} (x, \theta) 
\cdot \mathbf{1}_{1 \le i_2 \le d_R}
+ \tfrac{1}{3}  \hat{V}_{j_1} V_0^{i_2} (x, \theta) \cdot \mathbf{1}_{d_R + 1 \le i_2 \le d}   \Bigr\} \cdot 
\mathcal{H}_{(i_1, i_2)}^{\widetilde{Y}} (\xi; \theta)  p^{\widetilde{Y}} (\xi ; \theta);   \label{weight_3} \\ 
& \mathbb{E}_{\theta} [ \partial_{i_1}  \delta_{\xi} (\widetilde{Y}) \eta_{0 j_1,1}]  \nonumber \\ 
& \quad =  \sum_{i_2 = 1}^{d} \Bigl\{ 
\tfrac{1}{6} V_{j_1}^{i_2} (x,  \theta)
\cdot \mathbf{1}_{1 \le i_2 \le d_R} 
+ \tfrac{1}{6} \hat{V}_{j_1} V_0^{i_2} (x,  \theta) \cdot \mathbf{1}_{d_R + 1 \le i_2 \le d}   \Bigr\} \cdot \mathcal{H}_{(i_1, i_2)}^{\widetilde{Y}} (\xi; \theta) p^{\widetilde{Y}} (\xi ; \theta);   \label{weight_7}  \\
& \mathbb{E}_{\theta} 
  [ \partial_{i_1} \delta_{\xi} (\widetilde{Y}) \eta_{j_1 0,1}]  \nonumber \\ 
& \quad  = \sum_{i_2 = 1}^{d} \Bigl\{ 
   \tfrac{1}{6}   V_{j_1}^{i_2} (x, \theta)
   \cdot \mathbf{1}_{1 \le i_2 \le d_R}
   + \tfrac{1}{8} \hat{V}_{j_1} V_0^{i_2} (x, \theta)
   \cdot \mathbf{1}_{d_R + 1 \le i_2 \le d} \Bigr\} \cdot 
   \mathcal{H}_{(i_1, i_2)}^{\widetilde{Y}} (\xi; \theta) 
   p^{\widetilde{Y}} (\xi ; \theta);  \label{weight_8} \\
&  \mathbb{E}_{\theta} \bigl[ \partial_{i_1} \delta_{\xi} (\widetilde{Y}) \zeta_{j_1 j_2,1} \bigr]
 = \sum_{i_2, i_3 = 1}^{d}  \biggl\{  \tfrac{1}{2} 
 V_{j_1}^{i_2} (x, \theta) V_{j_2}^{i_3} (x, \theta) \cdot \mathbf{1}_{1 \le i_2, i_3 \le d_R} \nonumber \\   
& \quad + \tfrac{1}{3}
\hat{V}_{j_1} V_0^{i_2} (x , \theta) V_{j_2}^{i_3}(x ,  \theta) 
\cdot \mathbf{1}_{{d_R + 1 \le i_2 \le d, 1 \le i_3 \le d_R}}  
+  \tfrac{1}{6} 
  V_{j_1}^{i_2}(x , \theta) 
  \hat{V}_{j_2} V_0^{i_3} (x , \theta) 
  \cdot \mathbf{1}_{{1 \le i_2 \le d_R, d_R + 1 \le i_3 \le d}} \nonumber \\
 & \quad +  \tfrac{1}{8}  
 \hat{V}_{j_1} V_0^{i_2} (x ,  \theta) 
 \hat{V}_{j_2} V_0^{i_3} (x , \theta)
 \cdot \mathbf{1}_{d_R + 1 \le i_2, i_3 \le d} 
 \biggr\} \cdot 
 \mathcal{H}_{(i_1, i_2, i_3)}^{\widetilde{Y}}
 (\xi ; \theta) p^{\widetilde{Y}} (\xi ; \theta);  \label{weight_4}  \\[0.1cm] 
&  \mathbb{E}_{\theta} \bigl[ \partial_{i_1} \partial_{i_2} \delta_{\xi} (\widetilde{Y}) \zeta_{j_1 j_2,1}  \zeta_{j_3 j_4, 1} \bigr]
  = \tfrac{1}{2} \mathcal{H}_{(i_1, i_2)}^{\widetilde{Y}} (\xi ; \theta)
   p^{\widetilde{Y}} (\xi ; \theta)
  \cdot \mathbf{1}_{j_1 = j_3, j_2 = j_4} + r_1 (x, \xi; \theta); \label{weight_5}  \\[0.1cm]
&  \mathbb{E}_{\theta} [\partial_{i_1} \partial_{i_2} \delta_{\xi} (\widetilde{Y})  \eta_{j_1 j_2,1} 
  \eta_{j_3 j_4,1} ]
  =  \tfrac{1}{12} \mathcal{H}_{(i_1, i_2)}^{\widetilde{Y}} (\xi ; \theta)  
   p^{\widetilde{Y}} (\xi ; \theta)
   \cdot \mathbf{1}_{j_1 = j_3, j_2 = j_4}  + r_2 (x, \xi; \theta),  \label{weight_6}
\end{align}
where the terms 
$r_i(x, \xi; \theta), \; i = 1, 2$ are given in the form of
\begin{align} \label{eq:r_remainder}
 \sum_{l=4, 6} \sum_{\substack{\alpha \in \{1, \ldots, d \}^l}} 
   v_{\alpha} (x,  \theta) \mathcal{H}^{\widetilde{Y}}_{\alpha} (\xi ; \theta)  p^{\widetilde{Y}} (\xi ; \theta)
\end{align}
for some $v_{\alpha} : \mathbb{R}^d \times \Theta \to \mathbb{R}$ being of the same structure as given in (\ref{eq:Xi_form}). 
\end{lemma}
Thus, applying Lemma \ref{representation_weight} to (\ref{eq:expansion_GE}), we obtain expression for $\Xi_l (x, \xi ; \theta), \, l = 1, 2$ as: 
\begin{align} \label{eq:Xi_1_cf}
& \Xi_1 (x, \xi; \theta) \nonumber \\ 
& = \sum_{j_1, j_2 = 1}^{d_R} \sum_{i_1 =1}^{d}  
\biggl\{
\hat{V}_{j_1} V_{j_2}^{i_1} (x, \theta) 
\cdot \mathbf{1}_{1 \le i_1 \le d_R}
+ \tfrac{1}{3}  \hat{V}_{j_1} \hat{V}_{j_2} V_0^{i_1} (x, \theta) \cdot \mathbf{1}_{d_R + 1 \le i_1 \le d}  \biggr\} 
\cdot \widetilde{\Xi}^{i_1}_{j_1 j_2} (x, \xi; \theta), 
\end{align}
where 
\begin{align*}
\widetilde{\Xi}^{i_1}_{j_1 j_2} (x, \xi; \theta) 
& =  \tfrac{1}{2}  
\sum_{i_2, i_3 =1}^{d_R}  
V_{j_1}^{i_2}(x , \theta) 
V_{j_2}^{i_3}(x , \theta)  \, 
\mathcal{H}_{(i_1,i_2,i_3)}^{\widetilde{Y}} (\xi; \theta) p^{\widetilde{Y}} (\xi ; \theta) \\ 
&  \quad  +  \tfrac{1}{3}
\sum_{\substack{d_R + 1 \leq  i_2 \leq d \\ 1 \leq i_3 \leq d_R}}  
\hat{V}_{j_1} V_0^{i_2} (x , \theta) V_{j_2}^{i_3}(x ,  \theta)  
\, \mathcal{H}_{(i_1,i_2,i_3)}^{\widetilde{Y}} (\xi; \theta) p^{\widetilde{Y}} (\xi ; \theta)   
\\
& \quad +  \tfrac{1}{6} \sum_{\substack{1 \leq  i_2 \leq d_R \\ d_R + 1 \leq i_3 \leq d}}   
V_{j_1}^{i_2}(x , \theta) 
\hat{V}_{j_2} V_0^{i_3} (x , \theta) 
\mathcal{H}_{(i_1,i_2,i_3)}^{\widetilde{Y}} (\xi; \theta)  p^{\widetilde{Y}} (\xi ; \theta) \\
 & \quad+  \tfrac{1}{8}  \sum_{i_2, i_3 = d_R + 1}^d 
 \hat{V}_{j_1} V_0^{i_2} (x ,  \theta) 
 \hat{V}_{j_2} V_0^{i_3} (x , \theta)
 \mathcal{H}_{(i_1,i_2,i_3)}^{\widetilde{Y}} (\xi; \theta) p^{\widetilde{Y}} (\xi ; \theta),  
\end{align*}
and  
\begin{align} 
\Xi_{2} (x, \xi; \theta) 
& = \tfrac{1}{2} \sum_{i_1, i_2 =1}^{d_R} \sum_{k = 1}^{d_R} \Bigl( \hat{V}_{k} V_{0}^{i_1} (x,  \theta)
 + \hat{V}_{0} V_{k}^{i_1} (x, \theta) \Bigr) V_k^{i_2}( x ,  \theta )   
 \mathcal{H}_{(i_1,i_2)}^{\widetilde{Y}} (\xi; \theta)  p^{\widetilde{Y}} (\xi ; \theta) 
\nonumber  \\
& \quad  +  \tfrac{1}{4} \sum_{i_1, i_2 =1}^{d_R} \sum_{k_1, k_2 = 1}^{{d_R}}   
\hat{V}_{k_1} V_{k_2}^{i_1}(x ,  \theta) \hat{V}_{k_1} V_{k_2}^{i_2}(x ,  \theta)  \mathcal{H}_{(i_1,i_2)}^{\widetilde{Y}} (\xi; \theta)  p^{\widetilde{Y}} (\xi ; \theta)   \nonumber  \\ 
& \quad +  \sum_{\substack{1 \leq i_1 \leq d_R \\ d_R + 1 \leq i_2 \leq d}}  \sum_{k = 1}^{d_R} 
  \Bigl( \tfrac{1}{3} \hat{V}_{k} V_{0}^{i_1} (x ,  \theta)
 + \tfrac{1}{6} \hat{V}_{0} V_{k}^{i_1} (x ,  \theta) \Bigr) 
 \hat{V}_k V_0^{i_2}(x , \theta)  
\mathcal{H}_{(i_1,i_2)}^{\widetilde{Y}} (\xi; \theta) p^{\widetilde{Y}} (\xi ; \theta) \nonumber 
\\
& \quad  + \tfrac{1}{6} \sum_{\substack{1 \leq i_1 \leq d_R \\ d_R + 1 \leq i_2 \leq d}} 
\sum_{k=1}^{d_R} 
V_k^{i_1} (x ,  \theta)
\Bigl(\hat{V}_0 \hat{V}_k V_0^{i_2} (x , \theta) + \hat{V}_k \hat{V}_0 V_0^{i_2} (x ,  \theta) \Bigr)
\mathcal{H}_{(i_1,i_2)}^{\widetilde{Y}} (\xi; \theta)   
p^{\widetilde{Y}} (\xi ; \theta) 
\nonumber \\
& \quad  +  \tfrac{1}{6} \sum_{\substack{1 \leq i_1 \leq d_R \\ d_R + 1 \leq  i_2 \leq d}} \sum_{k_1, k_2 = 1}^{{d_R}} \hat{V}_{k_1} V_{k_2}^{i_1}(x,  \theta) \hat{V}_{k_1} \hat{V}_{k_2} V_0^{i_2}(x ,  \theta)
\mathcal{H}_{(i_1,i_2)}^{\widetilde{Y}} (\xi; \theta) p^{\widetilde{Y}} (\xi ; \theta) 
\nonumber \\
& \quad + \sum_{i_1, i_2 = d_R + 1}^d  
 \sum_{k=1}^{d_R} 
 \hat{V}_k V_0^{i_1}  (x, \theta)
  \Bigl( \tfrac{1}{6} \hat{V}_0 \hat{V}_k V_0^{i_2} (x,  \theta) 
 +   \tfrac{1}{8}\hat{V}_k \hat{V}_0 V_0^{i_2} (x,  \theta) \Bigr) 
 \mathcal{H}_{(i_1,i_2)}^{\widetilde{Y}} (\xi; \theta) p^{\widetilde{Y}} (\xi ; \theta)  
\nonumber  \\
& \quad  + \tfrac{1}{24} \sum_{i_1, i_2 = d_R + 1}^d \sum_{k_1, k_2 = 1}^{{d_R}} 
  \hat{V}_{k_1} \hat{V}_{k_2} V_0^{i_1}(x,  \theta) \hat{V}_{k_1} \hat{V}_{k_2} V_0^{i_2}(x ,  \theta) 
  \mathcal{H}_{(i_1,i_2)}^{\widetilde{Y}} (\xi; \theta) p^{\widetilde{Y}} (\xi ; \theta)  \nonumber \\ 
& \quad + \sum_{l=4,6} \sum_{\substack{\alpha \in \{1, \ldots, d \}^l}}  w_{\alpha} (x, \theta) \times 
\mathcal{H}^{\widetilde{Y}}_{\alpha} (\xi ; \theta)  p^{\widetilde{Y}} (\xi ; \theta),   
\label{eq:Xi_2_cf}
\end{align} 
with $w_{\alpha}$ having the same structure as $\nu_\alpha$ in (\ref{eq:Xi_form}). 
\\ 

We consider the terms $\Xi_l (x, \xi ; \theta)$ for $l \ge 3$, and observe that they have the general form as (\ref{eq:Xi_form}). 
First, from $\widetilde{Y}^\varepsilon$ given in (\ref{eq:tilde_Y}), each of the terms in ${\Xi}_l (x, \xi; \theta), \, l \ge 3$ is given as: 
\begin{align*}
 v_{\beta}^\alpha (x, \theta) \cdot  
\mathbb{E}_{\theta} 
 \Bigl[ \partial^{\alpha} \delta_{\xi} (\widetilde{Y})  
 I_\beta (1) \Bigr]
\end{align*} 
for some multi-indices $\alpha \in  \{1, \ldots, d\}^k, \; k \ge 1$ and $\beta \in \{0,1, \ldots, d_R \}^{\nu}, \; \nu \geq 1$, and $v_{\beta}^\alpha : \mathbb{R}^d \times \Theta \to \mathbb{R}$ being the same structure as $\nu_\alpha $ in (\ref{eq:Xi_form}) under condition (\ref{assump:coeff}). 
We note that the multiple stochastic integral $I_{\beta} (1)$ comes from application of It\^o's formula to the products of combinations of the following random variables: 
$$
B_{k_1, 1}, \quad   \widetilde{B}_{k_2, 2},  \quad   
I_{(k_3 , 0)} (1), \quad  I_{(0, k_4)} (1), \qquad  0 \le k_1, k_2, k_3, k_4 \le d_R.
$$
Then, iterative use of duality formula (\ref{eq:duality_formula}) (in particular, Lemma 2.5 in \cite{iguchi2021second}) yields 
\begin{align} \label{eq:ibp_general}
v_{\beta}^\alpha (x, \theta)  \cdot 
\mathbb{E}_{\theta} \Bigl[ \partial^{\alpha} \delta_{\xi} (\widetilde{Y})  I_\beta (1) \Bigr]
= 
\sum_{\gamma \in \{1, \ldots, d\}^{n(\beta)}} 
v_{\beta}^\alpha (x, \theta)  v_{\gamma} (x,  \theta) \cdot 
\mathbb{E}_{\theta} \Bigl[ \partial^{\alpha} \partial^{\gamma}  \delta_{\xi} (\widetilde{Y}) \Bigr], 
\end{align}
for some $v_{\gamma} : \mathbb{R}^d \times \Theta \to \mathbb{R}$ again characterised as $\nu_\alpha$ in (\ref{eq:Xi_form}), where $n(\beta)$ is the number of non-zero elements in $\beta$. Thus, we conclude that ${\Xi}_l (x, \xi; \theta), \, l \ge 3$, admits the general form (\ref{eq:Xi_form}) from the following formula (e.g., see Theorem 9.2 in \cite{ike:14}): 
%
\begin{align} 
\label{eq:deriv_density}
 \mathbb{E}_{\theta} [ \partial^{\beta} \delta_{\xi} (\widetilde{Y})]
 = (-1)^{|\beta|}  
 \partial^\xi_\beta p^{\widetilde{Y}} (\xi ; \theta) 
 = \mathcal{H}_{\beta}^{\widetilde{Y}} (\xi ; \theta)  
 \cdot p^{\widetilde{Y}} (\xi ; \theta), \quad \beta \in \{1, \ldots, d \}^l, \; l \geq 1.  
\end{align}
\subsubsection{Step 3. Specification of the Residual Term}
We consider the term $R_J  (x, \xi ; \theta)$ specified as (\ref{eq:remainder_bd}). Each of the terms in $\textstyle{\mathbb{E}_{\theta} [ \frac{\partial^J}{\partial \eta^J}  \delta_{\xi} (\widetilde{Y}^{\eta}) ] }$ is given as:   
\begin{align*}
   \mathbb{E}_{\theta}
   \Bigl[ \partial^{\alpha} \delta_{\xi} (\widetilde{Y}^{\eta}) 
    G_{\alpha}^{\eta} (x; \theta)
   \Bigr],   
\end{align*} 
where $\alpha \in \{ 1, \ldots , d \}^k, \;  k = \nu_J, \ldots, J$ for some integer $\nu_J < J$, and $G_{\alpha}^{\eta} (x; \theta)  \in \mathbb{D}^{\infty}$ is given by products of $\widetilde{Y}^{(\beta_1), \alpha_1}, \ldots, \widetilde{Y}^{(\beta_k), \alpha_k}$ with $\beta \in \{1, 2, 3\}^k$. In particular, under  conditions (\ref{assump:param_space})--(\ref{assump:coeff}), there exists a constant $c >0$ independent of $x \in \mathbb{R}^d$ and $\theta \in \Theta$ such that $\| G_{\alpha}^{\eta} (x; \theta) \|_{k,p} \leq c$ for any $k \in \mathbb{N}$ and $p \geq 1$. Since $\{ \widetilde{Y}^{ \eta} \}_{\eta \in (0,1)}$ is non-degenerate uniformly in $\eta \in (0,1)$ from Lemma \ref{lemma:uniform_nondeg}, one can apply the Malliavin integration by parts and obtain the following estimate due to the upper bound (\ref{eq:estimate_weight}): 
\begin{gather}
\Bigl| \mathbb{E}_{\theta} \Bigl[ \partial^{\alpha} \delta_{\xi} (\widetilde{Y}^{\eta}) G_{\alpha}^{\eta} (x; \theta)  \Bigr] \Bigl|
 =  \Bigl| \mathbb{E}_{\theta} \Bigl[ \mathbf{1}_{\widetilde{Y}^{\eta} \geq \xi} 
H_{(\alpha_1, \ldots, \alpha_k, 1, \ldots, d)} (\widetilde{Y}^{\eta}, G_{\alpha}^{\eta} (x ; \theta) ) \Bigr] \Bigl| \nonumber \\[0.1cm]
\leq  \bigl\| H_{(\alpha_1, \ldots, \alpha_k, 1, \ldots, d)} (\widetilde{Y}^{\eta}, G_{\alpha}^{\eta} (x; \theta) )  \bigr\|_p 
\leq c \, \| ( \det \sigma^{\widetilde{Y}^{\eta}} )^{-1} \|_q^m 
\leq C,  \nonumber 
\end{gather} 
for some $p, q \geq 2$ and positive constants $C, c, m$ independent of $\xi, x \in \mathbb{R}^d, \theta \in \Theta$. 
Thus, there exists a constant $C > 0$ such that %
\begin{align*}
    \sup_{x, \xi \in \mathbb{R}^d, \theta \in \Theta}
    \bigl| R_J (x, \xi ; \theta)  \bigr| 
    \leq  C.
\end{align*}
\subsubsection{Final Step. Complete Form of Density Expansion}
Finally, substituting $\sqrt{\Delta}$ for $\varepsilon$ in the formula (\ref{X_density}) with (\ref{eq:expansion_GE}), (\ref{eq:Xi_form}), (\ref{eq:Xi_1_cf}) and (\ref{eq:Xi_2_cf}), we obtain the density expansion (\ref{eq:density_expansion}) in the main text. We note that 
$\Xi_l (x, \xi ; \theta) |_{\xi = m_{x, \theta, \sqrt{\Delta}} (y)}$ corresponds to 
$ \Psi_l (\Delta, x, y ; \theta) \times p_\Delta^{\widetilde{X}^{\mrm{LG}}} (x, y ; \theta)$, $ 1 \le l \le J-1$. The proof of Lemma \ref{lemma:AE_scheme} is now complete.
\subsection{Proof of Technical Results for Lemma \ref{lemma:AE_scheme}}  \label{appendix:technical_Malliavin} 
\subsubsection{Proof of Lemma \ref{lemma:uniform_nondeg}}  \label{appendix:uniform_nondeg}
\noindent We will show that, for any $t > 0$, $(x, \theta) \in \mathbb{R}^d \times \Theta$, it holds that 
\begin{align*}
\textstyle{\sup_{\varepsilon \in (0,1)} \| (\det \sigma^{\widetilde{Y}^{\varepsilon}})^{-1} \|_p  < \infty }
\end{align*}
for all $p \in (1, \infty)$. The Malliavin derivative of $\widetilde{Y}^{ \varepsilon}$ is given via
$D_{k,s } \widetilde{Y}^{\varepsilon} 
\equiv  F_{k, s}^{(1)} +  F_{k, s}^{(2),  \varepsilon}$, $1 \le k \leq d_R$, $s > 0$, where  
\begin{align*}
F_{k,s}^{(1)}
& = 
\begin{bmatrix}
  D_{k,s} \widetilde{Y}_{R} \\
  D_{k,s} \widetilde{Y}_{S}
\end{bmatrix}  
= 
\begin{bmatrix}
 V_{R, k} (x, \theta)  \times \mathbf{1}_{s \in [0,1]} \\[0.1cm]  
 \hat{V}_k  V_{S, 0} (x,  \theta) (1 - s) \times \mathbf{1}_{s \in [0,1]} 
\end{bmatrix},  \ \ 
F_{k, s}^{(2),  \varepsilon}  
 = \sum_{l = 1}^3 \varepsilon^l \cdot D_{k, s} \widetilde{Y}^{(l)}. 
%
%
\end{align*}
Under conditions (\ref{assump:hypo1})--(\ref{assump:hypo2}),  there exists a  constant $c>0$ such that 
\begin{align} \label{eq:positive_definite}
  \sup_{ | \xi |  = 1, \xi \in \mathbb{R}^d}  \bigl\langle Q(s) \xi, \,  \xi \bigr\rangle \geq c,  \qquad    0 < s < 1,  
\end{align}
where $Q(s)$ is a $d \times d$ matrix with elements given as 
$$
Q_{ij}(s) = \sum_{k=1}^{d_R} F_{k, s}^{(1), i}  F_{k, s}^{(1), j},  
\qquad  1 \le i, j \le d.
$$ 
Notice that the matrix $Q(s)$ is independent of $\varepsilon \in (0,1)$ and the positive definiteness holds uniformly in $\varepsilon$. Furthermore, under conditions (\ref{assump:param_space})--(\ref{assump:coeff}), there exists a constant $C  > 0$ independent of $(x, \theta) \in \mathbb{R}^d \times \Theta$ such that for any $p \in (1, \infty)$, 
\begin{align} \label{eq:moment_bound}
\sup_{s \leq 1, \;  \varepsilon \in (0,1) } \mathbb{E}_{\theta} [|F_{k,s}^{(2),  \varepsilon}  |^p] 
\leq C. 
\end{align} 
From Lemma 5 (or equivalently Proposition 9) in \cite{bally:03}, (\ref{eq:positive_definite}) and (\ref{eq:moment_bound}) lead to  the following result. For every 
$(x, \theta) \in \mathbb{R}^d \times \Theta$, we have  $\textstyle{\sup_{\varepsilon \in (0,1)} \| (\det \sigma^{\widetilde{Y}^{\varepsilon}})^{-1} \|_p  < \infty }$ for any $p \in (1, \infty)$. The proof is now complete.
%
%
%
%
%
\subsubsection{Proof of Lemma \ref{representation_weight}}  \label{pf_representation_weight}
\noindent 
{\bf Derivation of (\ref{weight_2}):} The duality formula (\ref{eq:duality_formula}) and (\ref{eq:deriv_density}) yield 
\begin{align}
  & \mathbb{E}_{\theta} 
   \Bigl[ \partial_{i_1} \delta_{\xi} (\widetilde{Y})
   \int_0^1 s dB_{j, s} 
   \Bigr] 
   =  \sum_{i_2 = 1}^d  
   \mathbb{E}_{\theta}
   \Bigl[ \partial_{i_2} \partial_{i_1} \delta_{\xi} (\widetilde{Y}) \int_0^1 D_{j,s}\widetilde{Y}^{ i_2} s ds 
   \Bigr] \nonumber \\ 
   &  =  \tfrac{1}{2}  \sum_{i_2 = 1}^{d_R} V_j^{i_2} (x, \theta)  \mathbb{E}_{\theta} 
   \bigl[ 
   \partial_{i_2}  \partial_{i_1}  \delta_{\xi} (\widetilde{Y}) 
   \bigr]
   + \tfrac{1}{6}  \sum_{i_2 = d_R + 1}^d  \hat{V}_j V_0^{i_2} (x, \theta) 
   \mathbb{E}_{\theta} 
   \bigl[
   \partial_{i_2}  \partial_{i_1}  \delta_{\xi} (\widetilde{Y}) 
   \bigr] \nonumber \\
   %
    & =  \tfrac{1}{2}  \sum_{i_2 = 1}^{d_R}  V_j^{i_2} (x, \theta) 
    \mathcal{H}_{(i_1, i_2)}^{\widetilde{Y}} (\xi ; \theta )
    p^{\widetilde{Y}} (\xi; \theta)  
    + \tfrac{1}{6}  \sum_{i_2 = d_R + 1}^d   \hat{V}_j V_0^{i_2} (x, \theta) 
    \mathcal{H}_{(i_1, i_2)}^{\widetilde{Y}} (\xi ; \theta) 
    p^{\widetilde{Y}}  (\xi; \theta),  \nonumber  
\end{align}
where we used
\begin{align} \label{Mall_deriv_Y}
  D_{j_2, s} \widetilde{Y} = 
 \begin{bmatrix} 
  V_{R, j_2} (x,  \theta) \times \mathbf{1}_{s \in [0,1]}  \\
 \hat{V}_{j_2} V_{S,0} (x, \theta) (1-s) \times \mathbf{1}_{s \in [0,1]} 
 \end{bmatrix}. 
\end{align} 
\\

\noindent 
{\bf Derivation of (\ref{weight_3}):} The duality formula (\ref{eq:duality_formula}), (\ref{eq:deriv_density}) and (\ref{Mall_deriv_Y}) yield  
\begin{align}
    & \mathbb{E}_{\theta} 
     \Bigl[ \partial_{i_1}  \delta_{\xi} (\widetilde{Y}) \int_0^1 B_{j, s} ds \Bigr] 
    = \mathbb{E}_{\theta}
    \Bigl[ \partial_{i_1}  \delta_{\xi} (\widetilde{Y}) 
    \int_0^1 (1-s) dB_{j, s} \Bigr] \nonumber \\
    & =  \tfrac{1}{2}  \sum_{i_2 = 1}^{d_R} V_j^{i_2} (x, \theta) 
    \mathbb{E}_{\theta} \bigl[
    \partial_{i_2}  \partial_{i_1}  \delta_{\xi} (\widetilde{Y}) \bigr]
    + \tfrac{1}{3}  \sum_{i_2 = d_R + 1}^d  \hat{V}_j V_0^{i_2} (x, \theta) \mathbb{E}_{\theta} \bigl[ 
    \partial_{i_2}  \partial_{i_1}  \delta_{\xi} (\widetilde{Y})
    \bigr] \nonumber \\ 
    & =  \tfrac{1}{2}  \sum_{i_2 = 1}^{d_R}  V_j^{i_2} (x,  \theta) 
    \mathcal{H}_{(i_1, i_2)}^{\widetilde{Y}} (\xi; \theta) 
    p^{\widetilde{Y}} (\xi; \theta )
    + \tfrac{1}{3}  \sum_{i_2 = d_R + 1}^d   
    \hat{V}_j V_0^{i_2} (x, \theta) \mathcal{H}_{(i_1, i_2)}^{\widetilde{Y}} (\xi; \theta )  
    p^{\widetilde{Y}} (\xi ; \theta).  \nonumber 
\end{align} 
\noindent
{\bf Derivation of (\ref{weight_7}):}
We have 
\begin{align}
    \mathbb{E}_{\theta} 
    [ \partial_{i_1} \delta_{\xi} (\widetilde{Y}) \eta_{0j, 1}]
    = \mathbb{E}_{\theta}
     [ \partial_{i_1} \delta_{\xi} (\widetilde{Y}) \zeta_{j0,1 } ]
     - \tfrac{1}{3}  \mathbb{E}_{\theta} [ \partial_{i_1} \delta_{\xi} (\widetilde{Y}) B_{j,1}],    
\end{align}
and 
\begin{align}
   &  \mathbb{E}_{\theta}
    [ \partial_{i_1} \delta_{\xi} (\widetilde{Y}) B_{j, 1} ] 
    = \sum_{i_2=1}^d 
   \mathbb{E}_{\theta}
   [ \partial_{i_2} \partial_{i_1} \delta_{\xi} (\widetilde{Y}) 
    \int_0^1 D_{j, s} \widetilde{Y}^{i_2} ds ]  \nonumber \\  
  & =  \sum_{i_2 = 1}^{d_R}  V_j^{i_2} (x; \theta)  
  \mathcal{H}_{(i_1, i_2)}^{\widetilde{Y}} (\xi ; \theta)    
   p^{\widetilde{Y}} (\xi ; \theta)   
  + \tfrac{1}{2}  \sum_{i_2 = d_R + 1}^d   
   \hat{V}_j V_0^{i_2} (x, \theta)
   \mathcal{H}_{(i_1, i_2)}^{\widetilde{Y}} (\xi ; \theta ) 
   p^{\widetilde{Y}} (\xi ; \theta).  \label{eq:BM_weight}
\end{align} 
Using (\ref{weight_3}) and (\ref{eq:BM_weight}), we obtain 
\begin{align*} 
     & \mathbb{E}_{\theta} [ \partial_{i_1} \delta_{\xi} (\widetilde{Y}) \eta_{0j, 1}]  \nonumber \\ 
     & = \tfrac{1}{6} \sum_{i_2 = 1}^{d_R}  V_j^{i_2} (x, \theta)  \mathcal{H}_{(i_1, i_2)}^{\widetilde{Y}} (\xi ; \theta)
     p^{\widetilde{Y}} (\xi ; \theta) 
     + \tfrac{1}{6} \sum_{i_2 = d_R + 1}^d  
     \hat{V}_j V_0^{i_2} (x,  \theta) 
     \mathcal{H}_{(i_1, i_2)}^{\widetilde{Y}} (\xi ; \theta)
     p^{\widetilde{Y}} (\xi ; \theta). 
\end{align*} 

\noindent
{\bf Derivation of (\ref{weight_8}):} Making use of (\ref{weight_3}) and (\ref{eq:BM_weight}),  we obtain
\begin{align*}
    & \mathbb{E}_{\theta} [ \partial_{i_1} \delta_{\xi} (\widetilde{Y}) \eta_{j0, 1}] 
    = \tfrac{1}{2} \mathbb{E}_{\theta} 
    [ \partial_{i_1} \delta_{\xi}  (\widetilde{Y})  \zeta_{j0, 1} ]
    - \tfrac{1}{12} \mathbb{E}_{\theta} 
    [ \partial_{i_1} \delta_{\xi} (\widetilde{Y}) B_{j, 1}]   \\ 
    & = \tfrac{1}{6} \sum_{i_2 = 1}^{d_R}  V_j^{i_2} (x, \theta)  
   \mathcal{H}_{(i_1, i_2)}^{\widetilde{Y}} (\xi ; \theta)    
   p^{\widetilde{Y}} (\xi ; \theta) +  \tfrac{1}{8} \sum_{i_2 = d_R + 1}^d 
    \hat{V}_j V_0^{i_2} (x,  \theta) 
    \mathcal{H}_{(i_1, i_2)}^{\widetilde{Y}} (\xi ; \theta )
    p^{\widetilde{Y}} (\xi ; \theta). 
\end{align*}

\noindent 
 {\bf Derivation of (\ref{weight_4}):} Recall: for $1 \le j_1, j_2 \le d_R$, 
\begin{align} \label{eq:rv_def}
   \zeta_{j_1 j_2, 1} 
   = \tfrac{1}{2} 
    \bigl\{ B_{j_1, 1} B_{j_2, 1} -  \mathbf{1}_{j_1 = j_2} \bigr\}
   + \tfrac{1}{2} B_{j_1, 1} \widetilde{B}_{j_2, 1} \times  \mathbf{1}_{j_1 <  j_2}
   - \tfrac{1}{2} B_{j_2, 1} \widetilde{B}_{j_1, 1} \times  \mathbf{1}_{j_1 >  j_2}.
\end{align}
Since $\widetilde{B}$ and $B$ are independent, we have 
\begin{align*}
  \mathbb{E}_{\theta} 
  \Bigl[ 
   \partial_{i_1} \delta_{\xi} (\widetilde{Y})   \zeta_{j_1 j_2, 1} 
  \Bigr]
  & =  \tfrac{1}{2} 
  \mathbb{E}_{\theta} 
  \Bigl[ \partial_{i_1} \delta_{\xi} (\widetilde{Y})  
  \bigl\{ B_{j_1, 1} B_{j_2, 1} - \mathbf{1}_{j_1 = j_2} 
  \bigr\} 
  \Bigr] \nonumber \\ 
  & =  \tfrac{1}{2} \mathbb{E}_{\theta} 
   \Bigl[ \partial_{i_1} \delta_{\xi} (\widetilde{Y}) 
     \Bigl\{ 
     \int_0^1 B_{j_1,s } d B_{j_2,s } +  \int_0^1 B_{j_2,s } d B_{j_1, s}  \Bigr\}
   \Bigr]. 
\end{align*}
The duality formula (\ref{eq:duality_formula}) and the Malliavin derivative (\ref{Mall_deriv_Y}) give
\begin{align*}
    &\mathbb{E}_{\theta} 
    \Bigl[ \partial_{i_1} \delta_{\xi} (\widetilde{Y}) 
    \int_0^1 B_{j_1,s } d B_{j_2, s}  
    \Bigr] 
    =  \sum_{i_2=1}^d  \mathbb{E}_{\theta}
    \Bigl[ \partial_{i_1} \partial_{i_2} \delta_{\xi} (\widetilde{Y})
    \int_0^1  D_{j_2, s} \widetilde{Y}^{i_2}   B_{j_1,s } d s  
    \Bigr]   \nonumber    \\ 
    & =  \sum_{i_2 = 1}^{d_R}  
    \mathbb{E}_{\theta} 
    \Bigl[ \partial_{i_2}  \partial_{i_1} \delta_{\xi} (\widetilde{Y}) \int_0^1 B_{j_1,s} d s 
    \Bigr]  V_{j_2}^{i_2} (x , \theta)  \nonumber \\
    & \quad  + \sum_{i_2 = d_R + 1}^d  \mathbb{E}_{\theta} \Bigl[ \partial_{i_2}  \partial_{i_1} \delta_{\xi} (\widetilde{Y}) \int_0^1 (1-s) B_{j_1, s} d s \Bigr]  \hat{V}_{j_2} V_{0}^{i_2} (x ,  \theta)  \nonumber \\
    & =  
    \tfrac{1}{2} \sum_{i_2, i_3 = 1}^{d_R} 
    \mathbb{E}_{\theta} 
    \Bigl[
    \partial_{i_3} \partial_{i_2}  \partial_{i_1} \delta_{\xi} (\widetilde{Y})  \Bigr]  
    V_{j_1}^{i_2}(x , \theta) V_{j_2}^{i_3} (x , \theta)  \nonumber \\
    & 
    \quad + \tfrac{1}{3} \sum_{\substack{d_R + 1 \leq  i_2 \leq d \\ 1 \leq i_3 \leq d_R}}  
    \mathbb{E}_{\theta} \Bigl[ \partial_{i_3} \partial_{i_2}  \partial_{i_1} \delta_{\xi} (\widetilde{Y})  \Bigr] 
    \hat{V}_{j_1} V_0^{i_2} (x , \theta) V_{j_2}^{i_3} (x , \theta) 
    \nonumber  \\
    & \quad  +  \tfrac{1}{6} \sum_{\substack{1 \leq  i_2 \leq d_R \\ d_R + 1 \leq i_3 \leq d}}  
    \mathbb{E}_{\theta}
    \Bigl[ 
    \partial_{i_3} \partial_{i_2}  \partial_{i_1} \delta_{\xi} (\widetilde{Y})  
    \Bigr] V_{j_1}^{i_2}(x , \theta) \hat{V}_{j_2} V_0^{i_3} (x , \theta) \nonumber \\
    & \quad  +  
    \tfrac{1}{8}  \sum_{i_2, i_3 = d_R + 1}^d   
    \mathbb{E}_{\theta} \Bigl[ \partial_{i_3} \partial_{i_2}  \partial_{i_1} \delta_{\xi} (\widetilde{Y})  \Bigr] 
    \hat{V}_{j_1} V_0^{i_2} (x , \theta) 
    \hat{V}_{j_2} V_0^{i_3} (x , \theta),  %
\end{align*}
where we used 
\begin{align*}
   \int_0^1 B_{j_1, s} ds =  \int_0^1 (1-s) d B_{j_1, s},   \ \ 
   \int_0^1 (1-s) B_{j_1, s} ds =  \tfrac{1}{2} \int_0^1 (1-s)^2  d B_{j_1, s}. 
\end{align*}
Thus, we obtain (\ref{weight_4}). 
%
%
%
\\

\noindent 
{\bf Derivation of (\ref{weight_5}):}
From (\ref{eq:rv_def}) we have:
\begin{align*}
    & \mathbb{E}_{\theta} 
    \bigl[ \partial_{i_1} \partial_{i_2} \delta_{\xi} 
    (\widetilde{Y})  \zeta_{j_1 j_2,1} \zeta_{j_3 j_4, 1} 
    \bigr] 
    = \tfrac{1}{4} \mathbb{E}_{\theta} 
    \Bigl[
    \partial_{i_1} \partial_{i_2} \delta_{\xi} (\widetilde{Y}) 
    \bigl\{ B_{j_1,1} B_{j_2, 1} -  \mathbf{1}_{j_1 = j_2} \bigr\}
    \bigl\{ B_{j_3, 1} B_{j_4, 1} -  \mathbf{1}_{j_3 = j_4} \bigr\}  
    \Bigr]  \\
    & \quad + \tfrac{1}{4} 
    \mathbb{E}_{\theta} 
    \bigl[ \partial_{i_1} \partial_{i_2} \delta_{\xi} (\widetilde{Y}) 
    B_{j_1,1} B_{j_3, 1} \bigr] 
    \times \mathbf{1}_{j_1, j_3 <  j_2 = j_4}   
    - \tfrac{1}{4} \mathbb{E}_{\theta} 
    \bigl[ \partial_{i_1} \partial_{i_2} \delta_{\xi} (\widetilde{Y}) 
    B_{j_1,1} B_{j_4, 1} 
    \bigr]  \times \mathbf{1}_{j_1, j_4 <  j_2 = j_3}   \\
    & \quad - \tfrac{1}{4}
    \mathbb{E}_{\theta} 
    \bigl[ \partial_{i_1} \partial_{i_2} \delta_{\xi} (\widetilde{Y}) B_{j_2,1} B_{j_3, 1} 
    \bigr]  \times \mathbf{1}_{j_2, j_3 <  j_1 = j_4}   
    +  \tfrac{1}{4} \mathbb{E}_{\theta} 
    \bigl[ \partial_{i_1} \partial_{i_2} \delta_{\xi} (\widetilde{Y})
    B_{j_2,1} B_{j_4, 1} 
    \bigr] \times \mathbf{1}_{j_2, j_4 <  j_1 = j_3}. 
\end{align*} 
It\^o's formula yields 
\begin{align*} 
\bigl\{ B_{j_1,1} B_{j_2, 1} -  \mathbf{1}_{j_1 = j_2} \bigr\}
\bigl\{ B_{j_3, 1} B_{j_4, 1} -  \mathbf{1}_{j_3 = j_4} \bigr\} 
=  \bigl\{ I_{(j_1, j_2)}(1) +  I_{(j_2, j_1)}(1)  \bigr\}
\bigl\{ I_{(j_3, j_4)}(1) +  I_{(j_4, j_3)}(1) \bigr\}, 
\end{align*}
and 
\begin{align} \label{product_weinerInt}
  & I_{(j_1, j_2)}(1) I_{(j_3, j_4)}(1) 
  =  I_{(j_1 , j_2, j_3, j_4)} (1)
   + I_{(j_1, j_3, j_2, j_4)} (1)
   + I_{(j_1, j_3, j_4, j_2)} (1)
   + I_{(j_3, j_1, j_4, j_2)} (1)  \nonumber \\[0.1cm]
  &  + I_{(0, j_4, j_2)} (1)  \times \mathbf{1}_{j_1 = j_3} 
   + I_{(j_1, 0, j_4)} (1) \times  \mathbf{1}_{j_2 = j_3} 
   + I_{(0, j_4, j_2)} (1) \times  \mathbf{1}_{j_1 = j_3} \nonumber \\[0.1cm]
  & + I_{(j_3, 0, j_2)} (1) \times  \mathbf{1}_{j_1 = j_4} 
    + I_{(j_1, j_3, 0)} (1) \times \mathbf{1}_{j_2 = j_4}  
    + I_{(j_3, j_1, 0)} (1) \times \mathbf{1}_{j_2 = j_4}  \nonumber \\
  & + \tfrac{1}{2}  \times \mathbf{1}_{j_1 = j_3, j_2 =  j_4}.
\end{align}
Making use of (\ref{product_weinerInt}) and the duality formula (\ref{eq:duality_formula}), we have
\begin{align*}
   & \mathbb{E}_{\theta}  
   \bigl[ 
   \partial_{i_1} \partial_{i_2} \delta_{\xi} (\widetilde{Y})  
   \bigl\{ B_{j_1, 1} B_{j_2, 1} -  \mathbf{1}_{j_1 = j_2} \bigr\}
   \bigl\{ B_{j_3, 1} B_{j_4, 1} -  \mathbf{1}_{j_3 = j_4} \bigr\}  
   \bigr]   \\[0.3cm]
   & \qquad =  \mathcal{H}_{(i_1, i_2)}^{\widetilde{Y}} (\xi ; \theta)  p^{\widetilde{Y}} (\xi; \theta) 
   \times  \bigl\{  \mathbf{1}_{j_1 = j_3, j_2 = j_4} 
   +  \mathbf{1}_{j_1 = j_4, j_2 = j_3} 
   \bigr\} \nonumber \\
   & \qquad\qquad  + \sum_{k=2, 4} \sum_{\alpha \in \{1, \ldots, d\}^k} 
   w_{\alpha} (x ,\theta) \times (-1)^{|\alpha|} \, 
   \partial^\xi_{(\alpha_1, \ldots, \alpha_k, i_1, i_2)} 
   p^{\widetilde{Y}} (\xi; \theta), \nonumber 
\end{align*}
where $w_{\alpha} : \mathbb{R}^d \times \Theta \to \mathbb{R}$ is of the same structure as the one described for $\nu_\alpha$ in (\ref{eq:r_remainder}) under condition (\ref{assump:coeff}). We also have 
\begin{align*}
   &  \mathbb{E}_{\theta} \bigl[
   \partial_{i_1} \partial_{i_2} \delta_{\xi} (\widetilde{Y}) 
   B_{j_1,1} B_{j_3, 1} 
   \bigr]  \times \mathbf{1}_{j_1, j_3 <  j_2 = j_4}    \\[0.3cm] 
   & = \mathbb{E}_{\theta} \bigl[ 
   \partial_{i_1} \partial_{i_2} \delta_{\xi} (\widetilde{Y}) 
   \bigr]  \times \mathbf{1}_{j_1 = j_3 < j_2 = j_4}
   + \  \sum_{\alpha \in \{1, \ldots, d \}^2} \widetilde{w}_{\alpha} (x , \theta)   
   \mathbb{E}_{\theta} \bigl[ 
   \partial^{\alpha} \partial_{i_1} \partial_{i_2}  \delta_{\xi} (\widetilde{Y}) 
   \bigr] \nonumber \\
   & =  \mathcal{H}_{(i_1, i_2)}^{\widetilde{Y}} (\xi ; \theta)
   p^{\widetilde{Y}} (\xi; \theta) \times \mathbf{1}_{j_1 = j_3 < j_2 = j_4} 
   + \sum_{\alpha \in \{1, \ldots, d \}^2} 
   \widetilde{w}_{\alpha} (x , \theta) \times  (-1)^{|\alpha|} \, 
   \partial^\xi_{(\alpha_1, \ldots, \alpha_k, i_1, i_2)}
   p^{\widetilde{Y}} (\xi; \theta),
 \end{align*} 
with some $\widetilde{w}_\alpha : \mathbb{R}^d \times \Theta \to \mathbb{R}$ having the same property as the above $w_\alpha$, where we used $B_{j_1, 1} B_{j_2, 1} = I_{(j_1, j_2)}(1) + I_{(j_2, j_1)}(1) + \mathbf{1}_{j_1 = j_2}$ for $1 \le j_1, j_2 \le  d_R$ in the first equality. Thus, we obtain (\ref{weight_5}). 
%
\\

\noindent 
{\bf Derivation of (\ref{weight_6}):}  
From the definition of $\textstyle{\eta_{j_1 j_2, 1} = \frac{1}{3} \zeta_{j_1 j_2, 1} - \widetilde{\eta}_{j_1 j_2, 1}}$,  it holds that 
\begin{align}
   & \mathbb{E}_{\theta} \bigl[ 
   \partial_{i_1} \partial_{i_2} \delta_{\xi} (\widetilde{Y}) 
   \eta_{j_1 j_2, 1} \eta_{j_3 j_4, 1} \bigr] 
   =  \tfrac{1}{9} \mathbb{E}_{\theta} \bigl[ 
   \partial_{i_1} \partial_{i_2} \delta_{\xi} (\widetilde{Y})
   \zeta_{j_1 j_2, 1}  \zeta_{j_3 j_4, 1} \bigr]
   + \mathbb{E}_{\theta} 
   \bigl[ \partial_{i_1} \partial_{i_2} \delta_{\xi} (\widetilde{Y})
   \widetilde{\eta}_{j_1 j_2, 1} \widetilde{\eta}_{j_3 j_4, 1}
   \bigr]  \nonumber \\[0.2cm]
   & = \tfrac{1}{18} \mathcal{H}_{(i_1, i_2)}^{\widetilde{Y}} (\xi ; \theta)  p^{\widetilde{Y}} (\xi; \theta) \times \mathbf{1}_{j_1 = j_3, j_2 = j_4} 
   + \sum_{k=2, 4} \sum_{\alpha \in \{1, \ldots, d \}^k} 
   v_{\alpha} (x , \theta) \cdot   (-1)^{|\alpha|} 
   \partial^\xi_{(\alpha_1, \ldots, \alpha_k, i_1, i_2)} 
   p^{\widetilde{Y}} (\xi; \theta) \nonumber \\
   & \qquad + \tfrac{1}{36} \mathcal{H}_{(i_1, i_2)}^{\widetilde{Y}} (\xi ; \theta)  p^{\widetilde{Y}} (\xi; \theta)  \times \mathbf{1}_{j_1 = j_3, j_2 = j_4} \nonumber \\[0.3cm]
   & =  \tfrac{1}{12} \mathcal{H}_{(i_1, i_2)}^{\widetilde{Y}} (\xi ; \theta)  p^{\widetilde{Y}} (\xi; \theta) \times  \mathbf{1}_{j_1 = j_3, j_2 = j_4}
   + \sum_{k=2, 4} 
   \sum_{\alpha \in \{1, \ldots, d \}^k} v_{\alpha} (x , \theta) \cdot   (-1)^{|\alpha|} \partial^\xi_{(\alpha_1, \ldots, \alpha_k, i_1, i_2)}
   p^{\widetilde{Y}} (\xi; \theta), \nonumber  
\end{align}
 where we used $\textstyle{\mathbb{E} [\widetilde{\eta}_{j_1 j_2, 1} \widetilde{\eta}_{j_3 j_4, 1}] = \frac{1}{36} \times \mathbf{1}_{j_1 = j_3, j_2 = j_4}}$ and (\ref{weight_5}). 
\section{Proof of Proposition \ref{prop:density_scheme}}
\label{appendix:local_wa_ver1}
\subsection{Proof of Properties (i)--(iii)}
\label{appendix:sub_wa_ver1}
We exploit the density expansion (\ref{eq:density_expansion}) obtained in Lemma \ref{lemma:AE_scheme}. Property (iii) holds immediately by identifying $R^J_2 (x, y ; \theta)$ with the residual term $R^J (x, y; \theta)$ in Lemma \ref{lemma:AE_scheme}. We make use of the following notation: for a multi-index $\alpha \in \{1, \ldots, d \}^l, \, l \ge 1$,
\begin{align*} 
 & \| \alpha  \|_R \equiv (\mrm{\#  \, of \, elements \, taking \, in \, the \, rough \, coordinates \, } \{1, \ldots,  d_R \});   \\
 & \| \alpha  \|_S \equiv (\mrm{\#  \, of \, elements \, taking \, in \, the \, smooth \, coordinates \, } \{d_R + 1, \ldots,  d \}).  
\end{align*}   
To prove Properties (i)--(ii), we first introduce the following result about the expression for $\Psi_3$ in the density expansion (\ref{eq:density_expansion}): 
\begin{lemma} \label{lemma:psi_3}
Let $\Delta > 0$, $x, y \in \mathbb{R}^d$ and $\theta \in \Theta$. It holds that
\begin{align}
\Psi_3 (\Delta, x, y;\theta) 
= \Psi_3^{\mrm{weak}} (\Delta, x, y; \theta)
 + \widetilde{\Psi}_3 (\Delta, x, y; \theta), 
\end{align} 
where $\Psi_3^{\mrm{weak}} (\Delta, x, y; \theta)$ is given in (\ref{eq:psi3_w}) in the main text, and each of the terms in $\widetilde{\Psi}_3 (\Delta, x, y; \theta)$ is given as: 
\begin{align} \label{eq:rem_rep}
   v_{\alpha} (x , \theta) 
   \mathcal{H}_{\alpha} (\Delta, x, y; \theta), 
\end{align} 
for a multi-index $\alpha \in \{1, \ldots, d \}^{l}, \; l \ge 3$, 
and the function $v_{\alpha}: \mathbb{R}^d \times \Theta \to \mathbb{R}$ 
is determined from the products of $V_j (\cdot ; \theta), \; 0 \le j \le d_R$ and theirs partial derivatives. 
%
\end{lemma}
We provide the proof in Section \ref{appendix:psi_3}. To determine the residual term $R_1^J (\Delta, x, y ; \theta)$, we introduce the following result whose proof is contained in Section \ref{appendix:wa_hermite}: 
\begin{lemma} \label{lemma:wa_hermite}
Let $\Delta > 0$, $x \in \mathbb{R}^d$, $\theta \in \Theta$ and $\alpha \in \{1, \ldots, d \}^l, \, l \ge 1$. For any $\varphi \in C_b^\infty (\mathbb{R}^d; \mathbb{R})$, there exist a constant $C > 0$ and $q \ge 1$ so that: 
\begin{align} 
\left| \int_{\mathbb{R}^d} \varphi (y) \mathcal{H}_\alpha (\Delta, x, y; \theta) p_\Delta^{\widetilde{X}^{\mrm{LG}}} (x, y ; \theta) dy  \right| \le C (1 + |x|^q) \cdot \sqrt{\Delta^{\| \alpha \|_R + 3 \| \alpha \|_S}}. 
\end{align}
\end{lemma}
Since $\Psi_l (\Delta, x, y ; \theta), \, l \ge 4$, involves Hermite polynomials $\mathcal{H}_\alpha (\Delta, x, y; \theta)$ with $\alpha \in \{1, \ldots, d \}^j, \, j \ge 2$ from the expression of $\Xi_{l} (x, \xi; \theta)$ in the expansion (\ref{eq:expansion_GE}), Lemma \ref{lemma:wa_hermite} yields:
\begin{align}
\left| \sum_{l = 4}^{J-1} 
\Delta^{\tfrac{l}{2}} 
\int_{\mathbb{R}^d} \varphi (y) \, 
\Psi_l (\Delta, x, y; \theta) \, 
p_\Delta^{\widetilde{X}^{\mrm{LG}}} (x, y ; \theta) 
dy \right| \le C (1 + |x|^q) \Delta^3, 
\end{align}
for some $C > 0$ and $q \ge 1$, which indicates that $\Delta^{{l}/{2}} \Psi_l (\Delta, x, y; \theta), \, 4 \le l \le J-1,$ are included in the residual $R_1^J (\Delta, x, y ; \theta)$. We also write: for $\Delta > 0$, $x, y \in \mathbb{R}^d$ and $\theta \in \Theta$, 
\begin{align*}
\Psi_l (\Delta, x, y; \theta) \equiv  \Psi_l^{\mrm{weak}} (\Delta, x, y; \theta) 
+ \widetilde{\Psi_l} (\Delta, x, y; \theta), \quad l = 1,2,
\end{align*}
where $\widetilde{\Psi_l} (\Delta, x, y; \theta)$ is explicitly determined from the definition of $\Psi_l (\Delta, x,  y; \theta)$ and $\Psi_l^{\mrm{weak}} (\Delta, x, y; \theta)$ in Lemma \ref{lemma:AE_scheme} and Proposition \ref{prop:density_scheme} given in the main text, respectively. Notice that every term in $\Delta^{{l}/{2}} \cdot \widetilde{\Psi}_l (\Delta, x,y; \theta), \, 1 \le l \le 3$, contains  
\begin{align}
 \Delta^{\tfrac{l}{2}} \cdot \mathcal{H}_\alpha (\Delta, x, y ; \theta), 
 \quad \alpha \in \{1 \ldots, d \}^j, \ \ j \ge 1, 
\end{align}
with $l$ and $\alpha$ satisfying $l + \| \alpha \|_R + 3 \| \alpha  \|_S \ge 6$. Thus, Lemma \ref{lemma:wa_hermite} yields:  
\begin{align}
\left| \sum_{l = 1}^{3} 
\Delta^{\tfrac{l}{2}} \int_{\mathbb{R}^d} \varphi (y) 
\, \widetilde{\Psi}_l (\Delta, x, y; \theta) \, 
p_\Delta^{\widetilde{X}^{\mrm{LG}}} (x, y ; \theta) 
dy \right| \le C (1 + |x|^q) \Delta^3,
\end{align}
for some $C > 0$ and $q \ge 1$. The proof is now complete. 
\subsubsection{Proof of Lemma \ref{lemma:psi_3}}
\label{appendix:psi_3}
Notice from the proof of Lemma \ref{lemma:AE_scheme} in Section \ref{pf_hypo_density} that 
$$ 
\Psi_3  (\Delta,x, y; \theta) p_\Delta^{\widetilde{X}^{\mrm{LG}}} (x, y; \theta)
= \Xi_3 (x, \xi; \theta) |_{\xi = m_{x, \theta, \sqrt{\Delta}} (y)},
$$
where $\Xi_3 (x, \xi; \theta)$ is determined from the expansion (\ref{eq:expansion_GE}), and given as: 
\begin{align}  \label{eq:Xi_3}  
\Xi_3 (x, \xi; \theta) 
& = \tfrac{1}{2} \sum_{i =1}^{d_R} \hat{V}_0 V_0^i  (x, \theta) 
   \mathbb{E}_{\theta} 
   \Bigl[ \partial_i \delta_\xi (\widetilde{Y})
   \Bigr] 
   + \widetilde{\Xi}_{3} (x, \xi ; \theta) \nonumber \\ 
& = \tfrac{1}{2} \sum_{i =1}^{d_R} \hat{V}_0 V_0^i  (x, \theta) 
   \mathcal{H}_{(i)}^{\widetilde{Y}}  (\xi; \theta) p^{\widetilde{Y}} (\xi; \theta) + \widetilde{\Xi}_{3} (x, \xi; \theta), 
\end{align}
where we have applied the formula (\ref{eq:deriv_density}) to the first term in the right hand side of (\ref{eq:Xi_3}), and we have set: 
\begin{align} \label{eq:Xi_g3}
 \widetilde{\Xi}_{3} (x, \xi; \theta) 
  & = \sum_{i_1, i_2 = 1}^d  
\mathbb{E}_{\theta} 
\Bigl[ \partial_{i_1} \partial_{i_2} \delta_{\xi} (\widetilde{Y})  \widetilde{Y}^{ (1), i_1} \widetilde{Y}^{ (2), i_2} \Bigr]  \nonumber \\ 
& \quad  
+ \tfrac{1}{3!} \sum_{1 \leq i_1, i_2, i_3  \leq d} 
 \mathbb{E}_{\theta} 
 \Bigl[ \partial_{i_1} \partial_{i_2} \partial_{i_3} \delta_{\xi} (\widetilde{Y}) \widetilde{Y}^{(1), i_1} \widetilde{Y}^{(1), i_2} \widetilde{Y}^{(1), i_3}
 \Bigr].  
\end{align}  
We will now show that each of the terms in 
$\widetilde{\Xi}_3 (x, \xi; \theta)$ 
%
is given as: 
\begin{align} \label{eq:target}
\nu_\alpha (x, \theta) \mathcal{H}_\alpha^{\widetilde{Y}} (\xi; \theta) p^{\widetilde{Y}} (\xi; \theta), 
\end{align}
where the multi-index $\alpha$ and $\nu_\alpha (\cdot, \cdot)$ are the same as those in (\ref{eq:rem_rep}). From (\ref{eq:Xi_g3}), each of the terms in 
$\widetilde{\Xi}_3 (x, \xi; \theta)$ is written as: 
\begin{align} \label{eq:term_Xi3}
\mathbb{E}_{\theta} 
 \Bigl[ \partial^{\alpha} \delta_{\xi} (\widetilde{Y})  
 I_\beta (1) \Bigr]
 v_{\beta}^\alpha (x,  \theta), 
\end{align} 
for some multi-indices $\alpha \in  \{1, \ldots, d\}^k, \; k = 2,3$ and $\beta \in \{0,1, \ldots, d_R \}^{\nu}, \; \nu \geq 1$ such that 
$k +  n(\beta) \geq 3$, and $v_{\beta}^\alpha :  \mathbb{R}^d \times \Theta \to \mathbb{R}$ being the same structure as $\nu_\alpha$ in (\ref{eq:rem_rep}). We recall $n (\beta)$ is the number of non-zero elements in the multi-index $\beta$. Thus, applying (\ref{eq:ibp_general}) and (\ref{eq:deriv_density}) to (\ref{eq:term_Xi3}), we conclude that the term (\ref{eq:term_Xi3}) is of the form (\ref{eq:target}). The proof is now complete. 
\subsubsection{Proof of Lemma \ref{lemma:wa_hermite}}
\label{appendix:wa_hermite}
 It follows from the definition of Hermite polynomials and the integration by parts that for $\varphi \in C_p^\infty (\mathbb{R}^d; \mathbb{R})$, 
\begin{align}
 & \int_{ \mathbb{R}^d } \varphi (y) \mathcal{H}_{\alpha}^{\widetilde{Y}} 
 \bigl( m_{x, \theta, \sqrt{\Delta}}(y)  ; \theta \bigr) \times \tfrac{1}{\sqrt{\Delta^{d_R + 3d_S}}} p^{\widetilde{Y}} \bigl( m_{x, \theta, \sqrt{\Delta}} (y)  ; \theta \bigr)  dy \nonumber  \\
 & = \int_{\mathbb{R}^d} \varphi \bigl( h_{x, \theta, \sqrt{\Delta}} (\xi)  \bigr) 
 (-1)^{| \alpha |}
 \times \partial^\xi_\alpha p^{\widetilde{Y}} (\xi ; \theta) d \xi \nonumber \\ 
 & = \sqrt{\Delta^{ \| \alpha \|_R  + 3 \| \alpha \|_S }} \int_{\mathbb{R}^d} \partial^z_{\alpha} 
 \varphi ( z ) |_{z= h_{x, \theta, \sqrt{\Delta}} (\xi)}
 \times  p^{\widetilde{Y}} (\xi ; \theta) d \xi \nonumber  \\ 
 & = \sqrt{\Delta^{\| \alpha \|_R  + 3 \| \alpha \|_S }} \cdot 
 \mathbb{E}_\theta  \bigl[ 
 \partial^\alpha
 \varphi \bigl( \widetilde{X}^{\mrm{LG}, x}_\Delta \bigr) 
 \bigr], 
 \label{eq:remainder_ibp}  
\end{align}
where we have used the change of variables in the second line, and the function $h_{x, \theta, \sqrt{\Delta}}: \mathbb{R}^d \to \mathbb{R}^d$ is defined as: 
\begin{align}
\xi = (\xi_R, \xi_S) \mapsto h_{x, \theta, \sqrt{\Delta}} (\xi) 
= \begin{bmatrix}
x_R + V_{R, 0} (x, \beta) \Delta + \sqrt{\Delta} \xi_R \\[0.2cm] 
x_S + V_{S, 0} (x, \gamma) \Delta + \tfrac{\Delta^2}{2} \hat{V}_0 V_{S, 0} (x, \theta) + \sqrt{\Delta^3} \xi_S 
\end{bmatrix}. 
\end{align}
Thus, from the property of the test function $\varphi 
\in C_p^\infty (\mathbb{R}^d ; \mathbb{R})$ and (\ref{eq:remainder_ibp}), we have that: 
\begin{align*}
\left| \int_{\mathbb{R}^d} \varphi (y) \mathcal{H}_\alpha (\Delta, x, y; \theta) p_\Delta^{\widetilde{X}^{\mrm{LG}}} (x, y ; \theta) dy  \right|
= \sqrt{\Delta^{\| \alpha \|_R + 3 \| \alpha \|_S }} 
\cdot 
\Bigl| \mathbb{E}_\theta  \bigl[ 
\partial^\alpha
\varphi \bigl( \widetilde{X}^{\mrm{LG}, x}_\Delta \bigr) \bigr] \Bigr|  
\le C ( 1 + |x|^q ), 
\end{align*} 
for some constants $C > 0$ and $q \ge  1$, where we exploited the following bound: under Assumptions (\ref{assump:param_space})--(\ref{assump:coeff}), for any
$\Delta > 0$, $x \in \mathbb{R}^d$, $\theta \in \Theta$, and $p \ge 1$, there exists $C > 0$ so that
\begin{align}
\mathbb{E}_\theta \Bigl[ 
\bigl| \widetilde{X}^{\mrm{LG}, x}_\Delta \bigr|^p \Bigr] \le C (1 + |x|^p).  
\end{align}
The proof of Lemma \ref{lemma:wa_hermite} is now complete.
%
\subsection{Proof of Property (iv)}
\noindent From the properties of Hermite polynomials, we have 
\begin{align*}
\int_{\mathbb{R}^d} 
\mathcal{H}_\alpha^{\widetilde{Y}}  (\xi; \theta)
\, p^{\widetilde{Y}} (\xi ; \theta) \, d\xi = 0,  \qquad  \theta \in \Theta. 
\end{align*}
%
%
Due to the definition of $\Psi^{\mrm{weak}} (\Delta, x, y ; \theta)$ involving $\mathcal{H}_\alpha^{\widetilde{Y}}  (\xi; \theta) |_{\xi = m_{x, \theta, \sqrt{\Delta}}(y)}$, change of variables gives: 
\begin{align*}
\int_{\mathbb{R}^d} 
 \Psi^{\mrm{weak}} (\Delta, x, y ; \theta) 
 \, p^{\widetilde{X}^{\mathrm{LG}}}_\Delta (x, y; \theta) dy = 0, 
\end{align*}
and now the proof is complete. 
\section{Proof of Proposition \ref{prop:local_WA_ver2}} \label{appendix:diff_density_weak}
\noindent A Taylor expansion of the exponential function gives: %
\begin{align} 
\label{eq:Taylor_exp}
  \exp \left( K (z) \right) - (1 + z) 
 = \sum_{i = 7}^{\nu} c_i z^i  
 + g (z) \tfrac{z^{8}}{7!},  \ \  z \in \mathbb{R},   
\end{align}
for some integer $\nu > 7$ and constants $c_i$, where 
\begin{align*}
    g (z) := \int_0^1 
     \exp \left( s K(z) \right) ds.
\end{align*}
Notice that given the definition of the function $z \mapsto K(z)$, there exists a constant $C > 0$ so that $\textstyle{\sup_{z \in \mathbb{R}} |g (z)| \leq C }$. Thus, we have 
\begin{align*} 
    & \left|
    \int_{\mathbb{R}^d}  \varphi (y)
    p^{\widetilde{X}^{\mathrm{LG}}}_\Delta (x, y ; \theta) \left( 1 + 
    \Psi^{\mrm{weak}} (\Delta, x, y; \theta) \right) dy 
    - \int_{\mathbb{R}^d}  \varphi (y)  
    \bar{p}^{\, \mathrm{II}}_\Delta (x, y ; \theta) dy  \right|  \nonumber \\
    & \leq \sum_{i=7}^{\nu} 
    \left| 
    \int_{\mathbb{R}^d} \varphi (y) c_i \bigl\{ \Psi^{\mrm{weak}} (\Delta, x, y ; \theta) \bigr\}^i p^{\widetilde{X}^{\mathrm{LG}}}_\Delta (x, y ; \theta)  dy 
    \right|  \nonumber \\ 
    & \qquad  +  \left| 
    \int_{\mathbb{R}^d} \varphi (y) 
    g \bigl(  \Psi^{\mrm{weak}} (\Delta, x, y ; \theta)  \bigr) 
    \tfrac{\bigl\{ \Psi^{\mrm{weak}} (\Delta, x, y ; \theta) \bigr\}^{8}}{7!}
    p^{\widetilde{X}^{\mathrm{LG}}}_\Delta (x, y ; \theta)  dy 
    \right|  \nonumber \\ 
    & \leq C (1 + |x|^q) \Delta^3,
\end{align*}
for some constants $C > 0$ $q \ge 1$, where we made use of
$\textstyle{\sup_{y \in \mathbb{R}^d} 
 \bigl| g \bigl(  \Psi^{\mrm{weak}} (\Delta, x, y ; \theta)  \bigr) \bigr| < \infty}$, of the polynomial growth of the test function $\varphi$, and of the following estimates:
\begin{itemize}
\item Under  conditions (\ref{assump:param_space})-(\ref{assump:coeff}), there exist constants $C > 0$ and $q_1, q_2 \ge 1$ independent of $x, y \in \mathbb{R}^d$ and $\theta \in \Theta$ such that  
\begin{align*}
\bigl| \Psi^{\mrm{weak}} (\Delta, x, y ; \theta)  \bigr|  
\leq  C (1 + |x|^{q_1} + |y|^{q_2}) \sqrt{\Delta}.    
\end{align*}
\item Under  conditions (\ref{assump:param_space})-(\ref{assump:coeff}), for $p \geq 1$, there exists a constant $C > 0$ such that
\begin{align*}
 \mathbb{E}_\theta \bigl[ \bigl| \widetilde{X}_{\Delta}^{\mathrm{LG}, x} \bigr|^p \bigr] 
 \leq  C (1 + |x|^p).   
 \end{align*}
\end{itemize}
The proof is now complete. 
\section{Proof of Theorem \ref{thm:density_approx_second}} \label{sec:pf_density_approx_second}
Let $\Delta > 0$, $x, y \in\mathbb{R}^{d_R}$ and $\theta \in \Theta$. We have that 
\begin{align} 
\label{eq:density_decomp}
   &   p_{e,\Delta}^{X} (x , y; \theta) 
   - \bar{p}^{\, \mathrm{II}, (M)}_{e,\Delta} (x , y; \theta) \nonumber \\ 
   & \qquad =  
    \big(
      p_{e,\Delta}^{X} (x , y; \theta) - \bar{p}^{\, \mathrm{I}, (M)}_{e,\Delta} (x , y; \theta)  
    \big)
   + \big(
        \bar{p}^{\, \mathrm{I}, (M)}_{e,\Delta} (x , y; \theta) 
        - \bar{p}^{\, \mathrm{II}, (M)}_{e,\Delta} (x , y; \theta) 
     \big).
\end{align} 
For the first term on the right-hand-side of (\ref{eq:density_decomp}), we obtain the upper bound (\ref{eq:density_approx_second}) given in the main text under conditions (\ref{assump:param_space})--(\ref{assump:hypo1}) from \cite{iguchi:21-2}. For the second term, we have 
\begin{align} 
\label{eq:density_decomp_20}
 & \bar{p}^{\, \mathrm{I}, (M)}_{e,\Delta} (x , y; \theta) 
 - \bar{p}^{\, \mathrm{II},(M)}_{e,\Delta} (x , y; \theta) = E_1 + E_2, 
\end{align}
where we have defined
\begin{align*}
 & E_1 = \int_{\mathbb{R}^{d_R \times (M-1)}} 
  \Bigl\{ 
  \prod_{i=1}^{M-1} \bar{p}_{e,\delta_M}^{\, \mathrm{I}} (x_{i-1}, x_i; \theta)
  \Bigr\} 
  \mathscr{R}_{\delta_M} (x_{M-1}, y; \theta)\, 
  d x_1 \cdots d x_{M-1}; 
  \\[0.2cm] 
  & E_2 =  \sum_{k=1}^{M-1} 
  \int_{\mathbb{R}^{d_R \times (M-1)}} 
  \Bigl\{
     \prod_{i=1}^{M-k-1} \bar{p}_{e,\delta_M}^{\, \mathrm{I}}
      (x_{i-1}, x_i; \theta)
  \Bigr\} \nonumber  \\
 & \qquad  \times 
  \mathscr{R}_{\delta_M} (x_{M-k-1}, x_{M-k}; \theta) 
  \Bigl\{
     \prod_{i=1}^{k} 
     \bar{p}_{e,\delta_M}^{\, \mathrm{II}} (x_{M-k+(i-1)}, 
     x_{M - k + i}; \theta) 
  \Bigr\}\,d x_1 \cdots d x_{M-1}. 
\end{align*} 
The following term is used in the above expressions:
\begin{align}
\mathscr{R}_s (z_1, z_2 ; \theta)   
 := \bar{p}_{e,s}^{\, \mathrm{I}} (z_1, z_2; \theta)  
    - \bar{p}_{e,s}^{\, \mathrm{II}} (z_1, z_2; \theta), 
    \qquad s \ge 0, \; z_1, z_2 \in \mathbb{R}^{d_R}, \; \theta \in \Theta.
\end{align} 
We obtain the following bounds for $E_1$ and $E_2$.
\begin{lemma} 
\label{lemma:bd_E1_E2}
Assume that conditions (\ref{assump:param_space})--(\ref{assump:hypo1}) hold. There exist constants $C_1, C_2, c_1, c_2 > 0$, $q \geq d_R/2$ and non-decreasing functions $h_1(\cdot)$, $h_2(\cdot)$ such that for any $x, y \in \mathbb{R}^{d_R}$, $\theta \in \Theta$,  
\begin{align}
 |E_1| 
 & \leq C_1 (\delta_M)^3 \frac{h_1 (\Delta)}{\Delta^{q}} e^{ - c_1 \frac{|y - x|^2 }{\Delta}},  \label{eq:E1}\\ 
|E_2| 
 & \leq C_2 (M-1)  (\delta_M)^3 \frac{h_2 (\Delta)}{\Delta^{q}} e^{ - c_2 \frac{|y - x|^2 }{\Delta}}. \label{eq:E2}
\end{align}
\end{lemma}
\noindent We provide the proof in Section  \ref{appendix:density_approx_second}. 
Hence, we have 
\begin{align*} 
 \big|\,
   \bar{p}_{e,\Delta}^{\, \mathrm{I}, (M)} (x , y; \theta) 
   - \bar{p}^{\, \mathrm{II}, (M)}_{e,\Delta} (x , y; \theta) 
 \,\big| 
& \leq M  C  (\delta_M)^3 \frac{h (\Delta)}{\Delta^{q}} e^{ - c \frac{|y - x|^2 }{\Delta}} 
= \frac{C'}{M^2} \frac{h(\Delta)}{\Delta^{q}} e^{ - c \frac{|y - x|^2 }{\Delta}},
\end{align*}
where $h(\cdot)$ is a non-decreasing function and $C, C', c > 0$ are some constants independent of $\theta \in \Theta, \Delta > 0$ and $x, y \in \mathbb{R}^{d_R}$. The proof is complete.

\subsection{Proof of Lemma \ref{lemma:bd_E1_E2}} \label{appendix:density_approx_second}
\subsubsection{Upper Bound (\ref{eq:E1})} \label{sec:bd_E_1}
Term $E_1$ is given as: 
\begin{align*}
    E_1 = \int_{\mathbb{R}^{d_R}}
     E_{1,1} (x, x_{M-1} ; \theta) 
     E_{1,2} (x_{M-1}, y; \theta)  d x_{M-1},
\end{align*}
where 
\begin{align*}
   E_{1,1} (x,  x_{M-1} ; \theta)
   & = \int_{\mathbb{R}^{d_R \times (M-2)}}  
   \prod_{i=1}^{M-1}
   \bar{p}^{\, \mathrm{I}}_{e, \delta_M} (x_{i-1}, x_i; \theta) dx_1 \cdots dx_{M-2};  \\
   E_{1,2} (x_{M-1}, y; \theta)
   & = \mathscr{R}_{\delta_M}  (x_{M-1}, y ; \theta). 
\end{align*}
From now on, we will derive upper bounds for terms $E_{1,1} (x, x_{M-1}; \theta)$ 
and 
$E_{1,2} (x_{M-1}, y; \theta)$. From \cite{iguchi:21-2}, we have:
\begin{align} 
\label{eq:bd_E11}
    \left| E_{1,1} (x,  x_{M-1} ; \theta) \right| 
    \leq  C \frac{h (\Delta)}{ \left( (M-1)\delta_M \right)^{d_R / 2} } 
    e^{-c  \frac{|x_{M-1} - x|^2}{(M-1) \delta_M}  }
\end{align}
for a non-decreasing function $h (\cdot)$ and constants $C, c > 0$ independent of $x, x_{M-1} \in \mathbb{R}^{d_R}$ and $\theta \in \Theta$. 
We consider the term $E_{1,2} (x_{M-1}, y; \theta)$. First, we note that there exists a random variable $\pi_{e} (t, x, B_t ; \theta)$, $x \in \mathbb{R}^{d_R}$, $t > 0$, $\theta \in \Theta$ such that for any appropriate test functions $f: \mathbb{R}^{d_R} \to \mathbb{R}$,
\begin{align*}
    \int_{\mathbb{R}^{d_R}} f (y) \bar{p}_{e, t}^{\, \mathrm{I}} (x, y ; \theta) dy 
    & = 
    \int_{\mathbb{R}^{d_R}} f (y) {p}_{e,t}^{\widetilde{X}^{\mathrm{EM}}} (x, y ; \theta)  \bigl\{ 1 + \Psi_e^{\mrm{weak}} (t, x, y ; \theta)  \bigr\} dy \nonumber \\
    & = \mathbb{E}_\theta \bigl[ f (\widetilde{X}_{e,t}^{\mathrm{EM},x})
     \bigl\{ 1 + \pi_{e} (t, x, B_t ; \theta) \bigr\}  \bigr], 
\end{align*}
where $\widetilde{X}_{e,t}^{\mathrm{EM},x}$ is the one-step Euler-Maruyama scheme with step-size $t$ starting from $x \in \mathbb{R}^{d_R}$. 
The stochastic weight $\pi_{e} (t, x, B_t ; \theta)$ is interpreted as the infinite-dimensional version of the weight function 
$y \mapsto \Psi_e^{\mrm{weak}} (t, x, y; \theta)$ and is given as: 
\begin{align} \label{eq:stochastic_weight} 
    \pi_{e} (t, x, B_t ; \theta) 
    & = \Psi_e^{\mrm{weak}} 
    (t, x, \widetilde{X}^{\mathrm{EM},x}_{e,t} ; \theta)  \\ 
    &\,\,= 
    \bigg\{\, \tfrac{\sqrt{t}}{2}  \sum_{i_1, i_2, i_3 =1}^{d_R} \sum_{k_1, k_2 = 1}^{d_R} \hat{V}_{k_1} V_{k_2}^{i_1} (x, \theta) V_{k_1}^{i_2}(x, \theta) V_{k_2}^{i_3}(x, \theta) \mathcal{H}_{(i_1,i_2,i_3)}^{\widetilde{Y}_{e}} (\xi ; \theta) \nonumber \\[0.2cm]
    & \quad +   \tfrac{t}{2} \sum_{i_1, i_2 =1}^{d_R} \sum_{k = 1}^{d_R} \Bigl( \hat{V}_{k} V_{0}^{i_1} (x,  \theta)
    + \hat{V}_{0} V_{k}^{i_1} (x, \theta) \Bigr) V_k^{i_2}(x, \theta)   \mathcal{H}^{\widetilde{Y}_{e}}_{(i_1,i_2)} (\xi; \theta)  \nonumber \\
    & \quad  +  \tfrac{t}{4} \sum_{i_1, i_2 =1}^{d_R} \sum_{k_1, k_2 = 1}^{d_R}
    \hat{V}_{k_1} V_{k_2}^{i_1}(x, \theta) \hat{V}_{k_1} V_{k_2}^{i_2}(x, \theta)  \mathcal{H}^{\widetilde{Y}_{e}}_{(i_1,i_2)}(\xi ; \theta)  \nonumber \\[0.2cm]   
    & \quad + \tfrac{\sqrt{t^3} }{2}  \sum_{i = 1}^{d_R} 
     \hat{V}_{0} V_{0}^{i}(x , \theta)   \mathcal{H}^{\widetilde{Y}_{e}}_{(i)}(\xi ; \theta)  
     \, \bigg\} 
    \bigg|_{\xi = \tfrac{1}{\sqrt{t}} 
     \sum_{k = 1}^{d_R} V_k (x ,  \theta) B_{k, t}},  \nonumber 
\end{align}
where we used: 
\begin{align*}
    m_{x, \theta, \sqrt{t}, 1} (\widetilde{X}_{e, t}^{\mathrm{EM},x})
    = \frac{\widetilde{X}_{e, t}^{\mathrm{EM},x} - x - V_0 (x ,  \theta) t}{\sqrt{t}} 
    = \tfrac{1}{\sqrt{t}} 
     \sum_{k = 1}^{d_R} V_k (x,  \theta) B_{k, t}. 
\end{align*}
%
%
%
We define $\Gamma : \mathbb{R} \to \mathbb{R}$ as 
$\Gamma (z) \equiv \exp \left( K (z) \right) - (1 + z)$. Then, 
$\mathscr{R}_{\delta_M} (x, y ; \theta)$, $x, y \in \mathbb{R}^{d_R}$, $\theta \in \Theta$, is given as: 
\begin{align*}
    \mathscr{R}_{\delta_M} (x, y ; \theta) 
    & = p_{e, \delta_M}^{\widetilde{X}^{\mathrm{EM}}} (x, y ; \theta) 
     \cdot \Gamma 
     \bigl( 
     \Psi_e^{\mrm{weak}} (\delta_M, x, y ;  \theta) 
     \bigr) \nonumber \\ 
    & = p_{e, \delta_M}^{\widetilde{X}^{\mathrm{EM}}} (x, y ; \theta)
     \cdot \mathbb{E}_{\theta} 
     \left[ 
     \Gamma \bigl( 
     \pi_{e} (\delta_M, x,  B_{\delta_M}  ;  \theta) 
     \bigr)
     | \widetilde{X}_{e, \delta_M}^{\mathrm{EM}, x} = y
     \right] \nonumber \\ 
    & = \mathbb{E}_\theta 
    \left[ \delta_y  (\widetilde{X}_{e, \delta_M}^{\mathrm{EM},x}) 
    \, \Gamma  
      \bigl( \pi_{e} (\delta_M, x,  B_{\delta_M}  ;  \theta)   \bigr) 
    \right],   
\end{align*}
where we have used (\ref{eq:stochastic_weight}) in the second line and (\ref{eq:conditional_density}) in the last line.  
Using the Malliavin integration by parts formula and H\"older's inequality, we obtain  
\begin{align*}
    | \mathscr{R}_{\delta_M} (x, y ; \theta) |
    & = \left| 
    \mathbb{E}_{\theta} 
    \bigl[
    \mathbf{1}_{\widetilde{X}_{e, \delta_M}^{\mathrm{EM}, x} \geq y} 
    H_{(1, \ldots, d_R)}  \bigl( \widetilde{X}_{e, \delta_M}^{\mathrm{EM}, x}, 
        \Gamma \bigl( 
        \pi_{e} (\delta_M, x,  B_{\delta_M}  ;  \theta) 
        \bigr) 
    \bigr]  
    \right| \\ 
    & \leq  
    \bigl\|  \mathbf{1}_{\widetilde{X}_{e, \delta_M}^{\mathrm{EM}, x} \geq y} \bigr\|_2  
    \bigl\| 
    H_{(1, \ldots, d_R)}  
    \bigl(
        \widetilde{X}_{e, \delta_M}^{\mathrm{EM}, x},  
        \Gamma \bigl(  
        \pi_{e} (\delta_M, x,  B_{\delta_M}  ;  \theta) 
        \bigr) 
    \bigr\|_2. 
\end{align*}
%
%
%
%
%
%
%
From Lemma F.3 in \cite{iguchi:21-2}, we have 
\begin{align*}  
    \bigl\|  \mathbf{1}_{\widetilde{X}_{e, \delta_M}^{\mathrm{EM}, x} \geq y} \bigr\|_2    
   \leq h (\Delta) e^{- c \tfrac{|y-x|^2}{\delta_M}},
\end{align*}   
where $h(\cdot)$ is a non-decreasing function and $c > 0$ is some constant.  Furthermore, the upper bound (\ref{eq:estimate_weight}) with Kusuoka-Stroock's estimate (Theorem 1.20 in \cite{kusu:84}) yields 
\begin{align*} 
    \bigl\|  H_{(1, \ldots, d_R)}  
    \bigl( \widetilde{X}_{e, \delta_M}^{\mathrm{EM}, x},   
    \Gamma \bigl(  
    \pi_{e} (\delta_M, x,  B_{\delta_M}  ;  \theta) 
    \bigr)
    \bigr\|_2
    & \leq \frac{C_1}{(\delta_M)^{d_R/2}} 
    \| \Gamma \bigl( 
    \pi_{e} (\delta_M, x,  B_{\delta_M}  ;  \theta) 
    \bigr)  \|_{d_R, 4} \nonumber \\   
    & \leq 
    \frac{C_2}{(\delta_M)^{d_R/2}} (\delta_M)^3
\end{align*}
for some constants $C_1 = C(\Delta), C_2 = C_2(\Delta) > 0$ that are independent of $(x, \theta) \in \mathbb{R}^{d_R} \times \Theta$ under conditions (\ref{assump:param_space})--(\ref{assump:coeff}). In the last inequality, we used the following bound: 
\begin{align} \label{eq:bd_weight_diff}
    \| \Gamma \bigl( 
    \pi_{e} (\delta_M, x,  B_{\delta_M}  ;  \theta) 
    \bigr)  \|_{d_R, 4}
    \le C (\delta_M)^3,
\end{align}
for some constant $C > 0$ independent of $x$ and $\theta$. 
(\ref{eq:bd_weight_diff}) is immediately obtained from the Taylor expansion (\ref{eq:Taylor_exp}) and the following estimate: Let $j \in \mathbb{N}$, $p = 2k, \; k \in \mathbb{N}$. Under conditions (\ref{assump:param_space})--(\ref{assump:coeff}), there exists a constant $C> 0$ such that for any $(x, \theta) \in \mathbb{R}^{d_R} \times \Theta$ and integer $i \geq 7$, 
\begin{align} \label{eq:bd_poly_weight}
  \Bigl\| 
  \bigl\{ 
  \pi_{e} (\delta_M, x,  B_{\delta_M}  ;  \theta) 
  \bigr\}^i
  \Bigr\|_{j,p} \leq C (\delta_M)^3.
\end{align}  
Notice that the estimate (\ref{eq:bd_poly_weight}) is immediately obtained by recursive application of H\"older's inequality with the following estimate. There exists constant $C >0$ independent of $(x, \theta) \in \mathbb{R}^{d_R} \times \Theta$ such that 
\begin{align*}
    \bigl\| 
    \pi_{e} (\delta_M, x,  B_{\delta_M}  ;  \theta)
    \bigr\|_{k,q} 
    \leq C  \sqrt{\delta_M}, 
\end{align*}
for any $k \in \mathbb{N}$, $q \in (1, \infty)$. 
Hence, we have
\begin{align} \label{eq:bd_R}
    \left| \mathscr{R}_{\delta_M} (x_{M-1}, y; \theta)  \right| 
    \leq C (\delta_M)^3 \frac{h(\Delta)}{(\delta_M)^{d_R / 2}} 
     e^{ -c \tfrac{|y - x_{M-1}|^2}{\delta_M} },
\end{align}
for a non-decreasing function $h(\cdot)$ and constants $C, c > 0$ independent of $x_{M-1}, y \in \mathbb{R}^{d_R}$ and $\theta \in \Theta$. 
Thus, from estimates (\ref{eq:bd_E11}), (\ref{eq:bd_R}), we obtain the upper bound for $E_1$ as: 
\begin{align*}
    \left|E_1 \right| 
    & \leq C_1 (\delta_M)^3 \frac{h (\Delta)}{\left( (M-1) \delta_M \right)^{d_R/2} (\delta_M)^{d_R/2} } \int_{\mathbb{R}^{d_R}} 
     e^{ -c_1 \tfrac{|x_{M-1} - x|^2}{(M-1)\delta_M}}
     e^{ -c_2 \tfrac{|y - x_{M-1}|^2}{\delta_M} }
     d x_{M-1}  \\[0.1cm]
    & \leq C_2 (\delta_M)^3 \frac{h (\Delta)}{\Delta^{d_R/2}}  
     e^{ -c_3 \tfrac{|y - x|^2}{\Delta} }, 
\end{align*}
where we have made use of the following result. There exist constants $C, c > 0$ such that for any constants $c_1, c_2 > 0$ and $0 < r < s$, $x, y \in \mathbb{R}^{d_R}$,  
\begin{align} \label{eq:conv_gaussian}
 \frac{1}{r^{{d_R}/2}} 
 \int_{\mathbb{R}^{d_R}} 
e^{ - c_1 \frac{|z-x|^2}{r} }
 e^{ - c_2 \frac{|y-z|^2}{s-r} } dz  
 \leq C \left( \frac{s-r}{s} \right)^{d_R/2} e^{ - c \frac{|y-x|^2}{s}}.
\end{align} 
The proof is now complete. 
\subsubsection{Upper Bound (\ref{eq:E2})}  \label{sec:bd_E_2}
\noindent Term $E_2$ is written as: 
\begin{align*}
    E_2 
    = \sum_{k=1}^{M-1} 
    \int_{\mathbb{R}^{d_R}}  F_{1,k} (x, x_{M-k-1} ; \theta)  
    F_{2,k} (x_{M-k-1}, y ; \theta)
    d x_{M-k-1}, 
\end{align*}
where 
\begin{align*}
 & F_{1,k} (x, x_{M-k-1} ; \theta)  
   :=  
   \int_{\mathbb{R}^{d_R \times (M- k -1)}} 
   \prod_{i=1}^{M-k-1}   
   \bar{p}^{\, \mathrm{I}}_{e, \delta_M} (x_{i-1}, x_{i};  \theta) 
   d x_1 \cdots d x_{M-k-2}; \\ 
 & F_{2,k} (x_{M-k-1}, y; \theta) 
 := \int_{\mathbb{R}^{d_R}} 
   \mathscr{R}_{\delta_M} (x_{M-k-1}, x_{M-k} ;\theta)
   \widetilde{F}_{2,k} (x_{M - k}, y ; \theta) d x_{M-k}, 
\end{align*}
for $1 \le j \le M-1$, where we have set: 
\begin{align} \label{eq:bd_F1}
     \widetilde{F}_{2,k} (x_{M - k}, y ; \theta) 
     = \int_{\mathbb{R}^{d_R \times (k-1) }} 
     \prod_{i=1}^k \bar{p}^{\, \mathrm{II}}_{e, \delta_M} (x_{M- k + (i-1)}, x_{ M - k + i };  \theta)  d x_{M- k + 1} \cdots 
     d x_{M-1},
\end{align}
under the notation $x_{M} = y \in \mathbb{R}^{d_R}$. From \cite{iguchi:21-2}, we  have
\begin{align*} 
 \left| F_{1, k} (x, x_{M-k-1}; \theta)  \right|  
 \leq C \frac{h (\Delta)}{((M - k- 1) \delta_M)^{d_R/2}}
     e^{ - c \tfrac{|x_{M-k-1} - x|^2 }{ (M - k- 1) \delta_M} },  
\end{align*}
for some non-decreasing function $h (\cdot)$ and constants $C, c > 0$ independent of $(x, \theta) \in \mathbb{R}^{d_R} \times \Theta$ under conditions (\ref{assump:param_space})--(\ref{assump:coeff}). 
To obtain an upper bound for $\widetilde{F}_{2,k} (x_{M-k-1}, y; \theta)$, we will express it as a generalised expectation. First, from the Markov property of Euler-Maruyama scheme and (\ref{eq:stochastic_weight}) we have for a measurable function $f : \mathbb{R}^{d_R} \to \mathbb{R}$, 
\begin{align} \label{eq:markov}
    & \int_{\mathbb{R}^{d_R}} \int_{\mathbb{R}^{d_R}} 
    f (y) \bar{p}^{\, \mathrm{II}}_{e, \delta_M} (x, \xi ; \theta) 
    \bar{p}^{\, \mathrm{II}}_{e, \delta_M} (\xi, y ; \theta)  d \xi dy \nonumber  \\ 
    & = \int_{\mathbb{R}^{d_R}} 
     \mathbb{E}_\theta 
     \bigl[ f (\widetilde{X}^{\mathrm{EM},\xi}_{e, \delta_M}) 
      \, 
      e^{ K \bigl( \Psi_e^{\mrm{weak}} (\delta_M, \xi, \widetilde{X}^{\mathrm{EM},\xi}_{e, \delta_M}; \theta ) \bigr)}
     \bigr]  
    \, \bar{p}^{\, \mathrm{II}}_{e, \delta_M} (x, \xi ; \theta) 
    \, d\xi  \nonumber \\
    & = \mathbb{E}_\theta 
    \Bigl[
     \mathbb{E}_\theta 
     \bigl[ f (\widetilde{X}^{\mathrm{EM},\xi}_{e, \delta_M}) 
      \, 
      e^{ K \bigl( \Psi_e^{\mrm{weak}} (\delta_M, \xi, \widetilde{X}^{\mathrm{EM},\xi}_{e, \delta_M}; \theta ) \bigr)}
     \bigr] |_{\xi = \widetilde{X}^{\mathrm{EM},x}_{e, \delta_M} } 
     \, 
     \times 
     e^{ K \bigl( \Psi_e^{\mrm{weak}} (\delta_M, x, 
     \widetilde{X}^{\mathrm{EM},x}_{e, \delta_M}; \theta ) \bigr) }
    \Bigr] \nonumber  \\ 
    & = \mathbb{E}_\theta 
    \Bigl[
     \mathbb{E}_\theta 
     \bigl[ f (\widetilde{X}^{\mathrm{EM},(M), x}_{e, 2 \delta_M}) 
      \, 
     e^{ K \bigl( \Psi_e^{\mrm{weak}} (\delta_M, \widetilde{X}^{\mathrm{EM}, x}_{e, \delta_M}, \widetilde{X}^{\mathrm{EM},(M), x}_{e, 2 \delta_M}; \theta) \bigr) } 
    | \mathcal{F}_{\delta_M} \bigr] 
    \, 
    \times 
     e^{ K \bigl( \Psi_e^{\mrm{weak}} (\delta_M, x , \widetilde{X}^{\mathrm{EM},x}_{e, \delta_M}; \theta ) \bigr) }
    \Bigr] \nonumber \\
    & = \mathbb{E}_\theta 
    \Bigl[
     f (\widetilde{X}^{\mathrm{EM},(M), x}_{e, 2 \delta_M}) 
      \, 
     e^{ K \bigl( 
     \pi_{e}
     ( \delta_M, \widetilde{X}^{\mathrm{EM}, x}_{e, \delta_M}, B_{ \delta_M}^{[2]} ; \theta) 
     \bigr) } 
     e^{ K \bigl( 
     \pi_{e}
     ( \delta_M, x, B_{ \delta_M}^{[1]} ; \theta) 
     \bigr) }
    \Bigr], 
\end{align}
where $\widetilde{X}_{e, k \delta_M}^{\mathrm{EM}, (M), x}$, $0 \le k \le M$,
is the Markov chain of Euler-Maruyama scheme at time $k \delta_M$, starting from point $x \in \mathbb{R}^{d_R}$ with time step of length $\delta_M$, and $B^{[l]}_{\delta_M} := B_{l \delta_M}  - B_{(l-1) \delta_M}$, $1 \le l \le M$.
Since it follows from (\ref{eq:expectaion_rep}) and (\ref{eq:stochastic_weight}) that
\begin{align} \label{eq:density2_ge}
    \bar{p}_{e, \delta_M}^{\, \mathrm{II}} (x_{M-1}, y ; \theta ) 
    = \mathbb{E}_\theta 
    \Bigl[ 
    \delta_y (\widetilde{X}_{e, \delta_M}^{\mathrm{EM}, x_{M-1}})
    \, e^{K \bigl( 
     \Psi_e^{\mrm{weak}} (\delta_M, x_{M-1}, \widetilde{X}_{e, \delta_M}^{\mathrm{EM}, x_{M-1}}; \theta)  \bigr) }  
    \Bigr],
\end{align}
(\ref{eq:density2_ge}) and recursive application of (\ref{eq:markov}) yield:
\begin{align*}
    \widetilde{F}_{2,k} (x_{M-k}, y; \theta) 
    & = 
    \mathbb{E}_\theta 
    \Bigl[
    \delta_y ( \widetilde{X}_{e, k \delta_M}^{\mathrm{EM}, (M), x_{M-k}} )  
    \,\prod_{l=1}^{k} 
    e^{  K \bigl( 
      \pi_{e}
      \bigl(\delta_M, \widetilde{X}_{e, (l-1) \delta_M}^{\mathrm{EM}, (M), x_{M-k}}, 
         B^{[M-k + l]}_{\delta_M} ; \theta 
      \bigr)  
    \bigr) }
    \Bigr]    \\ 
    & \equiv 
    \mathbb{E}_\theta 
    \left[
    \delta_y ( \widetilde{X}_{e, k \delta_M}^{\mathrm{EM}, (M), x_{M-k}} )  
    \, G^k_{\delta_M, x_{M-k}, \theta}  
    \right]. 
\end{align*}
Use of the Malliavin integration by parts formula and H\"older's inequality gives: %
\begin{align*}
   | \widetilde{F}_{2, k} (x_{M-k}, y; \theta) |  
   & = \Bigl| \mathbb{E}_\theta [
   \mathbf{1}_{ \widetilde{X}^{\mathrm{EM}, (M), x_{M-k}}_{e, k \delta_M}  \geq y}  
   H_{(1, \ldots, d_R)} 
      \bigl(
      \widetilde{X}^{\mathrm{EM}, (M), x_{M-k}}_{e, k \delta_M}, G^{k}_{\delta_M, x_{M-k}; \theta}
      \bigr)
   ] 
   \Bigr|   \nonumber   \\ 
   & \leq  
   \bigl\| \mathbf{1}_{ \widetilde{X}^{\mathrm{EM}, (M), x_{M-k}}_{e, k \delta_M}  \geq y } \bigr\|_2
   \left\| 
   H_{(1, \ldots, d_R)} \bigl( \widetilde{X}^{\mathrm{EM}, (M), x_{M-k}}_{e, k \delta_M}, G^{k}_{\delta_M, x_{M-k}; \theta}  \bigr) 
   \right\|_2. 
\end{align*} 
By Lemma F.1 in \cite{iguchi:21-2} again, we have 
\begin{align*}
      \bigl\| 
      \mathbf{1}_{ \widetilde{X}^{\mathrm{EM}, (M), x_{M-k}}_{e, k \delta_M}  \geq y  } 
      \bigr\|_2 
      \leq   h (\Delta)  e^{-c \frac{y - x_{M-k}}{k \delta_M}},
\end{align*} 
for a non-decreasing function $h(\cdot)$ and a constant $c>0$ that are independent of $x_{M-k}, y \in \mathbb{R}^{d_R}$ and $\theta \in \Theta$ under conditions (\ref{assump:param_space})--(\ref{assump:coeff}). The upper bound (\ref{eq:estimate_weight}) with Kusuoka-Stroock's estimate (Theorem 1.20 in \cite{kusu:84}) yields that there exist constants $C_1 , C_2, c > 0$ independent of $x_{M-k}, y \in \mathbb{R}^{d_R}$ and $\theta \in \Theta$ such that 
\begin{align*}
 \left\| 
 H_{(1, \ldots, d_R)}
 \bigl( 
   \widetilde{X}^{\mathrm{EM}, (M), x_{M-k}}_{e, k \delta_M} 
   ,  G^{k}_{\delta_M, x_{M-k}, \theta}
 \bigr) 
 \right\|_2 
 & \leq  \tfrac{C_1}{\bigl(k\delta_M \bigr)^{d_R/2}} 
  \left\|  G^{k}_{\delta_M, x_{M-k}, \theta}      \right\|_{d_R, 4} \nonumber \\
 & \leq  \tfrac{C_2}{\bigl(k\delta_M \bigr)^{d_R/2}},
\end{align*}
where on the last inequality we have made use of the estimate derived as follows. Let $j, k \in \mathbb{N}$. Under conditions (\ref{assump:param_space})--(\ref{assump:coeff}), there exists a constant $C = C (\Delta) > 0$ independent of $x_{M-k}$ and $\theta$ such that 
\begin{align} \label{eq:bd_exp_weight}
    \| G^k_{\delta_M, x_{M-k}, \theta} \|_{j, 2k} \leq C.
\end{align} 
We will give the proof of (\ref{eq:bd_exp_weight}) in the end of this subsection. Thus, the upper bound of $\widetilde{F}_{2,k} (x_{M-k}, y ; \theta)$ is given as:
\begin{align*}
    | \widetilde{F}_{2,k} (x_{M-k}, y ; \theta) | 
    \leq C \frac{h (\Delta)}{(k \delta_M)^{d_R/2}}
     e^{-c \tfrac{| y - x_{M-k}|^2}{k \delta_M}},
\end{align*}
for a non-decreasing constant $h (\cdot)$ and some constants $C, c > 0$ independent of $x_{M-k}, y$ and $\theta$.  Combining the estimate of $\widetilde{F}_{2,k} (x_{M-k}, y ; \theta)$ with that of (\ref{eq:bd_R}), we obtain:  
\begin{align} \label{eq:bd_F2}
    & |F_{2,k} (x_{M-k-1}, y; \theta) | \nonumber  \\  
    &\leq  C_1 (\delta_M)^3 \frac{h (\Delta)}{(\delta_M)^{d_R/2} (k \delta_M)^{d_R/2}}\int_{\mathbb{R}^{d_R}} 
     e^{- c_1 \tfrac{| x_{M-k} - x_{M - k- 1} |^2}{\delta_M} }
     e^{- c_2 \tfrac{| y  - x_{M - k} |^2}{ k \delta_M} }
     d x_{M-k}  \nonumber  \\
    & \leq 
    C_2 (\delta_M)^3 \frac{h (\Delta)}{( ( k + 1) \delta_M)^{d_R /2} }
      e^{- c_3 \tfrac{| y - x_{M - k- 1} |^2}{ (k + 1) \delta_M} }, 
\end{align}
for some positive constants $C_1, C_2, c_1, c_2, c_3$ independent of $x_{M-k-1} \in \mathbb{R}^{d_R}$, where we used estimate (\ref{eq:conv_gaussian}) on the last inequality. Finally, thanks to (\ref{eq:conv_gaussian}) and the upper bounds (\ref{eq:bd_F1}), (\ref{eq:bd_F2}), we have 
\begin{align*}
    |E_2| 
    & \leq \sum_{k = 1}^{M-1} 
      \int_{\mathbb{R}^{d_R}} 
      | F_{1, k} (x_{M-k-1} ; \theta) | 
      | F_{2, k} (x_{M-k-1} ; \theta) | d x_{M-k-1} \\ 
    & \leq  C (M-1) (\delta_M)^3 \frac{h (\Delta)}{ \Delta^{d_R/2} }
      e^{- c  \tfrac{| y - x  |^2}{ \Delta }}, 
\end{align*}
for a non-decreasing function $h (\Delta)$ and some constants $C, c > 0$ independent of $x, y \in \mathbb{R}^{d_R}$ and $\theta \in \Theta$. The proof is now complete. 
\\

\noindent
({\it Proof of} (\ref{eq:bd_exp_weight})).
Recall that $\Gamma (z) := e^{K(z)} - (1 + z), \, z \in \mathbb{R}$, with Taylor expansion (\ref{eq:Taylor_exp}). We have
\begin{align*}
    & G^k_{\delta_M, x_{M-k}, \theta}
     = \prod_{l=1}^k 
    e^{
    K \bigl( 
      \pi_{e}
      ( 
        \delta_M, \widetilde{X}_{e, (l-1) \delta_M}^{\mathrm{EM}, (M), x_{M-k}},  B^{[M-k + l]}_{\delta_M} ; \theta 
      )  
    \bigr)  
    }  \\ 
    & = \prod_{l=1}^k 
     \Bigl\{ 1 
     + \pi_{e}
      ( 
        \delta_M, \widetilde{X}_{e, (l-1) \delta_M}^{\mathrm{EM}, (M), x_{M-k}},  B^{[M-k + l]}_{\delta_M} ; \theta 
      )
     + \Gamma \bigl(  
     \pi_{e}
      ( 
    \delta_M, \widetilde{X}_{e, (l-1) \delta_M}^{\mathrm{EM}, (M), x_{M-k}},  B^{[M-k + l]}_{\delta_M} ; \theta 
      ) 
      \bigr)
     \Bigr\} \\
     & \equiv \prod_{l=1}^k g^l_{\delta_M, x_{M-k}, \theta}. 
\end{align*}
Similarly, we define
\begin{align*}
    \widetilde{G}^k_{\delta_M, x_{M-k}, \theta} 
    & = \prod_{l=1}^k 
     \Bigl\{ 1 + 
     \pi_{e}
      ( 
       \delta_M, \widetilde{X}_{e, (l-1) \delta_M}^{\mathrm{EM}, (M), x_{M-k}},  B^{[M-k + l]}_{\delta_M} ; \theta 
      )
     \Bigr\} 
     \equiv \prod_{l=1}^k \widetilde{g}^l_{\delta_M, x_{M-k}, \theta}.
\end{align*}
Note that for any $j, e \in \mathbb{N}$
\begin{align} \label{eq:G_norm}
    \| {G}^k_{\delta_M, x_{M-k}, \theta}   \|_{j, 2e}
    \le \| \widetilde{G}^k_{\delta_M, x_{M-k}, \theta}   \|_{j, 2e} 
     + \| G^k_{\delta_M, x_{M-k}, \theta} 
    - \widetilde{G}^k_{\delta_M, x_{M-k}, \theta}   \|_{j, 2e}.
\end{align}
 Due to Lemma 2 in \cite{igu:22}, the first term of the right-hand side of (\ref{eq:G_norm}) 
is bounded by some positive constant $C = C(\Delta)$ independent of $(x_{M-k}, \theta) \in \mathbb{R}^{d_R} \times \Theta$ under conditions (\ref{assump:param_space})--(\ref{assump:coeff}). 
We consider the second term of the right hand side of (\ref{eq:G_norm}). Since it holds  
\begin{align*}
    &  G^k_{\delta_M, x_{M-k}, \theta} 
            - \widetilde{G}^k_{\delta_M, x_{M-k}, \theta}   \\ 
    & = \sum_{m = 1}^k  
      \biggl\{ \prod_{l=1}^{k-m} g^l_{\delta_M, x_{M-k}, \theta}  \biggr\}
      \, \Gamma \bigl( 
      \pi_{e }  
      (\delta_M, \widetilde{X}_{e, (k-m) \delta_M}^{\mathrm{EM}, (M), x_{M-k} }, B_{\delta_M}^{[M- m + 1]}   ; \theta)  
      \bigr) 
      \, 
      \biggl\{ \prod_{l=1}^{m-1} \widetilde{g}^{k - l +1}_{\delta_M, x_{M-k}, \theta}  \biggr\},
\end{align*}
we have
\begin{align*}
& \| G^k_{\delta_M, x_{M-k}, \theta} 
        - \widetilde{G}^k_{\delta_M, x_{M-k}, \theta}  \|_{j, 2e} \\ 
& \le \sum_{m = 1}^k 
\left\| \prod_{l=1}^{k-m} g^l_{\delta_M, x_{M-k}, \theta} \right\|_{j, p}
\bigl\| 
\Gamma \bigl( 
\pi_{e } (\delta_M, \widetilde{X}_{e, (k-m) \delta_M}^{\mathrm{EM}, (M), x_{M-k} }, B_{\delta_M}^{[M- m + 1]}   ; \theta)  
\bigr)
\bigr\|_{j, q}
\left\| \prod_{l=1}^{m-1} \widetilde{g}^{k - l +1}_{\delta_M, x_{M-k}, \theta} \right\|_{j, r}, 
\end{align*}
for some positive even integers $p,q,r$, where we used H\"older's inequality.   
Under conditions (\ref{assump:param_space})--(\ref{assump:coeff}), Lemma 2 in \cite{igu:22} and the uniformly boundedness property of the function $z \mapsto \exp (K(z))$ yield 
\begin{align*}
 \left\| \prod_{l=1}^{k-m} g^l_{\delta_M, x_{M-k}, \theta} \right\|_{j, p}
 \le C_1,   \quad 
 \left\| \prod_{l=1}^{m-1} \widetilde{g}^{k - l +1}_{\delta_M, x_{M-k}, \theta} \right\|_{j, r}
 \le C_2, \ \ 1 \le m \le k, 
\end{align*}
for some  constants $C_1, C_2 > 0$ independent of $x_{M-k}$ and $\theta$. 
Furthermore, applying a similar discussion as the one for obtaining the bound (\ref{eq:bd_weight_diff}), we have: 
\begin{align*}  
\Bigl\|
\Gamma  
\bigl( 
\pi_{e }  
(\delta_M, \widetilde{X}_{e, (k-m) \delta_M}^{\mathrm{EM}, (M), x_{M-k} }, B_{\delta_M}^{[M- m + 1]}   ; \theta)  
\bigr)  
\Bigr\|_{j, q} \le C (\delta_M)^3, 
 \end{align*} 
for some constant $C > 0$ independent of $x_{M-k}$ and $\theta$ under conditions (\ref{assump:param_space})--(\ref{assump:coeff}). 
The proof is now complete. 
\section{Additional Simulation Study Under High Frequency Observations Regime} \label{sec:fn} 
We consider the stochastic FitzHugh-Nagumo model used to describe the dynamics of action potential generation within an neuronal axon. The model is specified as a $2$-dimensional hypo-elliptic SDE driven by a scalar Brownian motion ($d = 2$, $d_S = d_R =1$):
\begin{align}
  \begin{aligned} \label{eq:FN}
    d X_t & = 
    \left[ \begin{array}{c}
      d X_{R,t} \\[0.1cm]
      d X_{S,t}
    \end{array} \right]
     = 
     \left[  \begin{array}{c}
        \gamma X_{S,t} - X_{R,t} + \alpha \\[0.1cm]
        \tfrac{1}{\varepsilon}
      \bigl( X_{S,t} - (X_{S,t})^3 - X_{R,t} - s \bigr)   
     \end{array}
     \right] dt 
     +
     \left[  \begin{array}{c}
     \sigma  \\[0.1cm]
     0  
     \end{array}
     \right] d B_t , \\ 
  \end{aligned}
\end{align}
where  $\gamma, \alpha, \varepsilon, s, \sigma$ are the parameters. Typically, $s$ is fixed in experiments, thus we will estimate parameter $\theta = (\gamma, \alpha, \varepsilon, \sigma)$ under the complete observations regime by making use of new contrast estimator and compare it against the local Gaussian contrast estimator proposed by \cite{glot:20}. The $\Phi_2$ term in the new contrast (\ref{eq:contrast}) for the model (\ref{eq:FN}) is given as: for $\Delta > 0, y \in \mathbb{R}^2, x = [x_R, x_S ]^\top \in \mathbb{R}^2$,  
\begin{gather*}
\Phi_2 (\Delta, x, y ; \theta) 
= - \tfrac{\sigma^2}{2} \times \mathcal{H}_{(1,1)} (\Delta, x, y; \theta)
+ \Bigl\{ \tfrac{\sigma^2}{2 \varepsilon} 
- \tfrac{\sigma^2}{6 \varepsilon^2} \bigl(1 - 3 (x_S)^2 \bigr) \Bigr\}
\times \mathcal{H}_{(1, 2)} (\Delta, x, y; \theta)  \\
+ 
\Bigl\{ 
- \tfrac{\sigma^2}{8 \varepsilon^2} 
+ \tfrac{\sigma^2}{8 \varepsilon^3} \bigl(1 - 3 (x_S)^2 \bigr) \Bigr\}
\times \mathcal{H}_{(2, 2)} (\Delta, x, y; \theta),  
\end{gather*}
with the Hermite polynomials $\mathcal{H}_{(i_1, i_2)} (\Delta, x, y; \theta)\, 1 \le i_1, i_2 \le 2$  defined in (\ref{eq:hermite}) in the main text. 
In our experiment, we fix $s= 0.01$ and set the true  parameter to
$$
\trueparam = (\gamma^\dagger, \alpha^\dagger, \varepsilon^\dagger, \sigma) = (1.50, 0.30, 0.10, 0.60).
$$ 
We generate `correct' synthetic datasets $Y_{\mathrm{FN}}$ over the time period $[0,100]$ under $\trueparam$ by using the local Gaussian scheme (\ref{eq:LG}) given in the main text with very small discretisation step $10^{-4}$. We will consider the following three scenarios for the number of datasets $n$ and the observation time step $\Delta_n$, by subsampling from the synthetic datasets $Y_{\mathrm{FN}}$: 
\begin{description}
\item[FN-1.] $(n, \Delta_n) = (5000, 0.02)$ with the time interval of observations $T \equiv n \Delta_n = 100$. 
\item[FN-2.] $(n, \Delta_n) = (10000, 0.01)$ with $T  = 100$.
\item[FN-3.] $(n, \Delta_n) = (20000, 0.005)$ with $T  = 100$.
\end{description}
We use the adaptive moments (Adam) optimiser to minimise the contrast estimators with the following algorithmic specifications:
(step-size) = $0.01$, (exponential decay rate for the first moment estimates) = $0.9$, (exponential decay rate for the second moment estimates) = $0.999$, (additive term for numerical stability) = $1 \times 10^{-8}$ and (number of iterations) = $20,000$. 
In Table \ref{table:est_fn}, we summarise the mean and standard deviation of the two contrast estimators as computed by $20$ trajectories under experimental design $\textbf{FN-3}$. 
For the drift parameters, the local Gaussian and the new contrast estimators provide similar accuracy, and give a very accurate estimation for $\varepsilon$, i.e.~the parameter at the drift of the smooth component, as also anticipated by the rapid convergence rate $\textstyle{\Delta_n^{1/2}{n}^{-1/2}}$ towards normality in Theorem~\ref{thm:asymp_norm}. 
Estimation of the diffusion parameter $\sigma$ is improved via use of the new contrast. We focus on the estimation  of $\sigma$. In Figure \ref{fig:rmes_fn}, we plot the root mean squared errors (RMSEs) for the two contrast estimators for $\sigma$ computed by $20$ replicates under the settings $\textbf{FN-1,2,3}$. The figure illustrates that the new estimator gives faster convergence to the true parameter than the local Gaussian contrast, which (arguably) relates to the weaker requirement of $\Delta_n=o(n^{-1/3})$ versus $\Delta_n=o(n^{-1/2})$ to obtain asymptotic normality for the estimators provided by the analytical results in this work -- see Theorem \ref{thm:asymp_norm} and Remark \ref{rem:vs_LocalGaussian} in the main text. 
\begin{table}
 \caption{\small Mean and standard deviation (in bracket) from $20$ parameter estimator replicates in the stochastic FitzHugh-Nagumo model (\ref{eq:FN}) under  scenario \textbf{FN-3}.}
 \label{table:est_fn}
 \centering 
  \begin{tabular}{cccc}
  \hline  
  Parameter & True value & {Local Gaussian contrast} & {New contrast} \\  
  \hline 
  $\gamma$ & 1.50 & 1.5085 (0.0711) 
  & 1.5040 (0.0717)  \\ 
  $\alpha$ & 0.30 & 0.3163 (0.0741) 
  & 0.3158 (0.0749) \\ 
  $\varepsilon$ & 0.10 & 0.1001 (0.0000) & 0.1000 (0.0001) \\
  $\sigma $ & 0.60 & 0.5918 (0.0012) & 0.6001 (0.0017)  \\
  \hline
  \end{tabular}
\end{table}

\begin{figure}[H] 
\caption{\small Root mean squared errors (RMSEs) obtained from $20$ replications of estimates for $\sigma$ in the stochastic FitzHugh-Nagumo model (\ref{eq:FN}), under scenarios \textbf{FN-1,2,3}.}
\label{fig:rmes_fn} 
\centering
\includegraphics[keepaspectratio, width=4.4cm]{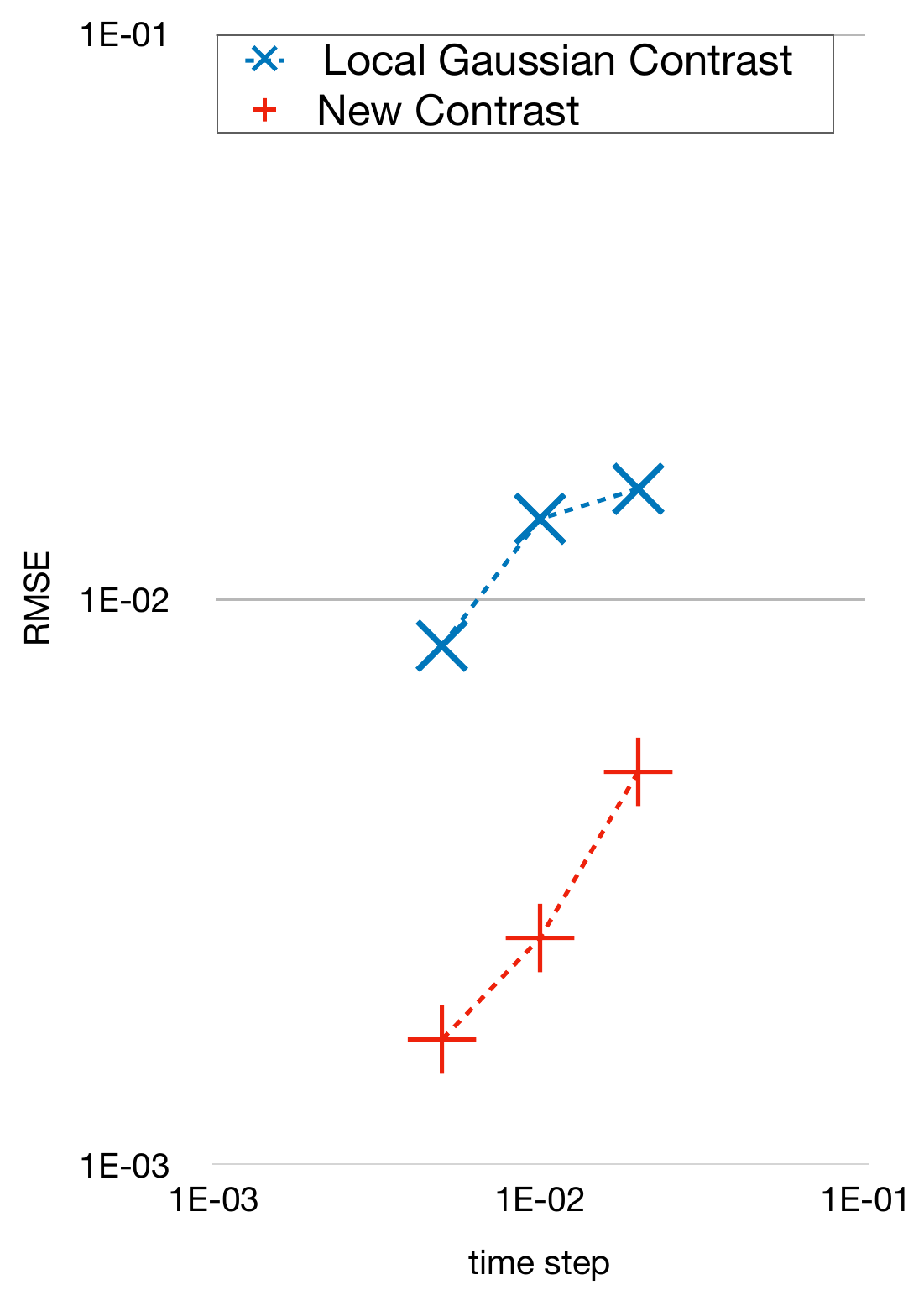}
\end{figure} 
\bibliographystyle{rss}
\bibliography{bernoulli} 

\begin{thebibliography}{37}
\expandafter\ifx\csname natexlab\endcsname\relax\def\natexlab#1{#1}\fi
\expandafter\ifx\csname url\endcsname\relax
  \def\url#1{\texttt{#1}}\fi
\expandafter\ifx\csname urlprefix\endcsname\relax\def\urlprefix{URL: }\fi

\bibitem[{Ableidinger et~al.(2017)Ableidinger, Buckwar and
  Hinterleitner}]{abl:17}
Ableidinger, M., Buckwar, E. and Hinterleitner, H. (2017) A stochastic version
  of the {J}ansen and {R}it neural mass model: {A}nalysis and numerics.
\newblock \textit{J. Math. Neurosci.}, \textbf{7}, 1--35.

\bibitem[{A{\"\i}t-Sahalia(2002)}]{ait:02}
A{\"\i}t-Sahalia, Y. (2002) Maximum likelihood estimation of discretely sampled
  diffusions: a closed-form approximation approach.
\newblock \textit{Econometrica}, \textbf{70}, 223--262.

\bibitem[{A{\"\i}t-Sahalia(2008)}]{ait:08}
--- (2008) Closed-form likelihood expansions for multivariate diffusions.
\newblock \textit{Ann. Statist.}, \textbf{36}, 906--937.

\bibitem[{Andrieu et~al.(2010)Andrieu, Doucet and Holenstein}]{andrieu:10}
Andrieu, C., Doucet, A. and Holenstein, R. (2010) Particle {M}arkov chain
  {M}onte {C}arlo methods.
\newblock \textit{J. R. Stat. Soc. Ser. B. Stat. Methodol.}, \textbf{72},
  269--342.

\bibitem[{Anonymous(1978)}]{anon:78}
Anonymous (1978) News and notes: influenza in a boarding school.
\newblock \textit{Br. Med. J.}, \textbf{1}, 586--590.

\bibitem[{Bally(2003)}]{bally:03}
Bally, V. (2003) An elementary introduction to {M}alliavin calculus. research
  report, rr-4718. inria.

\bibitem[{Beskos et~al.(2015)Beskos, Dureau and Kalogeropoulos}]{besk:15}
Beskos, A., Dureau, J. and Kalogeropoulos, K. (2015) Bayesian inference for
  partially observed stochastic differential equations driven by fractional
  {B}rownian motion.
\newblock \textit{Biometrika}, \textbf{102}, 809--827.

\bibitem[{Betancourt(2017)}]{beta:17}
Betancourt, M. (2017) A conceptual introduction to {H}amiltonian {M}onte
  {C}arlo.
\newblock \textit{arXiv preprint arXiv:1701.02434}.

\bibitem[{Buckwar et~al.(2020)Buckwar, Tamborrino and Tubikanec}]{buc:20}
Buckwar, E., Tamborrino, M. and Tubikanec, I. (2020) Spectral density-based and
  measure-preserving {ABC} for partially observed diffusion processes. an
  illustration on {H}amiltonian {SDE}s.
\newblock \textit{Stat. Comput.}, \textbf{30}, 627--648.

\bibitem[{Ditlevsen and Samson(2019)}]{dit:19}
Ditlevsen, S. and Samson, A. (2019) Hypoelliptic diffusions: filtering and
  inference from complete and partial observations.
\newblock \textit{J. R. Stat. Soc. Ser. B. Stat. Methodol.}, \textbf{81},
  361--384.

\bibitem[{Genon-Catalot and Jacod(1993)}]{genon:93}
Genon-Catalot, V. and Jacod, J. (1993) On the estimation of the diffusion
  coefficient for multi-dimensional diffusion processes.
\newblock In \textit{Annales de l'IHP Probabilit{\'e}s et statistiques},
  vol.~29, 119--151.

\bibitem[{Gloter and Yoshida(2020)}]{glot:20}
Gloter, A. and Yoshida, N. (2020) Adaptive and non-adaptive estimation for
  degenerate diffusion processes.
\newblock \textit{arXiv preprint arXiv:2002.10164}.

\bibitem[{Gloter and Yoshida(2021)}]{glot:21}
--- (2021) Adaptive estimation for degenerate diffusion processes.
\newblock \textit{Electron. J. Stat.}, \textbf{15}, 1424--1472.

\bibitem[{Gobet and Labart(2008)}]{gobet:08}
Gobet, E. and Labart, C. (2008) Sharp estimates for the convergence of the
  density of the {E}uler scheme in small time.
\newblock \textit{Electron. Commun. Probab.}, \textbf{13}, 352--363.

\bibitem[{Graham et~al.(2022)Graham, Thiery and Beskos}]{graham2022manifold}
Graham, M.~M., Thiery, A.~H. and Beskos, A. (2022) Manifold {M}arkov chain
  {M}onte {C}arlo methods for {B}ayesian inference in diffusion models.
\newblock \textit{J. R. Stat. Soc. Ser. B. Stat. Methodol.}, \textbf{84},
  1229--1256.

\bibitem[{Hall and Heyde(2014)}]{hall:14}
Hall, P. and Heyde, C.~C. (2014) \textit{Martingale Limit Theory and Its
  Application}.
\newblock Academic press.

\bibitem[{Hoffman et~al.(2014)Hoffman, Gelman et~al.}]{hoff:14}
Hoffman, M.~D., Gelman, A. et~al. (2014) The {N}o-{U}-{T}urn sampler:
  adaptively setting path lengths in hamiltonian monte carlo.
\newblock \textit{J. Mach. Learn. Res.}, \textbf{15}, 1593--1623.

\bibitem[{Iguchi and Yamada(2021{\natexlab{a}})}]{iguchi:21-2}
Iguchi, Y. and Yamada, T. (2021{\natexlab{a}}) Operator splitting around
  {E}uler-{M}aruyama scheme and high order discretization of heat kernels.
\newblock \textit{ESAIM: Math. Model. Numer. Anal.}, \textbf{55}, S323--S367.

\bibitem[{Iguchi and Yamada(2021{\natexlab{b}})}]{iguchi2021second}
--- (2021{\natexlab{b}}) A second-order discretization for degenerate systems
  of stochastic differential equations.
\newblock \textit{IMA J. Numer. Anal.}, \textbf{41}, 2782--2829.

\bibitem[{Iguchi and Yamada(2022)}]{igu:22}
--- (2022) Weak approximation of {SDE}s for tempered distributions and
  applications.
\newblock \textit{Adv. Comput. Math}, \textbf{48}, 269--342.

\bibitem[{Ikeda and Watanabe(2014)}]{ike:14}
Ikeda, N. and Watanabe, S. (2014) \textit{Stochastic {D}ifferential {E}quations
  and {D}iffusion {P}rocesses}.
\newblock Elsevier.

\bibitem[{Kessler(1997)}]{kessler1997estimation}
Kessler, M. (1997) Estimation of an ergodic diffusion from discrete
  observations.
\newblock \textit{Scand. J. Stat.}, \textbf{24}, 211--229.

\bibitem[{Kingma and Ba(2015)}]{kigma:15}
Kingma, D.~P. and Ba, J. (2015) Adam: {A} method for stochastic optimization.
\newblock In \textit{3rd International Conference on Learning Representations,
  {ICLR}} (eds. Y.~Bengio and Y.~LeCun).

\bibitem[{Kloeden and Platen(1992)}]{kloe:92}
Kloeden, P.~E. and Platen, E. (1992) \textit{Numerical {S}olution of
  {S}tochastic {D}ifferential {E}quations}.
\newblock Springer.

\bibitem[{Kusuoka and Stroock(1984)}]{kusu:84}
Kusuoka, S. and Stroock, D. (1984) Applications of the {M}alliavin calculus,
  {p}art {I}.
\newblock In \textit{North-Holland Mathematical Library}, vol.~32, 271--306.
  Elsevier.

\bibitem[{Li(2013)}]{li:13}
Li, C. (2013) Maximum-likelihood estimation for diffusion processes via
  closed-form density expansions.
\newblock \textit{Ann. Statist.}, \textbf{41}, 1350--1380.

\bibitem[{Melnykova(2020)}]{melnykova2020parametric}
Melnykova, A. (2020) Parametric inference for hypoelliptic ergodic diffusions
  with full observations.
\newblock \textit{Stat. Inference Stoch. Process.}, \textbf{23}, 595--635.

\bibitem[{Milstein and Tretyakov(2021)}]{mil:21}
Milstein, G.~N. and Tretyakov, M.~V. (2021) \textit{Stochastic {N}umerics for
  {M}athematical {P}hysics}.
\newblock Springer Cham.

\bibitem[{Nualart(2006)}]{nual:06}
Nualart, D. (2006) \textit{The Malliavin {C}alculus and {R}elated {T}opics}.
\newblock Springer Berlin, Heidelberg.

\bibitem[{Papaspiliopoulos et~al.(2007)Papaspiliopoulos, Roberts and
  Sk{\"o}ld}]{papa:07}
Papaspiliopoulos, O., Roberts, G.~O. and Sk{\"o}ld, M. (2007) A general
  framework for the parametrization of hierarchical models.
\newblock \textit{Stat. Sci.}, \textbf{22}, 59--73.

\bibitem[{Pavliotis(2014)}]{pavl:14}
Pavliotis, G.~A. (2014) \textit{Stochastic {P}rocesses and {A}pplications:
  {D}iffusion {P}rocesses, the {F}okker-{P}lanck and {L}angevin {E}quations},
  vol.~60.
\newblock Springer.

\bibitem[{Pokern et~al.(2009)Pokern, Stuart and Wiberg}]{poke:09}
Pokern, Y., Stuart, A.~M. and Wiberg, P. (2009) Parameter estimation for
  partially observed hypoelliptic diffusions.
\newblock \textit{J. R. Stat. Soc. Ser. B. Stat. Methodol.}, \textbf{71},
  49--73.

\bibitem[{Prakasa~Rao(1988)}]{prak:88}
Prakasa~Rao, B. L.~S. (1988) Statistical inference from sampled data for
  stochastic processes.
\newblock In \textit{Statistical inference from stochastic processes ({I}thaca,
  {NY}, 1987)}, vol.~80 of \textit{Contemp. Math.}, 249--284. Amer. Math. Soc.,
  Providence, RI.

\bibitem[{Stramer et~al.(2010)Stramer, Bognar and Schneider}]{stra:10}
Stramer, O., Bognar, M. and Schneider, P. (2010) Bayesian inference for
  discretely sampled {M}arkov processes with closed-form likelihood expansions.
\newblock \textit{J. Financ. Econom.}, \textbf{8}, 450--480.

\bibitem[{Uchida and Yoshida(2012)}]{uchi:12}
Uchida, M. and Yoshida, N. (2012) Adaptive estimation of an ergodic diffusion
  process based on sampled data.
\newblock \textit{Stoch. Process. Their Appl.}, \textbf{122}, 2885--2924.

\bibitem[{Vehtari et~al.(2021)Vehtari, Gelman, Simpson, Carpenter and
  B{\"u}rkner}]{veh:21}
Vehtari, A., Gelman, A., Simpson, D., Carpenter, B. and B{\"u}rkner, P.-C.
  (2021) Rank-normalization, folding, and localization: {A}n improved \^{R} for
  assessing convergence of {MCMC} (with discussion).
\newblock \textit{Bayesian Anal.}, \textbf{16}, 667--718.

\bibitem[{Watanabe(1987)}]{watanabe1987analysis}
Watanabe, S. (1987) Analysis of {W}iener functionals ({M}alliavin calculus) and
  its applications to heat kernels.
\newblock \textit{Ann. Probab.}, \textbf{15}, 1--39.

\end{thebibliography}

\end{document}